\documentclass[12pt]{article}
\usepackage{amsmath}
\usepackage{graphicx}
\usepackage{enumerate}
\usepackage{natbib}
\usepackage{url} 
\usepackage{comment}

\usepackage{graphicx}
\graphicspath{ {./figs/} }
\usepackage{subcaption}

\newcommand{\blind}{0}

\usepackage{amssymb,amsfonts,mathrsfs,amsmath,amsthm,mathtools,bm,dsfont, thmtools,bbm}\allowdisplaybreaks
\usepackage[dvipsnames]{xcolor}
\usepackage{array, graphicx, graphics,float,multirow,tabularx,tikz,pgfplots, longtable,threeparttable}
\usepackage{setspace}
\usepackage[colorlinks,citecolor=blue,urlcolor=blue, runcolor=blue, filecolor=cyan]{hyperref}
\usepackage{footnotebackref}
\usepackage{bibentry, natbib} 
\usepackage{paralist}  
\usepackage{appendix} 
\usepackage{mathrsfs}
\usepackage{enumitem}
\usepackage{authblk}
\usepackage{caption}
\usepackage{subcaption}
\usepackage{titlesec}
\usepackage{textcase,relsize}
\usepackage{datetime}
\usepackage{booktabs}

\declaretheoremstyle[notefont=\bfseries,notebraces={}{},%
    headpunct={},postheadspace=1em]{mystyle}
\declaretheorem[style=mystyle,numbered=no,name=Assumption]{asmp-hand}
\declaretheorem[style=mystyle,numbered=no,name=Condition]{cond-hand}
\declaretheorem[style=mystyle,numbered=no,name=Example]{exmp-hand}

\usepackage{titlesec}

\titleformat{\paragraph}
{\normalfont\normalsize\bfseries}{\theparagraph}{1em}{}
\titlespacing*{\paragraph}
{0pt}{3.25ex plus 1ex minus .2ex}{1.5ex plus .2ex}

			\def \bfA {\mathbf{A}}
			\def \bfB {\mathbf{B}}
			\def \bfC {\mathbf{C}}
			
\def \bbE {\mathbb{E}}	\def \calE {\mathcal{E}}		\def \bfE {\mathbf{E}}
			\def \bfF {\mathbf{F}}
	\def \calG {\mathcal{G}}		
			\def \bfH {\mathbf{H}}
	\def \calI {\mathcal{I}}		\def \bfI {\mathbf{I}}
	\def \calJ {\mathcal{J}}		
	\def \calK {\mathcal{K}}		
	\def \calL {\mathcal{L}}		\def \bfL {\mathbf{L}}
			\def \bfM {\mathbf{M}}
	\def \calN {\mathcal{N}}		
	\def \calO {\mathcal{O}}		\def \bfO {\mathbf{O}}
\def \bbP {\mathbb{P}}	\def \calP {\mathcal{P}}		
			\def \bfQ {\mathbf{Q}}
\def \bbR {\mathbb{R}}			\def \bfR {\mathbf{R}}
	\def \calS {\mathcal{S}}		\def \bfS {\mathbf{S}}
	\def \calT {\mathcal{T}}		
			\def \bfU {\mathbf{U}}
			\def \bfV {\mathbf{V}}
			\def \bfW {\mathbf{W}}
			\def \bfX {\mathbf{X}}
			\def \bfY {\mathbf{Y}}
\def \bbZ {\mathbb{Z}}	\def \calZ {\mathcal{Z}}		\def \bfZ {\mathbf{Z}}

			\def \bfc {\mathbf{c}}
			
			\def \bfe {\mathbf{e}}
			\def \bff {\mathbf{f}}

			\def \bfs {\mathbf{s}}

			\def \bfx {\mathbf{x}}
			\def \bfy {\mathbf{y}}

\def \bbOne{\mathbbm{1}}

\def \Mstar {M^{\star}}

\def \Var {\mathrm{Var}}

\def \bfS {\bm{S}}
\def \bfY {\mathbf{Y}}

\def\calE{{\cal  E}} 
 
\def\calG{{\cal  G}} 
 
\def\calI{{\cal  I}} 
\def\calJ{{\cal  J}} 
\def\calK{{\cal  K}} 
\def\calL{{\cal  L}} 
 
\def\calN{{\cal  N}} 
\def\calO{{\cal  O}} 
\def\calP{{\cal  P}}

\def\calS{{\cal  S}} 
\def\calT{{\cal  T}}

\def\calZ{{\cal  Z}} 

\def \vect {\mathrm{vec}}

\def \rank {\mathrm{rank}}

\def \diag {\mathrm{diag}}

\def \inj {\mathrm{inj}}

\def \full {\mathrm{full}}

\def \proj {\mathrm{proj}}

\def \op {\mathrm{op}}

\newcommand{\norm}[1]{\left\Vert#1\right\Vert}

\DeclareMathOperator*{\argmin}{arg\,min}

\numberwithin{equation}{section}

\theoremstyle{definition}
\newtheorem{remark}{Remark}
\newtheorem{example}{Example} [section]

\theoremstyle{plain}
\newtheorem{theorem}{Theorem}[section]
\newtheorem{proposition}[theorem]{Proposition}
\newtheorem{lemma}[theorem]{Lemma}
\newtheorem{corollary}[theorem]{Corollary}
\newtheorem{claim}{Claim}
\newtheorem{condition}{Condition}
\newtheorem{assumption}{Assumption}[section]
\newtheorem*{assumption*}{Assumption}

\def \bbE {\mathbb{E}}	\def \calE {\mathcal{E}}		\def \bfE {\mathbf{E}}
\def \calP {\mathcal{P}}

\allowdisplaybreaks

\usepackage{algorithm,algpseudocode}


   
\date{\today}
\newdateformat{mydate}{\monthname[\THEMONTH] \THEYEAR}

\begin{document}

\def\spacingset#1{\renewcommand{\baselinestretch}%
{#1}\small\normalsize} \spacingset{1}
  
\if0\blind
{ 
\title{\bf Uncertainty Quantification for Ranking with Heterogeneous Preferences}
    \author[1]{Jianqing Fan}
  \author[2]{Hyukjun Kwon}
  \author[3]{Xiaonan Zhu}
  \affil[1,2,3]{Department of Operations Research and Financial Engineering, Princeton University}
  \maketitle
} \fi

\if1\blind
{   
 \bigskip
 \bigskip
 \bigskip
\begin{center}
\spacingset{1.3}
    {\LARGE\bf Uncertainty Quantification for Ranking with Heterogeneous Preferences}
\end{center}
   \medskip
} \fi

\bigskip
\begin{abstract}
 This paper studies human preference learning based on partially revealed choice behavior and formulates the problem as a generalized Bradley--Terry--Luce (BTL) ranking model that accounts for heterogeneous preferences. Specifically, we assume that each user is associated with a nonparametric preference function, and each item is characterized by a low-dimensional latent feature vector --- their interaction defines the underlying  low-rank score matrix. 
 In this formulation, we propose an indirect regularization method for collaboratively learning the score matrix, which ensures entrywise $\ell_\infty$-norm error control --- a novel contribution to the heterogeneous preference learning literature. This technique is based on sieve approximation and can be extended to a broader class of binary choice models where a smooth link function is adopted. In addition, by applying a single step of the Newton--Raphson method, we debias the regularized estimator and establish uncertainty quantification for item scores and rankings of items, both for the aggregated and individual preferences. Extensive simulation results from synthetic and real datasets corroborate our theoretical findings.
\end{abstract}

\noindent%
\vfill

\newpage
\spacingset{1.9} 


\section{Introduction}\label{sec:intro}

Learning human preferences from revealed choice behavior is a profound and fundamental goal in a wide range of theoretical and practical disciplines. Examples include the Bradley--Terry--Luce (BTL) model \citep{bradley1952rank,luce1959individual,mcfadden1972conditional}, web search \citep{dwork2001rank,wang2016learning}, assortment optimization \citep{aouad2018approximability, chen2020dynamic,lee2024lowassortment}, recommendation systems \citep{baltrunas2010group,li2023estimating}, education \citep{avery2013revealed}, voting \citep{plackett1975analysis,mattei2013preflib}, and instruction tuning used in recent large language models \citep{ouyang2022training,lee2024low}. However, in contrast to the homogeneous preference assumption typically adopted in the literature, users in real-world applications exhibit diverse preferences. This paper studies collaborative learning of heterogeneous, user-specific preferences based on user choices revealed for only a small fraction of items.

We formulate the heterogeneous preference learning problem as a generalized BTL ranking model, where users are allowed to have distinct item score vectors \citep[see][]{park2015preference,katz2018nonparametric,negahban2018learning,li2020neural}. Specifically, we assume that each user $i$ is associated with an unknown, nonparametric preference function $g_i(\cdot)$, each item $j$ is characterized by a low-dimensional latent feature vector ${\bm\eta}_j$, and the matrix $\big[g_i({\bm\eta}_j)\big]_{i,j}$ defines the underlying preference scores for user $i$ across item $j$. For each user $i$, the item pairs selected for comparisons follow an Erd\H{o}s--R\'enyi random graph with probability $p_i>0$. Once an item pair $(j, j')$ is sampled, user $i$ chooses $j$ over $j'$ with probability $\exp(g_i({\bm\eta}_j))/(\exp(g_i({\bm\eta}_j))+\exp(g_i({\bm\eta}_{j'})))$. Collecting these item-comparison samples, we aim to estimate the underlying preference scores.

We focus on a scenario where the nonparametric preference functions $\{g_i(\cdot)\}_{i \in [d_1]}$ satisfy certain regularity conditions, such as sufficient smoothness, which control the complexity of the score matrix and make regularization methods effective for collaborative learning. Specifically, we adopt the \textit{nuclear norm regularization}, which has been widely used and studied in low-rank models; see, for example, \cite{beck:2009}, \cite{candes2009exact}, \cite{cai:2010}, \cite{candes2010matrix}, \cite{mazumder:2010},  \cite{koltchinskii:2011}, \cite{ma:2011}, \cite{negahban2012restricted}, \cite{parikh2014proximal}, \cite{chen2020noisy}, and \cite{fan2025covariates}.

This paper makes a novel contribution to the heterogeneous preference learning literature by establishing \textit{entrywise error control} and \textit{uncertainty quantification} under the generalized BTL model.  While comparable results have been achieved in ranking literature under the assumption of homogeneous preferences \citep[e.g., see][]{chen2015spectral,chen2019spectral,fan2023spectral,fan2024covariate,fan2024uncertainty,fan2024ranking}, such results have not been established in the heterogeneous setting due to its inherent complexity. Moreover, as we collaboratively learn the preference score matrix, the required sample size for each user is substantially smaller than that in the ranking literature.

Our key strategy is the use of a \textit{reparameterized nonconvex surrogate} of the original convex regularized problem --- a powerful technique that has been widely adopted in matrix completion models \citep[e.g.,][]{chen:2019inference,chen2020noisy}, and causal inference models \citep[e.g.,][]{choi2023norank,choi2024inference,choi2024matrix}. This technique establishes an entrywise error bound by i) analyzing the error of nonconvex gradient descent iterates using the leave-one-out technique and ii) showing the closeness of the nonconvex gradient iterate and the solution of the convex regularization problem. 

However, as noted in \cite{fan2025covariates}, the maximum likelihood estimation (MLE) method, which is commonly adopted in BTL models \citep[e.g., see][]{park2015preference,negahban2018learning}, introduces significant technical challenges in the nonconvex surrogate strategy. Specifically, the closeness of the two solution concepts of the regularized MLE problem and its reparameterized nonconvex problem is not guaranteed. To address this issue, we propose a novel regularization approach, which we term \textit{indirect regularization.} 

Usually, when a parameter of interest is assumed to satisfy a certain structural assumption, such as sparsity or low-rankness, a penalty is imposed directly on the parameter to encourage the resulting estimator to conform to the assumed structure, which is \textit{direct regularization.} In contrast, although we impose the regularity conditions on the score matrix, leading to approximate low-rankness, we do not directly regularize the score matrix. In our model, the score matrix defines an item comparison probability matrix through a smooth function. We ``indirectly'' regularize this induced probability matrix, instead of the original preference score matrix. Specifically, we propose an \textit{indirectly regularized least squares formulation} and its reparameterized nonconvex surrogate, in lieu of the directly regularized MLE. We show that this indirect regularization effectively controls the structure of the original score matrix estimator, similar to the direct regularization. Moreover, unlike the direct regularization, this approach ensures the desired transversality between the original regularized problem and its reparameterized counterpart, leading to entrywise error control.

Our strategy is motivated by an observation that the (approximate) low-rank structure of the score matrix is well-transferred to that of the probability matrix as long as a smooth link function is adopted. We formalize this observation in light of sieve approximation theory \citep[see][]{chen2007large,fan2016projected,chernozhukov2023inference,choi2024inference}. Our approach offers a novel perspective for a broad class of high-dimensional BTL models where a smooth link function is assumed. In these models, MLE-based methods have been central, as the likelihood function can be easily specified using a known link function, and theoretical error bounds are guaranteed to some extent --- for example, empirical loss bounds in \cite{park2015preference} and Euclidean error bounds in \cite{negahban2018learning}. In this context, this paper introduces the regularized least square formulation as a compelling alternative, achieving an entrywise error bound --- another important and useful error bound --- and deriving uncertainty quantification results based on it. 

Uncertainty quantification for item preference scores and rankings is another key contribution of this paper. For aggregated preferences, we remove the regularization bias by applying a single step of the Newton--Raphson method to the log-likelihood, leveraging entrywise error bounds and local strong convexity. Remarkably, we obtain comparable results even for individual preferences --- highlighting a unique strength of our approach. To achieve this, following the Newton--Raphson debiasing, we further apply spectral projection to refine the estimator.

Our theory applies to a broad class of human preference learning problems where users have heterogeneous preferences, including the aforementioned applications. We highlight two particularly timely and impactful applications: (i) large-scale conference reviewing, and (ii) the ranking of large language models (LLMs).

In 2024,  NeurIPS, a machine learning conference, had 16,671 submissions, and in 2023, ICML received 6,538 submissions from 18,535 authors. These burden the referee system significantly and impact the quality of reviews with huge individual noise. Many recent works have analyzed the peer-review system and proposed various approaches \citep[see][]{su2021you, shah2022challenges, su2024analysis, yan2024isotonic, fernandes2025peer}. Our framework offers an alternative: by collecting random pairwise comparisons from reviewers, we estimate underlying preference scores, which can then be used for accept/reject decisions.

Our theory is also applicable to ranking LLMs, a task that has gained significant attention due to its rapid growth \citep[e.g.,][]{chatzi2024prediction, dhurandhar2024ranking, maia2024efficient, wang2024ranking}. Many works adopt the BTL model with (regularized) MLE. In contrast, our regularized least squares and debiasing approach offers entrywise error control and uncertainty quantification. Additionally, our framework extends to low-rank reward model training for reinforcement learning from human feedback  in LLMs \citep[e.g.,][]{ouyang2022training, lee2024low, park2024rlhf}.

 
The rest of the paper is structured as follows. Section~\ref{sec:model} introduces our model. Section~\ref{sec:discussion} presents the estimation procedure and its motivation. Section~\ref{sec:errorbounds} provides the convergence rates of our estimator. Section~\ref{sec:UQ} elaborates on our debiasing method and the resulting uncertainty quantification results. Section~\ref{sec:numericalstudies} reports numerical experiments using the reel-watching dataset from \textit{Kuaishou}, a Chinese short-video platform. All proofs are provided in the appendix.

Let $\norm{\cdot}$, $\norm{\cdot}_{\mathrm{F}}$, $\norm{\cdot}_*$, and $\norm{\cdot}_{\infty}$ denote the matrix operator norm, Frobenius norm, nuclear norm, and entrywise $\ell_{\infty}$ norm, respectively. Also, we use $\norm{\cdot}_{2,\infty}$ to denote the largest $\ell_2$ norm of all rows of a matrix. We write $\sigma_{\max}(\cdot)$ and $\sigma_{\min}(\cdot)$ to represent the largest and smallest (nonzero) singular values of a matrix, respectively, and $\sigma_{j}(\cdot)$ to represent the $j$th largest singular value of a matrix. For two sequences $a_{n}$ and $b_{n}$, we denote $a_{n} \ll b_{n}$ (or $b_{n} \gg a_{n}$) if $a_{n} =o(b_{n})$, $a_{n} \lesssim b_{n}$ (or $b_{n} \gtrsim a_{n}$) if $a_{n} =O(b_{n})$, and $a_{n} \asymp b_{n}$ if $a_{n} \lesssim b_{n}$ and $a_{n} \gtrsim b_{n}$. For each natural number $n$, we denote $[n]=\{1, \ldots, n\}$. Also, let $\bfe^N_j$, for $j=1, \ldots, N$, denote the standard basis vectors in $\bbR^N$. When there is no risk of confusion, we will drop the superscript $N$ and write $\bfe_j$. For a matrix $\bfX$, let $\bfX_{n, \cdot}$ and $\bfX_{\cdot, n}$ denote $n$th row and column of $\bfX$, respectively. 

 \section{Model}\label{sec:model}
Suppose that there are $d_1$ users and $d_2$ items. We assume that each item $j$ is associated with an $r$-dimensional latent feature vector ${\bm \eta}_j$ where $r \ll \min\{d_1,d_2\}$, and each user $i$ has a unknown, nonparametric preference function $g_i:\bbR^r \rightarrow \bbR$. Then, the preference score of user $i$ for item $j$ is defined as $g_i({\bm\eta}_j)$. We denote the preference score matrix by ${\bf\Theta}^\star\coloneqq[g_i({\bm\eta}_j)]_{i \in [d_1], j\in[d_2]}$, which governs the comparison behavior of each user.
Specifically, when user $i$ is presented with a pair of distinct items $(j, j')$, the probability that the user $i$ prefers item $j$ over $j'$, denoted as $j \succ_i j'$, is
\begin{align}
  \bbP\{j \succ_i j'\}  =\frac{\exp(\Theta^\star_{i,j}) }{\exp(\Theta^\star_{i,j})+\exp(\Theta^\star_{i,j'})}= \frac{1}{1+\exp(-(\Theta^\star_{i,j}-\Theta^\star_{i,j'}))} \coloneqq  \sigma(\Theta^\star_{i,j}-\Theta^\star_{i,j'}).\label{eq:BTLfunction}
\end{align}
The probability \eqref{eq:BTLfunction} follows from Luce's choice axiom \citep{luce1959individual} and is widely adopted 
in a wide range of discrete response models.

It is evident that the probability \eqref{eq:BTLfunction}, induced by ${\bf\Theta^{\star}}$, is invariant to row-wise shifts in ${\bf\Theta^{\star}}$. That is, any score matrix such that $\bf\Theta^\star+\bfc \bf1^\top$, with $\bfc \in \bbR^{d_1}$, leads to the equivalent probability $\sigma(\Theta^\star_{i,j}+c_i-(\Theta^\star_{i,j'}-c_i))=\sigma(\Theta^\star_{i,j}-\Theta^\star_{i,j'})$. To address the identification issue, we focus on the centered scores by imposing Assumption \ref{asp:identification}, namely, each individual's average preference score over all $d_2$ items is normalized at zero. We note that this assumption does not cause any loss of generality with respect to item rankings --- whether for an individual user or the aggregated ranking across multiple users. 
\begin{assumption}[Identification]\label{asp:identification}
   ${\bf\Theta}^\star \bf1=\bf0$, namely, row sum is zero.
\end{assumption}


For each user $i \in [d_1]$, we define a directed graph $\calG_i = ([d_2], E_i)$ where $[d_2]$ is the node set and $E_i \subseteq \{i\} \times [d_2] \times [d_2]$ denotes the edge set, indicating the pairs of items compared by user $i$. Without loss of generality, we consider only edges $(i,j,j')$ such that $j<j'$ throughout this paper.

\begin{assumption}[Random sampling and choices]\label{asp:randomness}
We assume the Erd{\H o}s--R{\' e}nyi random graph sampling for each user, with heterogeneous probabilities. That is, for each user $i \in [d_1]$, an ordered triplet $(i,j,j') \in \{i\}\times [d_2] \times [d_2]$ with $j<j'$ independently belongs to the edge set $E_i$ with probability $p_i>0$. Denoting $p_{\min}\coloneqq \min_i p_i$ and $p_{\max}\coloneqq \max_i p_i$, we assume $p_{\min}>c p_{\max}$ for some constant $c>0.$ Denote the mean as $\bar{p}\coloneqq (d_1)^{-1}\sum_{i\in[d_1]}p_i$. Also, when $(i,j, j')\in E_i$, we randomly observe $j \succ_i j'$ with probability specified in \eqref{eq:BTLfunction}.  
\end{assumption}

We begin by defining the $d_1 \times d_2(d_2-1)/2$ dimensional score ``gap'' matrix $\bfM^\star$. To do so, we first index all item pairs according to a pre-defined \textit{lexicographic} order:
 \begin{align*}
     \underbrace{(1,2)-(1,3)-\cdots-(1,d_2)}_{d_2-1}-\underbrace{(2,3)-(2,4)-\cdots-(2,d_2)}_{d_2-2}-\cdots-\underbrace{(d_2-1,d_2)}_{1},
\end{align*}  
and denote this order by $k=\calL(j,j')$, a function from the indices $k =1,\ldots, d_2(d_2-1)/2$ to the set of all ordered item pairs $(j,j')$ with $j<j'$. We then define $\Mstar_{i,k} \coloneqq \Theta^\star_{i,j}-\Theta^\star_{i,j'},$ for $k=\calL(j,j')$ and
the indicator of a sampled edge as $\delta_{i,k}\coloneqq \bbOne\{\text{$(i,\calL^{-1}(k)) \in \bigcup_{i=1}^{d_1}E_i$}\}$. When $\delta_{i,k}=1$ with $k=\calL(j,j')$, we define the comparison outcome $y_{i,k}\coloneqq \bbOne\{j \succ_i j'\}$ and define $S\coloneqq \{(i,k)|(i,\calL^{-1}(k))\in \bigcup_{i=1}^{d_1}E_i\}$.





\section{Estimation Strategy}\label{sec:discussion}

Without any structural constraints on ${\bf\Theta}^\star$, collaborative learning is generally infeasible. By imposing certain regularity conditions on preference functions $\{g_i(\cdot)\}_{i \in [d_1]}$, we can control the complexity of the parameter ${\bf\Theta}^\star$, making regularization methods, such as \textit{nuclear norm regularization}, effective.

We aim to establish an entrywise error bound for the score matrix ${\bf\Theta}^\star$ --- an important result that has not been established in the heterogeneous preference learning model \citep[e.g.,][]{park2015preference,katz2018nonparametric,negahban2018learning,li2020neural}. Unlike the Euclidean (or Frobenius norm) error bound, the entrywise error bound provides  uniform control on the  estimation error across individual entries --- for example, enabling more informative recommendations in online marketplaces. Additionally, it plays a crucial role in other theoretical results of this paper, uncertainty quantification for item score gaps and item rankings, which also represent new contributions to the preference learning literature.

Our strategy for the entrywise error control is utilizing a \textit{reparameterized nonconvex surrogate} of the convex, nuclear norm regularized problem, which has recently gained popularity in low-rank matrix completion models, causal inference models, and mixed membership models. This technique achieves an entrywise error bound by i) analyzing the error of nonconvex gradient descent iterates using the leave-one-out technique and ii) transferring this error to the solution of the original problem by establishing the closeness of the two solution concepts.

\subsection{Technical challenges in nonconvex surrogate for regularized MLE}\label{sec:tech_challenges}

In BTL models, the maximum likelihood estimation (MLE) methods are commonly adopted. However, as noted in \cite{fan2025covariates}, directly applying the nonconvex surrogate strategy to the regularized MLE introduces significant technical challenges. 

As mentioned, a key component of this strategy is establishing that the gap between the two solution concepts of the convex regularized problem and its corresponding nonconvex problem is exceedingly small. Remarkably, \cite{chen2020noisy} show that it suffices to bound this gap \textit{only within a low-dimensional subspace associated with the regularization}, as long as the gradient of the nonconvex problem is sufficiently small. This strategy significantly facilitates the proof, and has been adopted in the subsequent works \citep[e.g.,][]{chen:2019inference,choi2023norank,choi2024inference,choi2024matrix, wang2025robust, fan2025covariates}. However, this result does not hold in general when the MLE loss is employed.

To further illustrate this issue, consider the regularized MLE that minimizes:
\begin{align}
  -\underbrace{ \sum_{(i,k)\in S} p_i^{-1} \left(y_{i,k} M_{i,k} - \log \left(1+ \exp\left( M_{i,k}  \right) \right)\right)}_{\coloneqq \ell(\bfM)} + \lambda \norm{\bfM}_* \label{eq:penalMLE}
\end{align}
where $S$ denotes the sample set introduced in Section \ref{sec:model} and $\lambda >0$ is a regularization parameter. The corresponding nonconvex function is then defined as: 
\begin{align}
   g(\bfA, \bfB)\coloneqq -\underbrace{ \sum_{(i,k)\in S} p_i^{-1} \left(y_{i,k} [\bfA \bfB^\top]_{i,k} - \log \left(1+ \exp\left( [\bfA \bfB^\top]_{i,k}  \right) \right)\right)}_{\coloneqq \ell(\bfA \bfB^\top)} + \frac{\lambda}{2} \norm{\bfA}_{\mathrm{F}}^2+ \frac{\lambda}{2} \norm{\bfB}_{\mathrm{F}}^2, \label{eq:nonconvexMLE}
\end{align}
where $\bfA \in \bbR^{d_1 \times q}$ and $\bfB \in \bbR^{d_2(d_2-1)/2 \times q}$ with $q \coloneqq \rank(\bfM^\star)$ are factorization of $\bfM^\star$. The reparameterization and the replacement of the regularization component are motivated by the following equation \citep{mazumder:2010}:
\begin{align*}
    \norm{\bfM}_*=\min_{\bfA \in \bbR^{d_1 \times q}, \bfB \in \bbR^{d_2(d_2-1)/2 \times q}, \bfA\bfB^\top=\bfM} \frac{1}{2}\left( \norm{\bfA}_{\mathrm{F}}^2+ \norm{\bfB}_{\mathrm{F}}^2\right).
\end{align*}

The crucial step in analysis is to show that the sufficiently small gradient $\norm{\nabla g(\bfA, \bfB)}_{\mathrm{F}}$ implies that the gradient of the loss component at $(\bfA, \bfB)$, i.e., $-\nabla \ell(\bfA \bfB^\top)$, has \textit{a bounded spectral deviation} from the subspaces of $\bfA$ and $\bfB$. That is, we need to show
\begin{align}
    \norm{P_{T^\perp} \left(  -\nabla \ell(\bfA \bfB^\top) \right)} &= \norm{P_{T^\perp} \left(  -\sum_{(i,k)\in S}p_i^{-1} \left( y_{i,k} -   \sigma([\bfA \bfB^\top]_{i,k}) \right) \bfe_i^{d_1} \bfe_k^{d_2(d_2-1)/2 \top}\right)}  < \lambda/2 \label{eq:spectral_deviation}
\end{align}
where $P_{T}(\cdot)$ and $P_{T^\perp}(\cdot)$ denote the projection operators onto the subspace spanned by $\bfA$ and $\bfB$, and its orthogonal space, respectively, and $\bfe_i^{d_1}$ and $\bfe_k^{d_2(d_2-1)/2}$ are $d_1$ and $d_2(d_2-1)/2$ dimensional standard basis vectors. When there is no risk of confusion, we will simply write $\bfe_i$ and $\bfe_k$ instead. Let $\widetilde{\bfM}$ and $(\widetilde{\bfA}, \widetilde{\bfB})$ denote the two solutions of \eqref{eq:penalMLE} and \eqref{eq:nonconvexMLE}.\footnote{To be precise, we use different solution concepts for the two problems. On one hand, $\widetilde{\bfM}$ denotes the minimizer of \eqref{eq:penalMLE}. On the other hand, for the nonconvex problem, we initiate  gradient descent iterations from the ground truth. Although the truth is unknown, this poses no issue, as we do not numerically compute the nonconvex solutions in estimation: they only serve as intermediate variables that facilitate the analysis of the convex estimator $\widetilde{\bfM}$. We run the iterations for polynomial times $t_0$ and select the iteration $t' <t_0$ where the gradient $\norm{\nabla g(\bfA^t, \bfB^t)}_{\mathrm{F}}$ is minimized over $t$. $(\widetilde{\bfA}, \widetilde{\bfB})$ is the iterates at this stopping point. For more details, we refer to \cite{chen:2019inference}, \cite{chen2020noisy}, \cite{choi2023norank}, \cite{choi2024matrix}, \cite{choi2024inference}, and \cite{fan2025covariates}.} Then, as long as \eqref{eq:spectral_deviation} holds at $(\widetilde{\bfA},\widetilde{\bfB})$, one can establish
$$
\norm{\calP_{T^\perp}(\widetilde{\bfM}-\widetilde{\bfA} \widetilde{\bfB}^\top)}_{\mathrm{F}} \leq \norm{\calP_{T}(\widetilde{\bfM}-\widetilde{\bfA} \widetilde{\bfB}^\top)}_{\mathrm{F}}
$$
and thus
$$\norm{\widetilde{\bfM}-\widetilde{\bfA} \widetilde{\bfB}^\top}_{\mathrm{F}} \leq \norm{\calP_{T^\perp}(\widetilde{\bfM}-\widetilde{\bfA} \widetilde{\bfB}^\top)}_{\mathrm{F}}+\norm{\calP_{T}(\widetilde{\bfM}-\widetilde{\bfA} \widetilde{\bfB}^\top)}_{\mathrm{F}} \leq  2\norm{\calP_{T}(\widetilde{\bfM}-\widetilde{\bfA} \widetilde{\bfB}^\top)}_{\mathrm{F}}.$$
As a result, instead of controlling the entire gap $\widetilde{\bfM}-\widetilde{\bfA} \widetilde{\bfB}^\top$, we can focus only on the gap within the low-rank subspaces, which is easier.

However, condition \eqref{eq:spectral_deviation} is difficult to satisfy in the MLE, mainly due to the nonlinearity of the link function (and its inverse) --- clearly, $\sigma(\bfA \bfB^\top)$ and $\bfA \bfB^\top$ span different subspaces, \textit{both in terms of direction and dimension}, and thus $\sigma(\bfA \bfB^\top)$ has a non-negligible deviation from the linear space of $\bfA$ and $\bfB$. 



\subsection{Indirect regularization as a solution}\label{subsec:leastsquare}

This paper proposes a novel solution to the aforementioned issue, which we call \textit{indirect regularization}. The problem in the regularized MLE can be summarized as the fact that $\sigma(\bfA \bfB^\top)$, which determines the spectrum of the gradient $-\nabla \ell(\bfA \bfB^\top)$ does not lie within the space associated with the regularization, i.e., the subspace spanned by $\bfA$ and $\bfB$. 

To address this issue, instead of \eqref{eq:penalMLE} and \eqref{eq:nonconvexMLE}, we consider minimizing the following:
\begin{align}
  \frac{1}{2} \sum_{(i,k)\in S} p_i^{-1}\left(y_{i,k} - L_{i,k} \right)^2 + \lambda \norm{\bfL}_*, \label{eq:penalLS}
\end{align}
where $L_{i,k}$ represents the probability $\bbP(y_{i,k}=1)$, rather than the score gap. Its nonconvex counterpart is defined as
\begin{align}
   f(\bfX, \bfY)\coloneqq \frac{1}{2}  \sum_{(i,k)\in S} p_i^{-1} \left(y_{i,k}- [\bfX \bfY^\top]_{i,k} \right)^2 + \frac{\lambda}{2} \norm{\bfX}_{\mathrm{F}}^2+ \frac{\lambda}{2} \norm{\bfY}_{\mathrm{F}}^2. \label{eq:nonconvexLS}
\end{align}
In this regularized least squares formulation, it is important to note that the regularization is applied not to $\bfM^\star$, but rather to the probability matrix $\sigma(\bfM^\star)$ which is induced by $\bfM^\star$ and $\bf\Theta^\star$ --- a motivation of the term ``indirect regularization.'' Moreover, the reparameterization in \eqref{eq:nonconvexLS} is also performed on this matrix. Importantly, the gradient of the square error loss is linear in $\bfX$ and $\bfY$, and thus we can easily show that this gradient lies within the subspaces of $\bfX$ and $\bfY$. Consequently, we can establish the closeness of the solutions of \eqref{eq:penalLS} and \eqref{eq:nonconvexLS}.


In a large class of BTL models --- among many, ranking models \citep{chen2019spectral,fan2023spectral,fan2024uncertainty,fan2024ranking}, preference learning models \citep{park2015preference,katz2018nonparametric,negahban2018learning,li2020neural}, mixed-membership models \citep{fan2025covariates}, and 1-bit matrix completion models \citep{chen2023statistical} --- the least squares formulation has been less popular than MLE approaches. Notably, this paper uncovers an overlooked advantage of the regularized least squares approach in this class of problems. It can serve as a cornerstone for achieving entrywise error control and uncertainty quantification through the nonconvex surrogate strategy, by enabling us to control the spectral deviation of the gradient.


However, another crucial --- and more challenging --- question remains. The nuclear norm regularization and its nonconvex surrogate are designed for low-rank estimation. Can this strategy be applied directly to $\sigma(\bfM^\star)$, rather than to $\bfM^\star$ or $\bf\Theta^\star$? Specifically, what is the rank of $\sigma(\bfM^\star)$?

\subsection{Low-rank approximation of $\sigma(\bfM^\star)$}\label{sec:low_rank_approximation}

The structure of $\sigma(\bfM^\star)$ depends on the properties of the nonparametric preference functions $\{g_i(\cdot)\}_{i \in [d_1]}$. Note that, for $k=\calL(j,j'),$ the score gap is $\Mstar_{i,k}= \Theta^\star_{i,j}-\Theta^\star_{i,j'}=g_i({\bm\eta}_{j})-g_i({\bm\eta}_{j'})$. Then, the probability \eqref{eq:BTLfunction} is
\begin{align}
  \bbP\{j \succ_i j'\}  & = \frac{1}{1+\exp(-(\Theta^\star_{i,j}-\Theta^\star_{i,j'}))}  = \sigma(g_i({\bm\eta}_{j})-g_i({\bm\eta}_{j'})) \coloneqq \widetilde{\sigma}_i(\widetilde{\bff}_k) \label{eq:tildesigmaf}
\end{align}
where $\widetilde{\bff}_k \coloneqq [{\bm\eta}_{j}^\top, {\bm\eta}_{j'}^\top]^\top$ is $2r$-dimensional, and $ \widetilde{\sigma}_i(\cdot)$ is such that $ \widetilde{\sigma}_i(\bfx,\bfy)=\sigma(g_i(\bfx)-g_i(\bfy))$ for $r$-dimensional vectors $\bfx$ and $\bfy.$  

Consider a sieve approximation of \eqref{eq:tildesigmaf}:
\begin{align*}
 \sigma(M^\star_{i,k}) =\widetilde{\sigma}_i(\widetilde{\bff}_k) = \sum_{l=1}^R \lambda_{i,l}\phi_l(\widetilde{\bff}_k)+\varepsilon_{i,k}={\bm\lambda}_i^\top {\bm\phi}_k+\varepsilon_{i,k}
\end{align*}
where ${\bm\lambda}_i \coloneqq (\lambda_{i,1}, \ldots, \lambda_{i,R})^\top \in \bbR^R$ is the sieve coefficient, ${\bm\phi}_k \coloneqq (\phi_1(\widetilde{\bff}_k), \ldots, \phi_R(\widetilde{\bff}_k))^\top \in \bbR^R$ is the basis functions, $R$ is the sieve dimension, and $\varepsilon_{i,k}$ is the sieve approximation error. Denote ${\bf\Lambda}$ as $d_1 \times R$ matrix that stacks ${\bm\lambda}_i^\top$, ${\bf\Phi}$ as $d_2(d_2-1)/2 \times R$ matrix that stacks ${\bm\phi}_k^\top$ and $\bm\calE$ as $d_1 \times d_2(d_2-1)/2$ matrix of $\varepsilon_{i,k}$. Then, in the matrix notation, we can write
\begin{align*}
      \sigma(\bfM^\star)  = \underbrace{{\bf\Lambda} {\bf\Phi}^\top}_{\coloneqq \bfL^\star}+ {\bm\calE}
\end{align*}
where $\rank(\bfL^\star)=R.$

Recent works in the low-rank matrix inference literature \citep{chernozhukov2023inference, choi2024inference} have studied approximate low-rankness, condition numbers, and incoherence properties of matrices of the form~\eqref{eq:tildesigmaf}, using sieve approximation techniques \citep[see][]{chen2007large, fan2016projected}. We adopt assumptions similar to those in this literature.

First, we consider sufficiently smooth functions $\{g_i(\cdot)\}_{i\in [d_1]}$. Then, it is reasonable to expect that the functions $\{\widetilde{\sigma}_i(\cdot)\}_{i\in [d_1]}$ will inherit this smoothness, due to the smoothness of $\sigma(\cdot).$ Specifically, we assume that $\{\widetilde{\sigma}_i(\cdot)\}_{i\in [d_1]}$ belong to a H\"older class: For some constant $C>0,$ and sufficiently large $a,b>0,$
\begin{align*}
     \left\lbrace h : \max_{b_1+ \cdots + b_{2r}=b}\left\vert \frac{\partial^{b}h(\bfx)}{\partial x_1^{b_1}\cdots\partial x_r^{b_{2r}}} - \frac{\partial^{b}h(\bfy)}{\partial y_1^{b_1}\cdots\partial y_r^{b_{2r}}} \right\vert \leq C \| \bfx - \bfy\|^{a}, \quad \text{for all $\bfx, \bfy$} \right  \rbrace .
\end{align*}
It is well-known that \citep[see][]{chen2007large,fan2016projected,chernozhukov2023inference,choi2024inference}, if the functions belong to this class and we use standard basis functions such as polynomials, trigonometric polynomials, and B-splines, the sieve approximation error is bounded by
	\[
		\max_{i,k}|\varepsilon_{i,k}| \lesssim R^{-s}, \ \ \ \ \\ \text{where } s \coloneqq  (a + b)/\dim(\widetilde{\bff}_k). 
	\]
As long as $\{\widetilde{\sigma}_i(\cdot)\}_{i\in [d_1]}$ is sufficiently smooth, the value of $s>0$ is large enough so that $\max_{i,k}|\varepsilon_{i,k}|$ remains sufficiently small, even when the sieve dimension $R$ grows slowly.

Second, we allow the condition number of $\bfL^\star$, denoted by $\kappa\coloneqq\sigma_{\max}(\bfL^\star)/\sigma_{\min}(\bfL^\star)$, to grow with $R$, provided that the growth is not too fast. This assumption is imposed to ensure the presence of ``spiked" singular values of $\bfL^\star$ \citep[see, e.g.,][]{abbe2020entrywise, chernozhukov2023inference}.

  
Lastly, we consider another important condition for low-rank matrix estimation --- namely, incoherence. Essentially, the incoherence of $\bfL^\star$ follows from the non-spikiness of $\bfL^\star$, which is a consequence of the boundedness of $\sigma(\cdot)$. As all entries of $\sigma(\bfM^\star)$ are bounded between 0 and 1, we will safely assume that $\max_{i,l}|\lambda_{i,l}|$ is bounded. The boundedness of $\max_{l,k}|\phi_l (\bff_k)|$ follows from the choice of basis functions. Define $\bfS_{\bf \Lambda}\coloneqq d_1^{-1}{\bf \Lambda}^\top{\bf \Lambda}$ and $\bfS_{\bf \Phi}\coloneqq (d_2(d_2-1)/2)^{-1}{\bf \Phi}^\top{\bf \Phi}$, and denote the singular value decomposition of $\bfL^\star$ as $\bfU^\star {\bf \Sigma}^\star \bfV^{\star\top}$ where ${\bf \Sigma}^\star =\diag (\sigma_1^\star, \ldots, \sigma_R^\star)$ with $ \sigma_1^\star\geq\cdots \geq\sigma_R^\star.$ Then, we can show
\begin{align}
    \norm{\bfU^\star}_{2, \infty} &\leq d_1^{-\frac{1}{2}} \norm{{\bf \Lambda}}_{2, \infty} \sigma_{\min}^{-\frac{1}{2}}(\bfS_{\bf\Lambda}) ; \label{eq:Uinco}\\
    \norm{\bfV^\star}_{2, \infty} &\leq (d_2(d_2-1)/2)^{-\frac{1}{2}}  \norm{{\bf \Phi}}_{2, \infty} \sigma_{\min}^{-\frac{1}{2}}(\bfS_{\bf\Phi}),\label{eq:Vinco}
\end{align}
by following, for example, the proof of Lemma 5.1 
in \cite{chernozhukov2023inference}. Denote the maximal value of the condition numbers of $\bfS_{\bf\Lambda}$ and $\bfS_{\bf\Phi}$ as $\mu$. Then, the singular vectors $\bfU^\star$ and $\bfV^\star$ will readily satisfy the incoherence condition with the incoherence parameter $\mu$. 

The following assumption formalizes the above discussion on the structure of $\sigma(\bfM^\star)$.\footnote{Although our assumption and those in \cite{chernozhukov2023inference, choi2024inference} appear similar, their motivations are quite different. \cite{chernozhukov2023inference} and \cite{choi2024inference} assume an approximately low-rank structure and its sieve representation --- rather than an exactly low-rank one --- as a modeling choice in their low-rank inference problems. In contrast, in this paper, the approximate low-rank structure arises as a consequence of adopting a smooth link function for binary outcomes, and we use the sieve representation to address the technical challenges in heterogeneous preference learning (Section \ref{sec:tech_challenges}), through the indirect regularization.}
\begin{assumption}\label{asp:sieve_assumptions}
    \begin{enumerate}
        \item[i)] $\max_{i,k} |\varepsilon_{i,k}| \lesssim R^{-s}  \ll  \sqrt{ \frac{1}{\max\{d_1,d_2(d_2-1)/2\}}}$. 
        \item[ii)] $\max_{i,l}|\lambda_{i,l}|<C$, $\max_{l,k}|\phi_l (\bff_k)|<C$, $\sigma_{\max}(\bfS_{\bf\Lambda})<C$, and $\sigma_{\max}(\bfS_{\bf\Phi})<C$ for some $C>0.$ Also, $\max\{ \sigma_{\max}(\bfS_{\bf\Lambda}) / \sigma_{\min}(\bfS_{\bf\Lambda}), \sigma_{\max}(\bfS_{\bf\Phi}) / \sigma_{\min}(\bfS_{\bf\Phi})\} < \mu$.
        \item[iii)] Denoting $\bar{d}\coloneqq d_1+d_2(d_2-1)/2$ and $\sigma^\star_{\min} \coloneqq\sigma_{\min}(\bfL^\star)$,
        \begin{align*}
         \frac{\bar{d}^2}{\bar{p} \min\{d_1, d_2(d_2-1)/2\} (\sigma_{\min}^\star)^2 } & \ll  \frac{1}{\kappa^4 \mu R \log (\bar{d})}, \\
  \frac{\kappa^2 \mu R}{\min\{d_1, d_2(d_2-1)/2\} }   \sqrt{\frac{\bar{d}\log^2(\bar{d})}{\bar{p}}} &\ll 1.
 \end{align*}
    \end{enumerate} 
\end{assumption}
For comparable assumptions in the low-rank inference literature, see Assumptions 4.6 and 5.1 in \cite{chernozhukov2023inference}, Assumptions 3.1 and 3.4 in \cite{choi2024inference}, as well as those in \cite{choi2024matrix}. Furthermore, we note that the incoherence property can be established by combining \eqref{eq:Uinco}, \eqref{eq:Vinco}, and Assumption \ref{asp:sieve_assumptions} ii).
\begin{lemma}\label{lem:incoherence}
    Under Assumption \ref{asp:sieve_assumptions} ii), we have
    \begin{align*}
         \norm{\bfU^\star}_{2, \infty} \lesssim  \sqrt{\frac{R \mu}{d_1}} \quad \text{and} \quad \norm{\bfV^\star}_{2, \infty} \lesssim  \sqrt{\frac{R \mu}{d_2(d_2-1)/2}}.
    \end{align*}
\end{lemma}

\subsection{Formal Estimation Procedure}\label{sec:est_procedure}

First, we conduct the nuclear norm regularized least squares using the observed item comparisons. This step estimates the probability $\sigma(\bfM^\star)$, controlling the complexity of the estimate. Second, we recover the score gap $\bfM^\star$ from the estimate in Step 1. Lastly, we take averages to estimate the original score matrix ${\bf\Theta}^\star$ from the estimated score gaps from Step 2. The formal procedure is as follows: 
\begin{enumerate}
    \item[{\bf Step 1:}] Solve the following problem:
\begin{align}
    \widehat{\bfL} \coloneqq \argmin_{\bfL \in \bbR^{d_1 \times d_2(d_2-1)/2}} \frac{1}{2} \sum_{(i,k)\in S} p_i^{-1}\left(y_{i,k} - L_{i,k} \right)^2 + \lambda \norm{\bfL}_* \label{eq:convexob}
\end{align}
where $\lambda>0$ is a regularization parameter.
\item[{\bf Step 2:}] Compute $\widehat{M}_{i,k}= \sigma^{-1}\left(\widehat{L}_{i,k} \right)$ for each $(i,k)$.
\item[{\bf Step 3:}] Take the averages $ \widehat{\bf\Theta}_{\cdot,j}=(d_2)^{-1} \left(\sum_{j'>j} \widehat{\bfM}_{\cdot, \calL(j,j')} -\sum_{j'<j} \widehat{\bfM}_{\cdot, \calL(j',j)} \right)$ for each $j\in[d_2].$
\end{enumerate}

\section{Euclidean and Entrywise Error Bounds}\label{sec:errorbounds}


\begin{theorem}\label{thm:errorboundforTheta}
 Suppose Assumption \ref{asp:identification}-\ref{asp:randomness} and \ref{asp:sieve_assumptions} hold. Suppose also that $\norm{{\bf\Theta}^{\star}}_{\infty}$ is bounded. Denote $\bar{d}\coloneqq d_1+d_2(d_2-1)/2$ and assume $\lambda=C_{\lambda} \sqrt{ \bar{d} / \bar{p}}$ for some large constant $C_{\lambda}>0$. Then, with probability at least $1-O(\bar{d}^{-10})$, we have
     \begin{align*}
      \frac{\|\widehat{\bf\Theta}-{\bf\Theta}^{\star}\|_{\mathrm{F}}}{\sqrt{d_1d_2}}   \lesssim    \sqrt{ \frac{\kappa^2 R \bar{d}}{\bar{p} d_1d_2^2  }}; \quad \quad     \norm{\widehat{\bf\Theta}-{\bf\Theta}^\star}_{\infty}  & \lesssim  \frac{\kappa^2 \mu R }{\min\{d_1, d_2(d_2-1)/2\}} \sqrt{\frac{  \bar{d}\log(\bar{d})}{\bar{p} }}.
\end{align*}
 \end{theorem}

This theorem shows that the estimation errors are evenly distributed across entries, in the sense that the entrywise error bound is of the similar order as the normalized Frobenius norm error bound. To illustrate this, suppose $d_1 \asymp d_2^2 \asymp \bar{d}$. In this case, the normalized Frobenius norm error bound  and entrywise error bound are given by $\sqrt{\tfrac{\kappa^2 R}{\bar{p}\bar{d}}}$ and $\kappa^2 \mu R \sqrt{\tfrac{\log(\bar{d})}{\bar{p}\bar{d}}},$
respectively, where the entrywise rate is slightly slower when $\kappa$, $\mu$ and $R$ is not too large.

It is also worth emphasizing that the required sample size --- i.e., the number of item comparisons --- in Theorem~\ref{thm:errorboundforTheta} can be significantly smaller than that typically assumed in the ranking literature. This is a consequence of the regularity conditions discussed in Section~\ref{sec:low_rank_approximation} --- these conditions essentially reduce the number of parameters to estimate from \( d_1 d_2(d_2 - 1)/2 \) to \( (d_1 + d_2(d_2 - 1)/2)R \). Specifically, assuming \( d_1 \asymp d_2^2 \) for simplicity, Assumption~\ref{asp:sieve_assumptions} iii) implies that the order of required sample size for user \( i \), i.e., \( p_i d_2(d_2 - 1)/2 \), depends only on the sieve dimension \( R \), the parameters \( \kappa \) and \( \mu \), and logarithmic factors --- rather than on the problem dimensions \( d_1 \) or \( d_2 \). In contrast, applying results from the ranking literature, for example, \citet{chen2019spectral}, to separately learn the preference of user $i$ requires a sample size of at least order $d_2$, ignoring model parameters and logarithmic factors.

In addition, the entrywise error bound in Theorem~\ref{thm:errorboundforTheta} leads to a corollary on top-$K$ item selection for both individual and aggregated preferences, a central topic in the ranking literature \citep[e.g.,][]{chen2019spectral}. This result is presented in the appendix (Section \ref{sec:topK}).

\section{Uncertainty Quantification}\label{sec:UQ}

\subsection{One-step Newton--Raphson-type debiasing}

Since our estimation procedure is based on a regularization method, it is evident that the resulting estimator inherently suffers from regularization bias, which complicates its distributional characterization. To address it, we employ a simple yet powerful debiasing scheme --- namely, the \textit{one-step Newton--Raphson debiasing} \citep[e.g., see][]{javanmard2013confidence,javanmard2014confidence}. To the best of our knowledge, this paper presents the first uncertainty quantification results in the heterogeneous preference learning literature \citep{park2015preference,katz2018nonparametric,negahban2018learning,li2020neural}.

Leveraging the facts that the entries of $\widehat\bfM$ are in a close neighborhood of the truth (see Corollary \ref{cor:errorboundforM}) and the log-likelihood function exhibits the local strong convexity, we can show that only a single iteration of Newton--Raphson update effectively absorbs the regularization bias. Recall the log-likelihood for the observed sample: 
\begin{align*}
    \ell (\bfM) = \sum_{(i,k)\in S} p_i^{-1} \left(y_{i,k} M_{i,k} - \log \left(1+ \exp\left( M_{i,k}  \right) \right)\right)
\end{align*}
for $\bfM \in \bbR^{d_1 \times d_2(d_2-1)/2}$. The gradient for entries $(i,k)\in S$ is then
\begin{align*}
    [\nabla_{\bfM}  \ell (\bfM)]_{i,k} =  p_i^{-1}\left( y_{i,k} - \frac{\exp\left( M_{i,k}  \right)}{1+ \exp\left( M_{i,k}  \right)} \right)  = p_i^{-1}\left( y_{i,k} - \sigma \left( M_{i,k}  \right)\right).
\end{align*}
Using the first and second gradients of the log-likelihood, we define the one-step Newton--Raphson debiased estimator as follows:
\begin{align}
    \widehat{M}^{\mathrm{NR}}_{i,k}  \coloneqq & \widehat{M}_{i,k} +p_i^{-1} \delta_{i,k} \left(\sigma'(\widehat{M}_{i,k}) \right)^{-1} \left( y_{i,k} - \sigma \left( \widehat{M}_{i,k}  \right) \right) \quad \text{for each $(i,k) \in [d_1] \times [d_2(d_2-1)/2]$}\label{eq:defofdebias}
\end{align}
where $\widehat{\bfM}$ is defined in Section \ref{sec:est_procedure} and $\sigma'(\cdot)$ is the derivative of the sigmoid function. 

\subsection{Inference for the aggregated preference}\label{sec:aggregated_inference}

While our theory can analyze individual preferences, the aggregated preference is always an important object for analysis. Suppose that a pair of items is given, and we are interested in testing whether a group of users prefers one item over the other. Typical examples of such groups include the entire set of users or a subset of users sharing attributes such as the same gender, age group, or region. We focus on the set of all users, as extensions to other cases are straightforward. 
\begin{theorem}\label{thm:asymptoticnormality}
    Suppose a pair of items $(j,j')$ with $j<j'$ is given. Let $k=\calL(j,j')$. Suppose the assumptions in Theorem \ref{thm:errorboundforTheta} hold, and assume further that
    \begin{align*}
        \left(\frac{\kappa \mu R }{\min\{d_1,d_2(d_2-1)/2\}}\right)^2 \frac{\bar{d} \log(\bar{d})}{\bar{p}  }  \ll \frac{1}{ \sqrt{\bar{p}d_1} }.
    \end{align*}
    Define
    \begin{align*}
      v^\star_k \coloneqq \frac{1}{d_1^2}  \sum_{i=1}^{d_1}  \frac{\delta_{i,k}}{p_i^2}\left(\frac{1}{\sigma(M^\star_{i,k})(1-\sigma(M^\star_{i,k}))}  \right).
    \end{align*}
    Then,
\begin{align*}
    &v_k^{\star -\frac{1}{2}}\left(\frac{1}{d_1} \sum_{i=1}^{d_1} \widehat{M}^{\mathrm{NR}}_{i,k}-\frac{1}{d_1} \sum_{i=1}^{d_1} M^\star_{i,k}\right)   \overset{\mathrm{d}}{\rightarrow} \calN \left( 0,1\right).
\end{align*}
\end{theorem}
 
We highlight that the variance $v^\star_k$ is in the form of $\bfc^\top \calI^{-1}\bfc$ where $\calI=-\bbE \nabla^2_\bfM \ell(\bfM^\star)$ is the Fisher information, and the vector $\bfc$ is such that $\bfc^\top \vect(\widehat{\bfM}^{\mathrm{d}}) = (d_1)^{-1}\sum_{i=1}^{d_1} \widehat{\bfM}^{\mathrm{d}}_{i,k}.$ Moreover, the variance $v^\star_k$ can be easily estimated by
\begin{align*}
    \widehat{v}_k \coloneqq \frac{1}{d_1^2}  \sum_{i=1}^{d_1} \frac{1}{p_i} \left(\frac{1}{\sigma(\widehat{M}_{i,k})(1-\sigma(\widehat{M}_{i,k}))}  \right),
\end{align*}
providing a feasible version of Theorem \ref{thm:asymptoticnormality}.   
\begin{proposition}\label{prop:feasibleCLT}
Under the assumptions in Theorem \ref{thm:asymptoticnormality}, we have $\widehat{v}_k  =v^\star_k + o(v^\star_k)$ with probability at least $1-O(\bar{d}^{-10})$. As a result,
\begin{align*}
    \widehat{v}_k^{-\frac{1}{2}}  \left(  \frac{1}{d_1} \sum_{i=1}^{d_1} \widehat{M}^{\mathrm{NR}}_{i,k}   - \frac{1}{d_1} \sum_{i=1}^{d_1}M^\star_{i,k}  \right)  \overset{\mathrm{d}}{\rightarrow} \calN (0, 1).
\end{align*}
\end{proposition}

\subsection{Inference for individual preferences}\label{sec:individualinference}

We are now interested in testing individual preference --- whether user $i$ prefers $j$ over $j'.$ This goal, which is typically more demanding than testing the ``averaged'' preferences, can be achieved by incorporating two additional steps: i) rank estimation and ii) sample splitting.

{\bf Rank estimation:} Recall that, for $k=\calL(j,j')$, we have $\Mstar_{i,k}=g_i({\bm\eta}_{j})-g_i({\bm\eta}_{j'})\coloneqq \widetilde{g}_i(\widetilde{\bff_k})$. By assuming a set of assumptions on $\{\widetilde{g}_i(\cdot)\}_{i\in[d_1]}$, similar to those in Section~\ref{sec:low_rank_approximation}, we can ensure an approximate low-rank structure and spiked singular values of $\bfM^\star$. This enables us to estimate the rank of the dominant low-rank component of $\bfM^\star$, denoted as $q$. The assumptions are elaborated in the appendix (Section~\ref{sec:LowrankM}), as they are similar to those in Section~\ref{sec:low_rank_approximation}. If the rank $q$ is consistently estimated using rank estimation methods \citep[e.g., see][]{bai2002determining,ahn2013eigenvalue, chernozhukov2023inference,choi2024inference} and the error bounds for $\bfM^\star$ (Corollary \ref{cor:errorboundforM}), we can estimate consistent singular spaces and project the debiased estimator $\widehat{\bfM}^{\mathrm{NR}}$ onto them to control the dimension of $\widehat{\bfM}^{\mathrm{NR}}$. For simplicity, this paper assumes that the rank $q$ is known, as in \cite{chernozhukov2023inference}.

{\bf Sample splitting:} Now, the distributional characterization of the projected estimator hinges on the analysis of the estimated singular subspaces. To this end, we will utilize the ``representation formula of spectral projectors'' \citep[see][]{xia2021normal,xia2021statistical}, which represents the singular subspace estimation error in terms of projected versions of the error $\widehat{\bfM}^{\mathrm{NR}}-\bfM^\star.$ A key technical requirement is that the debiased estimator must satisfy a certain error bound, i.e., $\norm{\widehat{\bfM}^{\mathrm{NR}}-\bfM^\star}\lesssim\sqrt{\bar{d}/\bar{p}}$ (up to log-terms and some parameter multiplications).  Sample splitting induces a sort of independence in the analysis, which allows us to invoke concentration inequalities to establish the bound. We emphasize that sample splitting is adopted purely for technical reasons, and our simulation results (see Section~\ref{sec:syntheticdataexperiment}) indicate that it does not lead to significant changes in estimation performance. Moreover, the use of sample splitting for debiasing regularized estimators is not uncommon in the low-rank inference literature \citep[e.g.,][]{xia2021statistical,chernozhukov2023inference}.

Recall that the set $S$ represents the entire sample. Without loss of generality, we assume that $|S|$ is even, and randomly split $S$ into two subsets, $S^1$ and $S^2$. Let $\delta^l_{i,k}$ denote the edge formation indicator $\delta_{i,k}$ where $(i,k) \in S^l$, for $l=1,2.$ Let $y^l_{i,k}$ denote the item comparison indicator $y^l_{i,k}=\bbOne\{j \succ_i j'\}$ where $(i,k)=(i,\calL(j,j'))\in S^l$ for $l=1,2.$ The formal inference procedure for individual preferences is as follows: 
\begin{enumerate}
    \item For the two subsamples, implement Steps~1 and 2 in Section~\ref{sec:est_procedure} to obtain \( \widehat{\bfM}^{1} \) and \( \widehat{\bfM}^{2} \), respectively. 
    \item Implement the debiasing step \eqref{eq:defofdebias} for each initial estimator in a cross-validated manner: 
    \begin{align*}
    \widehat{M}^{\mathrm{NR},1}_{i,k}  \coloneqq & \widehat{M}^1_{i,k} +\frac{2}{p_i} \delta^2_{i,k} \left(\sigma'(\widehat{\bfM}^1_{i,k}) \right)^{-1} \left( y^2_{i,k} - \sigma \left( \widehat{M}^1_{i,k}  \right) \right) \quad \text{for each $(i,k) \in [d_1] \times [d_2(d_2-1)/2]$};\\
    \widehat{M}^{\mathrm{NR},2}_{i,k}  \coloneqq & \widehat{M}^2_{i,k} +\frac{2}{p_i} \delta^1_{i,k} \left(\sigma'(\widehat{\bfM}^2_{i,k}) \right)^{-1} \left( y^1_{i,k} - \sigma \left( \widehat{M}^2_{i,k}  \right) \right) \quad \text{for each $(i,k) \in [d_1] \times [d_2(d_2-1)/2]$}.
    \end{align*}
    \item For each debiased estimator $\widehat{\bfM}^{\mathrm{NR},l}$ ($l=1,2$) compute its best rank-$q$ approximation 
    \begin{align*}
        \widehat{\bfM}^{\mathrm{proj},l} \coloneqq \argmin_{\rank(\bfM) \leq q}\norm{ \widehat{\bfM}^{\mathrm{NR},l}-\bfM}_{\mathrm{F}}^2 =\widehat{\bfU}^{\mathrm{NR},l} \widehat{\bfU}^{\mathrm{NR},l \top} \widehat{\bfM}^{\mathrm{NR},l} \widehat{\bfV}^{\mathrm{NR},l} \widehat{\bfV}^{\mathrm{NR},l \top}
    \end{align*}
    where $\widehat{\bfU}^{\mathrm{NR},l}$ and $\widehat{\bfV}^{\mathrm{NR},l}$ are the top-$q$ left and right singular vectors of $\widehat{\bfM}^{\mathrm{NR},l}.$ The final estimator is then 
    \begin{align*}
        \widehat{\bfM}^{\proj} \coloneqq \frac{1}{2}\widehat{\bfM}^{\mathrm{proj},1} + \frac{1}{2}\widehat{\bfM}^{\mathrm{proj},2}.
    \end{align*}
\end{enumerate}

\begin{theorem}\label{thm:asymptoticnormality_indiv}
    Suppose user $i$ and a pair of items $(j,j')$ with $j<j'$ are given. Let $k=\calL(j,j')$. Suppose the assumptions in Theorem \ref{thm:errorboundforTheta} hold. Additionally, suppose that technical assumptions \ref{asp:approx_low_rank_M}–\ref{asp:SSVforM} and \ref{asp:inference_indiv} are satisfied for the chosen rank $q$. Define
 \begin{align*}
   w^\star_{i,k} &\coloneqq  \underbrace{\bfV^{\bfM^\star }_{k,\cdot} \left(\sum_{k'=1}^{d_2(d_2-1)/2} \delta_{i,k'}\frac{1}{p_i^2\sigma(M^\star_{i,k'})(1-\sigma(M^\star_{i,k'}))}   \bfV^{\bfM^\star \top}_{k',\cdot}\bfV^{\bfM^\star}_{k',\cdot}\right)\bfV^{\bfM^\star \top}_{k,\cdot}}_{\coloneqq w^{\star,1}_{i,k}}\\
          &\quad + \underbrace{\bfU^{\bfM^\star}_{i,\cdot}\left(\sum_{i'=1}^{d_1} \delta_{i',k}\left(\frac{1}{p_{i'}^2\sigma(M^\star_{i',k})(1-\sigma(M^\star_{i',k}))} \right)\bfU^{\bfM^\star \top}_{i',\cdot}\bfU^{\bfM^\star}_{i',\cdot}\right)\bfU^{\bfM^\star \top}_{i,\cdot}}_{\coloneqq w^{\star,2}_{i,k}} 
      \end{align*}
where $\bfU^{\bfM^\star}$ and $\bfV^{\bfM^\star}$ are the top-$q$ left and right singular vectors of $\bfM^\star$, respectively. 
\begin{align*}
    w_{i,k}^{\star -\frac{1}{2}} \left( \widehat{M}^{\proj}_{i,k}-M^\star_{i,k}\right)   \overset{\mathrm{d}}{\rightarrow} \calN \left(0, 1 \right).
\end{align*}
\end{theorem}
We also highlight that, as in the inference for the aggregated preference, the variance $w^\star_{i,k}$ has the form of a linear combination of the inverse of the Fisher information, where the linear combination arises from the singular subspace projection. We estimate $w^\star_{i,k}$ by
\begin{align*}
    \widehat{w}_{i,k}  &\coloneqq \underbrace{\widehat{\bfV}_{k,\cdot}^{\mathrm{NR},1}\left(\sum_{k'=1}^{d_2(d_2-1)/2}  \frac{1}{p_i\sigma(\widehat{M}_{i,k'})(1-\sigma(\widehat{M}_{i,k'}))} \widehat{\bfV}_{k',\cdot}^{\mathrm{NR},1 \top}\widehat{\bfV}_{k',\cdot}^{\mathrm{NR},1} \right)\widehat{\bfV}^{\mathrm{NR},1 \top}_{k,\cdot}}_{\coloneqq \widehat{w}^1_{i,k}}\\
          &\quad + \underbrace{\widehat{\bfU}_{i,\cdot}^{\mathrm{NR},1}\left(\sum_{i'=1}^{d_1} \left(\frac{1}{p_i \sigma(\widehat{M}_{i',k})(1-\sigma(\widehat{M}_{i',k}))}  \right)\widehat{\bfU}^{\mathrm{NR},1 \top}_{i',\cdot}\widehat{\bfU}_{i',\cdot}^{\mathrm{NR},1}\right)\widehat{\bfU}^{\mathrm{NR},1 \top}_{i,\cdot}}_{\coloneqq \widehat{w}^2_{i,k}}.
\end{align*}
Alternatively, we may use $\widehat{\bfU}^{\mathrm{NR},2}$ and $\widehat{\bfV}^{\mathrm{NR},2}$ for the variance estimation. Based on these variance estimators, a feasible version of Theorem \ref{thm:asymptoticnormality_indiv} follows.
\begin{proposition}\label{prop:feasibleCLT_indiv}
Suppose the assumptions in Theorem \ref{thm:asymptoticnormality_indiv} hold. Suppose also that Assumption \ref{asp:forvarianceestimation} holds. Then, $\widehat{w}_{i,k}  =w^\star_{i,k} + o(w^\star_{i,k})$ with probability at least $1-O(\bar{d}^{-9})$. As a result,
\begin{align*}
    \widehat{w}_{i,k}^{-\frac{1}{2}}  \left(  \widehat{M}^{\proj}_{i,k}  -M^\star_{i,k} \right)  \overset{\mathrm{d}}{\rightarrow} \calN (0, 1).
\end{align*}
\end{proposition}

\subsection{Ranking inferences}\label{sec:ranking_inference}

In the ranking literature, where a homogeneous score ``vector'' for items is assumed, another important uncertainty quantification is about the ranking of items \citep[see, e.g.,][]{gao2023uncertainty, fan2023spectral, fan2024covariate, fan2024uncertainty, fan2024ranking}. In our setting, we are interested in extending this type of uncertainty quantification to individual preferences. Specifically, this section considers a two-sided confidence interval (CI) of item rankings for individual preferences. As the main idea for the aggregated preference is similar, we present it in the appendix (Section \ref{sec:market_ranking}). In addition, one-sided CIs and their applications \citep[see,][]{fan2024ranking} --- such as the \textit{top-$K$ placement test} (testing whether an item is ranked $K$th or higher) and \textit{sure screening of top-$K$ candidates} (constructing a confidence set that includes all top-$K$ items) --- are analyzed in the appendix (Section~\ref{sec:additionalrankinginferences}).


When analyzing the distribution of the score gap estimator for the individual preference (Section \ref{sec:individualinference}), we derive its linear expansion (see the proof of Theorem \ref{thm:asymptoticnormality_indiv}): for user $i$ and items $j \neq j'$,
\begin{align}
      \widehat{M}^{\proj}_{i,\bar{\calL}(j,j')}-M^\star_{i,\bar{\calL}(j,j')}\approx &\sum_{k'=1}^{d_2(d_2-1)/2} \underbrace{\frac{1}{p_i}\delta_{i,k'}\left( y_{i,k'}- \sigma(M^\star_{i,k'}) \right) \left(\sigma'(M^\star_{i,k'}) \right)^{-1}\bfV^{\bfM^\star}_{k',\cdot}\bfV^{\bfM^\star \top}_{\bar{\calL}(j,j'),\cdot}}_{\coloneqq \xi_{i,\bar{\calL}(j,j'); k'}}\nonumber\\
      &\quad + \sum_{i'=1}^{d_1} \underbrace{\frac{1}{p_{i'}}\bfU^{\bfM^\star}_{i,\cdot}\bfU^{\bfM^\star \top}_{i',\cdot} \delta_{i',\bar{\calL}(j,j')}\left( y_{i',\bar{\calL}(j,j')}- \sigma(M^\star_{i',\bar{\calL}(j,j')}) \right) \left(\sigma'(M^\star_{i',\bar{\calL}(j,j')}) \right)^{-1} }_{\coloneqq \nu_{i,\bar{\calL}(j,j'); i'}} \label{eq:Mprojlinearexpansion}
      \end{align}
where $\bar{\calL}(j, j')=\calL(j, j')$ if $j<j'$, and $\bar{\calL}(j, j')=\calL(j', j)$ if $j>j'$. We can leverage the Gaussian multiplier bootstrap method for this linear expansion in order to construct the simultaneous CIs \citep{chernozhuokov2022improved}. 
 
Suppose that two subsets $\calJ$ and $\calK$ such that $\calJ \subset \calK \subset [d_2]$ are given. Here, $\mathcal{J}$ denotes the set of items whose ranks are to be estimated among the items in $\mathcal{K}$. In the literature, we typically take $\mathcal{K} = [d_2]$ \citep{fan2024covariate,fan2024ranking}. Denote the ranking of item $j$ in $\calK$ for user $i$ based on the $i$th row of ${\bf\Theta}^\star$ by $r^\calK_{i,j}$. For simplicity, we assume that there is no tie. Suppose also that we have simultaneous score gap CIs $\{[C^i_L(j,j'), C^i_U(j,j')] \}_{j \in\calJ, j' \in \calK \setminus \{j\}}$ such that $\Theta^\star_{i,j}-\Theta^\star_{i,j'} \in [C^i_L(j,j'), C^i_U(j,j')], \, \forall j \in \calJ, \, \forall j'\in \calK \setminus\{j\}$ with probability at least $1-\alpha.$ Then, we will have the following relation \citep[see][]{ fan2024covariate, fan2024ranking}, which leads to the ranking CIs for all $j \in \calJ$:
\begin{align}
    1-\alpha & \leq \bbP \left( \Theta^\star_{i,j}-\Theta^\star_{i,j'}\in [C^i_L(j,j'), C^i_U(j,j')], \, \, \forall j \in \calJ, \, \forall j' \in \calK \setminus \{j\} \right) \nonumber \\
    & \leq \bbP \left( \bigcap_{j \in \calJ} \bigg\{1+ \underbrace{\sum_{j' \in \calK \setminus \{j\} }\bbOne \{C^i_U(j,j')<0\}}_{\text{``better than item $j$ in $\calK$ for $i$''}} \leq r^\calK_{i,j} \leq |\calK| - \underbrace{\sum_{j' \in \calK \setminus \{j\}}\bbOne \{C^i_L(j,j')>0\}}_{\text{``worse than item $j$ in $\calK$ for $i$''}}\bigg\} \right).\label{eq:rankinginterval_indiv}
\end{align}

In order to construct such simultaneous score gap CIs, we need to obtain the $(1-\alpha)$th quantile of the distribution of
\begin{align*}
    \widehat{\calT}_{\calJ, \calK}^i \coloneqq \max_{j \in \calJ}\max_{j' \in \calK \setminus \{j\}} \bigg| \frac{ 1 }{\sqrt{\widehat{w}_{i,\bar{\calL}(j,j')}}} \left(\widehat{M}^\proj_{i,\bar{\calL}(j,j')}-M^\star_{i,\bar{\calL}(j,j')} \right)\bigg|.
\end{align*} 
To analyze $\widehat{\calT}^i_{\calJ, \calK}$, motivated by the linear expansion \eqref{eq:Mprojlinearexpansion}, we define its population counterpart
\begin{align*}
    \calT^{i}_{\calJ, \calK} \coloneqq \max_{j \in \calJ}\max_{j' \in \calK \setminus \{j\}} \Bigg| \frac{ 1 }{\sqrt{w^\star_{i,\bar{\calL}(j,j')}}} \left(\sum_{k'=1}^{d_2(d_2-1)/2}  \xi_{i,\bar{\calL}(j,j'); k'} + \sum_{i'=1}^{d_1} \nu_{i,\bar{\calL}(j,j'); i'}\right)\Bigg|.
\end{align*} 
and its bootstrap version
\begin{align*}
    \calG^{i}_{\calJ, \calK} \coloneqq \max_{j \in \calJ}\max_{j' \in \calK \setminus \{j\}} \Bigg| \frac{ 1 }{\sqrt{w^\star_{i,\bar{\calL}(j,j')}}} \left(\sum_{k'=1}^{d_2(d_2-1)/2}  \xi_{i,\bar{\calL}(j,j'); k'}Z^{\xi}_{k'} + \sum_{i'=1}^{d_1} \nu_{i,\bar{\calL}(j,j'); i'}Z^{\nu}_{i'}\right)\Bigg|
\end{align*} 
where $\{Z^\xi_{k'}, Z^\nu_{i'}\}_{k'\geq 1,i'\geq 1}$ are i.i.d. standard normal random variables. Let $\calG^i_{\calJ, \calK;(1-\alpha)}$ be the $(1-\alpha)$th quantile of $\calG^i_{\calJ, \calK}.$ The feasible version for $\calG^i_{\calJ, \calK}$ is naturally defined as:
\begin{align*}
    \widehat{\calG}^i_{\calJ, \calK} \coloneqq \max_{j \in \calJ}\max_{j' \in \calK \setminus\{j\}} \Bigg| \frac{ 1 }{\sqrt{\widehat{w}_{i,\bar{\calL}(j,j')}}} \left(\sum_{k'=1}^{d_2(d_2-1)/2}  \widehat{\xi}_{i,\bar{\calL}(j,j'); k'}Z^{\xi}_{k'} + \sum_{i'=1}^{d_1} \widehat{\nu}_{i,\bar{\calL}(j,j'); i'}Z^{\nu}_{i'}\right)\Bigg|
\end{align*}
where 
\begin{align*}
    \widehat{\xi}_{i,\bar{\calL}(j,j'); k'} &=\begin{cases}
        \frac{1}{p_i}\delta^1_{i,k'}\left( y^1_{i,k'}- \sigma(\widehat{M}^{2}_{i,k'}) \right) \left(\sigma'(\widehat{M}^2_{i,k'}) \right)^{-1}\widehat{\bfV}^{\mathrm{NR},2}_{k',\cdot}\widehat{\bfV}^{\mathrm{NR},2 \top}_{\bar{\calL}(j,j'),\cdot} \quad \textit{if $(i,k')\in S^1$,}\\
        \frac{1}{p_i}\delta^2_{i,k'}\left( y^2_{i,k'}- \sigma(\widehat{M}^{1}_{i,k'}) \right) \left(\sigma'(\widehat{M}^1_{i,k'}) \right)^{-1}\widehat{\bfV}^{\mathrm{NR},1}_{k',\cdot}\widehat{\bfV}^{\mathrm{NR},1 \top}_{\bar{\calL}(j,j'),\cdot} \quad \textit{if $(i,k')\in S^2$,}
    \end{cases} \\
    \widehat{\nu}_{i,\bar{\calL}(j,j');i'} &=\begin{cases}
        \frac{1}{p_{i'}}\widehat{\bfU}^{\mathrm{NR},2}_{i,\cdot}\widehat{\bfU}^{\mathrm{NR},2 \top}_{i',\cdot} \delta^1_{i',\bar{\calL}(j,j')}\left( y^1_{i',\bar{\calL}(j,j')}- \sigma(\widehat{M}^2_{i',\bar{\calL}(j,j')}) \right) \left(\sigma'(\widehat{M}^2_{i',\bar{\calL}(j,j')}) \right)^{-1}\quad \textit{if $(i',\bar{\calL}(j,j'))\in S^1$,}\\
        \frac{1}{p_{i'}}\widehat{\bfU}^{\mathrm{NR},1}_{i,\cdot}\widehat{\bfU}^{\mathrm{NR},1 \top}_{i',\cdot} \delta^2_{i',\bar{\calL}(j,j')}\left( y^2_{i',\bar{\calL}(j,j')}- \sigma(\widehat{M}^1_{i',\bar{\calL}(j,j')}) \right) \left(\sigma'(\widehat{M}^1_{i',\bar{\calL}(j,j')}) \right)^{-1}\quad \textit{if $(i',\bar{\calL}(j,j'))\in S^2$}.
    \end{cases}
\end{align*}
The subsample sets $S^1$ and $S^2$ are as defined in Section \ref{sec:individualinference}. The following result validates the bootstrap approach for constructing simultaneous CIs. Let $\mu(\bfM^\star)$ denote the incoherence parameter of $\bfM^\star$ which is defined in the proof of Theorem \ref{thm:asymptoticnormality_indiv} (Section \ref{sec:proofofindivUQ}). 
\begin{theorem}\label{thm:individualbootstrap}
    Suppose that user $i\in [d_1]$ is given. Let $\widehat{\calG}^i_{\calJ, \calK;(1-\alpha)}$ be the $(1-\alpha)$th quantile of $\widehat{\calG}^i_{\calJ, \calK}.$ Suppose that the assumptions in Proposition \ref{prop:feasibleCLT_indiv} and Assumption \ref{asp:indivrankingbootstrap} hold. Particularly, for all $(i,k)$ such that $k = \bar{\calL}(j,j')$ where $j\in \calJ$ and $j' \in \calK \setminus \{j\}$, Assumption \ref{asp:inference_indiv} holds. For simplicity, assume $d_1 \asymp d_2(d_2-1)/2$. Then, we have
\begin{align*}
    |\bbP \left(\widehat{\calT}^i_{\calJ, \calK} > \widehat{\calG}^i_{\calJ, \calK;(1-\alpha)}\right) -\alpha|  \lesssim \left(\frac{(\mu(\bfM^\star))^4q^4\log^5(\bar{d}d_2)}{\bar{p}\bar{d}} \right)^{\frac{1}{4}}+o(1)
\end{align*}
with probability at least $1-O(\bar{d}^{-8})$.
\end{theorem}
We then define the CIs for the score gaps for all $j \in \calJ$ and $j' \in \calK \setminus \{j\}$ as follows:
\begin{align*}
    &C^i_L(j,j') = (-1)^{1-\bbOne \{j<j'\} }\widehat{M}^\proj_{i,\bar{\calL}(j,j')} -\widehat{\calG}^i_{\calJ, \calK;(1-\alpha)}\times \sqrt{\widehat{w}_{i,\bar{\calL}(j,j')}}; \\
    &C^i_U(j,j') = (-1)^{1-\bbOne \{j<j'\} }\widehat{M}^\proj_{i,\bar{\calL}(j,j')} +\widehat{\calG}^i_{\calJ, \calK;(1-\alpha)}\times \sqrt{\widehat{w}_{i,\bar{\calL}(j,j')}}.
\end{align*}
In addition, the $(1-\alpha) \times 100$\% simultaneous CIs for $\{r^\calK_{i,j}\}_{j \in \calJ}$ are defined as $\{[\widehat{r}_{i,j}^{U}, \widehat{r}_{i,j}^{L}]\}_{j \in \calJ}$ where
\begin{align*}
    \widehat{r}_{i,j}^U \coloneqq 1+  \sum_{j' \in \calK \setminus \{j\}}\bbOne \{C^i_U(j,j')<0\} \quad   \text{and} \quad \widehat{r}_{i,j}^L \coloneqq |\calK| -  \sum_{j' \in \calK \setminus \{j\}}\bbOne \{C^i_L(j,j')>0\} \quad \text{for each $j \in \calJ$}
\end{align*}
as shown in \eqref{eq:rankinginterval_indiv}.

\section{Numerical Studies}\label{sec:numericalstudies}

In this section, we examine our theoretical results using real data --- the reel-watching dataset from \textit{Kuaishou}, a Chinese short-video-sharing mobile application \citep{gao2022kuairec}.\footnote{In the appendix, Section \ref{sec:syntheticdataexperiment} reports additional numerical experiments using synthetic data to illustrate convergence rates and inferential theory.} The original \textit{Kuaishou} dataset contains $d_1 = 1,411$ users and $d_2 = 3,327$ items (videos), together with complete information on user-video interactions. As a measure of user preference, we use the \textit{watch ratio}, defined as the watched duration divided by the length of the video. The watch ratio exhibits substantial heterogeneity, reflecting diverse user preferences.   See the left panel of Figure~\ref{fig:watchratio}, which shows the distribution of the raw watch ratios, clipped at $20$. 

To improve the robustness of the estimation, we discretize the watch ratios into $14$ baskets: ${[0, 0.1], (0.1, 0.2], \ldots, (0.9, 1], (1, 1.2], (1.2, 1.5], (1.5, 2], (2, \infty)}$, marking each interval a score respectively  $0, \cdots, 13$. The right panel of Figure~\ref{fig:watchratio} displays the frequency of scores in each basket. For example, a score of $13$, corresponding to the watch ratio basket $(2, \infty)$, indicates that the user watched the video more than twice, suggesting a strong interest in the content.  Conversely, a low score such as 1, corresponding to the watch ratio basket $(0.1, 0.2]$ implies that the user lost interest and stopped watching shortly after starting. Based on this interpretation, for unobserved items, we can construct a recommendation system by leveraging the estimated rankings to suggest highly ranked videos tailored to each user. 

\begin{figure}[ht!]
\centering
\begin{subfigure}[b]{0.49\textwidth}
\centering
\includegraphics[width=\textwidth]{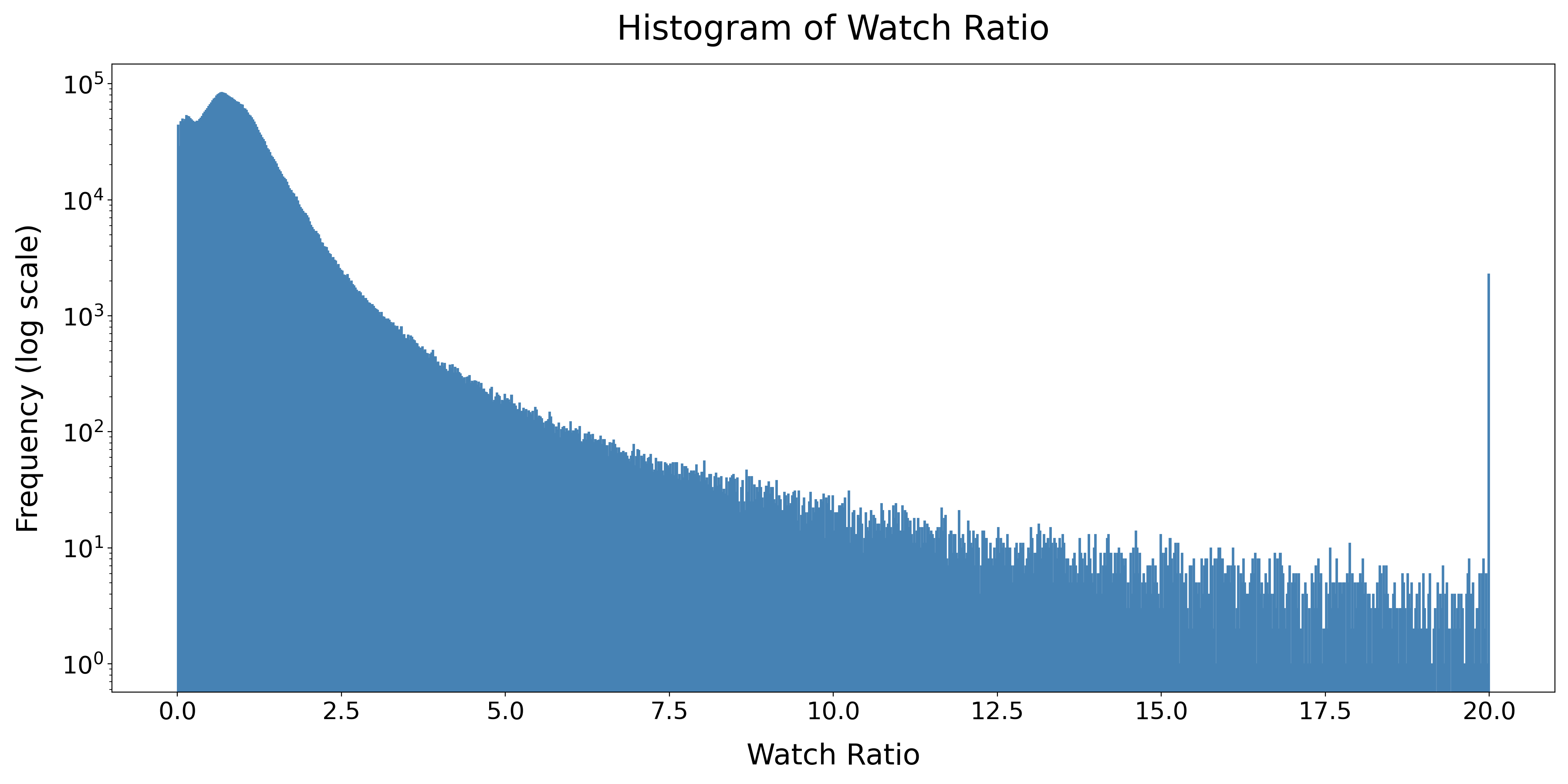}
\end{subfigure}
\hfill
\begin{subfigure}[b]{0.49\textwidth}
\centering
\includegraphics[width=\textwidth]{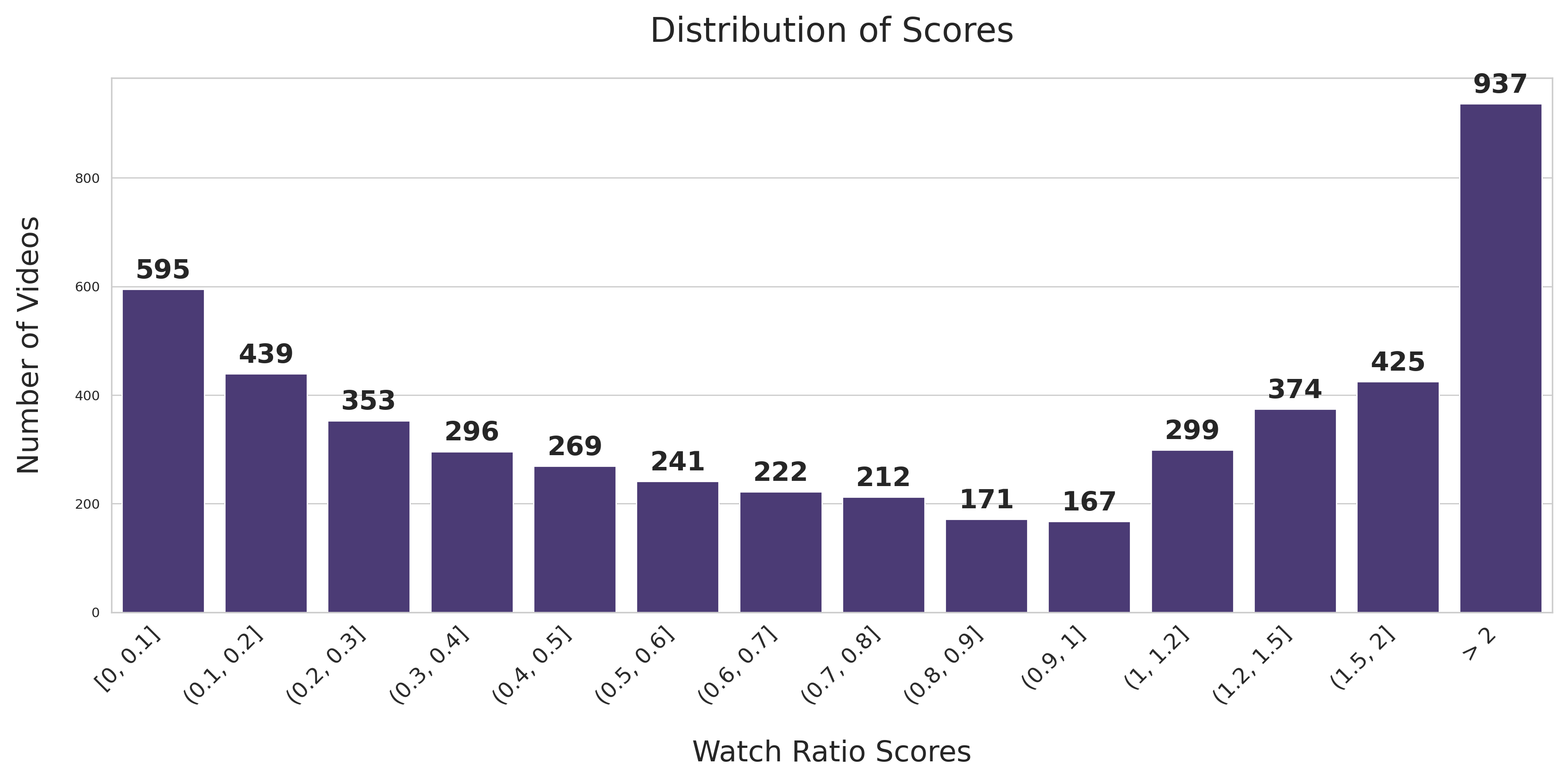}
\end{subfigure}
\caption{
Histogram of the watch ratios and their assigned baskets and scores.
}
\label{fig:watchratio}
\end{figure}

We randomly select $780$ users and $40$ videos from the dataset and construct a $780\times 40$ dimensional score matrix. Each row is then demeaned, serving as a ``quasi-true'' score matrix $\boldsymbol{\Theta}^\star$. With a fixed observation probability $p_i = p = 0.5$, we generate pairwise comparison graphs following Assumption \ref{asp:randomness}, namely, any two videos being compared by user $i$ with probability 0.5. Based on the observation comparisons, we follow the ranking inference procedures presented in Section \ref{sec:ranking_inference} and Section \ref{sec:market_ranking}. Specifically, we construct the confidence intervals for item rankings, at both aggregated and individual levels, setting $\alpha=0.05$ and $\calJ=\calK=[d_2].$ In addition, we use the same value of the regularization parameter $\lambda$ and the rank estimation method as in the synthetic data experiment (see Section \ref{sec:syntheticdataexperiment}). 

We also compute the true rank for each item based on the quasi-true score matrix ${\bf\Theta}^\star$ and its average. For the individual case, since the scores are discretized into only $14$ baskets, ties occur frequently. Therefore, instead of reporting a single rank value, we represent the true rank as an interval, whose length corresponds to the number of items having the same score. For example, User 11 gives the same rating of Video 17 as two other videos, resulting in a rank interval $[2, 4]$.

\begin{table}[!htp]
\centering
\small
\begin{tabular}{c@{\hskip 0.5cm}ccc@{\hskip 1cm}ccc}
\toprule
\multirow{2}{*}{\textbf{Item ID}} & \multicolumn{3}{c@{\hskip 1cm}}{\textbf{Aggregated preference}} & \multicolumn{3}{c}{\textbf{Individual preference (User 11)}} \\
\cmidrule(r){2-4} \cmidrule(l){5-7}
 & True Rank & CI & CI Length & True Rank & CI & CI Length \\
\midrule
23 & 1  & [1, 1]   & 0  & [1, 1]     & [1, 8]  & 7 \\
34  & 2  & [2, 5]   & 3  & [5, 9]    & [1, 10]   & 9 \\
17 & 3  & [2, 8]  & 6  & [2, 4]    & [1, 10]  & 9 \\
22 & 4  & [2, 8]  & 6 & [5, 9]     & [1, 8]  & 7 \\
33  & 5  & [2, 7]  & 5 & [13, 14]     & [2, 20]  & 18 \\
\\
6 & 36 & [36, 37] & 1  & [37, 39]   & [32, 37] & 5 \\
18 & 37 & [36, 38] & 2  & [33, 33]   & [34, 38] & 4 \\
29 & 38 & [37, 39] & 2  & [40, 40]   & [39, 40] & 1 \\
19 & 39 & [38, 40] & 2  & [37, 39]   & [35, 40] & 5 \\
21 & 40 & [39, 40] & 1  & [37, 39]   & [36, 40] & 4  \\
\bottomrule
\end{tabular}
\caption{Confidence intervals for the ranks of selected items using the aggregated preference and User 11's preference.}
\label{tab:kuai_ci}
\end{table}

 
We estimate the ranks of all videos based on a synthetic data generated according to the generalized BTL model.  The results for the ranking confidence intervals are presented in Table~\ref{tab:kuai_ci}. For the individual level, we randomly select User 11. The items are sorted by the true rank at the aggregated level, and we report the results only for the top 5 and bottom 5 items to save space. The confidence intervals for the aggregated-level ranking perform very well --- covering the true ranks while maintaining short lengths --- indicating that our method provides accurate and reliable rankings at the aggregated level. At the individual level, the estimated confidence intervals also cover the corresponding true ranks, although they tend to be slightly wider than those at the aggregate level due to larger variance. Nonetheless, the intervals remain reasonably informative.

 \section*{Acknowledgments} 
Fan's research is supported by the Office of Naval Research (ONR) grants N00014-22-1-2340 and N00014-25-1-2317, and by the National Science Foundation (NSF) grants DMS-2210833 and DMS-2412029.

 \onehalfspacing
\bibliographystyle{apalike}
\bibliography{reference}

\newpage

\appendix

{\LARGE \bf
\begin{center}
    APPENDIX
\end{center}
}

This supplementary appendix contains additional materials omitted from the main text, including synthetic data experiments (\ref{sec:syntheticdataexperiment}), additional theoretical assumptions and results (\ref{sec:topK}--\ref{sec:additionalrankinginferences}), and technical proofs (\ref{sec:proofs}). 

\section{Synthetic Data Experiments}\label{sec:syntheticdataexperiment}
We generate $\bf\Theta^\star$ using the following data generating process: for $(i,j)\in [d_1]\times [d_2]$,
	\begin{align}
        \Theta^\star_{i,j} =g_i({\bm\eta}_j)={\bm\alpha}_i'{\bm\beta}_j+ h_i\left( \zeta_j  \right), \quad \text{ where }\ \ h_i(\zeta_j)\coloneqq \sum_{m=1}^{\infty} \frac{|W_{i,m}|}{m^2}\cdot \sin(m\zeta_j). \label{eq:simul_DGP}
	\end{align}
Here, $W_{i,m}$ and $\zeta_j$ are independently generated from standard normal distribution, and ${\bm\alpha}_i'=[1, a_i]$ and ${\bm\beta}_j'=[b^1_j, b^2_j]$, where $a_i$, $b^1_j$, and $b^2_j$ are independently drawn from $\mathrm{Uniform}[0,1]$. Once $\bf\Theta^\star$ is generated, we demean and rescale each row to ensure $\bf\Theta^\star {\bf1}={\bf0}$ and $\norm{\bf\Theta^\star}_{\infty}=1.5$, so that each row is centered and normalized. Throughout the experiments, we set $\lambda=\sqrt{0.5\bar{d}/\bar{p}}$.

{\bf Convergence rates.} We first examine the convergence rates for our score matrix estimator in the Frobenius norm and the entrywise max norm (see Theorem \ref{thm:errorboundforTheta}). We also report the estimation error for the aggregated preference. The number of items is chosen as $d_2 \in \{20, 30, \ldots, 70, 80\}$, and the number of users is set to $d_1 = d_2(d_2 - 1)/2$. In the first experiment, we set $p_i = 0.8$ for one third of the users and $p_i = 0.4$ for the remaining users. In the second experiment, we instead use the lower probabilities $0.4$ and $0.2$ in place of $0.8$ and $0.4$, respectively. For each design, we repeat the procedure 20 times, regenerating $\mathbf{\Theta}^\star$ in each replication. Figure \ref{fig:error_bound} reports the experimental results, showing that all errors converge to zero as the dimension and the number of evaluators increase.

\begin{figure}[htbp]
    \centering
    \includegraphics[width=\textwidth, height=6cm]{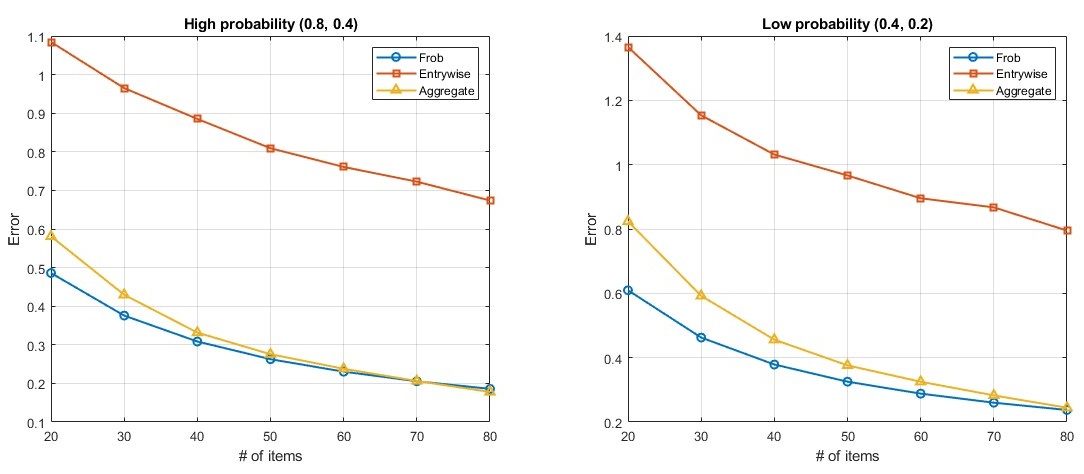}
    \caption{The left panel corresponds to the first experiment with higher probabilities 
($0.8$ for one third of the users and $0.4$ for the remaining users), while the right panel corresponds to the lower probabilities 
(0.4 and 0.2 in lieu of 0.8 and 0.4, respectively). ``Frob'' denotes the simulation average of the estimation error measured as 
$\lVert \widehat{\mathbf{\Theta}} - \mathbf{\Theta}^{\star} \rVert_{\mathrm{F}} / \sqrt{d_1 d_2}$, ``Entrywise'' denotes the simulation average of $\lVert \widehat{\mathbf{\Theta}} - \mathbf{\Theta}^\star \rVert_{\infty} / \lVert \mathbf{\Theta}^\star \rVert_{\infty}$, and ``Aggregate'' denotes the simulation average of
$\big\lVert \sum_{i=1}^{d_1} 
  (\widehat{\mathbf{\Theta}}_{i,\cdot} - \mathbf{\Theta}^{\star}_{i,\cdot}) 
\big\rVert \big/ 
\big\lVert \sum_{i=1}^{d_1} \mathbf{\Theta}^{\star}_{i,\cdot} 
\big\rVert$.}
    \label{fig:error_bound}
\end{figure}

{\bf Asymptotic normality.} 
We now validate the asymptotic normality results for both aggregated and individual score gaps (see Propositions \ref{prop:feasibleCLT} and \ref{prop:feasibleCLT_indiv}). We continue to use the same DGP as before. In both cases, we estimate the score gap between Item 1 and Item 2. Also, for the individual case, we select User 1. These choices do not restrict generality since $\mathbf{\Theta}^\star$ is randomly regenerated each time. For the aggregated score gap, we assume a homogeneous probability $p_i=p$ for each $i$, with $p \in \{0.4, 0.8\}$. For the individual score gap case, we use $p \in \{0.6, 0.8\}$. The number of items is chosen as $d_2 \in \{20, 40\}$, and the number of users is set to $d_1 = d_2(d_2-1)/2$, as before. For each setting, we conduct 500 iterations, regenerating $\mathbf{\Theta}^\star$ each time. For the individual case, we employ a simple rank estimator: the rank $q$ is defined as the number of singular values of $\widehat{\mathbf{M}}$ that exceed $10\%$ of the largest singular value of $\widehat{\mathbf{M}}$. Figure \ref{fig:asymp_normal_market} and \ref{fig:asymp_normal_indiv} present the results, with the histograms closely approximating the standard normal distribution.

\begin{figure}[htbp]
    \centering
    \begin{subfigure}{0.45\textwidth}
        \centering
        \includegraphics[width=\linewidth, height=4cm]{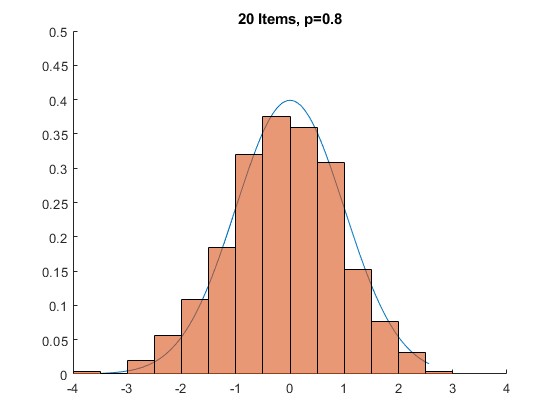}
    \end{subfigure}
     \begin{subfigure}{0.45\textwidth}
        \centering
        \includegraphics[width=\linewidth, height=4cm]{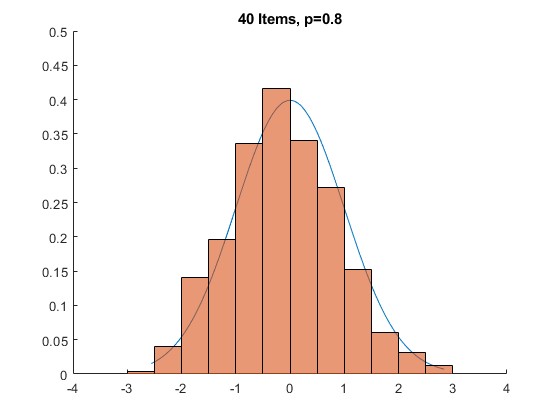}
    \end{subfigure}

    \begin{subfigure}{0.45\textwidth}
        \centering
        \includegraphics[width=\linewidth, height=4cm]{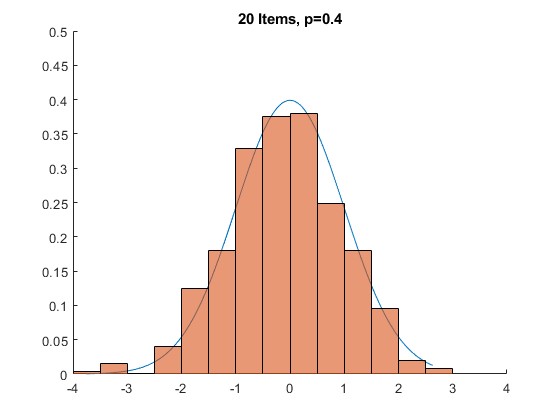}
    \end{subfigure}
     \begin{subfigure}{0.45\textwidth}
        \centering
        \includegraphics[width=\linewidth, height=4cm]{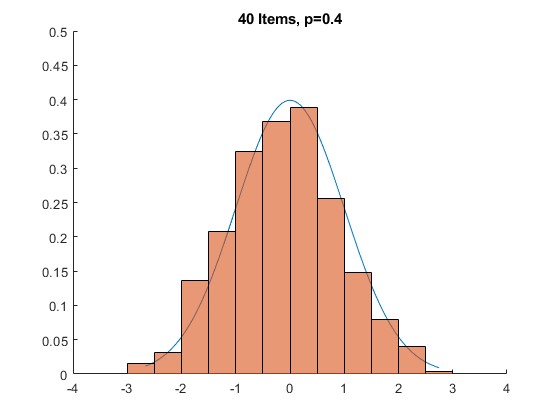}
    \end{subfigure}

   \caption{Histograms for the aggregated score gap. The first and second rows correspond to probabilities $p=0.8$ and $p=0.4$, respectively. From left to right, the number of items $d_2$ increases from 20 to 40. Each panel depicts the standardized error $\widehat{v}_k^{-1/2} \left( \tfrac{1}{d_1} \sum_{i=1}^{d_1} \widehat{M}^{\mathrm{NR}}_{i,1} 
- \tfrac{1}{d_1} \sum_{i=1}^{d_1} \Mstar_{i,1} \right)$. The solid blue line represents the standard normal distribution.} 
    \label{fig:asymp_normal_market}
\end{figure}

\begin{figure}[htbp]
    \centering
    \begin{subfigure}{0.45\textwidth}
        \centering
        \includegraphics[width=\linewidth, height=3.8cm]{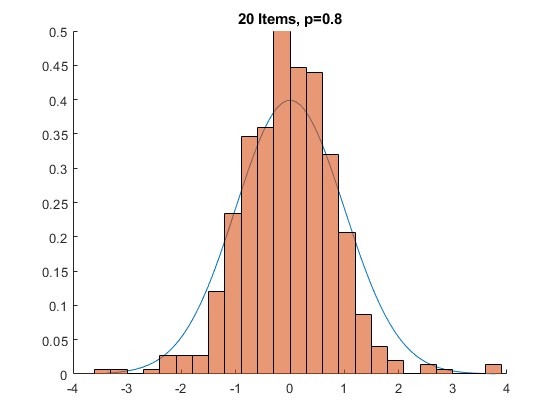}
    \end{subfigure}
     \begin{subfigure}{0.45\textwidth}
        \centering
        \includegraphics[width=\linewidth, height=3.8cm]{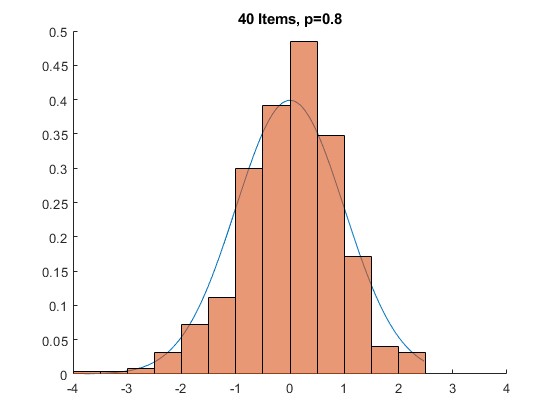}
    \end{subfigure}

    \begin{subfigure}{0.45\textwidth}
        \centering
        \includegraphics[width=\linewidth, height=3.8cm]{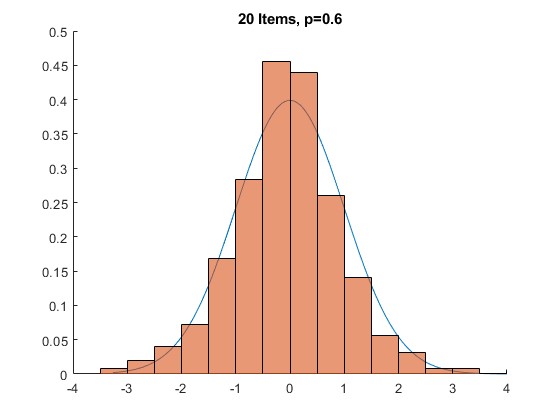}
    \end{subfigure}
     \begin{subfigure}{0.45\textwidth}
        \centering
        \includegraphics[width=\linewidth, height=3.8cm]{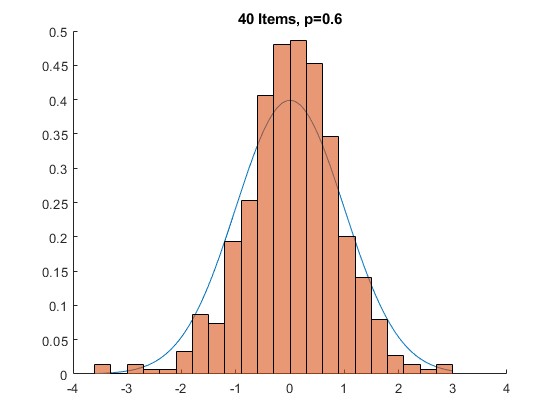}
    \end{subfigure}

    \caption{Histograms for the individual score gap. As in Figure \ref{fig:asymp_normal_market}, the first and second rows correspond to probabilities $p=0.8$ and $p=0.6$, respectively, and from left to right, $d_2$ increases from 20 to 40. Each panel depicts the standardized error $\widehat{w}_{1,1}^{-\frac{1}{2}}  \left(  \widehat{M}^{\proj}_{1,1}  -M^\star_{1,1} \right)$ for each setting. Again, the solid blue line represents the standard normal distribution.}  
    \label{fig:asymp_normal_indiv}
\end{figure}

{\bf Ranking confidence intervals.} Next, we perform experiments for the ranking confidence intervals for aggregated and individual levels (see Section \ref{sec:ranking_inference} and Section \ref{sec:market_ranking}). We maintain the setting from the previous experiments for asymptotic normality, except for choosing $d_2$ in $\{20, 30, 40\}$ and $p=0.8$ for all experiment settings. For both cases, we construct the confidence ranking interval for Item 1 out of all items. That is, we set $\calJ=\{1\}$ and $\calK=[d_2]$ for both cases. Again, the true $\bf\Theta^\star$ is being regenerated, so we do not lose generality. To compute the critical values, we conduct the bootstrap 2,000 times, targeting the confidence level $(1-\alpha) \times 100\%$ at 95\%. The entire procedure is repeated 500 times. Table \ref{tab:simul_ranking_interval} provides the simulation results. The true ranking is almost always contained within the constructed confidence interval throughout experiments, which is consistent with the conservative construction of the confidence intervals for tanks.  Note that the ranges of item ranks increase linearly with $d_2$, yet the confidence interval length increases much more slowly than the number of items in both the aggregated and individual cases (see the columns with title {\sl length/$d_2$}). Thus, as the dimension grows, the confidence intervals become increasingly informative.

\begin{table}[!htp]
\centering
\small
\begin{tabular}{c@{\hskip 0.5cm}ccc@{\hskip 1cm}ccc}
\toprule
\multirow{2}{*}{\textbf{\# of items}} & \multicolumn{3}{c@{\hskip 1cm}}{\textbf{Aggregated preference}} & \multicolumn{3}{c}{\textbf{Individual preference (User 1)}} \\
\cmidrule(r){2-4} \cmidrule(l){5-7}
 & CI Length & Length/$d_2$ & Coverage & CI Length & Length/$d_2$ & Coverage \\
\midrule
$d_2=20$ &  10.2360  & 51.18\%    & 1   & 16.3840  & 81.92\%   & 1\\
$d_2=30$  & 11.5460  & 38.49\%   & 1  & 19.1400    & 63.80\%   & 0.994 \\
$d_2=40$  & 12.4460  & 31.11\%   & 1  & 21.7340    & 54.33\%  & 1 \\
\bottomrule
\end{tabular}
\caption{Ranking confidence intervals. ``CI Length'' reports the average length of the ranking confidence interval for item~1, 
averaged over 500 repetitions. 
``Length/$d_2$'' is the ratio of this average length to the number of items. 
``Coverage'' denotes the proportion of repetitions in which the true rank of item~1 
is contained within the constructed confidence interval. }
\label{tab:simul_ranking_interval}
\end{table}
 
{\bf Sample splitting versus no sample splitting.} Lastly, we examine the efficiency loss incurred by sample splitting for uncertainty quantification of individual preferences (see Section~\ref{sec:individualinference}). To this end, we define the full-sample estimator $\widehat{\mathbf{M}}^{\proj,\full}$, constructed in the same way as $\widehat{\mathbf{M}}^\proj$ but without sample splitting. That is, 
\begin{align*}
        \widehat{\bfM}^{\proj,\full} \coloneqq \argmin_{\rank(\bfM) \leq q}\lVert \widehat{\bfM}^{\mathrm{NR}}-\bfM\rVert_{\mathrm{F}}^2 =\widehat{\bfU}^{\mathrm{NR}} \widehat{\bfU}^{\mathrm{NR} \top} \widehat{\bfM}^{\mathrm{NR}} \widehat{\bfV}^{\mathrm{NR}} \widehat{\bfV}^{\mathrm{NR} \top}
    \end{align*}
    where $\widehat{\bfM}^{\mathrm{NR}}$ is defined in \eqref{eq:defofdebias}, and   $\widehat{\bfU}^{\mathrm{NR}}$ and $\widehat{\bfV}^{\mathrm{NR}}$ are the top-$q$ left and right singular vectors of $\widehat{\bfM}^{\mathrm{NR}}.$ We compute the Frobenius norm and the operator norm error of $\widehat{\bfM}^\proj$ and $\widehat{\bfM}^{\proj,\full}$. We use the same DGP as in the previous experiments, with $p_i = p \in \{0.4, 0.8\}$, 
$d_2 \in \{20,30,40,50,60\}$, and $d_1 = d_2(d_2-1)/2$. 
For each setting, the procedure is repeated 20 times. 
Figure~\ref{fig:comparison_ss_noss} reports the results. 
We find that sample splitting has little effect on the operator norm error in most settings. 
Moreover, the Frobenius norm errors are quite close, and the gap tends to shrink as the number of items increases.

\begin{figure}[htbp]
    \centering
    \includegraphics[width=\textwidth, height=7cm]{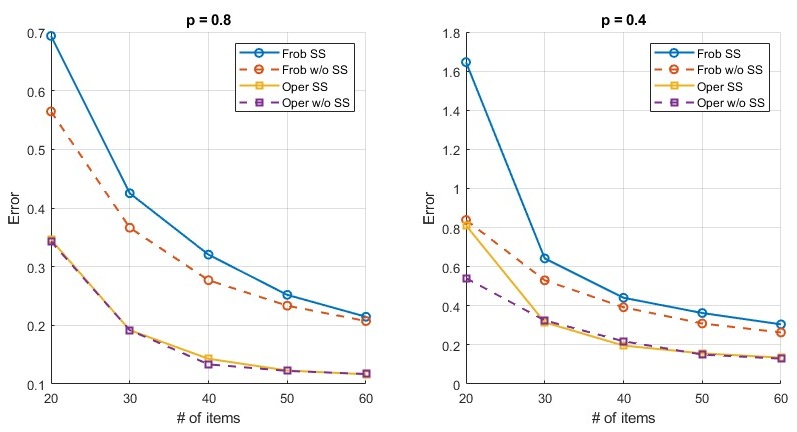}
    \caption{The left panel corresponds to the first experiment with higher probability $p=0.8$ and the right panel corresponds to the lower probability $p=0.4$.``Frob SS'' and ``Frob w/o SS'' denote the simulation average of the estimation error $\lVert\widehat{\bfM}^\proj-\bfM^\star\rVert_{\mathrm{F}}/\sqrt{d_1 d_2(d_2-1)/2}$ and $\lVert\widehat{\bfM}^{\proj,\full}-\bfM^\star\rVert_{\mathrm{F}}/\sqrt{d_1 d_2(d_2-1)/2}$, respectively. Likewise, ``Oper SS'' and ``Oper w/o SS'' denote the simulation average of the estimation error $\lVert\widehat{\bfM}^\proj-\bfM^\star\rVert/\sqrt{d_1 d_2(d_2-1)/2}$ and $\lVert\widehat{\bfM}^{\proj,\full}-\bfM^\star\rVert/\sqrt{d_1 d_2(d_2-1)/2}$, respectively.}
    \label{fig:comparison_ss_noss}
\end{figure}

\section{Top-$K$ item selection}\label{sec:topK}

Building on the error bound in Theorem~\ref{thm:errorboundforTheta}, we present a corollary on top-$K$ item selection for both individual and aggregated preferences. For each user $i\in [d_1]$, let $\Theta^\star_{i,(l)}$ be the score of $i$'s $l$th preferred item, and define $\Delta_{i,K} \coloneqq \Theta^\star_{i,(K)}-\Theta^\star_{i,(K+1)}$ for each $i \in [d_1].$ Also, define the aggregated evaluation of item $j \in [d_2]$ as $\bar{\Theta}^\star_j \coloneqq (d_1)^{-1} \sum_{i=1}^{d_1} \Theta^\star_{i,j}$ and $\bar{\Theta}^\star_{(l)}$ be the score of the $l$th preferred item. Denote $\Delta_K \coloneqq \bar{\Theta}^\star_{(K)}-\bar{\Theta}^\star_{(K+1)}.$
 
\begin{corollary}\label{cor:topK}
    Suppose the assumptions in Theorem \ref{thm:errorboundforTheta} hold and that
    \begin{align*}
        \Delta_{i,K} > C  \frac{\kappa^2 \mu R   }{\min\{d_1, d_2(d_2-1)/2\}} \sqrt{\frac{  \bar{d}\log(\bar{d})}{\bar{p} }}
    \end{align*}
for some sufficiently large $C>0$. Then, the set of top-$K$ items based on the $i$th row of \( \widehat{\bf\Theta} \) coincides with the top-$K$ items under user~$i$'s true preference, with probability at least \( 1 - O(\bar{d}^{-10}) \). Also, if we assume that
\begin{align*}
        \Delta_{K} > C  \frac{\kappa^2 \mu R  }{\min\{d_1, d_2(d_2-1)/2\}} \sqrt{\frac{  \bar{d}\log(\bar{d})}{ \bar{p}}}
    \end{align*}
for some sufficiently large $C >0$, then the set of top-$K$ items based on the average $d_1^{-1} \sum_{i=1}^{d_1}\widehat{\bf\Theta}_{i,\cdot}$ coincides with the top-$K$ items for the aggregated preference, with probability at least $1-O(\bar{d}^{-10})$.
\end{corollary}

\begin{remark}
Corollary \ref{cor:topK} relies on a sufficiently large score gap between the $K$th and $K+1$th ranked items. Even when this gap condition is inapplicable, we can still conduct inference on the top-$K$ items by leveraging the uncertainty quantification results developed in Section~\ref{sec:ranking_inference} and in Section~\ref{sec:additionalrankinginferences}.
\end{remark}
 
\section{Low-rank approximation of $\bfM^\star$}\label{sec:LowrankM}

Consider a sieve approximation of $\bfM^\star$:
\begin{align*}
 M^\star_{i,k} = \widetilde{g}_i(\widetilde{\bff}_k) = \sum_{l=1}^q \gamma_{i,l}\psi_l(\widetilde{\bff}_k)+\xi_{i,k}={\bm\gamma}_i^\top {\bm\psi}_k+\xi_{i,k}
\end{align*}
where $\widetilde{\bff}_k \coloneqq [{\bm\eta}_{j}^\top, {\bm\eta}_{j'}^\top]^\top$ is $2r$-dimensional, $ \widetilde{g}_i(\cdot)$ is such that $ \widetilde{g}_i(\bfx,\bfy)=g_i(\bfx)-g_i(\bfy)$ for $r$-dimensional vectors $\bfx$ and $\bfy$, ${\bm\gamma}_i \coloneqq (\gamma_{i,1}, \ldots, \gamma_{i,q})^\top \in \bbR^q$ is the sieve coefficient, ${\bm\psi}_k \coloneqq (\psi_1(\widetilde{\bff}_k), \ldots, \psi_q(\widetilde{\bff}_k))^\top \in \bbR^q$ is the basis functions, $q$ is the sieve dimension, and $\xi_{i,j}$ is the sieve approximation error. Denote ${\bf\Gamma}$ as $d_1 \times q$ matrix that stacks ${\bm\gamma}_i^\top$, ${\bf\Psi}$ as $d_2(d_2-1)/2 \times q$ matrix that stacks ${\bm\psi}_k^\top$ and $\bm\Xi$ as $d_1 \times d_2(d_2-1)/2$ matrix of $\xi_{i,k}$. Then, in the matrix notation, we can write
\begin{align*}
      \bfM^\star  =  {\bf\Gamma} {\bf\Psi}^\top + {\bm\Xi}.
\end{align*}

First,  we assume that $\{\widetilde{g}_i(\cdot)\}_{i\in [d_1]}$ belong to a H\"older class: For some $a,b,C>0,$
\begin{align*}
     \left\lbrace h : \max_{b_1+ \cdots + b_r=b}\left\vert \frac{\partial^{b}h(x)}{\partial x_1^{b_1}\cdots\partial x_r^{b_r}} - \frac{\partial^{b}h(y)}{\partial y_1^{b_1}\cdots\partial y_r^{b_r}} \right\vert \leq C \| \bfx - \bfy\|^{a}, \quad \text{for all $\bfx, \bfy$} \right  \rbrace .
\end{align*}
Similar to Section~\ref{sec:low_rank_approximation}, we assume that $q$ grows slowly to capture the dominant part of $\bfM^\star$ and the sieve approximation error becomes negligible:
\begin{assumption}\label{asp:approx_low_rank_M}
$\max_{i,k} |\xi_{i,k}| \ll  \sqrt{ \frac{1}{\max\{d_1,d_2(d_2-1)/2\}}}$.
\end{assumption}

Second, we assume $\sigma_{\max}({\bf\Gamma} {\bf\Psi}^\top)$ has standard strength and its condition number, denoted by $\vartheta$, grows slowly. Specifically, we assume the following:
\begin{assumption}\label{asp:SSVforM}
    $\sigma_{\max}({\bf\Gamma} {\bf\Psi}^\top)\asymp \sqrt{d_1d_2(d_2-1)/2}$ and $\vartheta \sqrt{\frac{\kappa^2 R \bar{d}}{ \bar{p}}}\ll \sqrt{d_1d_2(d_2-1)/2}$.
\end{assumption}

We note that Assumption \ref{asp:approx_low_rank_M}, Assumption \ref{asp:SSVforM}, and Corollary \ref{cor:errorboundforM} together imply that the leading $q$ singular values of $\bfM^\star$ grow faster than the sieve approximation error and the Frobenius norm estimation error in $\widehat{\bfM}$. As a result, the singular values of $\widehat{\bfM}$ can be used to estimate $q$ by following the rank estimation methods \citep[e.g., see][]{bai2002determining,ahn2013eigenvalue, chernozhukov2023inference,choi2024inference}.

 \section{Ranking inference for the aggregated preference}\label{sec:market_ranking}

This section elaborates on ranking inference for the aggregated preference. The analysis is similar to that for individual preferences (Section~\ref{sec:ranking_inference}) but is comparatively simpler. For simplicity, we set $\calK=[d_2]$ in this section. In the aggregated preference case, from the proof of Theorem \ref{thm:asymptoticnormality}, we have, for $j\neq j'$,
\begin{align}
    &\frac{1}{d_1} \sum_{i=1}^{d_1}\widehat{M}^{\mathrm{NR}}_{i,\bar{\calL}(j,j')}  -\frac{1}{d_1} \sum_{i=1}^{d_1}M^\star_{i,\bar{\calL}(j,j')} \nonumber \\
    &  \approx   \frac{1}{d_1} \sum_{i=1}^{d_1}  \underbrace{\frac{1}{p_i} \delta_{i,\bar{\calL}(j,j')} \left(\sigma'(M^\star_{i,\bar{\calL}(j,j')}) \right)^{-1} \left( y_{i,\bar{\calL}(j,j')}- \sigma(M^\star_{i,\bar{\calL}(j,j')}) \right)}_{\coloneqq \xi_{i,\bar{\calL}(j,j')}}\label{eq:linearexpansion}
\end{align}
where $\bar{\calL}(j, j')=\calL(j, j')$ if $j<j'$, and $\bar{\calL}(j, j')=\calL(j', j)$ if $j>j'$.

Denote the aggregated average score for item $j$ and its ranking induced by this score as $\bar{\Theta}^\star_{j}$ and $\bar{r}_{j}$, resp. We assume that there is no tie. Suppose that, for a subset of items $\calJ\in[d_2]$, we have simultaneous score gap CIs $\{[C_L(j,j'), C_U(j,j')] \}_{j \in\calJ, j'\neq j}$ such that $\bar{\Theta}^\star_{j}-\bar{\Theta}^\star_{j'} \in [C_L(j,j'), C_U(j,j')], \, \forall j \in \calJ, \, \forall j'\neq j$ with probability at least $1-\alpha.$ Then, we will have the following relation, which leads to the ranking CIs for all $j \in \calJ$:
\begin{align}
    1-\alpha & \leq \bbP \left( \bar{\Theta}^\star_{j}-\bar{\Theta}^\star_{j'}\in [C_L(j,j'), C_U(j,j')], \, \, \forall j \in \calJ, \, \forall j'\neq j \right) \nonumber \\
    & \leq \bbP \left( \bigcap_{j \in \calJ} \bigg\{1+ \underbrace{\sum_{j' \neq j}\bbOne \{C_U(j,j')<0\}}_{\text{``better than item $j$''}} \leq \bar{r}_{j} \leq d_2 - \underbrace{\sum_{j' \neq j}\bbOne \{C_L(j,j')>0\}}_{\text{``worse than item $j$''}}\bigg\} \right).\label{eq:rankinginterval}
\end{align}

In order to construct these simultaneous score gap CIs, the $(1-\alpha)$th quantile of the distribution of
\begin{align*}
    \widehat{\calT}_{\calJ} \coloneqq \max_{j \in \calJ}\max_{j' \neq j} \bigg| \frac{ 1 }{\sqrt{\widehat{v}_{\bar{\calL}(j,j')}}}\left(\frac{1}{d_1} \sum_{i=1}^{d_1}\widehat{M}^{\mathrm{NR}}_{i,\bar{\calL}(j,j')}  -\frac{1}{d_1} \sum_{i=1}^{d_1}M^\star_{i,\bar{\calL}(j,j')} \right)\bigg|
\end{align*}
is needed. Motivated by the linear expansion \eqref{eq:linearexpansion}, we define its population counterpart
\begin{align*}
    \calT_{\calJ}\coloneqq \max_{j \in \calJ}\max_{j' \neq j}\bigg|\frac{1}{\sqrt{v^\star_{\bar{\calL}(j,j')}}} \frac{1}{d_1}\sum_{i=1}^{d_1}  \xi_{i, \bar{\calL}(j,j')} \bigg|
\end{align*}
and its bootstrap counterpart
\begin{align*}
    \calG_\calJ\coloneqq \max_{j \in \calJ}\max_{j' \neq j} \bigg|\frac{1}{\sqrt{v^\star_{\bar{\calL}(j,j')}}} \frac{1}{d_1}\sum_{i=1}^{d_1}  \xi_{i, \bar{\calL}(j,j')} Z_i\bigg|
\end{align*}
where $\{Z_i\}_{i\geq1}$ are i.i.d. standard normal random variables. Let $\calG_{\calJ;(1-\alpha)}$ be the $(1-\alpha)$th quantile of $\calG_{\calJ}.$ The Gaussian multiplier bootstrap \citep[Theorem 2.2 in][]{chernozhuokov2022improved} guarantees:
\begin{align}
    |\bbP \left(\calT_\calJ > \calG_{\calJ;(1-\alpha)}\right) -\alpha|  \lesssim \left(\frac{\log^5(d_1d_2)}{\bar{p}d_1} \right)^{\frac{1}{4}}. \label{eq:bootstrap}
\end{align}
The feasible version for $\calG_\calJ$ is naturally defined as:
\begin{align*}
    \widehat{\calG}_\calJ \coloneqq \max_{j \in \calJ}\max_{j' \neq j} \bigg|\frac{1}{\sqrt{\widehat{v}_{\bar{\calL}(j,j')}}} \frac{1}{d_1}\sum_{i=1}^{d_1} \widehat{\xi}_{i,\bar{\calL}(j,j')} Z_i\bigg|
\end{align*}
where $\widehat{\xi}_{i,\bar{\calL}(j,j')}=p_i^{-1} \delta_{i,\bar{\calL}(j,j')} \left(\sigma'(\widehat{M}_{i,\bar{\calL}(j,j')}) \right)^{-1} \left( y_{i,\bar{\calL}(j,j')}- \sigma(\widehat{M}_{i,\bar{\calL}(j,j')}) \right)$. Then, we have the following result that validates the bootstrap approach.
\begin{theorem}\label{thm:marketbootstrap}
    Let $\widehat{\calG}_{\calJ;(1-\alpha)}$ be the $(1-\alpha)$th quantile of $\widehat{\calG}_\calJ.$ Suppose that the assumptions in Theorem \ref{thm:asymptoticnormality} hold. In addition, assume that $\bar{p}d_1 \gg \log^3 (\bar{d})$. Then, we have
\begin{align*}
    |\bbP \left(\widehat{\calT_\calJ} > \widehat{\calG}_{\calJ;(1-\alpha)}\right) -\alpha|  \lesssim \left(\frac{\log^5(d_1d_2)}{\bar{p}d_1} \right)^{\frac{1}{4}}+o(1)
\end{align*}
with probability at least $1-O(\bar{d}^{-9})$.
\end{theorem}

We then use the bootstrap quantile $\widehat{\calG}_{\calJ;(1-\alpha)}$ for constructing the above simultaneous CIs $\{[C_L(j,j'), C_U(j,j')] \}_{j \in \calJ, j'\neq j}$. That is, we define the CIs for all $j\in \calJ$ and $j'\neq j$ as follows:
\begin{align*}
    &C_L(j,j') \coloneqq (-1)^{1-\bbOne\{j<j'\}} \frac{1}{d_1}\sum_{i=1}^{d_1}\widehat{M}^{\mathrm{NR}}_{i,\bar{\calL}(j,j')} -\widehat{\calG}_{\calJ;(1-\alpha)}\times \sqrt{\widehat{v}_{\bar{\calL}(j,j')}}; \\
    &C_U(j,j') \coloneqq (-1)^{1-\bbOne\{j<j'\}} \frac{1}{d_1}\sum_{i=1}^{d_1}\widehat{M}^{\mathrm{NR}}_{i,\bar{\calL}(j,j')} +\widehat{\calG}_{\calJ;(1-\alpha)}\times \sqrt{\widehat{v}_{\bar{\calL}(j,j'))}}.
\end{align*}

Lastly, the $(1-\alpha)\times 100$\% simultaneous CIs for $\{\bar{r}_j\}_{j \in \calJ}$ are defined as $\{[\widehat{r}_j^U, \widehat{r}_j^L]\}_{j \in \calJ}$ where
\begin{align}
    \widehat{r}_j^U \coloneqq 1+  \sum_{j' \neq j}\bbOne \{C_U(j,j')<0\} \quad   \text{and} \quad \widehat{r}_j^L \coloneqq d_2 -  \sum_{j' \neq j}\bbOne \{C_L(j,j')>0\} \quad \text{for each $j \in \calJ$},\label{eq:rankinginferenceLU}
\end{align}
as shown in \eqref{eq:rankinginterval}. 

\section{Additional Results for Ranking Inferences}\label{sec:additionalrankinginferences}
In Section~\ref{sec:ranking_inference} and Section~\ref{sec:market_ranking}, we construct two-sided CIs for item rankings for both individual and aggregated preferences. This section considers hypothesis testing based on one-sided confidence intervals. For brevity, we will focus on the aggregated preference.\footnote{As in Section~\ref{sec:ranking_inference} and Section~\ref{sec:market_ranking}, the main difference in the individual preference case is the use of the linear expansion \eqref{eq:Mprojlinearexpansion}.} As in Section~\ref{sec:market_ranking}, assume that $\calK=[d_2].$

\begin{example}[Top-$K$ placement test]\label{ex:topKplacement}
    Suppose that we are interested in testing whether a given item $j$ belongs to the set of top-$K$ items with respect to the true aggregated preference $(d_1)^{-1}\sum_{i \in [d_1]} {\bf\Theta}^\star_{i,\cdot}$. To do so, we define 
    \begin{align*}
    \widetilde{\calT}_{j} \coloneqq \max_{j' \neq j}   \frac{ 1 }{\sqrt{\widehat{v}_{\bar{\calL}(j,j')}}}  (-1)^{\bbOne\{j <j'\} }\left(\frac{1}{d_1} \sum_{i=1}^{d_1}\widehat{M}^{\mathrm{NR}}_{i,\bar{\calL}(j,j')}  -\frac{1}{d_1} \sum_{i=1}^{d_1}M^\star_{i,\bar{\calL}(j,j')} \right) .
\end{align*}
and its bootstrap counterpart
\begin{align*}
    \widetilde{\calG}_j \coloneqq  \max_{j' \neq j}  \frac{1}{\sqrt{\widehat{v}_{\bar{\calL}(j,j')}}} (-1)^{\bbOne\{j <j'\} } \frac{1}{d_1}\sum_{i=1}^{d_1} p_i^{-1} \delta_{i,\bar{\calL}(j,j')} \left(\sigma'(\widehat{M}_{i,\bar{\calL}(j,j')}) \right)^{-1} \left( y_{i,\bar{\calL}(j,j')}- \sigma(\widehat{M}_{i,\bar{\calL}(j,j')}) \right) Z_i 
\end{align*}
where $\{Z_i\}_{i\geq1}$ are i.i.d. standard normal random variables. Denote the $(1-\alpha)$th quantile of $\widetilde{\calG}_j$ as $\widetilde{\calG}_{j; (1-\alpha)}$ and define
\begin{align*}
    C_L(j,j') &\coloneqq (-1)^{\bbOne\{j<j'\}} \frac{1}{d_1}\sum_{i=1}^{d_1}\widehat{M}^{\mathrm{NR}}_{i,\bar{\calL}(j,j')} -\widetilde{\calG}_{j;(1-\alpha)}\times \sqrt{\widehat{v}_{\bar{\calL}(j,j'))}}; \\
    \widetilde{r}_j^U &\coloneqq 1+  \sum_{j' \neq j}\bbOne \{C_L(j,j')>0\}.
\end{align*}
Consider the hypotheses
\begin{align*}
    H_0: \bar{r}_j \leq K \quad \text{versus} \quad H_1: \bar{r}_j > K.
\end{align*}
Similar to Section~\ref{sec:market_ranking}, we can show that $\widetilde{r}_j^U \leq \bar{r}_j$ with probability approaching $1-\alpha$. This implies that, under the null hypothesis, $\widetilde{r}_j^U \leq K$ holds with probability approaching at least $1-\alpha$. Therefore, we can reject the null hypothesis if $\widetilde{r}_j^U > K$.
\end{example}

 \begin{example}[Sure screening of top-$K$ candidates]\label{ex:surescreening}
 Suppose now that we are interested in constructing a confidence set that includes all top-$K$ items with respect to the true aggregated preference. For this purpose, we set $\calJ=\calK=[d_2]$, and define 
    \begin{align*}
    \widetilde{\calT}  \coloneqq \max_{j \in \calJ}\max_{j' \neq j}   \frac{ 1 }{\sqrt{\widehat{v}_{\bar{\calL}(j,j')}}}  (-1)^{\bbOne\{j <j'\} }\left(\frac{1}{d_1} \sum_{i=1}^{d_1}\widehat{M}^{\mathrm{NR}}_{i,\bar{\calL}(j,j')}  -\frac{1}{d_1} \sum_{i=1}^{d_1}M^\star_{i,\bar{\calL}(j,j')} \right) .
\end{align*}
and its bootstrap counterpart
\begin{align*}
    \widetilde{\calG}  \coloneqq \max_{j \in \calJ} \max_{j' \neq j}  \frac{1}{\sqrt{\widehat{v}_{\bar{\calL}(j,j')}}} (-1)^{\bbOne\{j <j'\} } \frac{1}{d_1}\sum_{i=1}^{d_1} p_i^{-1} \delta_{i,\bar{\calL}(j,j')} \left(\sigma'(\widehat{M}_{i,\bar{\calL}(j,j')}) \right)^{-1} \left( y_{i,\bar{\calL}(j,j')}- \sigma(\widehat{M}_{i,\bar{\calL}(j,j')}) \right) Z_i 
\end{align*}
where $\{Z_i\}_{i\geq1}$ are i.i.d. standard normal random variables. Denote the $(1-\alpha)$th quantile of $\widetilde{\calG}$ as $\widetilde{\calG}_{(1-\alpha)}$ and define
\begin{align*}
    C_L(j,j') &\coloneqq (-1)^{\bbOne\{j<j'\}} \frac{1}{d_1}\sum_{i=1}^{d_1}\widehat{M}^{\mathrm{NR}}_{i,\bar{\calL}(j,j')} -\widetilde{\calG}_{(1-\alpha)}\times \sqrt{\widehat{v}_{\bar{\calL}(j,j'))}}; \\
    \widetilde{r}_j^U &\coloneqq 1+  \sum_{j' \neq j}\bbOne \{C_L(j,j')>0\}.
\end{align*}
Lastly, define
$$\widetilde{S}_K \coloneqq \bigg\{j \in[d_2]:\widetilde{r}_j^U\leq K  \bigg\}.$$
Then, similar to Section~\ref{sec:market_ranking}, $\widetilde{r}_j^U \leq \bar{r}_j$ for all $j\in\calJ=[d_2]$ with probability approaching $1-\alpha$, which implies that all item $j$ such that $\bar{r}_j \leq K$ belong to $\widetilde{S}_K$ with probability approaching at least $1-\alpha$.
\end{example}

\section{Proofs of Main Results}\label{sec:proofs}

\begin{theorem}\label{thm:errorboundsforL}
 Suppose the assumptions in Theorem \ref{thm:errorboundforTheta} hold. Then, with probability at least $1-O(\bar{d}^{-10})$, we have
    \begin{align*}
        \norm{\widehat{\bfL}-\bfL^{\star}}_{\mathrm{F}} &\lesssim   \sqrt{\frac{ \kappa^2 R\bar{d}}{\bar{p}   }}; \quad
       \norm{\widehat{\bfL}-\bfL^{\star}} \lesssim   \sqrt{\frac{\kappa^2 \bar{d}}{\bar{p}   }};\\
    \norm{\widehat{\bfL}-\bfL^{\star}}_{\infty} &\lesssim  \frac{\kappa^2 \mu R  }{\min\{d_1, d_2(d_2-1)/2\}} \sqrt{\frac{  \bar{d}\log(\bar{d})}{\bar{p}  }}.
    \end{align*}
\end{theorem}

\begin{proof} 
 We prove the entrywise max norm error bound, as the other two cases can be proved similarly. This proof uses the nonconvex gradient descent iterates defined in Section \ref{sec:sectionA}. First, note that \eqref{Prelim2}, Lemma \ref{LemmaB1} and Lemma \ref{LemmaB11} imply
\begin{align*}
     \norm{\bfX^t \bfY^{t \top} - \bfL^\star}_{\infty} &\leq \norm{\bfX^t \bfH^t - \bfX^\star}_{\infty} \norm{\bfY^t \bfH^t}_{\infty} + \norm{\bfX^\star}_{\infty} \norm{\bfY^t \bfH^t - \bfY^\star}_{\infty} \\
     &\lesssim \frac{\mu R \sigma^\star_{\max} }{\min\{d_1, d_2(d_2-1)/2\}} \kappa \sqrt{\frac{  \bar{d}\log(\bar{d})}{\bar{p} (\sigma^\star_{\min})^2 }}
\end{align*}
uniformly for all $t \geq 0$ with probability at least $1-O(\bar{d}^{-10})$. 

Second, we pick $t_* \coloneqq \argmin_{0\leq t < t_0}\norm{\nabla f (\bfX^t,\bfY^t)}_{\mathrm{F}}$. Then, Lemma \ref{LemmaA1} and Lemma \ref{LemmaB1} yield
\begin{align*}
    \norm{\bfX^{t_*} \bfY^{t_* \top}-\widehat{\bfL}}_{\infty} \leq \norm{\bfX^{t_*} \bfY^{t_* \top}-\widehat{\bfL}}_{\mathrm{F}} \lesssim \frac{\kappa \lambda}{\bar{d}^5} \ll \frac{ \mu R \sigma^\star_{\max} }{\min\{d_1, d_2(d_2-1)/2\}} \kappa\sqrt{\frac{  \bar{d}\log(\bar{d})}{p (\sigma^\star_{\min})^2 }}.
\end{align*}
Therefore, we reach 
\begin{align*}
    \norm{\widehat{\bfL}-\bfL^{\star}}_{\infty} \leq \norm{\bfX^{t_*} \bfY^{t_* \top}-\widehat{\bfL}}_{\infty}+ \norm{\bfX^{t_*} \bfY^{t_* \top}- \bfL^\star}_{\infty} \lesssim \frac{ \mu R \sigma^\star_{\max} }{\min\{d_1, d_2(d_2-1)/2\}} \kappa \sqrt{\frac{  \bar{d}\log(\bar{d})}{p (\sigma^\star_{\min})^2 }}
\end{align*}
with probability at least $1-O(\bar{d}^{-10})$.   
\end{proof}

 \begin{corollary}\label{cor:errorboundforM}
Under the assumption in Theorem \ref{thm:errorboundforTheta}, we have with probability at least $1-O(\bar{d}^{-10})$,
\begin{align*}
      \norm{\widehat{\bfM}-\bfM^{\star}}_{\mathrm{F}}   \lesssim    \sqrt{ \frac{\kappa^2 R \bar{d}}{\bar{p}}}; \quad \quad  \norm{\widehat{\bfM}-\bfM^{\star}}_{\infty}    \lesssim  \frac{\kappa^2\mu R }{\min\{d_1, d_2(d_2-1)/2\}} \sqrt{\frac{  \bar{d}\log(\bar{d})}{\bar{p}   }}.
\end{align*}
 \end{corollary}
 
\begin{proof}
Define $s(x)\coloneqq \sigma^{-1}(x)$ and recall that $\widehat{\bfM}\coloneqq s(\widehat{\bfL})$. Note that, in our analysis, $s'(L^\star_{i,k})$, $s'(\widehat{L}_{i,k})$, $s''(L^\star_{i,k})$, and $s''(\widehat{L}_{i,k})$ are bounded for all $(i,k)$ with probability at least $1-O(\bar{d}^{-10})$. This is because $\norm{{\bf\Theta}^\star}_{\infty}$ is assumed to be bounded, and thus $\norm{\bfM^\star}_{\infty}$ is also bounded. By this fact and the small sieve approximation error assumption, we can conclude that the entries $L^\star_{i,k}$ are also bounded away from 0 and 1, for all $(i,k)$. In addition, the entrywise error control in  Theorem \ref{thm:errorboundsforL} ensures that the entries $\widehat{L}_{i,k}$ are also bounded away from 0 and 1 with probability at least $1-O(\bar{d}^{-10})$. 

Next, Theorem \ref{thm:errorboundsforL} and the assumptions therein imply that
\begin{align*}
 |  \widehat{L}_{i,k}-L^\star_{i,k}- \varepsilon_{i,k} |\leq  \norm{\widehat{\bfL}-\bfL^{\star}}_{\infty}+ \norm{\bm\calE}_{\infty} \lesssim \frac{\mu R \sigma^\star_{\max} }{\min\{d_1, d_2(d_2-1)/2\}} \kappa \sqrt{\frac{  \bar{d}\log(\bar{d})}{p (\sigma^\star_{\min})^2 }} \ll 1.
\end{align*}
Therefore, we can write, with probability at least $1-O(\bar{d}^{-10})$,   
\begin{align*}
    \norm{\widehat{\bfM}-\bfM^{\star}}_{\mathrm{F}}^2 & = \sum_{i=1}^{d_1} \sum_{k=1}^{d_2(d_2-1)/2} \left( s(\widehat{L}_{i,k})- s(L^\star_{i,k}+ \varepsilon_{i,k}) \right)^2\\
    & = \sum_{i=1}^{d_1} \sum_{k=1}^{d_2(d_2-1)/2} \left( s'(L^\star_{i,k}+ \varepsilon_{i,k}) (\widehat{L}_{i,k}-L^\star_{i,k}- \varepsilon_{i,k}) +\frac{1}{2}s''(x_{i,k}) (\widehat{L}_{i,k}-L^\star_{i,k}- \varepsilon_{i,k})^2 \right)^2 \\
    & \leq 2 \sum_{i=1}^{d_1} \sum_{k=1}^{d_2(d_2-1)/2}  \left(s'(L^\star_{i,k}+ \varepsilon_{i,k})\right)^2 (\widehat{L}_{i,k}-L^\star_{i,k}- \varepsilon_{i,k})^2 +\left(\frac{1}{2}s''(x_{i,k})\right)^2 (\widehat{L}_{i,k}-L^\star_{i,k}- \varepsilon_{i,k})^4  \\
    & \lesssim \sum_{i=1}^{d_1} \sum_{k=1}^{d_2(d_2-1)/2} (\widehat{L}_{i,k}-L^\star_{i,k}- \varepsilon_{i,k})^2  \leq \sum_{i=1}^{d_1} \sum_{k=1}^{d_2(d_2-1)/2} 2(\widehat{L}_{i,k}-L^\star_{i,k} )^2  + 2\varepsilon_{i,k}^2\\
    & \lesssim  \norm{\widehat{\bfL}-\bfL^{\star}}_{\mathrm{F}}^2 + \norm{\bm\calE}_{\mathrm{F}}^2  \lesssim \norm{\widehat{\bfL}-\bfL^{\star}}_{\mathrm{F}}^2
\end{align*}
for some $x_{i,k}$ between $\widehat{L}_{i,k}$ and $L^\star_{i,k}+ \varepsilon_{i,k}.$ The entrywise error bound can be established in a similar manner.

\end{proof}

\subsection{Proof of Theorem \ref{thm:errorboundforTheta}}
\begin{proof}[Proof of Theorem \ref{thm:errorboundforTheta}]
    For the Frobenius norm error bound, we have
\begin{align*}
    \norm{\widehat{\bf\Theta}-{\bf\Theta}^{\star}}_{\mathrm{F}}^2 &= \sum_{i=1}^{d_1} \sum_{j=1}^{d_2} \left(\frac{1}{d_2} \left(\sum_{j'>j} \widehat{\bfM}_{i, \calL(j,j')} -\sum_{j'<j} \widehat{\bfM}_{i, \calL(j',j)} \right)-\frac{1}{d_2} \left(\sum_{j'>j} \bfM^\star_{i, \calL(j,j')} -\sum_{j'<j} \bfM^\star_{i, \calL(j',j)} \right) \right)^2 \\
    &\leq \frac{1}{d_2} \sum_{i=1}^{d_1} \sum_{j=1}^{d_2}  \left( \sum_{j'>j} \left(\widehat{\bfM}_{i, \calL(j,j')} -  \bfM^\star_{i, \calL(j,j')}\right)^2 +  \sum_{j'<j} \left(\widehat{\bfM}_{i, \calL(j',j)}-\bfM^\star_{i, \calL(j',j)} \right)^2\right) \\
    & \leq  \frac{2}{d_2}\norm{\widehat{\bfM}-\bfM^{\star}}_{\mathrm{F}}^2
\end{align*}
where the second line uses Cauchy--Schwarz inequality. As a result, Corollary \ref{cor:errorboundforM} yields
\begin{align*}
    \norm{\widehat{\bf\Theta}-{\bf\Theta}^{\star}}_{\mathrm{F}} \lesssim \frac{1}{\sqrt{d_2}}\norm{\widehat{\bfM}-\bfM^{\star}}_{\mathrm{F}} \lesssim  \frac{1}{\sqrt{d_2}}  \sqrt{ \frac{\kappa^2 R \bar{d}}{\bar{p} }}.
\end{align*}
The entrywise error bound also follows from Corollary \ref{cor:errorboundforM}.
\end{proof}

\subsection{Proof of Theorem \ref{thm:asymptoticnormality}}
\begin{proof}[Proof of Theorem \ref{thm:asymptoticnormality}]
   We begin by applying the Taylor expansion for each $(i,k)$,
\begin{align}
    \widehat{M}^{\mathrm{NR}}_{i,k}  &= M^\star_{i,k} + \widehat{M}_{i,k}-M^\star_{i,k}+ \frac{1}{p_i} \delta_{i,k} \left(\sigma'(\widehat{M}_{i,k}) \right)^{-1} \left( y_{i,k}- \sigma(M^\star_{i,k}) \right) \nonumber\\
    & \quad + \frac{1}{p_i} \delta_{i,k} \left(\sigma'(\widehat{M}_{i,k}) \right)^{-1} \left( \sigma(M^\star_{i,k})- \sigma(\widehat{M}_{i,k}) \right) \nonumber\\
    & = M^\star_{i,k} + \widehat{M}_{i,k}-M^\star_{i,k}+ \frac{1}{p_i} \delta_{i,k} \left(\sigma'(\widehat{M}_{i,k}) \right)^{-1} \left( y_{i,k}- \sigma(M^\star_{i,k}) \right) \nonumber\\
    & \quad + \frac{1}{p_i} \delta_{i,k} \left(\sigma'(\widehat{M}_{i,k}) \right)^{-1}  \sigma'(\widehat{M}_{i,k}) \left(   M^\star_{i,k} -  \widehat{M}_{i,k}  \right) + \frac{1}{2p_i} \delta_{i,k} \left(\sigma'(\widehat{M}_{i,k}) \right)^{-1}  \sigma''(x_{i,k}) \left(   M^\star_{i,k} -  \widehat{M}_{i,k}  \right)^2 \nonumber\\
    & = M^\star_{i,k} + \frac{1}{p_i} \delta_{i,k} \left(\sigma'(\widehat{M}_{i,k}) \right)^{-1} \left( y_{i,k}- \sigma(M^\star_{i,k}) \right)+ \left(1-\frac{1}{p_i} \delta_{i,k}\right)\left(\widehat{M}_{i,k}-M^\star_{i,k}\right)\nonumber \\
    & \quad  + \frac{1}{2p_i} \delta_{i,k} \left(\sigma'(\widehat{M}_{i,k}) \right)^{-1}  \sigma''(x_{i,k}) \left(   M^\star_{i,k} -  \widehat{M}_{i,k}  \right)^2 \label{eq:Mddecomposition}
\end{align}
where the penultimate equality is due to Taylor's expansion, with $x_{i,k}$ lying between $\widehat{M}_{i,k}$ and $M^\star_{i,k}$.  As a result, by taking an average over users, we can write
\begin{align*}
      \frac{1}{d_1} \sum_{i=1}^{d_1} \widehat{M}^{\mathrm{NR}}_{i,k}  = \bar{\Theta}^\star_{j}-\bar{\Theta}^\star_{j'} + \calZ + \Delta 
\end{align*}
where
\begin{align*}
    \calZ\coloneqq\frac{1}{d_1} \sum_{i=1}^{d_1}\frac{1}{p_i} \delta_{i,k} \left(\sigma'(\widehat{M}_{i,k}) \right)^{-1} \left( y_{i,k}- \sigma(M^\star_{i,k}) \right)
\end{align*}
and
\begin{align*}
    \Delta\coloneqq\frac{1}{d_1} \sum_{i=1}^{d_1} \left(1-\frac{1}{p_i} \delta_{i,k}\right)\left(\widehat{M}_{i,k}-M^\star_{i,k}\right) + \frac{1}{d_1} \sum_{i=1}^{d_1} \frac{1}{2p_i} \delta_{i,k} \left(\sigma'(\widehat{M}_{i,k}) \right)^{-1}  \sigma''(x_{i,k}) \left(   M^\star_{i,k} -  \widehat{M}_{i,k}  \right)^2.
\end{align*}

First, we derive asymptotic normality from $\calZ$ by applying Lindeberg's CLT. Define $\rho(x) \coloneqq 1/(x(1-x))$ and write
\begin{align*}
    \left(\sigma'(\widehat{M}_{i,k}) \right)^{-1} = \left(\sigma(\widehat{M}_{i,k})\left(1-\sigma(\widehat{M}_{i,k})\right) \right)^{-1} = \left(\widehat{L}_{i,k}\left(1-\widehat{L}_{i,k}\right) \right)^{-1} = \rho(\widehat{L}_{i,k}).
\end{align*}
As mentioned in the proof of Corollary \ref{cor:errorboundforM}, $\widehat{L}{i,k}$ is bounded away from 0 and 1 with probability at least $1 - O(\bar{d}^{-10})$ for all $(i,k)$. 
This proof uses the nonconvex and leave-one-out gradient descent iterates defined in Section \ref{sec:sectionA} and \ref{sec:sectionB}. Define $t_* \coloneqq \argmin_{0\leq t < t_0}\norm{\nabla f (\bfX^t,\bfY^t)}_{\mathrm{F}}$. Using them, we decompose as follows:
\begin{align*}
    \rho(\widehat{L}_{i,k})& = \rho(L^\star_{i,k}+\varepsilon_{i,k})+ \underbrace{\rho(\widehat{L}_{i,k})- \rho(\bfX_{i,\cdot}^{t_* } \bfY^{t_* \top}_{k,\cdot})}_{\coloneqq a^1_{i,k}} + \underbrace{\rho(\bfX_{i,\cdot}^{t_*} \bfY^{t_* \top}_{k,\cdot})-\rho(\bfX_{i,\cdot}^{t_*,(d_1+k)} \bfY^{t_*,(d_1+k) \top}_{k,\cdot})}_{\coloneqq a^2_{i,k}} \\
    & \quad + \underbrace{\rho(\bfX_{i,\cdot}^{t_*,(d_1+k)} \bfY^{t_*,(d_1+k) \top}_{k,\cdot})-\rho(L^\star_{i,k}+\varepsilon_{i,k})}_{\coloneqq a^3_{i,k}}.
\end{align*}
 Then, we can rewrite $\calZ$ as
 \begin{align*}
     \calZ &= \underbrace{\frac{1}{d_1} \sum_{i=1}^{d_1}\frac{1}{p_i} \delta_{i,k} \rho(L^\star_{i,k}+\varepsilon_{i,k}) \left( y_{i,k}- \sigma(M^\star_{i,k}) \right)}_{\coloneqq b^\star} + \underbrace{\frac{1}{d_1} \sum_{i=1}^{d_1}\frac{1}{p_i} \delta_{i,k}a^1_{i,k}\left( y_{i,k}- \sigma(M^\star_{i,k}) \right)}_{\coloneqq b_1} \\
     & \quad +   \underbrace{\frac{1}{d_1} \sum_{i=1}^{d_1}\frac{1}{p_i} \delta_{i,k}a^2_{i,k}\left( y_{i,k}- \sigma(M^\star_{i,k}) \right)}_{\coloneqq b_2}+  \underbrace{\frac{1}{d_1} \sum_{i=1}^{d_1}\frac{1}{p_i} \delta_{i,k}a^3_{i,k}\left( y_{i,k}- \sigma(M^\star_{i,k}) \right)}_{\coloneqq b_3}.
 \end{align*}
We aim to bound $b_1$, $b_2$, and $b_3$, and then apply Lindeberg's CLT to $b^\star.$ For $b_1$, use the mean value theorem and Cauchy--Schwarz inequality to obtain
\begin{align*}
    b_1 & = \frac{1}{d_1} \sum_{i=1}^{d_1} \rho'(x_{i,k})(\widehat{L}_{i,k}-\bfX_{i,\cdot}^{t_*} \bfY^{t_* \top}_{k,\cdot}) \frac{1}{p_i}\delta_{i,k} \left( y_{i,k}- \sigma(M^\star_{i,k}) \right) \\
    & \leq \frac{1}{d_1} \sqrt{\sum_{i=1}^{d_1}(\rho'(x_{i,k}))^2(\widehat{L}_{i,k}-\bfX_{i,\cdot}^{t_*} \bfY^{t_* \top}_{k,\cdot})^2  } \sqrt{\sum_{i=1}^{d_1} \frac{1}{p_i^2}\delta_{i,k} \left( y_{i,k}- \sigma(M^\star_{i,k}) \right)^2 } \\
    & \lesssim \frac{1}{\bar{p}d_1} \norm{ \widehat{\bfL}- \bfX^{t_*}\bfY^{t_* \top} }_{\mathrm{F}} \sqrt{\sum_{i=1}^{d_1} \delta_{i,k} \left( y_{i,k}- \sigma(M^\star_{i,k}) \right)^2 } \\
    & \lesssim \frac{1}{\sqrt{\bar{p}d_1}} \norm{ \widehat{\bfL}- \bfX^{t_*}\bfY^{t_* \top} }_{\mathrm{F}}  \lesssim  \frac{1}{\sqrt{\bar{p}d_1}} \frac{\kappa \lambda}{\bar{d}^5}  \ll \frac{1}{\sqrt{\bar{p}d_1}}
\end{align*}
where $x_{i,k}$ is between $\widehat{L}_{i,k}$ and $\bfX_{i,\cdot}^{t_*} \bfY^{t_* \top}_{k,\cdot}$, and $C,C'>0$ are some constants. The fourth line follows from the Chernoff bound (cf. Lemma \ref{lem:numberofedges}). The fifth relation is from Lemma \ref{LemmaA1} and Lemma \ref{LemmaB1}. Therefore, this bound for $b_1$ holds with probability at least $1-O(\bar{d}^{-10}).$

Before turning to $b_2$, we record the following decomposition:
\begin{align}
&\bfX_{i,\cdot}^{t_*} \bfY^{t_* \top}_{k,\cdot}-\bfX_{i,\cdot}^{t_*,(d_1+k)} \bfY^{t_*,(d_1+k) \top}_{k,\cdot} \nonumber \\
& \quad = -\bfX_{i,\cdot}^{t_*,(d_1+k)} \bfR^{t_*,(d_1+k)} (\bfY^{t_*,(d_1+k) }_{k,\cdot} \bfR^{t_*,(d_1+k)}- \bfY^{t_*}_{k,\cdot}\bfH^{t_*})^\top  + (\bfX_{i,\cdot}^{t_*}\bfH^{t_*} - \bfX_{i,\cdot}^{t_*,(d_1+k)} \bfR^{t_*,(d_1+k)})(\bfY^{t_*}_{k,\cdot}\bfH^{t_*})^\top \label{eq:XYXY}
\end{align}
where the matrices $\bfH^{t_*}$ and $\bfR^{t_*,(d_1+k)}$ are defined in Section \ref{sec:sectionA} and Section \ref{sec:sectionB}. Now, we proceed to
\begin{align*}
    b_2 &= \frac{1}{d_1} \sum_{i=1}^{d_1} \frac{1}{p_i}\delta_{i,k}(\rho(\bfX_{i,\cdot}^{t_*} \bfY^{t_* \top}_{k,\cdot})-\rho(\bfX_{i,\cdot}^{t_*,(d_1+k)} \bfY^{t_*,(d_1+k) \top}_{k,\cdot}))\left( y_{i,k}- \sigma(M^\star_{i,k}) \right)\\
    &= \underbrace{\frac{1}{d_1} \sum_{i=1}^{d_1} \frac{1}{p_i}\delta_{i,k}\rho'(\bfX_{i,\cdot}^{t_*,(d_1+k)} \bfY^{t_*,(d_1+k) \top}_{k,\cdot})(\bfX_{i,\cdot}^{t_*} \bfY^{t_* \top}_{k,\cdot}-\bfX_{i,\cdot}^{t_*,(d_1+k)} \bfY^{t_*,(d_1+k) \top}_{k,\cdot})\left( y_{i,k}- \sigma(M^\star_{i,k}) \right)}_{\coloneqq b_{21}}\\
    & \quad + \underbrace{\frac{1}{d_1} \sum_{i=1}^{d_1} \frac{1}{p_i}\delta_{i,k}\rho''(x_{i,k})(\bfX_{i,\cdot}^{t_*} \bfY^{t_* \top}_{k,\cdot}-\bfX_{i,\cdot}^{t_*,(d_1+k)} \bfY^{t_*,(d_1+k) \top}_{k,\cdot})^2\left( y_{i,k}- \sigma(M^\star_{i,k}) \right)}_{\coloneqq b_{22}}
\end{align*}
where $x_{i,k}$ lies between $\bfX_{i,\cdot}^{t_*} \bfY^{t_* \top}_{k,\cdot}$ and $\bfX_{i,\cdot}^{t_*,(d_1+k)} \bfY^{t_*,(d_1+k) \top}_{k,\cdot}$.

For $b_{22}$, note that the higher-order terms can be bounded as
\begin{align*}
    b_{22} &\lesssim \max_i \max_k \big\lvert \bfX_{i,\cdot}^{t_*} \bfY^{t_* \top}_{k,\cdot}-\bfX_{i,\cdot}^{t_*,(d_1+k)} \bfY^{t_*,(d_1+k) \top}_{k,\cdot} \big\rvert^2 \\
     &\lesssim \norm{\bfX^{t_*,(d_1+k)}}_{2, \infty}^2 \norm{\bfY^{t_*,(d_1+k) }  \bfR^{t_*,(d_1+k)}- \bfY^{t_*} \bfH^{t_*}}_{2, \infty}^2  + \norm{\bfX^{t_*}\bfH^{t_*} - \bfX^{t_*,(d_1+k)} \bfR^{t_*,(d_1+k)}}_{2, \infty}^2 \norm{\bfY^{t_*}}_{2,\infty}^2\\
     & \lesssim \norm{\bfF^\star}_{2, \infty}^2 \norm{\bfF^{t_*,(d_1+k) }  \bfR^{t_*,(d_1+k)}- \bfF^{t_*} \bfH^{t_*}}_{2, \infty}^2 \\
     & \lesssim \left(\frac{\mu R \sigma^\star_{\max}}{\min\{d_1,d_2(d_2-1)/2\}}\right)^2 \frac{\bar{d} \log(\bar{d})}{\bar{p} (\sigma^\star_{\min})^2 }  \ll \frac{1}{ \sqrt{\bar{p}d_1} }
\end{align*}
 with probability at least $1-O(\bar{d}^{-10})$, using the Chernoff bound, \eqref{eq:XYXY}, \eqref{Prelim2}, Lemma \ref{LemmaB1}, and Lemma \ref{LemmaB11}. The last relation is the assumption in Theorem \ref{thm:asymptoticnormality}. 
 
 For $b_{21}$, use \eqref{eq:XYXY} and decompose it into two terms:
 \begin{align*}
     b_{21} &=  \frac{1}{d_1}  (\bfY^{t_*,(d_1+k) }_{k,\cdot} \bfR^{t_*,(d_1+k)}- \bfY^{t_*}_{k,\cdot}\bfH^{t_*})  \sum_{i=1}^{d_1} (-\bfX_{i,\cdot}^{t_*,(d_1+k)} \bfR^{t_*,(d_1+k)} )\frac{1}{p_i} \delta_{i,k}\rho'(\bfX_{i,\cdot}^{t_*,(d_1+k)} \bfY^{t_*,(d_1+k) \top}_{k,\cdot})\\
     & \quad  \quad \times  \left( y_{i,k}- \sigma(M^\star_{i,k}) \right)\\
     & \quad + \frac{1}{d_1} (\bfY^{t_*}_{k,\cdot}\bfH^{t_*})^\top \sum_{i=1}^{d_1} (\bfX_{i,\cdot}^{t_*}\bfH^{t_*} - \bfX_{i,\cdot}^{t_*,(d_1+k)} \bfR^{t_*,(d_1+k)}) \frac{1}{p_i}\delta_{i,k}\rho'(\bfX_{i,\cdot}^{t_*,(d_1+k)} \bfY^{t_*,(d_1+k) \top}_{k,\cdot})\\
     & \quad \quad \times  \left( y_{i,k}- \sigma(M^\star_{i,k}) \right).
 \end{align*}
For the first term, say $b_{211}$, we intend to invoke the matrix Bernstein inequality \citep[][Theorem 6.1.1]{tropp:2015}. Note that
\begin{align*}
    b_{211} &\coloneqq \frac{1}{d_1}  (\bfY^{t_*,(d_1+k) }_{k,\cdot} \bfR^{t_*,(d_1+k)}- \bfY^{t_*}_{k,\cdot}\bfH^{t_*}) \sum_{i=1}^{d_1}  (-\bfX_{i,\cdot}^{t_*,(d_1+k)} \bfR^{t_*,(d_1+k)} ) \frac{1}{p_i}\delta_{i,k}\rho'(\bfX_{i,\cdot}^{t_*,(d_1+k)} \bfY^{t_*,(d_1+k) \top}_{k,\cdot}) \\
    & \quad \times \left( y_{i,k}- \sigma(M^\star_{i,k}) \right) \\
    & \leq \frac{1}{d_1}  \norm{\bfY^{t_*,(d_1+k) }_{k,\cdot} \bfR^{t_*,(d_1+k)}- \bfY^{t_*}_{k,\cdot}\bfH^{t_*}}\\
    & \quad \times \norm{\sum_{i=1}^{d_1} \underbrace{(-\bfX_{i,\cdot}^{t_*,(d_1+k)} \bfR^{t_*,(d_1+k)} )\frac{1}{p_i}\delta_{i,k}\rho'(\bfX_{i,\cdot}^{t_*,(d_1+k)} \bfY^{t_*,(d_1+k) \top}_{k,\cdot})   \left( y_{i,k}- \sigma(M^\star_{i,k}) \right)}_{\coloneqq \bfS_i}}.
\end{align*}
 Note that $\bbE [\bfS_i]=\bf0$, and $L \coloneqq \max_i \norm{\bfS_i} \lesssim \bar{p}^{-1}\norm{\bfX^\star}_{2, \infty}$ by the lemmas in Section \ref{sec:sectionB}. Also, we have 
\begin{align*}
    V \coloneqq \max\bigg\{\norm{\sum_{i=1}^{d_1} \bbE \bfS_i^\top \bfS_i}, \norm{\sum_{i=1}^{d_1} \bbE \bfS_i \bfS_i^\top}  \bigg\} \lesssim  \frac{1}{\bar{p}} \norm{\bfX^\star}_{\mathrm{F}}^2.
\end{align*}
Therefore, we have 
\begin{align*}
    \sum_{i=1}^{d_1} S_i \lesssim \sqrt{V \log (d_1)}+ L \log (d_1) \lesssim \sqrt{\frac{d_1}{\bar{p}} \log (d_1)}\norm{\bfX^\star}_{2, \infty}
\end{align*}
with probability at least $1-O(\bar{d}^{-10}).$  The lemmas in Section \ref{sec:sectionB} and \eqref{Prelim2} yield
\begin{align*}
    b_{211} &\lesssim \frac{1}{d_1} \norm{\bfY^{t_*,(d_1+k) } \bfR^{t_*,(d_1+k)}- \bfY^{t_*}\bfH^{t_*}}_{\mathrm{F}} \sqrt{\frac{d_1}{\bar{p}} \log (\bar{d})}\norm{\bfX^\star}_{2, \infty} \\
    & \lesssim     \sqrt{\frac{\bar{d} \log(\bar{d})}{\bar{p} (\sigma^\star_{\min})^2 }}\sqrt{ \frac{\log (\bar{d})}{\bar{p}d_1}}\norm{\bfF^\star}_{2, \infty}^2 \ll \frac{1}{\sqrt{\bar{p}d_1}}
\end{align*}
with probability at least $1-O(\bar{d}^{-10}).$ For the second term, say $b_{212}$, use the lemmas in Section \ref{sec:sectionB} and Cauchy--Schwarz inequality to obtain
\begin{align*}
  b_{212} &\coloneqq  \frac{1}{d_1} (\bfY^{t_*}_{k,\cdot}\bfH^{t_*})^\top \sum_{i=1}^{d_1} (\bfX_{i,\cdot}^{t_*}\bfH^{t_*} - \bfX_{i,\cdot}^{t_*,(d_1+k)} \bfR^{t_*,(d_1+k)}) \frac{1}{p_i}\delta_{i,k}\rho'(\bfX_{i,\cdot}^{t_*,(d_1+k)} \bfY^{t_*,(d_1+k) \top}_{k,\cdot})\\
     & \quad \quad \times  \left( y_{i,k}- \sigma(M^\star_{i,k}) \right)\\
     & \lesssim \frac{1}{d_1}\norm{\bfY^\star}_{2, \infty}\sqrt{\sum_{i=1}^{d_1} \left(\bfX_{i,\cdot}^{t_*}\bfH^{t_*} - \bfX_{i,\cdot}^{t_*,(d_1+k)} \bfR^{t_*,(d_1+k)}\right)^2 }\\
     & \quad \times \sqrt{ \sum_{i=1}^{d_1}\frac{1}{p_i^2} \delta_{i,k}\left(\rho'(\bfX_{i,\cdot}^{t_*,(d_1+k)} \bfY^{t_*,(d_1+k) \top}_{k,\cdot}) \right)^2 \left( y_{i,k}- \sigma(M^\star_{i,k}) \right)^2 }\\
     & \lesssim \frac{1}{\sqrt{\bar{p}d_1}}  \sqrt{\frac{\bar{d} \log(\bar{d})}{\bar{p} (\sigma^\star_{\min})^2 }}\norm{\bfF^\star}_{2, \infty}^2  \ll \frac{1}{\sqrt{\bar{p}d_1}}.
\end{align*}

For $b_3$, we invoke the Bernstein inequality \citep[][Theorem 6.1.1]{tropp:2015}. Note that
\begin{align*}
    b_3 = \frac{1}{d_1} \sum_{i=1}^{d_1}  \underbrace{(\rho(\bfX_{i,\cdot}^{t_*,(d_1+k)} \bfY^{t_*,(d_1+k) \top}_{k,\cdot})-\rho(L^\star_{i,k}+\varepsilon_{i,k}))\frac{1}{p_i}\delta_{i,k}\left( y_{i,k}- \sigma(M^\star_{i,k}) \right) }_{\coloneqq s_i}
\end{align*}
with $\bbE[s_j]=0.$ The mean value theorem, the lemmas in Section \ref{sec:sectionB}, and the small sieve approximation error assumption yield
\begin{align*}
   & \rho(\bfX_{i,\cdot}^{t_*,(d_1+k)} \bfY^{t_*,(d_1+k) \top}_{k,\cdot})-\rho(L^\star_{i,k}+\varepsilon_{i,k}) = \rho'(x_{i,k}) (\bfX_{i,\cdot}^{t_*,(d_1+k)} \bfY^{t_*,(d_1+k) \top}_{k,\cdot}- L^\star_{i,k}-\varepsilon_{i,k} ) \\
   & \quad \lesssim  \frac{\kappa \mu R \sigma^\star_{\max} }{\min\{d_1, d_2(d_2-1)/2\}} \sqrt{\frac{  \bar{d}\log(\bar{d})}{\bar{p} (\sigma^\star_{\min})^2 }} +\frac{1}{\sqrt{\max\{d_1, d_2(d_2-1)/2\}}}\\
   & \quad \lesssim \frac{\kappa \mu R \sigma^\star_{\max} }{\min\{d_1, d_2(d_2-1)/2\}} \sqrt{\frac{  \bar{d}\log(\bar{d})}{\bar{p} (\sigma^\star_{\min})^2 }}
\end{align*}
where $x_{i,k}$ is lying between $\bfX_{i,\cdot}^{t_*,(d_1+k)} \bfY^{t_*,(d_1+k) \top}_{k,\cdot}$ and $L^\star_{i,k}+\varepsilon_{i,k}$. Therefore, with probability at least $1-O(\bar{d}^{-10})$,
\begin{align*}
    L \coloneqq \max_i \norm{s_i} \lesssim \frac{\kappa \mu R \sigma^\star_{\max} }{\min\{d_1, d_2(d_2-1)/2\}} \sqrt{\frac{  \bar{d}\log(\bar{d})}{\bar{p} (\sigma^\star_{\min})^2 }}
\end{align*}
and
\begin{align*}
    V \coloneqq  \norm{\sum_{i=1}^{d_1} \bbE s_i^2} &\lesssim \frac{d_1}{\bar{p}}  \left(\frac{\kappa \mu R \sigma^\star_{\max} }{\min\{d_1, d_2(d_2-1)/2\}} \sqrt{\frac{  \bar{d}\log(\bar{d})}{\bar{p} (\sigma^\star_{\min})^2 }} \right)^2,
\end{align*}
which leads to
\begin{align*}
    \norm{ \sum_{i=1}^{d_1} s_i} \lesssim \sqrt{\frac{d_1}{\bar{p}} \log (\bar{d})}  \frac{\kappa \mu R \sigma^\star_{\max} }{\min\{d_1, d_2(d_2-1)/2\}} \sqrt{\frac{  \bar{d}\log(\bar{d})}{\bar{p}(\sigma^\star_{\min})^2 }} 
\end{align*}

Therefore, 
\begin{align*}
    b_3 &\lesssim  \sqrt{\frac{\log (\bar{d})}{\bar{p}d_1}}  \frac{\kappa \mu R \sigma^\star_{\max} }{\min\{d_1, d_2(d_2-1)/2\}} \sqrt{\frac{  \bar{d}\log(\bar{d})}{\bar{p} (\sigma^\star_{\min})^2 }}  \ll \frac{1}{\sqrt{\bar{p}d_1}}.
\end{align*}

The bounds on $b_1$, $b_2$, and $b_3$ yield, with probability at least $1-O(\bar{d}^{-10}),$
\begin{align*}
    \calZ = \frac{1}{d_1} \sum_{i=1}^{d_1} \frac{1}{p_i}\delta_{i,k} \rho(L^\star_{i,k}+\varepsilon_{i,k}) \left( y_{i,k}- \sigma(M^\star_{i,k}) \right) + o(\frac{1}{\sqrt{\bar{p}d_1}}).
\end{align*}

We now establish that $\Delta$ is negligible. Recall its definition
\begin{align*}
    \Delta=\underbrace{\frac{1}{d_1} \sum_{i=1}^{d_1} \left(1-\frac{1}{p_i} \delta_{i,k}\right)\left(\widehat{M}_{i,k}-M^\star_{i,k}\right)}_{\coloneqq \Delta_{1}} + \underbrace{\frac{1}{d_1} \sum_{i=1}^{d_1} \frac{1}{2p_i} \delta_{i,k} \left(\sigma'(\widehat{M}_{i,k}) \right)^{-1}  \sigma''(x_{i,k}) \left(   M^\star_{i,k} -  \widehat{M}_{i,k}  \right)^2}_{\coloneqq \Delta_{2}}.
\end{align*}
For the second term $\Delta_{2}$, the Chernoff bound and Corollary \ref{cor:errorboundforM} imply, with probability at least $1-O(\bar{d}^{-10})$,
\begin{align*}
    \Delta_{2} &\lesssim \norm{\widehat{\bfM}-\bfM^{\star}}_{\infty}^2  \lesssim \left( \frac{\kappa \mu R \sigma^\star_{\max} }{\min\{d_1, d_2(d_2-1)/2\}} \sqrt{\frac{  \bar{d}\log(\bar{d})}{\bar{p} (\sigma^\star_{\min})^2 }}\right)^2 \ll \frac{1}{\sqrt{\bar{p}d_1}}.
\end{align*}

Now, we consider $\Delta_{1}.$ Denote $s(x)= \sigma^{-1}(x)$ and apply the Taylor expansion
\begin{align*}
   \Delta_{1} &=  \frac{1}{d_1} \sum_{i=1}^{d_1} \left(1-\frac{1}{p_i}\delta_{i,k}\right)\left(\widehat{M}_{i,k}-M^\star_{i,k}\right) = \frac{1}{d_1} \sum_{i=1}^{d_1} \left(1-\frac{1}{p_i}\delta_{i,k}\right)\left(s(\widehat{L}_{i,k})-s(L^\star_{i,k}+\varepsilon_{i,k})\right) \\
   & = \underbrace{\frac{1}{d_1} \sum_{i=1}^{d_1} \left(1-\frac{1}{p_i}\delta_{i,k}\right) s'(L^\star_{i,k}+\varepsilon_{i,k})
    \left(\widehat{L}_{i,k}-L^\star_{i,k}-\varepsilon_{i,k}\right)}_{\coloneqq \Delta_{11}}\\
    & \quad + \underbrace{\frac{1}{2d_1} \sum_{i=1}^{d_1} \left(1-\frac{1}{p_i}\delta_{i,k}\right) s''(x_{i,k})
    \left(\widehat{L}_{i,k}-L^\star_{i,k}-\varepsilon_{i,k}\right)^2}_{\coloneqq \Delta_{12}}
\end{align*}
where $x_{i,k}$ lies between $\widehat{L}_{i,k}$ and $L^\star_{i,k}+\varepsilon_{i,k}$. As mentioned in the proof of Corollary \ref{cor:errorboundforM}, $s'(\cdot)$ and $s''(\cdot)$ are bounded with high probability in our analysis. We can bound the higher order term $\Delta_{12}$ by applying the Chernoff bound, with probability at least $1-O(\bar{d}^{-10}),$
\begin{align*}
  \Delta_{12} &\coloneqq  \frac{1}{2 d_1} \sum_{i=1}^{d_1}  s''(x_{i,k})
    \left(\widehat{L}_{i,k}-L^\star_{i,k}-\varepsilon_{i,k}\right)^2 -\frac{1}{2d_1} \sum_{i=1}^{d_1} \frac{1}{p_i}\delta_{i,k} s''(x_{i,k})
    \left(\widehat{L}_{i,k}-L^\star_{i,k}-\varepsilon_{i,k}\right)^2 \\
    & \lesssim  \norm{\widehat{\bfL}-\bfL^\star}_{\infty}^2 +\norm{\bm\calE}_{\infty}^2 \\
    & \lesssim  \left( \frac{\kappa \mu R \sigma^\star_{\max} }{\min\{d_1, d_2(d_2-1)/2\}} \sqrt{\frac{  \bar{d}\log(\bar{d})}{\bar{p} (\sigma^\star_{\min})^2 }}\right)^2 + \frac{1}{\max\{d_1, d_2(d_2-1)/2\}} \\
    & \ll \frac{1}{\sqrt{\bar{p}d_1}}.
\end{align*}

For the term $\Delta_{11}$, we decompose it into three terms and bound them separately.
\begin{align*}
    \Delta_{11} \coloneqq& \frac{1}{d_1} \sum_{i=1}^{d_1} \left(1-\frac{1}{p_i}\delta_{i,k}\right) s'(L^\star_{i,k}+\varepsilon_{i,k})
    \left(\widehat{L}_{i,k}-L^\star_{i,k}-\varepsilon_{i,k}\right) \\
    =& \underbrace{\frac{1}{d_1} \sum_{i=1}^{d_1} \left(1-\frac{1}{p_i}\delta_{i,k}\right) s'(L^\star_{i,k}+\varepsilon_{i,k})
    \left(\widehat{L}_{i,k}-\bfX_{i,\cdot}^{t_*}\bfY^{t_* \top}_{k,\cdot}\right)}_{\coloneqq \Delta_{111}}\\
     &  + \underbrace{\frac{1}{d_1} \sum_{i=1}^{d_1} \left(1-\frac{1}{p_i}\delta_{i,k}\right) s'(L^\star_{i,k}+\varepsilon_{i,k})
    \left(\bfX_{i,\cdot}^{t_*}\bfY^{t_* \top}_{k,\cdot}-\bfX_{i,\cdot}^{t_*,(d_1+k)}\bfY^{t_*,(d_1+k) \top}_{k,\cdot}\right)}_{\coloneqq \Delta_{112}}\\
     &   + \underbrace{\frac{1}{d_1} \sum_{i=1}^{d_1} \left(1-\frac{1}{p_i}\delta_{i,k}\right) s'(L^\star_{i,k}+\varepsilon_{i,k})
    \left( \bfX_{i,\cdot}^{t_*,(d_1+k)}\bfY^{t_*,(d_1+k) \top}_{k,\cdot}- L^\star_{i,k}-\varepsilon_{i,k}\right)}_{\coloneqq \Delta_{113}}.
\end{align*}
 First, use Cauchy--Schwarz inequality and have, with probability at least $1-O(\bar{d}^{-10})$,
\begin{align*}
    \Delta_{111} &= \frac{1}{d_1} \sum_{i=1}^{d_1} \left(1-\frac{1}{p_i}\delta_{i,k}\right) s'(L^\star_{i,k}+\varepsilon_{i,k})
    \left(\widehat{L}_{i,k}-\bfX_{i,\cdot}^{t_*}\bfY^{t_* \top}_{k,\cdot}\right) \\
    & \leq \frac{1}{d_1} \sqrt{ \sum_{i=1}^{d_1} \left(1-\frac{1}{p_i}\delta_{i,k}\right)^2 } \sqrt{\sum_{i=1}^{d_1} \left(s'(L^\star_{i,k}+\varepsilon_{i,k})\right)^2
    \left(\widehat{L}_{i,k}-\bfX_{i,\cdot}^{t_*}\bfY^{t_* \top}_{k,\cdot}\right)^2 } \\
    & \lesssim     \frac{1}{\sqrt{\bar{p}d_1}}\norm{\widehat{\bfL}- \bfX^{t_*}\bfY^{t_* \top}}_{\mathrm{F}}  \lesssim \frac{1}{\sqrt{\bar{p}d_1}} \frac{\kappa \lambda}{\bar{d}^5} \ll \frac{1}{\sqrt{\bar{p}d_1}}
\end{align*}
where the fourth relation follows from Lemma \ref{LemmaA1} and Lemma \ref{LemmaB1}.

Using \eqref{eq:XYXY}, decompose $\Delta_{112}$ further and write
\begin{align*}
    \Delta_{112}&=\frac{1}{d_1} \sum_{i=1}^{d_1} \left(1-\frac{1}{p_i}\delta_{i,k}\right) s'(L^\star_{i,k}+\varepsilon_{i,k})
    \left(\bfX_{i,\cdot}^{t_*}\bfY^{t_* \top}_{k,\cdot}-\bfX_{i,\cdot}^{t_*,(d_1+k)}\bfY^{t_*,(d_1+k) \top}_{k,\cdot}\right) \\
    &=\underbrace{\frac{1}{d_1} \sum_{i=1}^{d_1} \left(1-\frac{1}{p_i}\delta_{i,k}\right) s'(L^\star_{i,k}+\varepsilon_{i,k})
    \left(-\bfX_{i,\cdot}^{t_*,(d_1+k)} \bfR^{t_*,(d_1+k)} \right)\left(\bfY^{t_*,(d_1+k) }_{k,\cdot} \bfR^{t_*,(d_1+k)}- \bfY^{t_*}_{k,\cdot}\bfH^{t_*}\right)^\top}_{\coloneqq \Delta_{1121}} \\
    & \quad +\underbrace{\frac{1}{d_1} \sum_{i=1}^{d_1} \left(1-\frac{1}{p_i}\delta_{i,k}\right) s'(L^\star_{i,k}+\varepsilon_{i,k})
    \left(\bfX_{i,\cdot}^{t_*}\bfH^{t_*} - \bfX_{i,\cdot}^{t_*,(d_1+k)} \bfR^{t_*,(d_1+k)}\right)\left(\bfY^{t_*}_{k,\cdot}\bfH^{t_*}\right)^\top}_{\coloneqq \Delta_{1122}} .
\end{align*}
For the first term, note that
\begin{align*}
    \Delta_{1121} \leq \frac{1}{d_1} \norm{\bfY^{t_*,(d_1+k) }_{k,\cdot} \bfR^{t_*,(d_1+k)}- \bfY^{t_*}_{k,\cdot}\bfH^{t_*}}\norm{\sum_{i=1}^{d_1} \underbrace{\left(1-\frac{1}{p_i}\delta_{i,k}\right) s'(L^\star_{i,k}+\varepsilon_{i,k})\bfX_{i,\cdot}^{t_*,(d_1+k)}}_{\coloneqq \bfS_i}}
\end{align*}
where we intend to apply the matrix Bernstein inequality \citep[][Theorem 6.1.1]{tropp:2015} for $\norm{\sum_{i=1}^{d_1}\bfS_i}$. Note that $\bbE [\bfS_i]=0$ and $L \coloneqq \max_i \norm{\bfS_i} \lesssim \bar{p}^{-1} \norm{\bfX^\star}_{2, \infty}$ by the lemmas in Section \ref{sec:sectionB}, with probability at least $1-O(\bar{d}^{-10}).$ Also, 
\begin{align*}
    V \coloneqq \max\bigg\{\norm{\sum_{i=1}^{d_1} \bbE \bfS_i^\top \bfS_i}, \norm{\sum_{i=1}^{d_1}\bbE \bfS_i \bfS_i^\top}  \bigg\} \lesssim \frac{1}{\bar{p}} \norm{\bfX^\star}_{\mathrm{F}}^2.
\end{align*}
Therefore, we have, with probability at least $1-O(\bar{d}^{-10})$
\begin{align*}
    \sum_{i=1}^{d_1} \bfS_i \lesssim \sqrt{\frac{d_1}{\bar{p}} \log (\bar{d})}\norm{\bfX^\star}_{2, \infty}.
\end{align*}
Invoke the lemmas in Section \ref{sec:sectionB} and \eqref{Prelim2}, 
\begin{align*}
   \Delta_{1121} &\lesssim \frac{1}{d_1} \norm{\bfY^{t_*,(d_1+k) } \bfR^{t_*,(d_1+k)}- \bfY^{t_*}\bfH^{t_*}}_{\mathrm{F}}\sqrt{\frac{d_1}{\bar{p}} \log (\bar{d})}\norm{\bfX^\star}_{2, \infty} \\
    & \lesssim \sqrt{\frac{\log (\bar{d})}{\bar{p}d_1} }   \sqrt{\frac{\bar{d} \log(\bar{d})}{\bar{p} (\sigma^\star_{\min})^2 }} \norm{\bfF^\star}_{2, \infty}^2 \ll \frac{1}{\sqrt{\bar{p}d_1}}.
\end{align*}
For the term $\Delta_{1122}$, applying the Cauchy--Schwarz inequality, the Chernoff bound, and the lemmas in Section \ref{sec:sectionB} and \eqref{Prelim2}, we have
\begin{align*}
    \Delta_{1122} &\coloneqq \frac{1}{d_1} \sum_{i=1}^{d_1} \left(1-\frac{1}{p_i}\delta_{i,k}\right) s'(L^\star_{i,k}+\varepsilon_{i,k})
    (\bfX_{i,\cdot}^{t_*}\bfH^{t_*} - \bfX_{i,\cdot}^{t_*,(d_1+k)} \bfR^{t_*,(d_1+k)})(\bfY^{t_*}_{k,\cdot}\bfH^{t_*})^\top  \\
    & \leq \frac{1}{d_1} \norm{\bfY^\star}_{2, \infty}\sqrt{\sum_{i=1}^{d_1} \left(1-\frac{1}{p_i}\delta_{i,k}\right)^2  } \sqrt{\sum_{i=1}^{d_1} \left(\bfX_{i,\cdot}^{t_*}\bfH^{t_*} - \bfX_{i,\cdot}^{t_*,(d_1+k)} \bfR^{t_*,(d_1+k)}\right)^2  } \\
    & \leq \frac{1}{\sqrt{\bar{p}d_1}}  \norm{\bfF^{t_*}\bfH^{t_*} - \bfF^{t_*,(d_1+k)} \bfR^{t_*,(d_1+k)}}_{\mathrm{F}}\norm{\bfF^\star}_{2, \infty}\\
    & \lesssim \frac{1}{\sqrt{\bar{p}d_1}}   \sqrt{\frac{\bar{d} \log(\bar{d})}{\bar{p} (\sigma^\star_{\min})^2 }}\norm{\bfF^\star}_{2, \infty}^2   \ll \frac{1}{\sqrt{\bar{p}d_1}},
\end{align*}
with probability at least $1-O(\bar{d}^{-10})$.

For $\Delta_{113}$, we invoke the Bernstein inequality \citep{tropp:2015}. Note that
\begin{align*}
   \Delta_{113} = \frac{1}{d_1} \sum_{i=1}^{d_1} \underbrace{\left(1-\frac{1}{p_i}\delta_{i,k}\right) s'(L^\star_{i,k}+\varepsilon_{i,k})
    \left( \bfX_{i,\cdot}^{t_*,(d_1+k)}\bfY^{t_*,(d_1+k) \top}_{k,\cdot}- L^\star_{i,k}-\varepsilon_{i,k}\right)}_{\coloneqq s_i}
\end{align*}
with $\bbE[s_i]=0$ and 
\begin{align*}
    L \coloneqq \max_i \norm{s_i}& \lesssim \frac{1}{\bar{p}}\frac{\kappa \mu R \sigma^\star_{\max} }{\min\{d_1, d_2(d_2-1)/2\}} \sqrt{\frac{  \bar{d}\log(\bar{d})}{\bar{p} (\sigma^\star_{\min})^2 }}
\end{align*}  
by the lemmas in Section \ref{sec:sectionB}, with probability at least $1-O(\bar{d}^{-10}).$ Also, 
\begin{align*}
    V \coloneqq   \norm{\sum_{i=1}^{d_1} \bbE s_i^2} &\lesssim \frac{1}{\bar{p}} d_1 \left(\frac{\kappa \mu R \sigma^\star_{\max} }{\min\{d_1, d_2(d_2-1)/2\}} \sqrt{\frac{  \bar{d}\log(\bar{d})}{\bar{p} (\sigma^\star_{\min})^2 }} \right)^2  ,
\end{align*}
which leads to
\begin{align*}
     \norm{\sum_{i=1}^{d_1} s_i} \lesssim \sqrt{\frac{d_1\log (\bar{d})}{\bar{p}} }  \frac{\kappa\mu R \sigma^\star_{\max} }{\min\{d_1, d_2(d_2-1)/2\}} \sqrt{\frac{  \bar{d}\log(\bar{d})}{\bar{p}(\sigma^\star_{\min})^2 }}, 
\end{align*}
with probability at least $1-O(\bar{d}^{-10}).$ Therefore, 
\begin{align*}
    b_3 &\lesssim \frac{1}{d_1}  \sqrt{\frac{d_1\log (\bar{d})}{\bar{p}} }  \frac{\kappa\mu R \sigma^\star_{\max} }{\min\{d_1, d_2(d_2-1)/2\}} \sqrt{\frac{  \bar{d}\log(\bar{d})}{\bar{p}(\sigma^\star_{\min})^2 }}  \ll \frac{1}{\sqrt{\bar{p}d_1}},
\end{align*}
with probability at least $1-O(\bar{d}^{-10}).$ 

Collecting the bounds for $\calZ$ and $\Delta$, we reach, with probability at least $1-O(\bar{d}^{-10})$
\begin{align*}
     \calZ+\Delta = \sum_{i=1}^{d_1} \underbrace{ \frac{1}{p_id_1}\delta_{i,k} \rho(L^\star_{i,k}+\varepsilon_{i,k}) \left( y_{i,k}- \sigma(M^\star_{i,k}) \right)  }_{\coloneqq Z_i}  + o(\frac{1}{\sqrt{\bar{p}d_1}}).
\end{align*}

Denote $\Omega \coloneqq \{\delta_{i,k}\}_{i\in[d_1], k\in[d_2(d_2-1)/2]}.$ It is easy to see that, with probability at least $1-O(\bar{d}^{-10}),$
\begin{align*}
    v_k^\star \coloneqq \sum_{i=1}^{d_1} \Var (Z_i|\Omega) = \sum_{i=1}^{d_1} \frac{1}{p_i^2 d_1^2} \delta_{i,k} \frac{1}{\sigma(M^\star_{i,k})(1-\sigma(M^\star_{i,k}))}     \asymp \frac{1}{ \bar{p}d_1}.
\end{align*}
Use Cauchy--Schwarz inequality and observe that for any $\epsilon>0,$
\begin{align*}
    (v^\star_k)^{-1}\sum_{i=1}^{d_1} \bbE [Z_i^2 \bbOne\{|Z_i|>\epsilon \sqrt{v^\star_k}\}|\Omega] &\leq  (v^\star_k)^{-1}\sum_{i=1}^{d_1} \sqrt{ \bbE [Z_i^4|\Omega] \bbE[\bbOne\{|Z_i|>\epsilon \sqrt{v^\star_k}\}|\Omega]  } .
\end{align*}
Note that for large $d_1$ (and thus large $\bar{p}d_1$), $\bbOne\{|Z_i|>\epsilon \sqrt{v^\star_k}\}=0$ for any $\epsilon>0.$ Then, by Linbeberg's CLT, conditioning on $\Omega$, we have $\sum_{i=1}^{d_1} Z_i \overset{d}{\rightarrow}\calN(0,v^\star_k).$ The unconditional CLT follows from the dominated convergence theorem. 
\end{proof}

\subsection{Proof of Proposition \ref{prop:feasibleCLT}}
\begin{proof}[Proof of Proposition \ref{prop:feasibleCLT}]
    It is enough to show that $\widehat{v}_k-v^\star_k = o(v^\star_k)=o(1/(\bar{p}d_1))$ with high probability. We begin by showing that $v^{\star}_{k}$ is close to $\bbE [v^{\star}_{k}]$ using the Bernstein inequality. Define
    \begin{align*}
        v^{\star}_{k}-\bbE [v^{\star}_{k}]=  \sum_{i=1}^{d_1}\underbrace{\frac{1}{d_1^2 p_{i}^2} (\delta_{i,k}-p_{i}) \left(\frac{1}{\sigma(M^\star_{i,k})(1-\sigma(M^\star_{i,k}))}\right)}_{\coloneqq s_{i}}
    \end{align*}
and note that $\bbE s_{i}=0$. Also, since $1/(\sigma(M^\star_{i,k})(1-\sigma(M^\star_{i,k})))$ is bounded, we have
\begin{align*}
    L \coloneqq \max_{i} \norm{s_{i}} \lesssim \frac{1}{\bar{p}^2 d_1^2} \quad   \text{and}    \quad
    V \coloneqq \norm{\sum_{i=1}^{d_1}\bbE s_{i}^2 } \lesssim \frac{1}{\bar{p}^3 d_1^3}.
\end{align*}
Therefore, with probability at least $1-O(\bar{d}^{-10}),$
\begin{align*}
   v^{\star}_{k}-\bbE [v^{\star}_{k}]=  \sum_{i=1}^{d_1}s_{i} \lesssim \frac{1}{\bar{p}d_1}\left( \sqrt{\frac{ \log (\bar{d})}{\bar{p}d_1}} + \frac{\log (\bar{d})}{\bar{p}d_1}\right) \ll \frac{1}{\bar{p}d_1}.
\end{align*}

    We now show that $\widehat{v}_k$ and $\bbE[v^\star_k]$ are close. By definition,
\begin{align*}
    \widehat{v}_k -\bbE[v^\star_k] &= \frac{1}{d_1^2}  \sum_{i=1}^{d_1} \frac{1}{p_i}   \left(\frac{1}{\sigma(\widehat{M}_{i,k})(1-\sigma(\widehat{M}_{i,k}))} - \frac{1}{\sigma(M^\star_{i,k})(1-\sigma(M^\star_{i,k}))} \right)\\
    &= \frac{1}{d_1^2}  \sum_{i=1}^{d_1} \frac{1}{p_i} (\rho(\widehat{L}_{i,k})-\rho(L^\star_{i,k}+\varepsilon_{i,k}) )
\end{align*}
where $\rho(x)=(x(1-x))^{-1}$. As mentioned in the proof of Corollary \ref{cor:errorboundforM}, $\widehat{L}_{i,k}$, $L^\star_{i,k}$, and $L^\star_{i,k}+\varepsilon_{i,k}$ are all bounded away from 0 and 1, with exceedingly high probability, implying that $\rho'(\cdot)$ is bounded in our analysis. By applying the mean value theorem and the entrywise error control in Theorem \ref{thm:errorboundsforL}, with probability at least $1-O(\bar{d}^{-10}),$ uniformly for all $(i,k)$,
\begin{align*}
    \rho(\widehat{L}_{i,k})-\rho(L^\star_{i,k}+\varepsilon_{i,k}) &= \rho'(x_{i,k}) (\widehat{L}_{i,k}-L^\star_{i,k}-\varepsilon_{i,k})\lesssim \norm{\widehat{\bfL}-\bfL^\star}_{\infty}+ \norm{\bm\calE}_{\infty} \\
    &\lesssim \frac{\kappa \mu R \sigma^\star_{\max} }{\min\{d_1, d_2(d_2-1)/2\}} \sqrt{\frac{  \bar{d}\log(\bar{d})}{\bar{p} (\sigma^\star_{\min})^2 }}+\frac{1}{\sqrt{\max\{d_1, d_2(d_2-1)/2\}}} \\
    & \lesssim \frac{\kappa \mu R \sigma^\star_{\max} }{\min\{d_1, d_2(d_2-1)/2\}} \sqrt{\frac{  \bar{d}\log(\bar{d})}{\bar{p} (\sigma^\star_{\min})^2 }}  
\end{align*}
for some $x_{i,k}$ lying between $\widehat{L}_{i,k}$ and $L^\star_{i,k}+\varepsilon_{i,k}$. Therefore, we have 
\begin{align*}
    \widehat{v}_k=v^\star_k + O\left(\frac{1}{\bar{p}d_1} \frac{\kappa \mu R \sigma^\star_{\max} }{\min\{d_1, d_2(d_2-1)/2\}} \sqrt{\frac{  \bar{d}\log(\bar{d})}{\bar{p} (\sigma^\star_{\min})^2 }}\right) = v^\star_k + o\left(\frac{1}{\bar{p}d_1} \right) = v^\star_k + o(v^\star_k)
\end{align*}
with probability at least $1-O(\bar{d}^{-10})$.
\end{proof}

\subsection{Proof of Theorem \ref{thm:asymptoticnormality_indiv}}\label{sec:proofofindivUQ}

We begin by presenting the additional assumptions imposed for the uncertainty quantification of individual preferences. Recall that the rank $q$ is selected such that Assumption \ref{asp:approx_low_rank_M} and \ref{asp:SSVforM} are satisfied. Let $\mu(\bfM^\star)$ denote the incoherence parameter of $\bfM^\star$, which satisfies
\begin{align*}
    \norm{\bfU^{\bfM^\star}}_{2, \infty}^2 \leq \frac{\mu(\bfM^\star) q}{ d_1}\quad \text{and} \quad \norm{\bfV^{\bfM^\star}}_{2, \infty}^2 \leq \frac{\mu(\bfM^\star) q}{ d_2(d_2-1)/2}.
\end{align*}

\begin{assumption}\label{asp:inference_indiv}
    \begin{align*}
        w^{\star,1}_{i,k} \asymp \frac{1}{\bar{p} d_2(d_2-1)/2};\quad  w^{\star,2}_{i,k} \asymp \frac{1}{\bar{p} d_1};\\
             \sqrt{\mu(\bfM^\star)q} \left(\frac{\kappa \mu R \sigma^\star_{\max} }{\min\{d_1, d_2(d_2-1)/2\}} \sqrt{\frac{  \bar{d}\log(\bar{d})}{\bar{p} (\sigma^\star_{\min})^2 }} \right)^2  \ll \sqrt{w^\star_{i,k}} ;\\
         \sqrt{ \frac{1}{\bar{p}}\frac{(\mu(\bfM^\star))^2q^2\log (\bar{d})}{\min\{d_1, d_2(d_2-1)/2\}} }  \frac{\kappa \mu R \sigma^\star_{\max} }{\min\{d_1, d_2(d_2-1)/2\}} \sqrt{\frac{  \bar{d}\log(\bar{d})}{\bar{p} (\sigma^\star_{\min})^2 }}\ll \sqrt{w^\star_{i,k}} ; \\
         \sigma_{\max}(\bfM^\star)    \frac{\kappa^2 R\bar{d}}{\bar{p}(\sigma_{\min}(\bfM^\star))^2  } \frac{\mu(\bfM^\star)q}{\min\{d_1, d_2(d_2-1)/2\}} \ll \sqrt{w^\star_{i,k}} ;\\
         \frac{\kappa^2 R\bar{d}}{\bar{p}\sigma_{\min}(\bfM^\star)  }    \sqrt{\frac{(\mu(\bfM^\star))^2q^2}{ d_1 d_2(d_2-1)/2 }}\ll \sqrt{w^\star_{i,k}}.
    \end{align*}
\end{assumption}

\begin{assumption}\label{asp:forvarianceestimation}
\begin{align*}
 \frac{(\mu(\bfM^\star))^2q^2\kappa^2 R\log (\bar{d}) \bar{d}}{\bar{p}(\sigma_{\min}(\bfM^\star))^2  } \frac{\max\{d_1,d_2(d_2-1)/2\}}{\min\{d_1,d_2(d_2-1)/2\}} \ll 1;\\
 \frac{(\mu(\bfM^\star))^4 q^4\log^2 (\bar{d})}{\bar{p}\min \{d_1, d_2(d_2-1)/2\} } \ll 1;\\
 \frac{(\mu(\bfM^\star))^2 q^2\kappa \mu R \sigma^\star_{\max} }{\min\{d_1, d_2(d_2-1)/2\}} \sqrt{\frac{  \bar{d}\log^2(\bar{d})}{\bar{p} (\sigma^\star_{\min})^2 }}\ll 1.
\end{align*}
\end{assumption}

\begin{proof}[Proof of Theorem \ref{thm:asymptoticnormality_indiv}]
For each $(i,k)$, consider the following decomposition
\begin{align}
    &\widehat{M}^{\mathrm{NR},1}_{i,k}-M^\star_{i,k} \nonumber \\ 
    &=   \underbrace{\widehat{M}^1_{i,k}-M^\star_{i,k}}_{\coloneqq \Delta^1_{1,i,k}}+ \underbrace{\frac{2}{p_i} \delta^2_{i,k} \left(\sigma'(\widehat{M}^1_{i,k}) \right)^{-1} \left( y^2_{i,k}- \sigma( M^\star_{i,k}) \right)}_{\coloneqq \Delta^1_{2,i,k}}  + \underbrace{\frac{2}{p_i} \delta^2_{i,k} \left(\sigma'(\widehat{M}^1_{i,k}) \right)^{-1} \left( \sigma( M^\star_{i,k})- \sigma(\widehat{M}^1_{i,k}) \right) }_{\coloneqq \Delta^1_{3,i,k}};\label{eq:decompositionMhatd1}\\
     &\widehat{M}^{\mathrm{NR},2}_{i,k} - M^\star_{i,k} \nonumber\\
    &=  \underbrace{\widehat{M}^2_{i,k}- M^\star_{i,k}}_{\coloneqq \Delta^2_{1,i,k}}+ \underbrace{\frac{2}{p_i} \delta^1_{i,k} \left(\sigma'(\widehat{M}^2_{i,k}) \right)^{-1} \left( y^1_{i,k}- \sigma( M^\star_{i,k}) \right) }_{\coloneqq \Delta^2_{2,i,k}} + \underbrace{\frac{2}{p_i} \delta^1_{i,k} \left(\sigma'(\widehat{M}^2_{i,k}) \right)^{-1} \left( \sigma( M^\star_{i,k})- \sigma(\widehat{M}^2_{i,k}) \right) }_{\coloneqq \Delta^2_{3,i,k}}.\label{eq:decompositionMhatd2}
\end{align}
Denote $\Delta^l=\Delta^l_1+\Delta^l_2+\Delta^l_3$ for $l=1,2.$ Also, we write the top-$q$ SVD of $\bfM^\star$ and $\widehat{\bfM}^{\mathrm{NR},l}$ ($l=1,2$) as $\bfU^{\bfM^\star} {\bf\Sigma}^{\bfM^\star} \bfV^{\bfM^\star \top}$ and $\widehat{\bfM}^{\mathrm{NR},l}$ ($l=1,2$) as $\widehat{\bfU}^{\mathrm{NR},l} \widehat{\bf\Sigma}^{\mathrm{NR},l} \widehat{\bfV}^{\mathrm{NR},l \top}$, respectively. By Eckart-Young-Mirsky Theorem, the best rank-$q$ approximation of the debiased estimator $\widehat{\bfM}^{\mathrm{NR},l}$ can be obtained as
\begin{align*}
    \widehat{\bfM}^{\proj,l} \coloneqq \argmin_{\rank(\bfM) \leq q}\norm{ \widehat{\bfM}^{\mathrm{NR},l}-\bfM}_{\mathrm{F}}^2 = \widehat{\bfU}^{\mathrm{NR},l} \widehat{\bfU}^{\mathrm{NR},l \top} \widehat{\bfM}^{\mathrm{NR},l}\widehat{\bfV}^{\mathrm{NR},l} \widehat{\bfV}^{\mathrm{NR},l\top} \quad \text{for $l=1,2.$}
\end{align*}

By the notations defined in the proof of Lemma \ref{lem:perturbedsingularvector}, for $l=1,2,$
\begin{align*}
    &\widehat{\bfW}^l\widehat{\bfW}^{l\top} \bfA^\star \widehat{\bfW}^l\widehat{\bfW}^{l\top}-\bfW^\star \bfW^{\star \top} \bfA^\star \bfW^\star \bfW^{\star \top}\\
    &= \begin{bmatrix}
        {\bf 0} & \widehat{\bfU}^{\mathrm{NR},l} \widehat{\bfU}^{\mathrm{NR},l\top} \bfM^\star \widehat{\bfV}^{\mathrm{NR},l} \widehat{\bfV}^{\mathrm{NR},l\top}-\bfM^\star\\
        (\widehat{\bfU}^{\mathrm{NR},l} \widehat{\bfU}^{\mathrm{NR},l\top} \bfM^\star \widehat{\bfV}^{\mathrm{NR},l} \widehat{\bfV}^{\mathrm{NR},l\top}-\bfM^\star)^\top &{\bf 0}
    \end{bmatrix}.
\end{align*}
We then have
\begin{align*}
    \frac{1}{2}\widehat{\bfM}^{\proj,1}_{i,k}+\frac{1}{2}\widehat{\bfM}^{\proj,2}_{i,k} - M^\star_{i,k} & = \sum_{l=1}^2\frac{1}{2}\bigg[\widehat{\bfW}^l\widehat{\bfW}^{l\top} \bfA^\star \widehat{\bfW}^l\widehat{\bfW}^{l\top}-\bfW^\star \bfW^{\star \top} \bfA^\star \bfW^\star \bfW^{\star \top}\bigg]_{i,k+d_1}\\
    &  \quad +  \sum_{l=1}^2\frac{1}{2}\bigg[  \widehat{\bfU}^{\mathrm{NR},l} \widehat{\bfU}^{\mathrm{NR},l\top}  (\widehat{\bfM}^{\mathrm{NR},l}-\bfM^\star)\widehat{\bfV}^{\mathrm{NR},l}\widehat{\bfV}^{\mathrm{NR},l} \bigg]_{i,k}
\end{align*}

From the proof of Lemma \ref{lem:perturbedsingularvector}, we can write, for $l=1,2$
\begin{align*}
    \widehat{\bfW}^l\widehat{\bfW}^{l\top}-\bfW^\star\bfW^{\star\top} = \sum_{m=1}^\infty  \calS_{\bfA^\star,m}(\widehat{\bfE}^l) 
\end{align*}
and thus
\begin{align*}
    &\widehat{\bfW}^l\widehat{\bfW}^{l\top} \bfA^\star \widehat{\bfW}^l\widehat{\bfW}^{l\top}-\bfW^\star \bfW^{\star \top} \bfA^\star \bfW^\star \bfW^{\star \top} \\
    & =\left(\calS_{\bfA^\star,1}(\widehat{\bfE}^l)\bfA^\star \bfW^\star \bfW^{\star \top} +  \bfW^\star \bfW^{\star \top} \bfA^\star \calS_{1}(\widehat{\bfE}^l) \right) \\
    & \quad + \sum_{m=2}^\infty \left(\calS_{\bfA^\star,m}(\widehat{\bfE}^l)\bfA^\star \bfW^\star \bfW^{\star \top} +  \bfW^\star \bfW^{\star \top} \bfA^\star \calS_{\bfA^\star,m}(\widehat{\bfE}^l) \right)\\
    & \quad + \left(\widehat{\bfW}^l\widehat{\bfW}^{l\top}-\bfW^\star\bfW^{\star\top} \right) \bfA^\star \left(\widehat{\bfW}^l\widehat{\bfW}^{l\top}-\bfW^\star\bfW^{\star\top} \right).
\end{align*}
Therefore, we have
\begin{align*}
    &\frac{1}{2}\widehat{\bfM}^{\proj,1}_{i,k}+\frac{1}{2}\widehat{\bfM}^{\proj,2}_{i,k}- M^\star_{i,k} \\
    & =\frac{1}{2}\sum_{l=1}^2\bigg[\calS_{\bfA^\star,1}(\widehat{\bfE}^l)\bfA^\star \bfW^\star \bfW^{\star \top} +  \bfW^\star \bfW^{\star \top} \bfA^\star \calS_{1}(\widehat{\bfE}^l) \bigg]_{i,k+d_1} \\
    & \quad + \frac{1}{2}\sum_{l=1}^2\bigg[\sum_{m=2}^\infty  \calS_{\bfA^\star,m}(\widehat{\bfE}^l)\bfA^\star \bfW^\star \bfW^{\star \top} +  \bfW^\star \bfW^{\star \top} \bfA^\star \calS_{\bfA^\star,m}(\widehat{\bfE}^l)  \bigg]_{i,k+d_1}\\
    & \quad + \frac{1}{2}\sum_{l=1}^2\bigg[\left(\widehat{\bfW}^l\widehat{\bfW}^{l\top}-\bfW^\star\bfW^{\star\top} \right) \bfA^\star \left(\widehat{\bfW}^l\widehat{\bfW}^{l\top}-\bfW^\star\bfW^{\star\top} \right)\bigg]_{i,k+d_1}.
\end{align*}

By the definition of $\calS_{\bfA^\star,1}(\widehat{\bfE}^l),$
\begin{align*}
   & \frac{1}{2}\sum_{l=1}^2\bigg[\calS_{\bfA^\star,1}(\widehat{\bfE}^l)\bfA^\star \bfW^\star \bfW^{\star \top} +  \bfW^\star \bfW^{\star \top} \bfA^\star \calS_{1}(\widehat{\bfE}^l) \bigg]_{i,k+d_1} \\
   &=\bigg[ \bfU^{\bfM^\star}_{\perp} \bfU^{\bfM^\star \top}_{\perp} \left(\frac{1}{2}\Delta^1+ \frac{1}{2}\Delta^2\right) \bfV^{\bfM^\star} \bfV^{\bfM^\star \top}\bigg]_{i,k} + \bigg[ \bfU^{\bfM^\star} \bfU^{\bfM^\star \top} \left(\frac{1}{2}\Delta^1+ \frac{1}{2}\Delta^2\right) \bfV^{\bfM^\star}_{\perp} \bfV^{\bfM^\star \top}_{\perp}\bigg]_{i,k}.
\end{align*}

Also, the following lemma follows from Lemma \ref{lem:perturbedsingularvector}.

\begin{lemma}\label{lem:negligible2}
     Under the assumptions in Theorem \ref{thm:asymptoticnormality_indiv}, we have, for $l=1,2$, with probability at least $1-O(\bar{d}^{-9})$,
     \begin{align*}
         \norm{(\widehat{\bfW}^l\widehat{\bfW}^{l\top}-\bfW^\star \bfW^{\star \top}) \bfA^\star (\widehat{\bfW}^l\widehat{\bfW}^{l\top}-\bfW^\star \bfW^{\star \top})}_{\infty}\lesssim \sigma_{\max}(\bfM^\star)    \frac{\kappa^2 R\bar{d}}{\bar{p}(\sigma_{\min}(\bfM^\star))^2  } \frac{\mu(\bfM^\star)q}{\min\{d_1, d_2(d_2-1)/2\}} .
     \end{align*}
\end{lemma}

Then, by Lemma \ref{lem:negligible2}, Lemma \ref{lem:negligible3}, and Lemma \ref{lem:negligible1}, we have with probability at least $1-O(\bar{d}^{-9}),$ for some $C>0,$
\begin{align*}
    &\frac{1}{2}\widehat{\bfM}^{\proj,1}_{i,k}+\frac{1}{2}\widehat{\bfM}^{\proj,2}_{i,k}- M^\star_{i,k} \\
     &=\sum_{k'=1}^{d_2(d_2-1)/2} \frac{1}{p_i}\delta_{i,k'}\left( y_{i,k'}- \sigma(M^\star_{i,k'}) \right) \left(\sigma'(M^\star_{i,k'}) \right)^{-1}\bfV^{\bfM^\star}_{k',\cdot}\bfV^{\bfM^\star \top}_{k,\cdot}\\
      &\quad + \sum_{i'=1}^{d_1} \bfU^{\bfM^\star}_{i,\cdot}\bfU^{\bfM^\star \top}_{i',\cdot} \frac{1}{p_{i'}}\delta_{i',k}\left( y_{i',k}- \sigma(M^\star_{i',k}) \right) \left(\sigma'(M^\star_{i',k}) \right)^{-1} \\
      & \quad +C \sqrt{\mu(\bfM^\star)q} \left(\frac{\kappa \mu R \sigma^\star_{\max} }{\min\{d_1, d_2(d_2-1)/2\}} \sqrt{\frac{  \bar{d}\log(\bar{d})}{\bar{p} (\sigma^\star_{\min})^2 }} \right)^2  \\
       & \quad + C \sqrt{ \frac{1}{\bar{p}}\frac{(\mu(\bfM^\star))^2q^2\log (\bar{d})}{\min\{d_1, d_2(d_2-1)/2\}} }  \frac{\kappa \mu R \sigma^\star_{\max} }{\min\{d_1, d_2(d_2-1)/2\}} \sqrt{\frac{  \bar{d}\log(\bar{d})}{\bar{p} (\sigma^\star_{\min})^2 }} \\
       & \quad + C \sigma_{\max}(\bfM^\star)    \frac{\kappa^2 R\bar{d}}{\bar{p}(\sigma_{\min}(\bfM^\star))^2  } \frac{\mu(\bfM^\star)q}{\min\{d_1, d_2(d_2-1)/2\}} \\
       & \quad + C\frac{\kappa^2 R\bar{d}}{\bar{p}\sigma_{\min}(\bfM^\star)  }    \sqrt{\frac{(\mu(\bfM^\star))^2q^2}{ d_1 d_2(d_2-1)/2 }} \\
       &=\sum_{k'=1}^{d_2(d_2-1)/2} \frac{1}{p_i}\delta_{i,k'}\left( y_{i,k'}- \sigma(M^\star_{i,k'}) \right) \left(\sigma'(M^\star_{i,k'}) \right)^{-1}\bfV^{\bfM^\star}_{k',\cdot}\bfV^{\bfM^\star \top}_{k,\cdot}\\
      &\quad + \sum_{i'=1}^{d_1} \bfU^{\bfM^\star}_{i,\cdot}\bfU^{\bfM^\star \top}_{i',\cdot} \frac{1}{p_{i'}}\delta_{i',k}\left( y_{i',k}- \sigma(M^\star_{i',k}) \right) \left(\sigma'(M^\star_{i',k}) \right)^{-1} \\
      & \quad + o\left(\sqrt{w^\star_{i,k}} \right).
\end{align*}
 Finally, we invoke Lemma \ref{lem:CLTSA1} and complete the proof.
\end{proof}

\subsubsection{Technical lemmas}

\begin{lemma}\label{lem:Mhatderrorbound}
    Under the assumptions in Theorem \ref{thm:asymptoticnormality_indiv}, with probability at least $1-O(\bar{d}^{-10})$, we have, for each $l=1,2$,
    \begin{align*}
        \norm{\widehat{\bfM}^{\mathrm{NR},l}-\bfM^\star} \lesssim \sqrt{\frac{\kappa^2 R\bar{d}}{\bar{p}  }}.
    \end{align*}
\end{lemma}
\begin{proof}
   We will maintain the notations in \eqref{eq:decompositionMhatd1} and \eqref{eq:decompositionMhatd2}. Without loss of generality, we focus on the case of $l=1$. By applying Corollary \ref{cor:errorboundforM}, we have
\begin{align}
 \norm{\Delta^1_1}=   \norm{\widehat{\bfM}^1-\bfM^\star} \leq \norm{\widehat{\bfM}^1-\bfM^\star}_{\mathrm{F}} \lesssim    \sqrt{ \frac{\kappa^2 R \bar{d}}{\bar{p}  }}  \label{eq:Delta1bound}
\end{align}
 with probability at least $1-O(\bar{d}^{-10})$. Also, by following the proof of Lemma \ref{lem:boundgradient}, we can bound, with probability at least $1-O(\bar{d}^{-10})$,
\begin{align}
    \norm{\Delta^1_2} \lesssim \sqrt{\frac{\bar{d}}{\bar{p}}}.  \label{eq:Delta2bound}
\end{align}
Lastly, to bound $\Delta^1_3$, we decompose it as
\begin{align*}
    \Delta^1_{3,i,k} = \underbrace{\frac{2}{p_i} \left(\delta^2_{i,k} -\frac{p_i}{2} \right) \left(\sigma'(\widehat{M}^1_{i,k}) \right)^{-1} \left( \sigma( M^\star_{i,k})- \sigma(\widehat{M}^1_{i,k}) \right)}_{\coloneqq \Delta^1_{31,i,k}}+\underbrace{\left(\sigma'(\widehat{M}^1_{i,k}) \right)^{-1} \left( \sigma( M^\star_{i,k})- \sigma(\widehat{M}^1_{i,k}) \right)}_{\Delta^1_{32,i,k}}. 
\end{align*}
By applying Theorem \ref{thm:errorboundsforL} and the small size of sieve approximation error assumption, we have
\begin{align*}
   \Delta^1_{32}\lesssim   \sqrt{\frac{ \kappa^2 R\bar{d}}{\bar{p}}}
\end{align*}
with probability at least $1-O(\bar{d}^{-10})$. For $\Delta^1_{31}$, we invoke the matrix Bernstein inequality \citep[Theorem 6.1.1 in][]{tropp:2015}. Define
\begin{align*}
    \Delta^1_{31} = \sum_{i=1}^{d_1} \sum_{k=1}^{d_2(d_2-1)/2} \underbrace{\frac{2}{p_i} \left(\delta^2_{i,k} -\frac{p_i}{2} \right) \left(\sigma'(\widehat{M}^1_{i,k}) \right)^{-1} \left(\sigma( M^\star_{i,k})- \sigma(\widehat{M}^1_{i,k}) \right)\bfe_i \bfe_k^\top}_{\coloneqq S_{i,k}}.
\end{align*}
Using Theorem \ref{thm:errorboundsforL} and the small size of sieve approximation error assumption once again, we have, with probability at least $1-O(\bar{d}^{-10})$, 
\begin{align*}
    L &\coloneqq \max_{i,k}\norm{S_{i,k}} \lesssim \frac{\kappa \mu R \sigma^\star_{\max} }{\bar{p}\min\{d_1, d_2(d_2-1)/2\}} \sqrt{\frac{  \bar{d}\log(\bar{d})}{\bar{p} (\sigma^\star_{\min})^2 }}.
\end{align*}
Similarly, we have, with probability at least $1-O(\bar{d}^{-10})$, 
\begin{align*}
    \norm{\sum_{i=1}^{d_1} \sum_{k=1}^{d_2(d_2-1)/2}\bbE S_{i,k}^\top S_{i,k}}& \lesssim \norm{ \frac{1}{\bar{p}} \max_{i,k}\left(\left(\sigma'(\widehat{M}^1_{i,k}) \right)^{-1} \left(\sigma( M^\star_{i,k})- \sigma(\widehat{M}^1_{i,k}) \right)\right)^2 \bar{d} \bfI  }\\
    &\lesssim \frac{\bar{d}}{\bar{p}} \left(\left(\frac{\kappa \mu R \sigma^\star_{\max} }{\min\{d_1, d_2(d_2-1)/2\}} \sqrt{\frac{  \bar{d}\log(\bar{d})}{\bar{p} (\sigma^\star_{\min})^2 }} \right)^2 + \frac{1}{\max\{d_1, d_2(d_2-1)/2\}}\right) \\
    &   \\
     &\lesssim \frac{\bar{d}}{\bar{p}} \left(\frac{\kappa \mu R \sigma^\star_{\max} }{\min\{d_1, d_2(d_2-1)/2\}} \sqrt{\frac{  \bar{d}\log(\bar{d})}{\bar{p} (\sigma^\star_{\min})^2 }} \right)^2.
\end{align*}
Similarly, we can bound $\norm{\sum_{i=1}^{d_1} \sum_{k=1}^{d_2(d_2-1)/2}\bbE S_{i,k} S_{i,k}^\top}.$ As a result, with probability at least $1-O(\bar{d}^{-10})$, 
\begin{align*}
    V &\coloneqq  \max\bigg\{\norm{\sum_{i=1}^{d_1} \sum_{k=1}^{d_2(d_2-1)/2}\bbE S_{i,k}^\top S_{i,k}}, \norm{\sum_{i=1}^{d_1} \sum_{k=1}^{d_2(d_2-1)/2}\bbE S_{i,k} S_{i,k}^\top}  \bigg\} \\
    &\lesssim \frac{\bar{d}}{\bar{p}} \left(\frac{\kappa \mu R \sigma^\star_{\max} }{\min\{d_1, d_2(d_2-1)/2\}} \sqrt{\frac{  \bar{d}\log(\bar{d})}{\bar{p} (\sigma^\star_{\min})^2 }} \right)^2.
\end{align*}
Therefore,
\begin{align*}
    \Delta^1_{31} \lesssim \sqrt{\frac{\bar{d}\log(\bar{d})}{\bar{p}}} \frac{\kappa \mu R \sigma^\star_{\max} }{\min\{d_1, d_2(d_2-1)/2\}} \sqrt{\frac{  \bar{d}\log(\bar{d})}{\bar{p} (\sigma^\star_{\min})^2 }}  \ll  \sqrt{\frac{\bar{d}}{\bar{p}}}
\end{align*}
with probability at least $1-O(\bar{d}^{-10})$.

Combining bounds for $\Delta^1_{1}$, $\Delta^1_{2}$, $\Delta^1_{31}$, and $\Delta^1_{32}$, we have, with probability at least $1-O(\bar{d}^{-10})$
\begin{align*}
    \Delta^1_1+\Delta^1_2+\Delta^1_3 \lesssim    \sqrt{\frac{\kappa^2 R\bar{d}}{\bar{p}  }}.
\end{align*} 
\end{proof}

 \begin{lemma}\label{lem:perturbedsingularvector}
    Under the assumptions in Theorem \ref{thm:asymptoticnormality_indiv}, with probability at least $1-O(\bar{d}^{-9})$, we have for each $l=1,2$,
    \begin{align*}
       & \max \bigg\{\norm{\widehat{\bfU}^{\mathrm{NR},l}\widehat{\bfU}^{\mathrm{NR},l \top}-\bfU^{\bfM^\star}\bfU^{\bfM^\star \top}}_{2, \infty},\norm{\widehat{\bfV}^{\mathrm{NR},l}\widehat{\bfV}^{\mathrm{NR},l \top}-\bfV^{\bfM^\star}\bfV^{\bfM^\star \top}}_{2, \infty}\bigg\}\\
        &  \quad \lesssim \sqrt{\frac{\kappa^2 R\bar{d}}{\bar{p}(\sigma_{\min}(\bfM^\star))^2  }}\sqrt{\frac{\mu(\bfM^\star)q}{\min\{d_1, d_2(d_2-1)/2\}}} .
    \end{align*}
\end{lemma}
\begin{proof}
Without loss of generality, we focus on the case of $l=1$ and follow the notation in \eqref{eq:decompositionMhatd1}. To bound the error of the estimated singular subspaces, we leverage the Representation formula of spectral projectors from, for example, \cite{xia2021normal} and \cite{xia2021statistical}. This proof follows the proof of Theorem 4 in \cite{xia2021statistical}, which analyzes the empirical singular spaces, with appropriate adaptations. Denoting $\Delta^1\coloneqq \Delta^1_1+\Delta^1_2+\Delta^1_3$, we define
\begin{align*}
    \widehat{\bfW}^1 &\coloneqq \begin{bmatrix}
        \widehat{\bfU}^{\mathrm{NR},1} & {\bf 0}\\
        {\bf 0} &  \widehat{\bfV}^{\mathrm{NR},1}
    \end{bmatrix} ;
    \quad    \widehat{\bfE}^1 \coloneqq \begin{bmatrix}
        {\bf 0} & \Delta^1  \\
         \Delta^{1\top}     &  {\bf 0}
    \end{bmatrix}
    \end{align*}
and
\begin{align*}
     \bfW^\star &\coloneqq \begin{bmatrix}
        \bfU^{\bfM^\star} & {\bf 0}\\
        {\bf 0} &  \bfV^{\bfM^\star}
    \end{bmatrix}; \quad \bfA^\star \coloneqq \begin{bmatrix}
        {\bf 0} & \bfM^\star  \\
         \bfM^{\star \top}       &  {\bf 0}
    \end{bmatrix}.
\end{align*}
Let $\bfU^{\bfM^\star}_\perp \in \bbR^{d_1\times(d_1-q)}$ and $\bfV^{\bfM^\star}_\perp \in \bbR^{d_2(d_2-1)/2 \times(d_2(d_2-1)/2-q)}$ be orthogonal to $\bfU^{\bfM^\star}$ and $\bfV^{\bfM^\star}$, respectively. That is, $[\bfU^{\bfM^\star}; \bfU^{\bfM^\star}_\perp]$ and $[\bfV^{\bfM^\star}; \bfV^{\bfM^\star}_\perp]$ are orthonormal. For any $s=1, 2, \cdots$, define
\begin{align*}
    \mathfrak{P}^{-s} \coloneqq \begin{cases}
        \begin{bmatrix}
            \bfU^{\bfM^\star} ({\bf\Sigma}^{\bfM^\star})^{-s} \bfU^{\bfM^\star \top} & {\bf 0} \\
            {\bf 0} &  \bfV^{\bfM^\star} ({\bf\Sigma}^{\bfM^\star})^{-s} \bfV^{\bfM^\star \top}
        \end{bmatrix}, \quad \text{if $s$ is even;} \\ \begin{bmatrix}
             {\bf 0}&\bfU^{\bfM^\star} ({\bf\Sigma}^{\bfM^\star})^{-s} \bfV^{\bfM^\star \top}  \\
             \bfV^{\bfM^\star} ({\bf\Sigma}^{\bfM^\star})^{-s} \bfU^{\bfM^\star \top} & {\bf 0}
        \end{bmatrix}, \quad \text{if $s$ is odd}
    \end{cases}
\end{align*}
and
\begin{align*}
    \mathfrak{P}^0 \coloneqq \mathfrak{P}^\perp \coloneqq  \begin{bmatrix}
            \bfU^{\bfM^\star}_\perp \bfU^{\bfM^\star \top}_\perp & {\bf 0} \\
            {\bf 0} &  \bfV^{\bfM^\star}_\perp \bfV^{\bfM^\star \top}_\perp
        \end{bmatrix}.
\end{align*}
Lastly, we define
\begin{align*}
    \calS_{\bfA^\star,m}(\widehat{\bfE}^1) \coloneqq \sum_{\bfs : s_1+\cdots+s_{m+1}=m}(-1)^{1+\tau(\bfs)} \mathfrak{P}^{-s_1} \widehat{\bfE}^1 \mathfrak{P}^{-s_2}\cdots \mathfrak{P}^{-s_m}\widehat{\bfE}^1 \mathfrak{P}^{-s_{m+1}}
\end{align*}
where $\bfs=(s_1, \ldots, s_{m+1})$ is such that $s_1, \cdots, s_{m+1}\geq 0$ are integers and $\tau(\bfs)= \sum_{m'=1}^{m+1}\bbOne\{s_{m'}>0\}$. Then, by Theorem 1 in \cite{xia2021normal}, Lemma \ref{lem:Mhatderrorbound}, and Assumption \ref{asp:SSVforM}, we have
\begin{align*}
    \widehat{\bfW}^1\widehat{\bfW}^{1\top}-\bfW^\star\bfW^{\star\top} = \sum_{m=1}^\infty  \calS_{\bfA^\star,m}(\widehat{\bfE}^1).
\end{align*}
Since
\begin{align*}
    \widehat{\bfW}^1\widehat{\bfW}^{1\top}-\bfW^\star\bfW^{\star\top} = \begin{bmatrix}
        \widehat{\bfU}^{\mathrm{NR},1}\widehat{\bfU}^{\mathrm{NR},1 \top} -\bfU^{\bfM^\star}\bfU^{\bfM^\star \top} & {\bf 0}\\
        {\bf 0} &  \widehat{\bfV}^{\mathrm{NR},1}\widehat{\bfV}^{\mathrm{NR},1\top}-\bfV^{\bfM^\star}\bfV^{\bfM^\star \top}
    \end{bmatrix},
\end{align*}
it is enough to prove for $\norm{\widehat{\bfW}^1\widehat{\bfW}^{1\top}-\bfW^\star\bfW^{\star\top}}_{2,\infty}.$ 

Note that for any $s\geq 1$,
\begin{align*}
    \max_{i \in [d_1]}\norm{\bfe_i^\top  \mathfrak{P}^{-s}} \leq \sqrt{\frac{\mu(\bfM^\star)q}{d_1}} (\sigma_{\min}(\bfM^\star))^{-s}\,\, \text{and} \,\, \max_{k \in [d_2(d_2-1)/2]}\norm{\bfe_{d_1+k}^\top  \mathfrak{P}^{-s}} \leq \sqrt{\frac{\mu(\bfM^\star)q}{d_2(d_2-1)/2}} (\sigma_{\min}(\bfM^\star))^{-s}.
\end{align*}
Also, for any $\bfs$ such that $s_1 \geq 1$, we have
\begin{align*}
 \max_{i \in [\bar{d}]} \norm{  \bfe_i^\top \mathfrak{P}^{-s_1} \widehat{\bfE}^1 \mathfrak{P}^{-s_2}\cdots \mathfrak{P}^{-s_m}\widehat{\bfE}^1 \mathfrak{P}^{-s_{m+1}}  } &\leq \max_{i \in [\bar{d}]} \norm{  \bfe_i^\top \mathfrak{P}^{-s_1}} \norm{\widehat{\bfE}^1 \mathfrak{P}^{-s_2}\cdots \mathfrak{P}^{-s_l}\widehat{\bfE}^1 \mathfrak{P}^{-s_{m+1}} } \\
 &\leq \max_{i \in [\bar{d}]} \norm{  \bfe_i^\top \mathfrak{P}^{-s_1}} \norm{\widehat{\bfE}^1}^m (\sigma_{\min}(\bfM^\star))^{-(m-s_1)}.
\end{align*}

By Lemma \ref{lem:Mhatderrorbound}, with probability at least $1-O(\bar{d}^{-10})$, for some constant $C>0,$
\begin{align}
    \norm{\widehat{\bfE}^1} \leq \underbrace{C\sqrt{\frac{\kappa^2 R\bar{d}}{\bar{p}  }}}_{\coloneqq \nu} .\label{eq:Ehatbound}
\end{align}

Therefore, if $s_1 \geq 1$, with probability at least $1-O(\bar{d}^{-10})$,
\begin{align*}
   \max_{i \in [\bar{d}]} \norm{  \bfe_i^\top \mathfrak{P}^{-s_1} \widehat{\bfE}^1 \mathfrak{P}^{-s_2}\cdots \mathfrak{P}^{-s_m}\widehat{\bfE}^1 \mathfrak{P}^{-s_{m+1}}  } \leq \left( \frac{\nu}{\sigma_{\min}(\bfM^\star)}\right)^m \sqrt{\frac{\mu(\bfM^\star)q}{\min\{d_1, d_2(d_2-1)/2\}}}.
\end{align*}

We now consider the case of $s_1=0.$ Because $s_1+\cdots + s_{m+1}=m$, there must exist $s_{m'}\geq 1$ for some $m'\geq 2$. Therefore, if we can bound $\norm{\mathfrak{P}^\perp (\mathfrak{P}^\perp \widehat{\bfE}^1 \mathfrak{P}^\perp)^t \widehat{\bfE}^1 \bfW^\star  }_{2, \infty}$ for $t\geq 0$, we can have the bound for any $\bfs.$ Here, we borrow Lemma 9 of \cite{xia2021statistical} with straightforward modifications.

\begin{lemma}[\cite{xia2021statistical}]\label{lem:orthogonalprojectionofEhat}
    Under the event \eqref{eq:Ehatbound}, there exist constant $C_1, C_2 >0$ such that, for any $t\geq 0$, with probability at least $1-O((t+1)\bar{d}^{-10})$,
    \begin{align*}
        \norm{\mathfrak{P}^\perp (\mathfrak{P}^\perp \widehat{\bfE}^1 \mathfrak{P}^\perp)^t \widehat{\bfE}^1 \bfW^\star  }_{2, \infty} \leq C_1 (C_2 \nu)^{t+1} \sqrt{\frac{\mu(\bfM^\star)q}{\min\{d_1, d_2(d_2-1)/2\}}}.
    \end{align*}
\end{lemma}

Choose $m_{\max} = \lceil 2 \log(\bar{d}) \rceil$. Then, for all $\bfs=(s_1, \ldots, s_{m+1})$ such that $s_1, \ldots, s_{m+1}\geq 0$ and $s_1+\cdots +s_{m+1}=m$, by Lemma \ref{lem:orthogonalprojectionofEhat} and \eqref{eq:Ehatbound},
\begin{align*}
    \max_{i \in [\bar{d}]} \norm{  \bfe_i^\top \mathfrak{P}^{-s_1} \widehat{\bfE}^1 \mathfrak{P}^{-s_2}\cdots \mathfrak{P}^{-s_m}\widehat{\bfE}^1 \mathfrak{P}^{-s_{m+1}}  } \leq C_1\left( \frac{C_2 \nu}{\sigma_{\min}(\bfM^\star)}\right)^m \sqrt{\frac{\mu(\bfM^\star)q}{\min\{d_1, d_2(d_2-1)/2\}}}
\end{align*}
for all $m \leq m_{\max}$, with probability at least $1-O(\log(\bar{d}) \bar{d}^{-10})$.

Therefore, 
\begin{align*}
    &\max_{i \in [d_1]} \norm{\bfe_i^\top (\widehat{\bfU}^{\mathrm{NR},1}\widehat{\bfU}^{\mathrm{NR},1 \top} -\bfU^{\bfM^\star}\bfU^{\bfM^\star \top}) }  \\
    &\quad = \max_{i \in [d_1]} \norm{\bfe_i^\top (\widehat{\bfW}^1\widehat{\bfW}^{1\top}-\bfW^\star\bfW^{\star\top}) }\\
    & \quad \leq  \max_{i \in [d_1]} \sum_{m=1}^{m_{\max}}\sum_{\bfs : s_1+\cdots+s_{m+1}=m} \norm{\bfe_i^\top \mathfrak{P}^{-s_1} \widehat{\bfE}^1 \mathfrak{P}^{-s_2}\cdots \mathfrak{P}^{-s_m}\widehat{\bfE}^1 \mathfrak{P}^{-s_{m+1}}} \\
    & \quad \quad + \sum_{m=m_{\max}+1}^{\infty}\sum_{\bfs : s_1+\cdots+s_{m+1}=m} \norm{ \mathfrak{P}^{-s_1} \widehat{\bfE}^1 \mathfrak{P}^{-s_2}\cdots \mathfrak{P}^{-s_m}\widehat{\bfE}^1 \mathfrak{P}^{-s_{m+1}}}_{2, \infty}.
\end{align*}

Since 
$\big\lvert \{(s_1, \ldots, s_{m+1}):s_1+\cdots +s_{m+1}=m, s_1, \ldots, s_{m+1}\in \bbZ^+   \} \big\rvert \leq 4^m,$ we have
\begin{align*}
    &\max_{i \in [d_1]} \norm{\bfe_i^\top (\widehat{\bfU}^{\mathrm{NR},1}\widehat{\bfU}^{\mathrm{NR},1 \top} -\bfU^{\bfM^\star}\bfU^{\bfM^\star \top}) } \\
    &\quad \leq \sum_{m=1}^{m_{\max}} C_1\left( \frac{4C_2 \nu}{\sigma_{\min}(\bfM^\star)}\right)^m \sqrt{\frac{\mu(\bfM^\star)q}{\min\{d_1, d_2(d_2-1)/2\}}}   +\sum_{m=m_{\max}+1}^{\infty} \left(\frac{4 \nu}{\sigma_{\min}(\bfM^\star)} \right)^m \\
    & \quad \lesssim \sqrt{\frac{\kappa^2 R\bar{d}}{\bar{p}(\sigma_{\min}(\bfM^\star))^2  }}\sqrt{\frac{\mu(\bfM^\star)q}{\min\{d_1, d_2(d_2-1)/2\}}} 
\end{align*}
with probability at least $1-O(\bar{d}^{-9})$, where the last line uses Lemma \ref{lem:Mhatderrorbound} and the geometric series sum. The same bound for $\norm{ \widehat{\bfV}^{\mathrm{NR},1}\widehat{\bfV}^{\mathrm{NR},1 \top} -\bfV^{\bfM^\star}\bfV^{\bfM^\star \top}}_{2, \infty}$ can be obtained similarly.  
\end{proof}

 \begin{lemma}\label{lem:projectederror}
        Under the assumptions in Theorem \ref{thm:asymptoticnormality_indiv}, with probability at least $1-O(\bar{d}^{-9})$, we have, for $l=1,2,$
        \begin{align*}
            &\max_{i,k} \big\lvert\widehat{\bfU}_{i,\cdot}^{\mathrm{NR},l} \widehat{\bfU}^{\mathrm{NR},l\top}  (\widehat{\bfM}^{\mathrm{NR},l}-\bfM^\star)\widehat{\bfV}^{\mathrm{NR},l}\widehat{\bfV}^{\mathrm{NR},l}_{k,\cdot}\big\rvert \\
            &\lesssim \sqrt{\frac{(\mu(\bfM^\star))^2q^2 \log(\bar{d})}{\bar{p}d_1d_2(d_2-1)/2}} + \mu(\bfM^\star)q \left(\frac{\kappa \mu R \sigma^\star_{\max} }{\min\{d_1, d_2(d_2-1)/2\}} \sqrt{\frac{  \bar{d}\log(\bar{d})}{\bar{p} (\sigma^\star_{\min})^2 }} \right)^2 \\
            & \quad + \frac{\kappa^2 R\bar{d}}{\bar{p}\sigma_{\min}(\bfM^\star)  } \frac{\mu(\bfM^\star)q}{\min\{d_1, d_2(d_2-1)/2\}}
        \end{align*}
\end{lemma}

\begin{proof}
Without loss of generality, we focus on the case of $l=1.$ By Lemma \ref{lem:Mhatderrorbound} and Lemma \ref{lem:perturbedsingularvector}, we have
\begin{align*}
    &\norm{ \widehat{\bfU}^{\mathrm{NR},1} \widehat{\bfU}^{\mathrm{NR},1\top}  (\widehat{\bfM}^{\mathrm{NR},1}-\bfM^\star)\widehat{\bfV}^{\mathrm{NR},1}\widehat{\bfV}^{\mathrm{NR},1\top}  }_{\infty} \\
    & \quad \leq  \norm{ \bfU^{\bfM^\star}\bfU^{\bfM^\star \top}  (\widehat{\bfM}^{\mathrm{NR},1}-\bfM^\star)\bfV^{\bfM^\star}\bfV^{\bfM^\star \top}  }_{\infty}\\
    & \quad \quad + \norm{ (\widehat{\bfU}^{\mathrm{NR},1}\widehat{\bfU}^{\mathrm{NR},1 \top} -\bfU^{\bfM^\star}\bfU^{\bfM^\star \top})  (\widehat{\bfM}^{\mathrm{NR},1}-\bfM^\star)\bfV^{\bfM^\star}\bfV^{\bfM^\star \top}  }_{\infty}\\
    & \quad \quad +  \norm{ \bfU^{\bfM^\star}\bfU^{\bfM^\star \top}  (\widehat{\bfM}^{\mathrm{NR},1}-\bfM^\star)(\widehat{\bfV}^{\mathrm{NR},1}\widehat{\bfV}^{\mathrm{NR},1 \top} -\bfV^{\bfM^\star}\bfV^{\bfM^\star \top})  }_{\infty}\\
    & \quad \quad + \norm{ (\widehat{\bfU}^{\mathrm{NR},1}\widehat{\bfU}^{\mathrm{NR},1 \top} -\bfU^{\bfM^\star}\bfU^{\bfM^\star \top})  (\widehat{\bfM}^{\mathrm{NR},1}-\bfM^\star)(\widehat{\bfV}^{\mathrm{NR},1}\widehat{\bfV}^{\mathrm{NR},1 \top} -\bfV^{\bfM^\star}\bfV^{\bfM^\star \top})  }_{\infty} \\
    & \quad \leq \norm{ \bfU^{\bfM^\star}\bfU^{\bfM^\star \top}  (\widehat{\bfM}^{\mathrm{NR},1}-\bfM^\star)\bfV^{\bfM^\star}\bfV^{\bfM^\star \top}  }_{\infty}   + C \frac{\kappa^2 R\bar{d}}{\bar{p}\sigma_{\min}(\bfM^\star)  } \frac{\mu(\bfM^\star)q}{\min\{d_1, d_2(d_2-1)/2\}}
\end{align*}
 with probability at least $1-O(\bar{d}^{-9})$, for some constant $C>0.$ We are left to bound $\norm{ \bfU^{\bfM^\star}\bfU^{\bfM^\star \top}  (\widehat{\bfM}^{\mathrm{NR},1}-\bfM^\star)\bfV^{\bfM^\star}\bfV^{\bfM^\star \top}  }_{\infty}.$  Note the following decomposition
 \begin{align*}
    \widehat{M}^{\mathrm{NR},1}_{i,k}- M^\star_{i,k}&= \underbrace{\frac{2}{p_i} \delta_{i,k}^2\left( y^2_{i,k}- \sigma( M^\star_{i,k}) \right) \left(\sigma'(\widehat{M}^1_{i,k}) \right)^{-1} }_{\coloneqq A^1_{i,k}}\nonumber \\
    & \quad  + \underbrace{\left(1-\frac{2}{p_i} \delta^2_{i,k}\right)\left(\widehat{M}^1_{i,k}- M^\star_{i,k}\right)- \left(1-\frac{2}{p_i} \delta^2_{i,k}\right) \left(\sigma'(\widehat{M}^1_{i,k}) \right)^{-1}  \frac{1}{2}\sigma''(x_{i,k}) \left(    M^\star_{i,k} -  \widehat{M}^1_{i,k}  \right)^2}_{\coloneqq B^1_{i,k}}  \\
    & \quad +\underbrace{\left(\sigma'(\widehat{M}^1_{i,k}) \right)^{-1}  \frac{1}{2}\sigma''(x_{i,k}) \left(    M^\star_{i,k} -  \widehat{M}^1_{i,k}  \right)^2}_{\coloneqq C^1_{i,k}}
 \end{align*}
 with $x_{i,k}$ lying between $\widehat{M}^1_{i,k}$ and $ M^\star_{i,k}$.  For any $(i,k)$, we can write
\begin{align*}
   &\bfU^{\bfM^\star}_{i,\cdot}\bfU^{\bfM^\star \top}  \bfA^1 \bfV^{\bfM^\star}\bfV^{\bfM^\star \top}_{k,\cdot}\\
   & \quad = \bfU^{\bfM^\star}_{i,\cdot} \left(\sum_{i'=1}^{d_1} \sum_{k'=1}^{d_2(d_2-1)/2}\underbrace{\frac{2}{p_{i'}} \delta_{i',k'}^2\left( y^2_{i',k'}- \sigma(M^\star_{i',k'}) \right) \left(\sigma'(\widehat{M}^1_{i',k'}) \right)^{-1} \bfU^{\bfM^\star \top}_{i',\cdot} \bfV^{\bfM^\star}_{k',\cdot}}_{\coloneqq \bfS_{i',k'}}   \right)
  \bfV^{\bfM^\star \top}_{k,\cdot}.
\end{align*}
We will bound the term $\sum_{i'=1}^{d_1} \sum_{k'=1}^{d_2(d_2-1)/2}\bfS_{i',k'}$ using the matrix Bernstein inequality \citep[][Theorem 6.1.1]{tropp:2015}. Note that $\bbE[\bfS_{i',k'}]={\bf 0}$. Also,
\begin{align*}
    L &\coloneqq \max_{i',k'}\norm{\bfS_{i',k'}} \lesssim \frac{1}{\bar{p}} \sqrt{\frac{\mu(\bfM^\star)q}{d_1}} \sqrt{\frac{\mu(\bfM^\star)q}{d_2(d_2-1)/2}} ;\\
    V &\coloneqq  \max\bigg\{\norm{\sum_{i'=1}^{d_1} \sum_{k'=1}^{d_2(d_2-1)/2}\bbE \bfS_{i',k'}^\top \bfS_{i',k'}}, \norm{\sum_{i'=1}^{d_1} \sum_{k'=1}^{d_2(d_2-1)/2}\bbE \bfS_{i',k'} \bfS_{i',k'}^\top}  \bigg\} \lesssim \frac{1}{\bar{p}}.
\end{align*}
 Therefore, with probability at least $1-O(\bar{d}^{-20}),$
\begin{align*}
    \norm{\sum_{i'=1}^{d_1} \sum_{k'=1}^{d_2(d_2-1)/2}  \bfS_{i',k'}} \lesssim \sqrt{\frac{1}{\bar{p}}\log (\bar{d})}
\end{align*}
and thus
\begin{align*}
\norm{\bfU^{\bfM^\star}_{i,\cdot}\bfU^{\bfM^\star \top}  \bfA^1 \bfV^{\bfM^\star}\bfV^{\bfM^\star \top}_{k,\cdot}} \lesssim \sqrt{\frac{(\mu(\bfM^\star))^2q^2 \log(\bar{d})}{\bar{p}d_1d_2(d_2-1)/2}}.
\end{align*}
Next, to bound $\max_{i,k}\norm{\bfU^{\bfM^\star}_{i,\cdot}\bfU^{\bfM^\star \top}  \bfB^1 \bfV^{\bfM^\star}\bfV^{\bfM^\star \top}_{k,\cdot}}$, similar to the previous case, we use the matrix Bernstein inequality \citep[][Theorem 6.1.1]{tropp:2015}. To do so, we write
\begin{align*}
    \sum_{i'=1}^{d_1} \sum_{k'=1}^{d_2(d_2-1)/2}\underbrace{\left(1-\frac{2}{p_i} \delta^2_{i,k}\right)\left(\widehat{M}^1_{i,k}- M^\star_{i,k}-\left(\sigma'(\widehat{M}^1_{i,k}) \right)^{-1}  \frac{1}{2}\sigma''(\widetilde{M}^1_{i,k}) \left(    M^\star_{i,k} -  \widehat{M}^1_{i,k}  \right)^2\right) \bfU^{\bfM^\star \top}_{i',\cdot} \bfV^{\bfM^\star}_{k',\cdot}}_{\coloneqq \bfS_{i',k'}}.
\end{align*}
 Again, $\bbE[\bfS_{i',k'}]={\bf 0}$. Also, by Corollary \ref{cor:errorboundforM}, with probability at least $1-O(\bar{d}^{-10}),$ 
\begin{align*}
    L &\coloneqq \max_{i',k'}\norm{\bfS_{i',k'}} \lesssim \frac{1}{\bar{p}} \sqrt{\frac{\mu(\bfM^\star)q}{d_1}} \sqrt{\frac{\mu(\bfM^\star)q}{d_2(d_2-1)/2}} \frac{\kappa \mu R \sigma^\star_{\max} }{\min\{d_1, d_2(d_2-1)/2\}} \sqrt{\frac{  \bar{d}\log(\bar{d})}{\bar{p} (\sigma^\star_{\min})^2 }} ;\\
    V &\coloneqq  \max\bigg\{\norm{\sum_{i'=1}^{d_1} \sum_{k'=1}^{d_2(d_2-1)/2}\bbE \bfS_{i',k'}^\top \bfS_{i',k'}}, \norm{\sum_{i'=1}^{d_1} \sum_{k'=1}^{d_2(d_2-1)/2}\bbE \bfS_{i',k'} \bfS_{i',k'}^\top}  \bigg\} \\
    & \lesssim \frac{1}{\bar{p}} \left(\frac{\kappa \mu R \sigma^\star_{\max} }{\min\{d_1, d_2(d_2-1)/2\}} \sqrt{\frac{  \bar{d}\log(\bar{d})}{\bar{p} (\sigma^\star_{\min})^2 }} \right)^2.
\end{align*}
 Therefore, with probability at least $1-O(\bar{d}^{-20}),$
\begin{align*}
    \norm{\sum_{i'=1}^{d_1} \sum_{k'=1}^{d_2(d_2-1)/2}  \bfS_{i',k'}} \lesssim \sqrt{\frac{1}{\bar{p}}\log (\bar{d})} \frac{\kappa \mu R \sigma^\star_{\max} }{\min\{d_1, d_2(d_2-1)/2\}} \sqrt{\frac{  \bar{d}\log(\bar{d})}{\bar{p} (\sigma^\star_{\min})^2 }}
\end{align*}
and thus, with probability at least $1-O(\bar{d}^{-9}),$
\begin{align*}
    \max_{i,k}\norm{\bfU^{\bfM^\star}_{i,\cdot}\bfU^{\bfM^\star \top}  \bfB^1 \bfV^{\bfM^\star}\bfV^{\bfM^\star \top}_{k,\cdot}} &\lesssim \sqrt{\frac{(\mu(\bfM^\star))^2q^2 \log(\bar{d})}{\bar{p}d_1d_2(d_2-1)/2}} \frac{\kappa \mu R \sigma^\star_{\max} }{\min\{d_1, d_2(d_2-1)/2\}} \sqrt{\frac{  \bar{d}\log(\bar{d})}{\bar{p} (\sigma^\star_{\min})^2 }}\\
    & \ll \sqrt{\frac{(\mu(\bfM^\star))^2q^2 \log(\bar{d})}{\bar{p}d_1d_2(d_2-1)/2}}  .
\end{align*}
Lastly, by Corollary \ref{cor:errorboundforM}, with probability at least $1-O(\bar{d}^{-9}),$ 
\begin{align*}
     \max_{i,k}\norm{\bfU^{\bfM^\star}_{i,\cdot}\bfU^{\bfM^\star \top}  \bfC^1 \bfV^{\bfM^\star}\bfV^{\bfM^\star \top}_{k,\cdot}}  &\leq \sqrt{\frac{\mu(\bfM^\star)q}{d_1}} \sqrt{\frac{\mu(\bfM^\star)q}{d_2(d_2-1)/2}} \norm{\bfC^1}\\
     & \lesssim  \mu(\bfM^\star)q \left(\frac{\kappa \mu R \sigma^\star_{\max} }{\min\{d_1, d_2(d_2-1)/2\}} \sqrt{\frac{  \bar{d}\log(\bar{d})}{\bar{p} (\sigma^\star_{\min})^2 }} \right)^2.
\end{align*}
\end{proof}

\begin{lemma}\label{lem:negligible3}
      From \eqref{eq:decompositionMhatd1} and \eqref{eq:decompositionMhatd2}, denote $\Delta^l\coloneqq \Delta^l_1+\Delta^l_2+\Delta^l_3$ for $l=1,2.$ Under the assumptions in Theorem \ref{thm:asymptoticnormality_indiv}, we have, uniformly for all $(i,k)$, with probability at least $1-O(\bar{d}^{-9})$, for some $C>0,$
      \begin{align*}
          & \bigg[ \bfU^{\bfM^\star}_{\perp} \bfU^{\bfM^\star \top}_{\perp} \left(\frac{1}{2}\Delta^1+ \frac{1}{2}\Delta^2\right) \bfV^{\bfM^\star} \bfV^{\bfM^\star \top}\bigg]_{i,k} + \bigg[ \bfU^{\bfM^\star} \bfU^{\bfM^\star \top} \left(\frac{1}{2}\Delta^1+ \frac{1}{2}\Delta^2\right) \bfV^{\bfM^\star}_{\perp} \bfV^{\bfM^\star \top}_{\perp}\bigg]_{i,k} \\
         &=\sum_{k'=1}^{d_2(d_2-1)/2} \frac{1}{p_i}\delta_{i,k'}\left( y_{i,k'}- \sigma(M^\star_{i,k'}) \right) \left(\sigma'(M^\star_{i,k'}) \right)^{-1}\bfV^{\bfM^\star}_{k',\cdot}\bfV^{\bfM^\star \top}_{k,\cdot}\\
      &\quad + \sum_{i'=1}^{d_1} \bfU^{\bfM^\star}_{i,\cdot}\bfU^{\bfM^\star \top}_{i',\cdot} \frac{1}{p_{i'}}\delta_{i',k}\left( y_{i',k}- \sigma(M^\star_{i',k}) \right) \left(\sigma'(M^\star_{i',k}) \right)^{-1} \\
      & \quad +C \sqrt{\mu(\bfM^\star)q} \left(\frac{\kappa \mu R \sigma^\star_{\max} }{\min\{d_1, d_2(d_2-1)/2\}} \sqrt{\frac{  \bar{d}\log(\bar{d})}{\bar{p} (\sigma^\star_{\min})^2 }} \right)^2  \\
       & \quad + C \sqrt{ \frac{1}{\bar{p}}\frac{(\mu(\bfM^\star))^2q^2\log (\bar{d})}{\min\{d_1, d_2(d_2-1)/2\}} }  \frac{\kappa \mu R \sigma^\star_{\max} }{\min\{d_1, d_2(d_2-1)/2\}} \sqrt{\frac{  \bar{d}\log(\bar{d})}{\bar{p} (\sigma^\star_{\min})^2 }}.
      \end{align*}
  \end{lemma}

  \begin{proof}
      By definition,
\begin{align*}
    \bfU^{\bfM^\star}_{\perp} \bfU^{\bfM^\star \top}_{\perp}= \bfI -\bfU^{\bfM^\star} \bfU^{\bfM^\star \top} \quad \text{and} \quad \bfV^{\bfM^\star}_{\perp} \bfV^{\bfM^\star \top}_{\perp}= \bfI -\bfV^{\bfM^\star} \bfV^{\bfM^\star \top}.
\end{align*}
Using these equalities,
\begin{align*}
    &\bigg[ \bfU^{\bfM^\star}_{\perp} \bfU^{\bfM^\star \top}_{\perp} \left(\frac{1}{2}\Delta^1+ \frac{1}{2}\Delta^2\right) \bfV^{\bfM^\star} \bfV^{\bfM^\star \top}\bigg]_{i,k} +\bigg[ \bfU^{\bfM^\star} \bfU^{\bfM^\star \top} \left(\frac{1}{2}\Delta^1+ \frac{1}{2}\Delta^2\right) \bfV^{\bfM^\star}_{\perp} \bfV^{\bfM^\star \top}_{\perp}\bigg]_{i,k} \\
    &= \bigg[  \left(\frac{1}{2}\Delta^1+ \frac{1}{2}\Delta^2\right) \bfV^{\bfM^\star} \bfV^{\bfM^\star \top}\bigg]_{i,k}+\bigg[  \bfU^{\bfM^\star} \bfU^{\bfM^\star \top}\left(\frac{1}{2}\Delta^1+ \frac{1}{2}\Delta^2\right) \bigg]_{i,k}\\
    &\quad - \bigg[ \bfU^{\bfM^\star}  \bfU^{\bfM^\star \top}  \left(\Delta^1+ \Delta^2\right) \bfV^{\bfM^\star} \bfV^{\bfM^\star \top}\bigg]_{i,k}.
\end{align*}
The last term is bounded in the proof of Lemma \ref{lem:projectederror} as, with probability at least $1-O(\bar{d}^{-9}),$
\begin{align*}
    &\norm{ \bfU^{\bfM^\star}  \bfU^{\bfM^\star \top}  \left( \Delta^1+ \Delta^2\right) \bfV^{\bfM^\star} \bfV^{\bfM^\star \top}}_{\infty} \\
    &\lesssim  \sqrt{\frac{(\mu(\bfM^\star))^2q^2 \log(\bar{d})}{\bar{p}d_1d_2(d_2-1)/2}} + \mu(\bfM^\star)q \left(\frac{\kappa \mu R \sigma^\star_{\max} }{\min\{d_1, d_2(d_2-1)/2\}} \sqrt{\frac{  \bar{d}\log(\bar{d})}{\bar{p} (\sigma^\star_{\min})^2 }} \right)^2.
\end{align*}
Now, we bound the first two terms. In the proof of Lemma \ref{lem:projectederror}, we define $\bfA^1$, $\bfB^1$, and $\bfC^1$. Define $\bfA^2$, $\bfB^2$, and $\bfC^2$ for $\widehat{\bfM}^{\mathrm{NR},2}-\bfM^\star$ analogously.
\begin{align*}
    & \bigg[ \frac{1}{2} (\Delta^1  +\Delta^2 )\bfV^{\bfM^\star} \bfV^{\bfM^\star \top}\bigg]_{i,k}  = \bigg[  \frac{1}{2}\left(\bfA^1+ \bfB^1+\bfC^1+\bfA^2+ \bfB^2+\bfC^2\right) \bfV^{\bfM^\star} \bfV^{\bfM^\star \top}\bigg]_{i,k};\\
    & \bigg[  \bfU^{\bfM^\star} \bfU^{\bfM^\star \top}\left(\frac{1}{2}\Delta^1+ \frac{1}{2}\Delta^2\right) \bigg]_{i,k} = \bigg[  \bfU^{\bfM^\star} \bfU^{\bfM^\star \top}\frac{1}{2}\left(\bfA^1+ \bfB^1+\bfC^1+\bfA^2+ \bfB^2+\bfC^2\right) \bigg]_{i,k}.
\end{align*}
First, by Corollary \ref{cor:errorboundforM}, with probability at least $1-O(\bar{d}^{-9})$,
\begin{align*}
    &\norm{ \frac{1}{2} (\bfC^1+\bfC^2) \bfV^{\bfM^\star} \bfV^{\bfM^\star \top} }_{\infty} \lesssim \sqrt{\mu(\bfM^\star)q} \left(\frac{\kappa \mu R \sigma^\star_{\max} }{\min\{d_1, d_2(d_2-1)/2\}} \sqrt{\frac{  \bar{d}\log(\bar{d})}{\bar{p} (\sigma^\star_{\min})^2 }} \right)^2;\\
    &\norm{ \bfU^{\bfM^\star} \bfU^{\bfM^\star \top} \frac{1}{2} (\bfC^1+\bfC^2)  }_{\infty} \lesssim \sqrt{\mu(\bfM^\star)q} \left(\frac{\kappa \mu R \sigma^\star_{\max} }{\min\{d_1, d_2(d_2-1)/2\}} \sqrt{\frac{  \bar{d}\log(\bar{d})}{\bar{p} (\sigma^\star_{\min})^2 }} \right)^2    .
\end{align*}
Second, we apply the matrix Bernstein inequality for the terms that include $\bfB^1$ and $\bfB2$. Note that
\begin{align*}
    \norm{\bigg[  \frac{1}{2} \bfB^1 \bfV^{\bfM^\star} \bfV^{\bfM^\star \top}\bigg]_{i,k}}   \leq  \frac{1}{2}\norm{\bfB^1_{i,\cdot} \bfV^{\bfM^\star}} \sqrt{\frac{\mu(\bfM^\star)q}{d_2(d_2-1)/2}}.
\end{align*}
To bound $ \frac{1}{2}\norm{\bfB^1_{i,\cdot} \bfV^{\bfM^\star}}$, consider
\begin{align*}
     \sum_{k=1}^{d_2(d_2-1)/2}\underbrace{\left(1-\frac{2}{p_i} \delta^2_{i,k}\right)\left(\widehat{M}^1_{i,k}- M^\star_{i,k}-\left(\sigma'(\widehat{M}^1_{i,k}) \right)^{-1}  \frac{1}{2}\sigma''(x_{j,k}) \left(    M^\star_{i,k} -  \widehat{M}^1_{i,k}  \right)^2\right)  \bfV^{\bfM^\star}_{k,\cdot}}_{\coloneqq \bfS_k}.
\end{align*}
 By the sample splitting, $\bbE[\bfS_k]={\bf 0}$. Also, by Corollary \ref{cor:errorboundforM},
\begin{align*}
    L &\coloneqq \max_k\norm{\bfS_k} \lesssim \frac{1}{\bar{p}}   \sqrt{\frac{\mu(\bfM^\star)q}{d_2(d_2-1)/2}} \frac{\kappa \mu R \sigma^\star_{\max} }{\min\{d_1, d_2(d_2-1)/2\}} \sqrt{\frac{  \bar{d}\log(\bar{d})}{\bar{p} (\sigma^\star_{\min})^2 }} ;\\
    V &\coloneqq  \max\bigg\{\norm{ \sum_{k=1}^{d_2(d_2-1)/2}\bbE \bfS_k^\top \bfS_k}, \norm{ \sum_{k=1}^{d_2(d_2-1)/2}\bbE \bfS_k \bfS_k^\top}  \bigg\} \\
    & \lesssim \frac{\mu(\bfM^\star)q}{\bar{p}} \left(\frac{\kappa \mu R \sigma^\star_{\max} }{\min\{d_1, d_2(d_2-1)/2\}} \sqrt{\frac{  \bar{d}\log(\bar{d})}{\bar{p} (\sigma^\star_{\min})^2 }} \right)^2.
\end{align*}
 Therefore, with probability at least $1-O(\bar{d}^{-10}),$
\begin{align*}
    \norm{ \sum_{k=1}^{d_2(d_2-1)}  \bfS_k} \lesssim \sqrt{\frac{\mu(\bfM^\star)q}{\bar{p}}\log (\bar{d})} \frac{\kappa \mu R \sigma^\star_{\max} }{\min\{d_1, d_2(d_2-1)/2\}} \sqrt{\frac{  \bar{d}\log(\bar{d})}{\bar{p} (\sigma^\star_{\min})^2 }}
\end{align*}
and thus
\begin{align*}
    \norm{  \frac{1}{2} \bfB^1 \bfV^{\bfM^\star} \bfV^{\bfM^\star \top} }_{\infty}   \lesssim  \sqrt{\frac{\mu(\bfM^\star)q}{\bar{p}}\log (\bar{d})} \sqrt{\frac{\mu(\bfM^\star)q}{d_2(d_2-1)/2}}\frac{\kappa \mu R \sigma^\star_{\max} }{\min\{d_1, d_2(d_2-1)/2\}} \sqrt{\frac{  \bar{d}\log(\bar{d})}{\bar{p} (\sigma^\star_{\min})^2 }} .
\end{align*}
The same bound can be obtained for $ \norm{\bfB^2 \bfV^{\bfM^\star} \bfV^{\bfM^\star \top} }_{\infty}.$ 

Similarly, with probability at least $1-O(\bar{d}^{-9})$,
\begin{align*}
    \norm{   \bfU^{\bfM^\star} \bfU^{\bfM^\star \top}\left(\frac{1}{2}\bfB^1+ \frac{1}{2}\bfB^2\right)  }_{\infty}\lesssim  \sqrt{\frac{\mu(\bfM^\star)q}{\bar{p}}\log (\bar{d})}  \sqrt{\frac{\mu(\bfM^\star)q}{d_1}} \frac{\kappa \mu R \sigma^\star_{\max} }{\min\{d_1, d_2(d_2-1)/2\}} \sqrt{\frac{  \bar{d}\log(\bar{d})}{p (\sigma^\star_{\min})^2 }}.
\end{align*}

We now focus on  
\begin{align*}
    \bigg[  \left(\frac{1}{2}\bfA^1+ \frac{1}{2}\bfA^2\right) \bfV^{\bfM^\star} \bfV^{\bfM^\star \top}\bigg]_{i,k} +   \bigg[\bfU^{\bfM^\star} \bfU^{\bfM^\star \top}  \left(\frac{1}{2}\bfA^1+ \frac{1}{2}\bfA^2\right)  \bigg]_{i,k}
\end{align*}
where the leading terms exist. By the definition of $\bfA^1$ and $\bfA^2,$
\begin{align*}
    &\bigg[  \left(\frac{1}{2} \bfA^1+ \frac{1}{2} \bfA^2\right) \bfV^{\bfM^\star} \bfV^{\bfM^\star \top}\bigg]_{i,k}  = \sum_{k'=1}^{d_2(d_2-1)/2} \frac{1}{p_i}\delta_{i,k'}^2\left( y^2_{i,k'}- \sigma(M^\star_{i,k'}) \right) \left(\sigma'(\widehat{M}^1_{i,k'}) \right)^{-1}\bfV^{\bfM^\star}_{k',\cdot}\bfV^{\bfM^\star \top}_{k,\cdot} \\
    & \quad +\sum_{k'=1}^{d_2(d_2-1)/2} \frac{1}{p_i}\delta_{i,k'}^1\left( y^1_{i,k'}- \sigma(M^\star_{i,k'}) \right) \left(\sigma'(\widehat{M}^2_{i,k'}) \right)^{-1}\bfV^{\bfM^\star}_{k',\cdot}\bfV^{\bfM^\star \top}_{k,\cdot}.
\end{align*}
By applying the Bernstein inequality, we can replace the estimates $\widehat{\bfM}^1$ and $\widehat{\bfM}^2$ with the true value $\bfM^\star.$ To see this, pick $[(1/2)\bfA^1\bfV^{\bfM^\star} \bfV^{\bfM^\star \top}]_{i,k}$ and observe
\begin{align*}
    &\bigg[  \frac{1}{2} \bfA^1 \bfV^{\bfM^\star} \bfV^{\bfM^\star \top}\bigg]_{i,k} \\
    &= \sum_{k'=1}^{d_2(d_2-1)/2} \frac{1}{p_i}\delta_{i,k'}^2\left( y^2_{i,k'}- \sigma(M^\star_{i,k'}) \right) \left(\sigma'(M^\star_{i,k'}) \right)^{-1}\bfV^{\bfM^\star}_{k',\cdot}\bfV^{\bfM^\star \top}_{k,\cdot}\\
    &\quad + \sum_{k'=1}^{d_2(d_2-1)/2} \underbrace{\frac{1}{p_i}\delta_{i,k'}^2\left( y^2_{i,k'}- \sigma(M^\star_{i,k'}) \right) \left(\left(\sigma'(\widehat{M}^1_{i,k'}) \right)^{-1}-\left(\sigma'(M^\star_{i,k'}) \right)^{-1}\right)\bfV^{\bfM^\star}_{k',\cdot}\bfV^{\bfM^\star \top}_{k,\cdot}}_{\coloneqq s_{k'}}.
\end{align*}

Note $\bbE[s_{k'}]=0$. Also, the continuity of the function $(\sigma'(\cdot))^{-1}$ and the entrywise error bound in Corollary \ref{cor:errorboundforM} lead to
\begin{align*}
    L &\coloneqq \max_{k'}\norm{s_{k'}} \lesssim \frac{1}{\bar{p}}  \frac{\mu(\bfM^\star)q}{d_2(d_2-1)} \frac{\mu R \sigma^\star_{\max} }{\min\{d_1, d_2(d_2-1)/2\}} \sqrt{\frac{  \bar{d}\log(\bar{d})}{\bar{p} (\sigma^\star_{\min})^2 }} ;\\
    V &\coloneqq  \norm{ \sum_{k'=1}^{d_2(d_2-1)/2}\bbE s_{k'}^2} \lesssim \frac{1}{\bar{p}} \left(\frac{\kappa \mu R \sigma^\star_{\max} }{\min\{d_1, d_2(d_2-1)/2\}} \sqrt{\frac{  \bar{d}\log(\bar{d})}{\bar{p} (\sigma^\star_{\min})^2 }} \right)^2  \frac{(\mu(\bfM^\star))^2 q^2}{d_2(d_2-1)/2} . 
\end{align*}
 Therefore, with probability at least $1-O(\bar{d}^{-9})$,
\begin{align*}
    \norm{\sum_{k'=1}^{d_2(d_2-1)/2}  s_{k'}} \lesssim \sqrt{ \frac{1}{\bar{p}}\frac{(\mu(\bfM^\star))^2q^2\log (\bar{d})}{d_2(d_2-1)/2} }  \frac{\kappa \mu R \sigma^\star_{\max} }{\min\{d_1, d_2(d_2-1)/2\}} \sqrt{\frac{  \bar{d}\log(\bar{d})}{\bar{p} (\sigma^\star_{\min})^2 }}.
\end{align*}

A similar proof can be employed to bound $[(1/2)\bfA^2\bfV^{\bfM^\star} \bfV^{\bfM^\star \top}]_{i,k}$,  $[\bfU^{\bfM^\star} \bfU^{\bfM^\star \top}(1/2)\bfA^1]_{i,k}$, and $[\bfU^{\bfM^\star} \bfU^{\bfM^\star \top}(1/2)\bfA^2 ]_{i,k}$, and we have
 \begin{align*}
    &\bigg[  \left(\frac{1}{2} \bfA^1+ \frac{1}{2} \bfA^2\right) \bfV^{\bfM^\star} \bfV^{\bfM^\star \top}\bigg]_{i,k} +   \bigg[\bfU^{\bfM^\star} \bfU^{\bfM^\star \top}  \left(\frac{1}{2} \bfA^1+ \frac{1}{2} \bfA^2\right)  \bigg]_{i,k}\\
   &=\sum_{k'=1}^{d_2(d_2-1)/2} \frac{1}{p_i}\delta_{i,k'}\left( y_{i,k'}- \sigma(M^\star_{i,k'}) \right) \left(\sigma'(M^\star_{i,k'}) \right)^{-1}\bfV^{\bfM^\star}_{k',\cdot}\bfV^{\bfM^\star \top}_{k,\cdot}\\
      &\quad + \sum_{i'=1}^{d_1} \bfU^{\bfM^\star}_{i,\cdot}\bfU^{\bfM^\star \top}_{i',\cdot} \frac{1}{p_{i'}}\delta_{i',k}\left( y_{i',k}- \sigma(M^\star_{i',k}) \right) \left(\sigma'(M^\star_{i',k}) \right)^{-1} \\
      & \quad + C \sqrt{ \frac{1}{\bar{p}}\frac{(\mu(\bfM^\star))^2q^2\log (\bar{d})}{\min\{d_1, d_2(d_2-1)/2\}} }  \frac{\kappa \mu R \sigma^\star_{\max} }{\min\{d_1, d_2(d_2-1)/2\}} \sqrt{\frac{  \bar{d}\log(\bar{d})}{\bar{p} (\sigma^\star_{\min})^2 }}
\end{align*}
for all $(i,k)$, with probability at least $1-O(\bar{d}^{-9})$.

Therefore, uniformly for all $(i,k)$, with probability at least $1-O(\bar{d}^{-9})$, for some $C>0,$
      \begin{align*}
          & \bigg[ \bfU^{\bfM^\star}_{\perp} \bfU^{\bfM^\star \top}_{\perp} \left(\frac{1}{2}\Delta^1+ \frac{1}{2}\Delta^2\right) \bfV^{\bfM^\star} \bfV^{\bfM^\star \top}\bigg]_{i,k} + \bigg[ \bfU^{\bfM^\star} \bfU^{\bfM^\star \top} \left(\frac{1}{2}\Delta^1+ \frac{1}{2}\Delta^2\right) \bfV^{\bfM^\star}_{\perp} \bfV^{\bfM^\star \top}_{\perp}\bigg]_{i,k} \\
         &=\sum_{k'=1}^{d_2(d_2-1)/2} \frac{1}{p_i}\delta_{i,k'}\left( y_{i,k'}- \sigma(M^\star_{i,k'}) \right) \left(\sigma'(M^\star_{i,k'}) \right)^{-1}\bfV^{\bfM^\star}_{k',\cdot}\bfV^{\bfM^\star \top}_{k,\cdot}\\
      &\quad + \sum_{i'=1}^{d_1} \bfU^{\bfM^\star}_{i,\cdot}\bfU^{\bfM^\star \top}_{i',\cdot} \frac{1}{p_{i'}}\delta_{i',k}\left( y_{i',k}- \sigma(M^\star_{i',k}) \right) \left(\sigma'(M^\star_{i',k}) \right)^{-1} \\
      & \quad +C \sqrt{\mu(\bfM^\star)q} \left(\frac{\kappa \mu R \sigma^\star_{\max} }{\min\{d_1, d_2(d_2-1)/2\}} \sqrt{\frac{  \bar{d}\log(\bar{d})}{\bar{p} (\sigma^\star_{\min})^2 }} \right)^2  \\
       & \quad + C \sqrt{ \frac{1}{\bar{p}}\frac{(\mu(\bfM^\star))^2q^2\log (\bar{d})}{\min\{d_1, d_2(d_2-1)/2\}} }  \frac{\kappa \mu R \sigma^\star_{\max} }{\min\{d_1, d_2(d_2-1)/2\}} \sqrt{\frac{  \bar{d}\log(\bar{d})}{\bar{p} (\sigma^\star_{\min})^2 }}.
      \end{align*}
  \end{proof}
 
\begin{lemma}\label{lem:CLTSA1}
     Under the assumptions in Theorem \ref{thm:asymptoticnormality_indiv},
     \begin{align*}
      & \sum_{k'=1}^{d_2(d_2-1)/2} \frac{1}{p_i}\delta_{i,k'}\left( y_{i,k'}- \sigma(M^\star_{i,k'}) \right) \left(\sigma'(M^\star_{i,k'}) \right)^{-1}\bfV^{\bfM^\star}_{k',\cdot}\bfV^{\bfM^\star \top}_{k,\cdot}\\
      &\quad + \sum_{i'=1}^{d_1} \bfU^{\bfM^\star}_{i,\cdot}\bfU^{\bfM^\star \top}_{i',\cdot} \frac{1}{p_{i'}}\delta_{i',k}\left( y_{i',k}- \sigma(M^\star_{i',k}) \right) \left(\sigma'(M^\star_{i',k}) \right)^{-1}   \overset{d}{\rightarrow} \calN(0,w^\star_{i,k})     \end{align*}
      where
    \begin{align*}
   w^\star_{i,k} &\coloneqq  \underbrace{\bfV^{\bfM^\star }_{k,\cdot} \left(\sum_{k'=1}^{d_2(d_2-1)/2} \delta_{i,k'}\frac{1}{p_i^2\sigma(M^\star_{i,k'})(1-\sigma(M^\star_{i,k'}))}   \bfV^{\bfM^\star \top}_{k',\cdot}\bfV^{\bfM^\star}_{k',\cdot}\right)\bfV^{\bfM^\star \top}_{k,\cdot}}_{\coloneqq w^{\star,1}_{i,k}}\\
          &\quad + \underbrace{\bfU^{\bfM^\star}_{i,\cdot}\left(\sum_{i'=1}^{d_1} \delta_{i',k}\left(\frac{1}{p_{i'}^2\sigma(M^\star_{i',k})(1-\sigma(M^\star_{i',k}))} \right)\bfU^{\bfM^\star \top}_{i',\cdot}\bfU^{\bfM^\star}_{i',\cdot}\right)\bfU^{\bfM^\star \top}_{i,\cdot}}_{\coloneqq w^{\star,2}_{i,k}} 
      \end{align*}
\end{lemma}
\begin{proof}
   Note that the two terms 
\begin{align*}
      &\sum_{k'=1}^{d_2(d_2-1)/2} \frac{1}{p_i}\delta_{i,k'}\left( y_{i,k'}- \sigma(M^\star_{i,k'}) \right) \left(\sigma'(M^\star_{i,k'}) \right)^{-1}\bfV^{\bfM^\star}_{k',\cdot}\bfV^{\bfM^\star \top}_{k,\cdot}; \quad \text{and}\\
      &\sum_{i'=1}^{d_1} \bfU^{\bfM^\star}_{i,\cdot}\bfU^{\bfM^\star \top}_{i',\cdot} \frac{1}{p_{i'}}\delta_{i',k}\left( y_{i',k}- \sigma(M^\star_{i',k}) \right) \left(\sigma'(M^\star_{i',k}) \right)^{-1}
     \end{align*}
are asymptotically independent in that they share only one random variable at the entry $(i,k)$, i.e., $p_i^{-1}\delta_{i,k}\left( y_{i,k}- \sigma( M^\star_{i,k}) \right) \left(\sigma'( M^\star_{i,k}) \right)^{-1}$.
Therefore, we will characterize the distributions of these two terms separately. Without loss of generality, we prove for
\begin{align*}
    \sum_{k'=1}^{d_2(d_2-1)/2} \underbrace{\frac{1}{p_i}\delta_{i,k'}\left( y_{i,k'}- \sigma(M^\star_{i,k'}) \right) \left(\sigma'(M^\star_{i,k'}) \right)^{-1}\bfV^{\bfM^\star}_{k',\cdot}\bfV^{\bfM^\star \top}_{k,\cdot}}_{\coloneqq Z_{k'}}.
\end{align*}
The proof is similar to that of Theorem \ref{thm:asymptoticnormality}. The difference is the presence of the singular vectors $\bfV^{\bfM^\star}_{k',\cdot}\bfV^{\bfM^\star \top}_{k,\cdot}$. Recall the definition of $\Omega$ from the proof of Theorem \ref{thm:asymptoticnormality} and note that
\begin{align*}
 w^{\star,1}_{i,k}=\sum_{k'=1}^{d_2(d_2-1)/2} \Var (Z_{k'}|\Omega) =  \bfV^{\bfM^\star}_{k,\cdot}\left(\sum_{k'=1}^{d_2(d_2-1)/2}\delta_{i,k'} \frac{1}{p_i^2 \sigma(M^\star_{i,k'})(1-\sigma(M^\star_{i,k'}))}   \bfV^{\bfM^\star \top}_{k',\cdot}\bfV^{\bfM^\star}_{k',\cdot}\right)\bfV^{\bfM^\star \top}_{k,\cdot}.
\end{align*}
 Use Cauchy--Schwarz inequality and observe that for any $\epsilon>0,$
\begin{align*}
    (w^{\star,1}_{i,k})^{-1}\sum_{k'=1}^{d_2(d_2-1)/2} \bbE [Z_{k'}^2 \bbOne\{|Z_{k'}|>\epsilon \sqrt{w^{\star,1}_{i,k}} \}|\Omega] &\leq (w^{\star,1}_{i,k})^{-1} \sum_{k'=1}^{d_2(d_2-1)/2} \sqrt{ \bbE [Z_{k'}^4|\Omega] \bbE[\bbOne\{|Z_{k'}|>\epsilon \sqrt{w^{\star,1}_{i,k}}\}|\Omega]  }.
\end{align*}
Note that for large $d_2(d_2-1)/2$, we have
\begin{align*}
    Z_{k'} \lesssim \frac{1}{\bar{p}}\frac{\mu(\bfM^\star) q}{d_2(d_2-1)/2}  \ll \epsilon \sqrt{w^{\star,1}_{i,k}} \asymp \epsilon\sqrt{\frac{1}{\bar{p}d_2(d_2-1)/2}  }
\end{align*}
and therefore $\bbOne\{|Z_{k'}|>\epsilon \sqrt{w^{\star,1}_{i,k}}\}=0$ with exceedingly high probability. By Linbeberg's CLT, we have, conditioning on $\Omega,$
\begin{align*}
      \sum_{k'=1}^{d_2(d_2-1)/2}Z_{k'} \overset{d}{\rightarrow} \calN(0,w^{\star,1}_{i,k}).
\end{align*}
The unconditional CLT is due to the dominated convergence theorem. An analogous proof shows that
\begin{align*}
   \sum_{i'=1}^{d_1} \bfU^{\bfM^\star}_{i,\cdot}\bfU^{\bfM^\star \top}_{i',\cdot} \frac{1}{p_{i'}}\delta_{i',k}\left( y_{i',k}- \sigma(M^\star_{i',k}) \right) \left(\sigma'(M^\star_{i',k}) \right)^{-1} \overset{d}{\rightarrow} \calN(0,w^{\star,2}_{i,k}).
\end{align*}
\end{proof}

\begin{lemma}\label{lem:negligible1}
    Follow the notations defined in the proof of Lemma \ref{lem:perturbedsingularvector}. Under the assumptions in Theorem \ref{thm:asymptoticnormality_indiv}, we have, for $l=1,2$, with probability at least $1-O(\bar{d}^{-9})$,
     \begin{align*}
          & \bigg[\sum_{m=2}^\infty \left(\calS_{\bfA^\star,m}(\widehat{\bfE}^l) \bfA^\star \bfW^\star \bfW^{\star \top}+\bfW^\star \bfW^{\star \top}\bfA^\star \calS_{\bfA^\star,m}(\widehat{\bfE}^l) \right) \bigg]_{i,k+d_1}  \lesssim \frac{\kappa^2 R\bar{d}}{\bar{p}\sigma_{\min}(\bfM^\star)  }    \sqrt{\frac{(\mu(\bfM^\star))^2q^2}{ d_1 d_2(d_2-1)/2 }}.
     \end{align*}
\end{lemma}
\begin{proof}
This lemma can be proved by following the proof of Lemma 5 in \cite{xia2021statistical} with straightforward modifications.  
 \end{proof}

\subsection{Proof of Proposition \ref{prop:feasibleCLT_indiv}}
The following lemma claims that the estimated singular vectors also satisfy the incoherence condition, which plays a crucial role in establishing the consistency of variance estimation.
\begin{lemma}\label{lem:inco_estimatedSV}
     Under the assumptions in Proposition \ref{prop:feasibleCLT_indiv}, with probability at least $1-O(\bar{d}^{-9})$, we have for each $l=1,2$,
    \begin{align*}
      \norm{\widehat{\bfU}^{\mathrm{NR},l}}_{2, \infty} \lesssim \sqrt{\frac{\mu(\bfM^\star)q}{ d_1 }}, \quad \text{and} \quad \norm{\widehat{\bfV}^{\mathrm{NR},l}}_{2, \infty} \lesssim \sqrt{\frac{\mu(\bfM^\star)q}{ d_2(d_2-1)/2 }}.
    \end{align*}
\end{lemma}
\begin{proof}
    Without loss of generality, we prove for $\norm{\widehat{\bfU}^{\mathrm{NR},1}}_{2, \infty}$. The other cases can be proved in a similar manner. Using Lemma \ref{lem:perturbedsingularvector}, we have, uniformly for all $i \in [d_1]$, with probability at least $1-O(\bar{d}^{-9}),$
    \begin{align*}
        &\bigg\lvert\norm{\widehat{\bfU}^{\mathrm{NR},l}_{i,\cdot}}_{\mathrm{F}}^2-\norm{\bfU^{\bfM^\star}_{i,\cdot}}_{\mathrm{F}}^2 \bigg\rvert = \bigg\lvert\norm{\widehat{\bfU}^{\mathrm{NR},l}_{i,\cdot}\widehat{\bfU}^{\mathrm{NR},l \top}}_{\mathrm{F}}^2-\norm{\bfU^{\bfM^\star}_{i,\cdot}\bfU^{\bfM^\star \top}}_{\mathrm{F}}^2\bigg\rvert\\
        & \quad  \leq \norm{\widehat{\bfU}^{\mathrm{NR},l}_{i,\cdot}\widehat{\bfU}^{\mathrm{NR},l \top}-\bfU^{\bfM^\star}_{i,\cdot}\bfU^{\bfM^\star \top}}_{\mathrm{F}}^2 + 2 \norm{\widehat{\bfU}^{\mathrm{NR},l}_{i,\cdot}\widehat{\bfU}^{\mathrm{NR},l \top}-\bfU^{\bfM^\star}_{i,\cdot}\bfU^{\bfM^\star \top}}_{\mathrm{F}} \norm{\bfU^{\bfM^\star}_{i,\cdot}\bfU^{\bfM^\star \top}}_{\mathrm{F}} \\
        & \quad \lesssim  \frac{\kappa^2 R\bar{d}}{\bar{p}(\sigma_{\min}(\bfM^\star))^2  }  \frac{\mu(\bfM^\star)q}{\min\{d_1, d_2(d_2-1)/2\}} + \sqrt{\frac{\kappa^2 R\bar{d}}{\bar{p}(\sigma_{\min}(\bfM^\star))^2  }}\sqrt{\frac{\mu(\bfM^\star)q}{\min\{d_1, d_2(d_2-1)/2\}}}  \sqrt{\frac{\mu(\bfM^\star)q}{d_1}} \\
        & \quad \ll \frac{\mu(\bfM^\star)q}{d_1}.
    \end{align*}  
Therefore,  uniformly for all $i \in [d_1]$, with probability at least $1-O(\bar{d}^{-9}),$
\begin{align*}
    \norm{\widehat{\bfU}^{\mathrm{NR},l}_{i,\cdot}}_{\mathrm{F}}^2 \leq \norm{\bfU^{\bfM^\star}_{i,\cdot}}_{\mathrm{F}}^2 + o\left(\frac{\mu(\bfM^\star)q}{d_1}\right) \lesssim \frac{\mu(\bfM^\star)q}{d_1}.
\end{align*}
    
\end{proof}

\begin{proof}[Proof of Proposition \ref{prop:feasibleCLT_indiv}]
    As in the proof of Proposition \ref{prop:feasibleCLT}, it suffices to show that $\widehat{w}_{i,k}-w^\star_{i,k}=o(w^\star_{i,k})$ with high probability. We will focus on establishing $\widehat{w}^2_{i,k}-w^{\star,2}_{i,k}=o(w^{\star,2}_{i,k})$ as the other case can be shown similarly. 

    We begin by showing that $w^{\star,2}_{i,k}$ is close to $\bbE [w^{\star,2}_{i,k}]$ using the Bernstein inequality. Define
    \begin{align*}
        w^{\star,2}_{i,k}-\bbE [w^{\star,2}_{i,k}]=  \sum_{i'=1}^{d_1}\underbrace{\frac{1}{p_{i'}^2} (\delta_{i',k}-p_{i'}) \left(\frac{1}{\sigma(M^\star_{i',k})(1-\sigma(M^\star_{i',k}))}\right)\bfU^{\bfM^\star}_{i,\cdot}\bfU^{\bfM^\star \top}_{i',\cdot}\bfU^{\bfM^\star}_{i',\cdot} \bfU^{\bfM^\star \top}_{i,\cdot}}_{\coloneqq s_{i'}}
    \end{align*}
and note that $\bbE s_{i'}=0$. Also, by the definition of the incoherence parameter $\mu(\bfM^\star),$
\begin{align*}
    L &\coloneqq \max_{i'} \norm{s_{i'}} \lesssim \frac{1}{\bar{p}^2} \frac{(\mu(\bfM^\star))^2 q^2}{d_1^2};\\
    V & \coloneqq \norm{\sum_{i'=1}^{d_1}\bbE s_{i'}^2 } \lesssim \frac{1}{\bar{p}^3}   \frac{(\mu(\bfM^\star))^4 q^4}{d_1^3}.
\end{align*}
Therefore, with probability at least $1-O(\bar{d}^{-10}),$
\begin{align*}
    w^{\star,2}_{i,k}-\bbE [w^{\star,2}_{i,k}]=  \sum_{i'=1}^{d_1}s_{i'} \lesssim \frac{1}{\bar{p}d_1}\left( \sqrt{\frac{(\mu(\bfM^\star))^4 q^4 \log (\bar{d})}{\bar{p}d_1}} + \frac{(\mu(\bfM^\star))^2 q^2 \log (\bar{d})}{\bar{p}d_1}\right) \ll \frac{1}{\bar{p}d_1}.
\end{align*}

Now, we show that $\widehat{w}_{i,k}$ is close to $\bbE [w^{\star,2}_{i,k}]$ with high probability. Observe that
\begin{align*}
&\widehat{\bfU}^{\mathrm{NR},1}_{i,\cdot}\left(\sum_{i'=1}^{d_1} \frac{1}{p_{i'}}  \left(\frac{1}{\sigma(\widehat{M}_{i',k})(1-\sigma(\widehat{M}_{i',k}))} \right)\widehat{\bfU}^{\mathrm{NR},1\top}_{i',\cdot}\widehat{\bfU}^{\mathrm{NR},1}_{i',\cdot}\right)\widehat{\bfU}^{\mathrm{NR},1\top}_{i,\cdot}\\
    &\quad  -\bfU^{\bfM^\star}_{i,\cdot}\left(\sum_{i'=1}^{d_1}\frac{1}{p_{i'}} \left(\frac{1}{\sigma(M^\star_{i',k})(1-\sigma(M^\star_{i',k}))}\right)\bfU^{\bfM^\star \top}_{i',\cdot}\bfU^{\bfM^\star}_{i',\cdot}\right)\bfU^{\bfM^\star \top}_{i,\cdot}\\
    & = \widehat{\bfU}^{\mathrm{NR},1}_{i,\cdot}\left(\sum_{i'=1}^{d_1} \frac{1}{p_{i'}} \left(\frac{1}{\sigma(\widehat{M}_{i',k})(1-\sigma(\widehat{M}_{i',k}))} \right)\widehat{\bfU}^{\mathrm{NR},1\top}_{i',\cdot}\widehat{\bfU}^{\mathrm{NR},1}_{i',\cdot}\right)\widehat{\bfU}^{\mathrm{NR},1\top}_{i,\cdot} \\
    & \quad - \widehat{\bfU}^{\mathrm{NR},1}_{i,\cdot}\left(\sum_{i'=1}^{d_1} \frac{1}{p_{i'}} \left(\frac{1}{\sigma(M^\star_{i',k})(1-\sigma(M^\star_{i',k}))}\right)\widehat{\bfU}^{\mathrm{NR},1\top}_{i',\cdot}\widehat{\bfU}^{\mathrm{NR},1}_{i',\cdot}\right)\widehat{\bfU}^{\mathrm{NR},1\top}_{i,\cdot}\\
    & \quad + \widehat{\bfU}^{\mathrm{NR},1}_{i,\cdot}\left(\sum_{i'=1}^{d_1} \frac{1}{p_{i'}} \left(\frac{1}{\sigma(M^\star_{i',k})(1-\sigma(M^\star_{i',k}))}\right)\widehat{\bfU}^{\mathrm{NR},1\top}_{i',\cdot}\widehat{\bfU}^{\mathrm{NR},1}_{i',\cdot}\right)\widehat{\bfU}^{\mathrm{NR},1\top}_{i,\cdot} \\
    & \quad -\bfU^{\bfM^\star}_{i,\cdot}\left(\sum_{i'=1}^{d_1}\frac{1}{p_{i'}} \left(\frac{1}{\sigma(M^\star_{i',k})(1-\sigma(M^\star_{i',k}))}\right)\bfU^{\bfM^\star \top}_{i',\cdot}\bfU^{\bfM^\star}_{i',\cdot}\right)\bfU^{\bfM^\star \top}_{i,\cdot}.
\end{align*}
 First, note that
\begin{align*}
&\widehat{\bfU}^{\mathrm{NR},1}_{i,\cdot}\left(\sum_{i'=1}^{d_1} \frac{1}{p_{i'}} \left(\frac{1}{\sigma(\widehat{M}_{i',k})(1-\sigma(\widehat{M}_{i',k}))} \right)\widehat{\bfU}^{\mathrm{NR},1\top}_{i',\cdot}\widehat{\bfU}^{\mathrm{NR},1}_{i',\cdot}\right)\widehat{\bfU}^{\mathrm{NR},1\top}_{i,\cdot} \\
    & \quad - \widehat{\bfU}^{\mathrm{NR},1}_{i,\cdot}\left(\sum_{i'=1}^{d_1} \frac{1}{p_{i'}} \left(\frac{1}{\sigma(M^\star_{i',k})(1-\sigma(M^\star_{i',k}))}\right)\widehat{\bfU}^{\mathrm{NR},1\top}_{i',\cdot}\widehat{\bfU}^{\mathrm{NR},1}_{i',\cdot}\right)\widehat{\bfU}^{\mathrm{NR},1\top}_{i,\cdot}\\
    & = \sum_{i'=1}^{d_1} \frac{1}{p_{i'}} \left(  \frac{1}{\sigma(\widehat{M}_{i',k})(1-\sigma(\widehat{M}_{i',k}))}  - \frac{1}{\sigma(M^\star_{i',k})(1-\sigma(M^\star_{i',k}))} \right)\widehat{\bfU}^{\mathrm{NR},1}_{i,\cdot}\widehat{\bfU}^{\mathrm{NR},1\top}_{i',\cdot}\widehat{\bfU}^{\mathrm{NR},1}_{i',\cdot}\widehat{\bfU}^{\mathrm{NR},1\top}_{i,\cdot}.
\end{align*}
As shown in the proof of Proposition \ref{prop:feasibleCLT}, with probability at least $1-O(\bar{d}^{-10}),$
\begin{align*}
  \max_{i,k}\bigg\lvert\frac{1}{\sigma(\widehat{M}_{i',k})(1-\sigma(\widehat{M}_{i',k}))}  - \frac{1}{\sigma(M^\star_{i',k})(1-\sigma(M^\star_{i',k}))} \bigg\rvert \lesssim \frac{\kappa \mu R \sigma^\star_{\max} }{\min\{d_1, d_2(d_2-1)/2\}} \sqrt{\frac{  \bar{d}\log(\bar{d})}{\bar{p} (\sigma^\star_{\min})^2 }} .
\end{align*}
Therefore, by Lemma \ref{lem:inco_estimatedSV}, with probability at least $1-O(\bar{d}^{-9}),$
\begin{align*}
    &\sum_{i'=1}^{d_1} \frac{1}{p_{i'}}  \left(  \frac{1}{\sigma(\widehat{M}_{i',k})(1-\sigma(\widehat{M}_{i',k}))}  - \frac{1}{\sigma(M^\star_{i',k})(1-\sigma(M^\star_{i',k}))} \right)\widehat{\bfU}^{\mathrm{NR},1}_{i,\cdot}\widehat{\bfU}^{\mathrm{NR},1\top}_{i',\cdot}\widehat{\bfU}^{\mathrm{NR},1}_{i',\cdot}\widehat{\bfU}^{\mathrm{NR},1\top}_{i,\cdot}\\
    &\lesssim \frac{1}{\bar{p}}d_1 \frac{\kappa \mu R \sigma^\star_{\max} }{\min\{d_1, d_2(d_2-1)/2\}} \sqrt{\frac{  \bar{d}\log(\bar{d})}{\bar{p} (\sigma^\star_{\min})^2 }} \frac{(\mu(\bfM^\star))^2 q^2}{d_1^2}\\
    & = \frac{1}{\bar{p} d_1} \frac{(\mu(\bfM^\star))^2 q^2\kappa \mu R \sigma^\star_{\max} }{\min\{d_1, d_2(d_2-1)/2\}} \sqrt{\frac{  \bar{d}\log(\bar{d})}{\bar{p} (\sigma^\star_{\min})^2 }} \ll  \frac{1}{\bar{p} d_1}.
\end{align*}
Next, using Lemma \ref{lem:perturbedsingularvector} and Lemma \ref{lem:inco_estimatedSV}, we can bound the error in the estimated singular space. Note that
\begin{align*}
    &\widehat{\bfU}^{\mathrm{NR},1}_{i,\cdot}\widehat{\bfU}^{\mathrm{NR},1\top}_{i',\cdot}\widehat{\bfU}^{\mathrm{NR},1}_{i',\cdot} \widehat{\bfU}^{\mathrm{NR},1\top}_{i,\cdot}-\bfU^{\bfM^\star}_{i,\cdot}\bfU^{\bfM^\star \top}_{i',\cdot}\bfU^{\bfM^\star}_{i',\cdot} \bfU^{\bfM^\star \top}_{i,\cdot} \\
    &=\widehat{\bfU}^{\mathrm{NR},1}_{i,\cdot}\widehat{\bfU}^{\mathrm{NR},1\top}_{i',\cdot} \left( \widehat{\bfU}^{\mathrm{NR},1}_{i',\cdot} \widehat{\bfU}^{\mathrm{NR},1\top}_{i,\cdot}-\bfU^{\bfM^\star}_{i',\cdot} \bfU^{\bfM^\star \top}_{i,\cdot}\right) + \left(\widehat{\bfU}^{\mathrm{NR},1}_{i,\cdot}\widehat{\bfU}^{\mathrm{NR},1\top}_{i',\cdot}- \bfU^{\bfM^\star}_{i,\cdot}\bfU^{\bfM^\star \top}_{i',\cdot}\right)\bfU^{\bfM^\star}_{i',\cdot} \bfU^{\bfM^\star \top}_{i,\cdot}
\end{align*}
and therefore 
\begin{align*}
& \sum_{i'=1}^{d_1} \frac{1}{p_{i'}}  \left(\frac{1}{\sigma(M^\star_{i',k})(1-\sigma(M^\star_{i',k}))}\right)\left(\widehat{\bfU}^{\mathrm{NR},1}_{i,\cdot}\widehat{\bfU}^{\mathrm{NR},1\top}_{i',\cdot}\widehat{\bfU}^{\mathrm{NR},1}_{i',\cdot}\widehat{\bfU}^{\mathrm{NR},1\top}_{i,\cdot} -\bfU^{\bfM^\star}_{i,\cdot}\bfU^{\bfM^\star \top}_{i',\cdot}\bfU^{\bfM^\star}_{i',\cdot} \bfU^{\bfM^\star \top}_{i,\cdot}\right)\\
& = \sum_{i'=1}^{d_1} \frac{1}{p_{i'}}  \left(\frac{1}{\sigma(M^\star_{i',k})(1-\sigma(M^\star_{i',k}))}\right) \widehat{\bfU}^{\mathrm{NR},1}_{i,\cdot}\widehat{\bfU}^{\mathrm{NR},1\top}_{i',\cdot} \left( \widehat{\bfU}^{\mathrm{NR},1}_{i',\cdot} \widehat{\bfU}^{\mathrm{NR},1\top}_{i,\cdot}-\bfU^{\bfM^\star}_{i',\cdot} \bfU^{\bfM^\star \top}_{i,\cdot}\right) \\
& \quad + \sum_{i'=1}^{d_1} \frac{1}{p_{i'}} \left(\frac{1}{\sigma(M^\star_{i',k})(1-\sigma(M^\star_{i',k}))}\right) \bfU^{\bfM^\star}_{i',\cdot} \bfU^{\bfM^\star \top}_{i,\cdot} \left(\widehat{\bfU}^{\mathrm{NR},1}_{i,\cdot}\widehat{\bfU}^{\mathrm{NR},1\top}_{i',\cdot}- \bfU^{\bfM^\star}_{i,\cdot}\bfU^{\bfM^\star \top}_{i',\cdot}\right).
\end{align*}
By applying the Cauchy--Schwarz inequality, we obtain
\begin{align*}
    &\sum_{i'=1}^{d_1} \frac{1}{p_{i'}}  \left(\frac{1}{\sigma(M^\star_{i',k})(1-\sigma(M^\star_{i',k}))}\right) \widehat{\bfU}^{\mathrm{NR},1}_{i,\cdot}\widehat{\bfU}^{\mathrm{NR},1\top}_{i',\cdot} \left( \widehat{\bfU}^{\mathrm{NR},1}_{i',\cdot} \widehat{\bfU}^{\mathrm{NR},1\top}_{i,\cdot}-\bfU^{\bfM^\star}_{i',\cdot} \bfU^{\bfM^\star \top}_{i,\cdot}\right)\\
    &\lesssim \frac{1}{\bar{p}} \sqrt{\frac{\mu(\bfM^\star)q}{d_1}}  \sqrt{\frac{\kappa^2 R\bar{d}}{\bar{p}(\sigma_{\min}(\bfM^\star))^2  }}\sqrt{\frac{\mu(\bfM^\star)q}{\min\{d_1, d_2(d_2-1)/2\}}} \\
    &\leq \frac{1}{\bar{p}d_1} \sqrt{\frac{(\mu(\bfM^\star))^2q^2\kappa^2 R\bar{d}}{\bar{p}(\sigma_{\min}(\bfM^\star))^2  }}\sqrt{\frac{\max\{d_1,d_2(d_2-1)/2\}}{\min\{d_1,d_2(d_2-1)/2\}}} \ll \frac{1}{\bar{p}d_1}.
\end{align*}
The term $$\sum_{i'=1}^{d_1} \frac{1}{p_{i'}} \left(\frac{1}{\sigma(M^\star_{i',k})(1-\sigma(M^\star_{i',k}))}\right) \bfU^{\bfM^\star}_{i',\cdot} \bfU^{\bfM^\star \top}_{i,\cdot} \left(\widehat{\bfU}^{\mathrm{NR},1}_{i,\cdot}\widehat{\bfU}^{\mathrm{NR},1\top}_{i',\cdot}- \bfU^{\bfM^\star}_{i,\cdot}\bfU^{\bfM^\star \top}_{i',\cdot}\right)$$ can be bounded in a similar way.
\end{proof}

\subsection{Proof of Theorem \ref{thm:individualbootstrap}} \label{sec:proofforindivranking}
Theorem \ref{thm:individualbootstrap} follows from the following Lemma \ref{lem:bootstrap_indiv}, \ref{lem:TTindiv}, and \ref{lem:GGindiv}. We impose the following technical assumption. 
\begin{assumption}\label{asp:indivrankingbootstrap}
\begin{align*}
      \mu(\bfM^\star)^5 q^5 \kappa^4 \mu^2 R^2 \log^4(\bar{d})  \ll \bar{p} \bar{d}, \quad \text{and}\quad  (\mu(\bfM^\star))^3q^3\kappa^2 R\bar{d}\log(\bar{d})         \ll \bar{p}(\sigma_{\min}(\bfM^\star))^2.
\end{align*}

\end{assumption}

\begin{lemma}\label{lem:bootstrap_indiv}
   Suppose the assumptions in Theorem \ref{thm:individualbootstrap} hold. For $\alpha \in (0,1),$ we have
    \begin{align*}
         \big| \bbP(\calT^{i}_{\calJ, \calK}\leq \calG^{i}_{\calJ, \calK;1-\alpha}) - \alpha \big| \lesssim \left( \frac{(\mu(\bfM^\star))^4 q^4 \log^5(\bar{d}d_2)}{\bar{p}\bar{d}}\right)^{\frac{1}{4}}. 
    \end{align*}
\end{lemma}
\begin{proof}
 By  definition,
\begin{align*}
    \calT^{i}_{\calJ,\calK} \coloneqq \max_{j \in \calJ} \max_{j' \in \calK \setminus \{j\}} \Bigg| \frac{1}{\sqrt{\bar{d}}}\left(\sum_{k'=1}^{d_2(d_2-1)/2}   \sqrt{\frac{\bar{d}}{w^\star_{i,\bar{\calL}(j,j')}}} \xi_{i,\bar{\calL}(j,j');k'}+\sum_{i'=1}^{d_1}   \sqrt{\frac{\bar{d}}{w^\star_{i,\bar{\calL}(j,j')}}} \nu_{i,\bar{\calL}(j,j');i'}  \right)\Bigg|.
\end{align*}

We verify Condition E and Condition M from \cite{chernozhuokov2022improved} and invoke Theorem 2.2 therein. Denote
\begin{align*}
    X_{i,\bar{\calL}(j,j');l}= \begin{cases}
        \sqrt{\frac{\bar{d}}{w^\star_{i,\bar{\calL}(j,j')}}} \xi_{i,\bar{\calL}(j,j');l} \quad \text{for $l=1,\ldots, d_2(d_2-1)/2$};\\
        \sqrt{\frac{\bar{d}}{w^\star_{i,\bar{\calL}(j,j')}}} \nu_{i,\bar{\calL}(j,j');l-d_2(d_2-1)/2} \quad \text{for $l=d_2(d_2-1)/2+1,\ldots, d_1$}.
    \end{cases}
\end{align*}
Then, it can be simplified to
\begin{align*}
    \calT^{i}_{\calJ, \calK}=  \max_{j \in \calJ} \max_{j' \in \calK \setminus \{j\}} \Bigg| \frac{1}{\sqrt{\bar{d}}} \sum_{l=1}^{\bar{d}} X_{i,\bar{\calL}(j,j');l} \Bigg|.
\end{align*}

Note that, for $i,j,j'$, and $l,$
\begin{align*}
    &\bbE [X_{i,\bar{\calL}(j,j');l}]=0; \quad \bbE \bigg[\exp \left( \frac{|X_{i,\bar{\calL}(j,j');l}|}{B_{\bar{d}}}\right)\bigg]\leq 2; \\
    & b_1^2 \leq \frac{1}{\bar{d}}\sum_{l=1}^{\bar{d}} \bbE [X_{i,\bar{\calL}(j,j');l}^2]; \quad \frac{1}{\bar{d}}\sum_{l=1}^{\bar{d}} \bbE [X_{i,\bar{\calL}(j,j');l}^4]\leq B_{\bar{d}}^2 b_2^2
\end{align*}
are satisfied with the choice of some constants $b_1, b_2 \asymp 1$ with $b_1 \leq b_2$ and some sequence $B_{\bar{d}}\geq 1$ for all $\bar{d}\geq 1$ such that $B_{\bar{d}} =C(\mu(\bfM^\star))^2 q^2/\sqrt{\bar{p}}$ for some sufficiently large $C>0$. Then, by Theorem 2.2 in \cite{chernozhuokov2022improved}, for $\alpha\in(0,1)$, we have $$\big| \bbP(\calT^{i}_{\calJ,\calK}\leq \calG^{i}_{\calJ,\calK;1-\alpha}) - \alpha \big| \lesssim \left( \frac{(\mu(\bfM^\star))^4 q^4\log^5(\bar{d}d_2)}{\bar{p}\bar{d}}\right)^{\frac{1}{4}}.$$

\end{proof}

\begin{lemma}\label{lem:TTindiv}
Suppose the assumptions in Theorem \ref{thm:individualbootstrap} hold. Then, 
    \begin{align*}
        \sup_{z\in \bbR}\big| \bbP(\widehat{\calT}^i_{\calJ,\calK} \leq z)- \bbP(\calT^{i}_{\calJ,\calK}\leq z) \big| \rightarrow 0.
    \end{align*}
\end{lemma}
\begin{proof}
    By the definition of $\widehat{\calT}^i_{\calJ,\calK}$ and $\calT^{i}_{\calJ,\calK}$,
\begin{align*}
&|\widehat{\calT}^i_{\calJ,\calK}-\calT^{i}_{\calJ,\calK}| \\
&\leq \max_{j \in \calJ} \max_{j' \in \calK \setminus \{j\}}\bigg|\frac{ 1 }{\sqrt{\widehat{w}_{i,\bar{\calL}(j,j')}}}\left(\widehat{M}^{\proj}_{i,\bar{\calL}(j,j')}  -M^\star_{i,\bar{\calL}(j,j')} \right)-\frac{1}{\sqrt{w^\star_{i,\bar{\calL}(j,j')}}} \left( \sum_{k'=1}^{d_2(d_2-1)/2} \xi_{i,\bar{\calL}(j,j');k'} + \sum_{i'=1}^{d_1}\nu_{i,\bar{\calL}(j,j');i'}\right)  \bigg|.    
\end{align*}
Note that
\begin{align*}
    &\bigg|\frac{ 1 }{\sqrt{\widehat{w}_{i,\bar{\calL}(j,j')}}}\left(\widehat{M}^{\proj}_{i,\bar{\calL}(j,j')}  -M^\star_{i,\bar{\calL}(j,j')} \right)-\frac{1}{\sqrt{w^\star_{i,\bar{\calL}(j,j')}}} \left( \sum_{k'=1}^{d_2(d_2-1)/2} \xi_{i,\bar{\calL}(j,j');k'} + \sum_{i'=1}^{d_1}\nu_{i,\bar{\calL}(j,j');i'}\right)  \bigg|\\
    &\leq \bigg|\frac{ 1 }{\sqrt{\widehat{w}_{i,\bar{\calL}(j,j')}}}\bigg|\bigg|\left(\widehat{M}^{\proj}_{i,\bar{\calL}(j,j')}  -M^\star_{i,\bar{\calL}(j,j')} \right)-  \left( \sum_{k'=1}^{d_2(d_2-1)/2} \xi_{i,\bar{\calL}(j,j');k'} + \sum_{i'=1}^{d_1}\nu_{i,\bar{\calL}(j,j');i'}\right)   \bigg|\\
    & \quad +\bigg|\frac{1}{\sqrt{w^\star_{i,\bar{\calL}(j,j')}}} \left( \sum_{k'=1}^{d_2(d_2-1)/2} \xi_{i,\bar{\calL}(j,j');k'} + \sum_{i'=1}^{d_1}\nu_{i,\bar{\calL}(j,j');i'}\right) \bigg| \bigg|1-\sqrt{\frac{w^\star_{i,\bar{\calL}(j,j')}}{\widehat{w}_{i,\bar{\calL}(j,j')}}} \bigg|.
\end{align*}
We can show that, with probability at least $1-O(\bar{d}^{-9})$,
\begin{align*}
   \max_{j \in \calJ} \max_{j' \in \calK \setminus \{j\}} \bigg|1-\sqrt{\frac{w^\star_{i,\bar{\calL}(j,j')}}{\widehat{w}_{i,\bar{\calL}(j,j')}}} \bigg| \leq \max_{j \in \calJ} \max_{j' \in \calK \setminus \{j\}} \bigg|1-\frac{w^\star_{i,\bar{\calL}(j,j')}}{\widehat{w}_{i,\bar{\calL}(j,j')}} \bigg| \ll 1/\sqrt{\log(\bar{d})}
\end{align*}
by following the proofs Proposition \ref{prop:feasibleCLT_indiv}.
As a result, with probability at least $1-O(\bar{d}^{-9})$,
 \begin{align*}
    &\max_{j \in \calJ} \max_{j' \in \calK \setminus \{j\}} \bigg|\frac{1}{\sqrt{w^\star_{i,\bar{\calL}(j,j')}}} \left( \sum_{k'=1}^{d_2(d_2-1)/2} \xi_{i,\bar{\calL}(j,j');k'} + \sum_{i'=1}^{d_1}\nu_{i,\bar{\calL}(j,j');i'}\right) \bigg| \bigg|1-\sqrt{\frac{w^\star_{i,\bar{\calL}(j,j')}}{\widehat{w}_{i,\bar{\calL}(j,j')}}} \bigg|=o(1).
 \end{align*}
In addition, by following the proof of Theorem \ref{thm:asymptoticnormality_indiv}, we have 
\begin{align*}
   \max_{j \in \calJ} \max_{j' \in \calK \setminus \{j\}} \bigg|\frac{ 1 }{\sqrt{\widehat{w}_{i,\bar{\calL}(j,j')}}}\bigg|\bigg|\left(\widehat{M}^{\proj}_{i,\bar{\calL}(j,j')}  -M^\star_{i,\bar{\calL}(j,j')} \right)-  \left( \sum_{k'=1}^{d_2(d_2-1)/2} \xi_{i,\bar{\calL}(j,j');k'} + \sum_{i'=1}^{d_1}\nu_{i,\bar{\calL}(j,j');i'}\right)   \bigg| =o(1)
\end{align*}
with probability at least $1-O(\bar{d}^{-9})$. Then, for any $\epsilon>0,$
\begin{align*}
    \sup_{z \in \bbR} \big| \bbP(\widehat{\calT}^i_{\calJ, \calK}\leq z)-\bbP(\calT^{i}_{\calJ,\calK} \leq z)  \big| \leq \bbP (|\widehat{\calT}^i_{\calJ,\calK}-\calT^{i}_{\calJ,\calK}|>\epsilon) + \sup_{z \in \bbR} \bbP (z < \calT^{i}_{\calJ,\calK} < z+\epsilon).
\end{align*}
\end{proof}

\begin{lemma}\label{lem:GGindiv}
    Suppose the assumptions in Theorem \ref{thm:individualbootstrap} hold.
    \begin{align*}
        \sup_{z\in \bbR}\big| \bbP(\widehat{\calG}^i_{\calJ,\calK}\leq z)- \bbP(\calG^{i}_{\calJ,\calK}\leq z) \big| \rightarrow 0.
    \end{align*}
\end{lemma}
\begin{proof}
By definition,
\begin{align*}
     |\widehat{\calG}^i_{\calJ,\calK}-\calG^{i}_{\calJ,\calK}| &\leq \max_{j \in \calJ} \max_{j' \in \calK \setminus \{j\}} \bigg|\frac{1}{\sqrt{w^\star_{i,\bar{\calL}(j,j')}}} \left( \sum_{k'=1}^{d_2(d_2-1)/2} \xi_{i,\bar{\calL}(j,j');k'} Z^\xi_{k'} + \sum_{i'=1}^{d_1}\nu_{i,\bar{\calL}(j,j');i'}Z_{i'}^\nu\right) \\
    & \quad -\frac{1}{\sqrt{\widehat{w}_{i,\bar{\calL}(j,j')}}} \left( \sum_{k'=1}^{d_2(d_2-1)/2} \widehat{\xi}_{i,\bar{\calL}(j,j');k'} Z^\xi_{k'} + \sum_{i'=1}^{d_1}\widehat{\nu}_{i,\bar{\calL}(j,j');i'}Z_{i'}^\nu\right) \bigg|.
\end{align*}
Note that 
\begin{align*}
   &  \bigg|\frac{1}{\sqrt{w^\star_{i,\bar{\calL}(j,j')}}} \left( \sum_{k'=1}^{d_2(d_2-1)/2} \xi_{i,\bar{\calL}(j,j');k'} Z^\xi_{k'} + \sum_{i'=1}^{d_1}\nu_{i,\bar{\calL}(j,j');i'}Z_{i'}^\nu\right) \\
    & \quad -\frac{1}{\sqrt{\widehat{w}_{i,\bar{\calL}(j,j')}}} \left( \sum_{k'=1}^{d_2(d_2-1)/2} \widehat{\xi}_{i,\bar{\calL}(j,j');k'} Z^\xi_{k'} + \sum_{i'=1}^{d_1}\widehat{\nu}_{i,\bar{\calL}(j,j');i'}Z_{i'}^\nu\right) \bigg|\\
    & =\bigg| \sum_{k'=1}^{d_2(d_2-1)/2} \Delta^\xi_{i,\bar{\calL}(j,j');k'} Z^\xi_{k'} + \sum_{i'=1}^{d_1} \Delta^\nu_{i,\bar{\calL}(j,j');i'} Z^\nu_{i'} \bigg|
\end{align*}
    where 
    \begin{align*}
        \Delta^\xi_{i,\bar{\calL}(j,j');k'} \coloneqq \frac{1}{\sqrt{\widehat{w}_{i,\bar{\calL}(j,j')}}}  \widehat{\xi}_{i,\bar{\calL}(j,j');k'}-\frac{1}{\sqrt{w^\star_{i,\bar{\calL}(j,j')}}}  \xi_{i,\bar{\calL}(j,j');k'};\\
        \Delta^\nu_{i,\bar{\calL}(j,j');i'} \coloneqq \frac{1}{\sqrt{\widehat{w}_{i,\bar{\calL}(j,j')}}}  \widehat{\nu}_{i,\bar{\calL}(j,j');i'}-\frac{1}{\sqrt{w^\star_{i,\bar{\calL}(j,j')}}}  \nu_{i,\bar{\calL}(j,j');i'}.
    \end{align*}

We focus on bounding $\sum_{k'=1}^{d_2(d_2-1)/2} \Delta^\xi_{i,\bar{\calL}(j,j');k'} Z^\xi_{k'}$ as the argument for $\sum_{i'=1}^{d_1} \Delta^\nu_{i,\bar{\calL}(j,j');i'} Z^\nu_{i'} $ is similar. We begin by a decomposition
\begin{align*}
    &\sum_{k'=1}^{d_2(d_2-1)/2} \Delta^\xi_{i,\bar{\calL}(j,j');k'} Z^\xi_{k'} \\
    &=  \underbrace{\frac{1}{\sqrt{\widehat{w}_{i,\bar{\calL}(j,j')}}} \sum_{k'=1}^{d_2(d_2-1)/2} (\widehat{\xi}_{i,\bar{\calL}(j,j');k'}-\xi_{i,\bar{\calL}(j,j');k'})Z^\xi_{k'}}_{\coloneqq a} + \underbrace{\frac{1}{\sqrt{\widehat{w}_{i,\bar{\calL}(j,j')}}} \sum_{k'=1}^{d_2(d_2-1)/2} \left(1-\sqrt{\frac{\widehat{w}_{i,\bar{\calL}(j,j')}}{w^\star_{i,\bar{\calL}(j,j')}}}\right)\xi_{i,\bar{\calL}(j,j');k'}Z^\xi_{k'}. }_{\coloneqq b}
\end{align*}
For $b$, we apply the incoherence condition and the Bernstein inequality \citep{koltchinskii:2011}. Note that $$\bbE \left(1-\sqrt{\frac{\widehat{w}_{i,\bar{\calL}(j,j')}}{w^\star_{i,\bar{\calL}(j,j')}}}\right)\xi_{i,\bar{\calL}(j,j');k'}Z^\xi_{k'}=0.$$ By following the proof of Proposition \ref{prop:feasibleCLT_indiv}, we have
\begin{align*}
   & \max_{j \in \calJ} \max_{j' \in \calK \setminus \{j\}} \bigg|1-\sqrt{\frac{w^\star_{i,\bar{\calL}(j,j')}}{\widehat{w}_{i,\bar{\calL}(j,j')}}} \bigg| \leq \max_{j \in \calJ} \max_{j' \in \calK \setminus \{j\}} \bigg|1-\frac{w^\star_{i,\bar{\calL}(j,j')}}{\widehat{w}_{i,\bar{\calL}(j,j')}} \bigg| \\
   & \quad  \lesssim  \sqrt{\frac{(\mu(\bfM^\star))^4 q^4 \log (\bar{d})}{\bar{p}\bar{d}}} +  (\mu(\bfM^\star))^2 q^2\kappa^2 \mu R    \sqrt{\frac{   \log(\bar{d})}{\bar{p} \bar{d}  }}  + \sqrt{\frac{(\mu(\bfM^\star))^2q^2\kappa^2 R\bar{d}}{\bar{p}(\sigma_{\min}(\bfM^\star))^2  }} 
\end{align*}
with probability at least $1-O(\bar{d}^{-9})$. By invoking Lemma \ref{lem:inco_estimatedSV}, we have, with probability at least $1-O(\bar{d}^{-9})$, 
\begin{align*}
    &\norm{\left(1-\sqrt{\frac{\widehat{w}_{i,\bar{\calL}(j,j')}}{w^\star_{i,\bar{\calL}(j,j')}}}\right)\xi_{i,\bar{\calL}(j,j');k'}Z^\xi_{k'}}_{\mathrm{subE}} \\
    & \quad \lesssim \left(\sqrt{\frac{(\mu(\bfM^\star))^4 q^4 \log (\bar{d})}{\bar{p}\bar{d}}} +  (\mu(\bfM^\star))^2 q^2\kappa^2 \mu R    \sqrt{\frac{   \log(\bar{d})}{\bar{p} \bar{d}  }}  + \sqrt{\frac{(\mu(\bfM^\star))^2q^2\kappa^2 R\bar{d}}{\bar{p}(\sigma_{\min}(\bfM^\star))^2  }}  \right) \frac{\mu(\bfM^\star)q}{ \bar{p}\bar{d}} ;\\
    &  \norm{\bbE \bigg[\sum_{k'=1}^{d_2(d_2-1)/2} \left(\left(1-\sqrt{\frac{\widehat{w}_{i,\bar{\calL}(j,j')}}{w^\star_{i,\bar{\calL}(j,j')}}}\right)\xi_{i,\bar{\calL}(j,j');k'}Z^\xi_{k'}\right)^2  \bigg]} \\
    & \quad \lesssim \left(\sqrt{\frac{(\mu(\bfM^\star))^4 q^4 \log (\bar{d})}{\bar{p}\bar{d}}} +  (\mu(\bfM^\star))^2 q^2\kappa^2 \mu R    \sqrt{\frac{   \log(\bar{d})}{\bar{p} \bar{d}  }}  + \sqrt{\frac{(\mu(\bfM^\star))^2q^2\kappa^2 R\bar{d}}{\bar{p}(\sigma_{\min}(\bfM^\star))^2  }}  \right)^2   
  \frac{\mu(\bfM^\star)q}{\bar{p}\bar{d}}
\end{align*}
where $\norm{\cdot}_{\mathrm{subE}}$ is the sub-exponential norm. Therefore, with probability at least $1-O(\bar{d}^{-9})$, 
\begin{align*}
    |b| &\lesssim   \left(\sqrt{\frac{(\mu(\bfM^\star))^4 q^4 \log (\bar{d})}{\bar{p}\bar{d}}} +  (\mu(\bfM^\star))^2 q^2\kappa^2 \mu R    \sqrt{\frac{   \log(\bar{d})}{\bar{p} \bar{d}  }}  + \sqrt{\frac{(\mu(\bfM^\star))^2q^2\kappa^2 R\bar{d}}{\bar{p}(\sigma_{\min}(\bfM^\star))^2  }}  \right) \\
    & \quad \times \left(\sqrt{\frac{\mu(\bfM^\star)q \log(\bar{d})}{ \bar{p}\bar{d}} }+ \frac{\mu(\bfM^\star)q}{ \bar{p}\bar{d}} \log^2(\bar{d}) \right) \sqrt{\bar{p}\bar{d}} \\
    & \lesssim   \left(\sqrt{\frac{(\mu(\bfM^\star))^4 q^4 \log (\bar{d})}{\bar{p}\bar{d}}} +  (\mu(\bfM^\star))^2 q^2\kappa^2 \mu R    \sqrt{\frac{   \log(\bar{d})}{\bar{p} \bar{d}  }}  + \sqrt{\frac{(\mu(\bfM^\star))^2q^2\kappa^2 R\bar{d}}{\bar{p}(\sigma_{\min}(\bfM^\star))^2  }}  \right) \\
    & \quad \times  \sqrt{ \mu(\bfM^\star)q \log(\bar{d})}  \ll 1.
\end{align*}
Now, we analyze the term $|a|.$ We first decompose $a$ into two terms:
\begin{align*}
    &\frac{1}{\sqrt{\widehat{w}_{i,\bar{\calL}(j,j')}}} \sum_{k'=1}^{d_2(d_2-1)/2} (\widehat{\xi}_{i,\bar{\calL}(j,j');k'}-\xi_{i,\bar{\calL}(j,j');k'})Z^\xi_{k'} \\
    &=\frac{1}{\sqrt{\widehat{w}_{i,\bar{\calL}(j,j')}}} \sum_{k'=1}^{d_2(d_2-1)/2} \sum_{(i,k')\in S^1}\frac{1}{p_i}\delta^1_{i,k'}\left(\frac{  y^1_{i,k'}- \sigma(\widehat{M}^{2}_{i,k'}) }{\sigma'(\widehat{M}^2_{i,k'})} \widehat{\bfV}^{\mathrm{NR},2}_{k',\cdot}\widehat{\bfV}^{\mathrm{NR},2 \top}_{\bar{\calL}(j,j'),\cdot}   - \frac{y^1_{i,k'}- \sigma(M^\star_{i,k'})}{\sigma'(M^\star_{i,k'})} \bfV^{\bfM^\star}_{k',\cdot}\bfV^{\bfM^\star \top}_{\bar{\calL}(j,j'),\cdot} \right)Z^\xi_{k'} \\
    &  + \frac{1}{\sqrt{\widehat{w}_{i,\bar{\calL}(j,j')}}} \sum_{k'=1}^{d_2(d_2-1)/2} \sum_{(i,k')\in S^2}\frac{1}{p_i}\delta^2_{i,k'}\left(\frac{  y^2_{i,k'}- \sigma(\widehat{M}^1_{i,k'}) }{\sigma'(\widehat{M}^1_{i,k'})} \widehat{\bfV}^{\mathrm{NR},1}_{k',\cdot}\widehat{\bfV}^{\mathrm{NR},1 \top}_{\bar{\calL}(j,j'),\cdot}   - \frac{y^2_{i,k'}- \sigma(M^\star_{i,k'})}{\sigma'(M^\star_{i,k'})} \bfV^{\bfM^\star}_{k',\cdot}\bfV^{\bfM^\star \top}_{\bar{\calL}(j,j'),\cdot} \right)Z^\xi_{k'}.
\end{align*}
We will focus on the first term since the arguments are similar. We further decompose this term into two:
\begin{align*}
    &\frac{1}{\sqrt{\widehat{w}_{i,\bar{\calL}(j,j')}}} \sum_{k'=1}^{d_2(d_2-1)/2} \sum_{(i,k')\in S^1}\frac{1}{p_i}\delta^1_{i,k'}\left(\frac{  y^1_{i,k'}- \sigma(\widehat{M}^{2}_{i,k'}) }{\sigma'(\widehat{M}^2_{i,k'})} \widehat{\bfV}^{\mathrm{NR},2}_{k',\cdot}\widehat{\bfV}^{\mathrm{NR},2 \top}_{\bar{\calL}(j,j'),\cdot}   - \frac{y^1_{i,k'}- \sigma(M^\star_{i,k'})}{\sigma'(M^\star_{i,k'})} \bfV^{\bfM^\star}_{k',\cdot}\bfV^{\bfM^\star \top}_{\bar{\calL}(j,j'),\cdot} \right) Z^\xi_{k'} \\
    &=\frac{1}{\sqrt{\widehat{w}_{i,\bar{\calL}(j,j')}}} \sum_{k'=1}^{d_2(d_2-1)/2} \sum_{(i,k')\in S^1}\frac{1}{p_i}\delta^1_{i,k'}\left(\frac{  y^1_{i,k'}- \sigma(\widehat{M}^{2}_{i,k'}) }{\sigma'(\widehat{M}^2_{i,k'})} -\frac{y^1_{i,k'}- \sigma(M^\star_{i,k'})}{\sigma'(M^\star_{i,k'})}\right)\widehat{\bfV}^{\mathrm{NR},2}_{k',\cdot}\widehat{\bfV}^{\mathrm{NR},2 \top}_{\bar{\calL}(j,j'),\cdot}Z^\xi_{k'} \\
    & \quad + \frac{1}{\sqrt{\widehat{w}_{i,\bar{\calL}(j,j')}}} \sum_{k'=1}^{d_2(d_2-1)/2} \sum_{(i,k')\in S^1}\frac{1}{p_i}\delta^1_{i,k'} \frac{y^1_{i,k'}- \sigma(M^\star_{i,k'})}{\sigma'(M^\star_{i,k'})} \left(\widehat{\bfV}^{\mathrm{NR},2}_{k',\cdot}\widehat{\bfV}^{\mathrm{NR},2 \top}_{\bar{\calL}(j,j'),\cdot}   -  \bfV^{\bfM^\star}_{k',\cdot}\bfV^{\bfM^\star \top}_{\bar{\calL}(j,j'),\cdot}  \right)Z^\xi_{k'}
\end{align*}
Now we bound the two terms in turn using the Bernstein inequality \citep{koltchinskii:2011}. In addition to the Bernstein inequality, we invoke the entrywise error bound of $\widehat{\bfM}-\bfM^\star$ to further bound the first term, and Lemma \ref{lem:perturbedsingularvector} for the second term. Note first that
\begin{align*}
    \bbE  \bigg[\frac{1}{p_i}\delta^1_{i,k'}\left(\frac{  y^1_{i,k'}- \sigma(\widehat{M}^{2}_{i,k'}) }{\sigma'(\widehat{M}^2_{i,k'})} -\frac{y^1_{i,k'}- \sigma(M^\star_{i,k'})}{\sigma'(M^\star_{i,k'})}\right)\widehat{\bfV}^{\mathrm{NR},2}_{k',\cdot}\widehat{\bfV}^{\mathrm{NR},2 \top}_{\bar{\calL}(j,j'),\cdot}Z^\xi_{k'}\bigg] =0,
\end{align*}
and, with probability at least $1-O(\bar{d}^{-10}),$
\begin{align*}
    \max_{i,k'} \bigg| \frac{  y^1_{i,k'}- \sigma(\widehat{M}^{2}_{i,k'}) }{\sigma'(\widehat{M}^2_{i,k'})} -\frac{y^1_{i,k'}- \sigma(M^\star_{i,k'})}{\sigma'(M^\star_{i,k'})} \bigg| \lesssim \frac{\kappa \mu R \sigma^\star_{\max} }{\min\{d_1, d_2(d_2-1)/2\}} \sqrt{\frac{  \bar{d}\log(\bar{d})}{\bar{p} (\sigma^\star_{\min})^2 }}.
\end{align*}
The second result can be obtained by following the proof of Proposition \ref{prop:feasibleCLT} analogously. Then, by Lemma \ref{lem:inco_estimatedSV}, with probability at least $1-O(\bar{d}^{-9}),$ 
\begin{align*}
    & \norm{\frac{1}{p_i}\delta^1_{i,k'}\left(\frac{  y^1_{i,k'}- \sigma(\widehat{M}^{2}_{i,k'}) }{\sigma'(\widehat{M}^2_{i,k'})} -\frac{y^1_{i,k'}- \sigma(M^\star_{i,k'})}{\sigma'(M^\star_{i,k'})}\right)\widehat{\bfV}^{\mathrm{NR},2}_{k',\cdot}\widehat{\bfV}^{\mathrm{NR},2 \top}_{\bar{\calL}(j,j'),\cdot}Z^\xi_{k'}}_{\mathrm{subE}} \\
    & \quad \lesssim \frac{1}{\bar{p}}\frac{\kappa \mu R \sigma^\star_{\max} }{\min\{d_1, d_2(d_2-1)/2\}} \sqrt{\frac{  \bar{d}\log(\bar{d})}{\bar{p} (\sigma^\star_{\min})^2 }} \frac{\mu(\bfM^\star)q}{\bar{d}};\\
    &  \norm{ \bbE \bigg[\sum_{k'=1}^{d_2(d_2-1)/2} \sum_{(i,k')\in S^1} \left( \frac{1}{p_i}\delta^1_{i,k'}\left(\frac{  y^1_{i,k'}- \sigma(\widehat{M}^{2}_{i,k'}) }{\sigma'(\widehat{M}^2_{i,k'})} -\frac{y^1_{i,k'}- \sigma(M^\star_{i,k'})}{\sigma'(M^\star_{i,k'})}\right)\widehat{\bfV}^{\mathrm{NR},2}_{k',\cdot}\widehat{\bfV}^{\mathrm{NR},2 \top}_{\bar{\calL}(j,j'),\cdot}Z^\xi_{k'} \right)^2 \bigg]} \\
    & \quad \lesssim \frac{1}{\bar{p}}\left(\frac{\kappa \mu R \sigma^\star_{\max} }{\min\{d_1, d_2(d_2-1)/2\}} \sqrt{\frac{  \bar{d}\log(\bar{d})}{\bar{p} (\sigma^\star_{\min})^2 }} \right)^2\frac{\mu(\bfM^\star)q}{\bar{d}}.
\end{align*}
Therefore, with probability at least $1-O(\bar{d}^{-10}),$ by the Bernstein inequality \citep{koltchinskii:2011},
\begin{align*}
    &\bigg|\frac{1}{\sqrt{\widehat{w}_{i,\bar{\calL}(j,j')}}} \sum_{k'=1}^{d_2(d_2-1)/2} \sum_{(i,k')\in S^1}\frac{1}{p_i}\delta^1_{i,k'}\left(\frac{  y^1_{i,k'}- \sigma(\widehat{M}^{2}_{i,k'}) }{\sigma'(\widehat{M}^2_{i,k'})} -\frac{y^1_{i,k'}- \sigma(M^\star_{i,k'})}{\sigma'(M^\star_{i,k'})}\right)\widehat{\bfV}^{\mathrm{NR},2}_{k',\cdot}\widehat{\bfV}^{\mathrm{NR},2 \top}_{\bar{\calL}(j,j'),\cdot}Z^\xi_{k'}\bigg|\\
    & \lesssim \sqrt{\frac{\mu(\bfM^\star)q \log(\bar{d})}{\bar{p}\bar{d}}} \frac{\kappa \mu R \sigma^\star_{\max} }{\min\{d_1, d_2(d_2-1)/2\}} \sqrt{\frac{  \bar{d}\log(\bar{d})}{\bar{p} (\sigma^\star_{\min})^2 }}\sqrt{\bar{p}\bar{d}}\ll 1.
\end{align*}
Now, we bound the second term:
\begin{align*}
    \frac{1}{\sqrt{\widehat{w}_{i,\bar{\calL}(j,j')}}} \sum_{k'=1}^{d_2(d_2-1)/2} \sum_{(i,k')\in S^1} \frac{1}{p_i}\delta^1_{i,k'} \frac{y^1_{i,k'}- \sigma(M^\star_{i,k'})}{\sigma'(M^\star_{i,k'})} \left(\widehat{\bfV}^{\mathrm{NR},2}_{k',\cdot}\widehat{\bfV}^{\mathrm{NR},2 \top}_{\bar{\calL}(j,j'),\cdot}   -  \bfV^{\bfM^\star}_{k',\cdot}\bfV^{\bfM^\star \top}_{\bar{\calL}(j,j'),\cdot}  \right)Z^\xi_{k'}
\end{align*}
where
\begin{align*}
    \bbE  \bigg[\frac{1}{p_i}\delta^1_{i,k'} \frac{y^1_{i,k'}- \sigma(M^\star_{i,k'})}{\sigma'(M^\star_{i,k'})} \left(\widehat{\bfV}^{\mathrm{NR},2}_{k',\cdot}\widehat{\bfV}^{\mathrm{NR},2 \top}_{\bar{\calL}(j,j'),\cdot}   -  \bfV^{\bfM^\star}_{k',\cdot}\bfV^{\bfM^\star \top}_{\bar{\calL}(j,j'),\cdot}  \right)Z^\xi_{k'}\bigg] =0.
\end{align*}
Note that, with probability at least $1-O(\bar{d}^{-9}),$
\begin{align*}
    &\norm{ \frac{1}{p_i}\delta^1_{i,k'} \frac{y^1_{i,k'}- \sigma(M^\star_{i,k'})}{\sigma'(M^\star_{i,k'})} \left(\widehat{\bfV}^{\mathrm{NR},2}_{k',\cdot}\widehat{\bfV}^{\mathrm{NR},2 \top}_{\bar{\calL}(j,j'),\cdot}   -  \bfV^{\bfM^\star}_{k',\cdot}\bfV^{\bfM^\star \top}_{\bar{\calL}(j,j'),\cdot}  \right)Z^\xi_{k'} }_{\mathrm{subE}} \lesssim \frac{\mu(\bfM^\star)q}{\bar{p}\bar{d}};  \\
    &  \norm{\bbE \bigg[\sum_{k'=1}^{d_2(d_2-1)/2} \sum_{(i,k')\in S^1} \left(\frac{1}{p_i}\delta^1_{i,k'} \frac{y^1_{i,k'}- \sigma(M^\star_{i,k'})}{\sigma'(M^\star_{i,k'})} \left(\widehat{\bfV}^{\mathrm{NR},2}_{k',\cdot}\widehat{\bfV}^{\mathrm{NR},2 \top}_{\bar{\calL}(j,j'),\cdot}   -  \bfV^{\bfM^\star}_{k',\cdot}\bfV^{\bfM^\star \top}_{\bar{\calL}(j,j'),\cdot}  \right)Z^\xi_{k'} \right)^2 \bigg]  } \\
    & \quad \lesssim \frac{1}{ \bar{p} }  \frac{\kappa^2 R\bar{d}}{\bar{p}(\sigma_{\min}(\bfM^\star))^2  }\frac{\mu(\bfM^\star)q}{\min\{d_1, d_2(d_2-1)/2\}}
\end{align*}
where the first result is from Lemma \ref{lem:inco_estimatedSV}, and the second result is from Lemma \ref{lem:perturbedsingularvector}. As a result, we have, with probability at least $1-O(\bar{d}^{-9}),$
\begin{align*}
    & \bigg| \frac{1}{\sqrt{\widehat{w}_{i,\bar{\calL}(j,j')}}} \sum_{k'=1}^{d_2(d_2-1)/2} \sum_{(i,k')\in S^1} \frac{1}{p_i}\delta^1_{i,k'} \frac{y^1_{i,k'}- \sigma(M^\star_{i,k'})}{\sigma'(M^\star_{i,k'})} \left(\widehat{\bfV}^{\mathrm{NR},2}_{k',\cdot}\widehat{\bfV}^{\mathrm{NR},2 \top}_{\bar{\calL}(j,j'),\cdot}   -  \bfV^{\bfM^\star}_{k',\cdot}\bfV^{\bfM^\star \top}_{\bar{\calL}(j,j'),\cdot}  \right)Z^\xi_{k'} \bigg| \\
    & \lesssim  \frac{\mu(\bfM^\star)q \log^2(\bar{d})}{\sqrt{\bar{p}\bar{d}}} + \sqrt{\frac{\kappa^2 R \mu(\bfM^\star)q \bar{d} \log(\bar{d})}{\bar{p}(\sigma_{\min}(\bfM^\star))^2  }} \ll 1.
\end{align*}


Then, we have, with probability at least $1-O(\bar{d}^{-8})$, for any $\epsilon>0,$ $|\widehat{\calG}^i_{\calJ,\calK}-\calG^{i}_{\calJ,\calK}|<\epsilon.$ Therefore, for any $\epsilon>0,$ we have
\begin{align*}
    \sup_{z \in \bbR} \big| \bbP(\widehat{\calG}^i_{\calJ,\calK}\leq z)-\bbP(\calG^{i}_{\calJ,\calK} \leq z)  \big| \leq \bbP (|\widehat{\calG}^i_{\calJ,\calK}-\calG^{i}_{\calJ,\calK}|>\epsilon) + \sup_{z \in \bbR} \bbP (z < \calG^{i}_{\calJ,\calK} < z+\epsilon).
\end{align*}
\end{proof}

\subsection{Proof of Theorem \ref{thm:marketbootstrap}} \label{sec:marketbootstrapproof}

Theorem \ref{thm:marketbootstrap} follows from the following Lemma \ref{lem:bootstrap_market}, \ref{lem:TTmarket}, and \ref{lem:GGmarket}.

\begin{lemma}\label{lem:bootstrap_market}
   Suppose that the assumptions in Theorem \ref{thm:marketbootstrap} hold. For $\alpha \in (0,1),$ we have
    \begin{align*}
         \big| \bbP(\calT_{\calJ}\leq \calG_{\calJ;1-\alpha}) - \alpha \big| \lesssim \left( \frac{\log^5(d_1d_2)}{\bar{p}d_1}\right)^{\frac{1}{4}}. 
    \end{align*}
\end{lemma}
\begin{proof}
 Recall the definition
\begin{align*}
    \calT_{\calJ}\coloneqq \max_{j \in \calJ} \max_{j' \neq j} \bigg| \frac{1}{\sqrt{d_1}}\sum_{i=1}^{d_1} \underbrace{ \frac{1}{\sqrt{v^\star_{\bar{\calL}(j,j')}}} \frac{1}{p_i \sqrt{d_1}} \delta_{i,\bar{\calL}(j,j')}  \left(\sigma'(M^\star_{i,\bar{\calL}(j,j')}) \right)^{-1} \left( y_{i,\bar{\calL}(j,j')}- \sigma(M^\star_{i,\bar{\calL}(j,j')}) \right)}_{\coloneqq X_{i,\bar{\calL}(j,j')}} \bigg|.
\end{align*}
We aim to verify Condition E and Condition M in \cite{chernozhuokov2022improved}. Note that
\begin{align*}
    \bbE [X_{i,\bar{\calL}(j,j')}]=0; \quad \bbE \bigg[\exp \left( \frac{|X_{i,\bar{\calL}(j,j')}|}{B_{d_1}}\right)\bigg]\leq 2; \quad  b_1^2 \leq \frac{1}{d_1}\sum_{i=1}^{d_1} \bbE [X_{i,\bar{\calL}(j,j')}^2]; \quad \frac{1}{d_1}\sum_{i=1}^{d_1} \bbE [X_{i,\bar{\calL}(j,j')}^4]\leq B_{d_1}^2 b_2^2,
\end{align*}
are satisfied with the choice of some constants $b_1, b_2 \asymp 1$ with $b_1 \leq b_2$ and some sequence $B_{d_1}\geq 1$ for all $d_1\geq 1$ such that $B_{d_1}= C /\sqrt{\bar{p}}$ for some sufficiently large $C>0$. Then, by Theorem 2.2 in \cite{chernozhuokov2022improved}, for $\alpha\in(0,1)$, we have $$\big| \bbP(\calT_{\calJ}\leq \calG_{\calJ;1-\alpha}) - \alpha \big| \lesssim \left( \frac{\log^5(d_1d_2)}{\bar{p}d_1}\right)^{\frac{1}{4}}.$$
\end{proof}

\begin{lemma}\label{lem:TTmarket}
Suppose the assumptions in Theorem \ref{thm:marketbootstrap} hold. Then,
    \begin{align*}
        \sup_{z\in \bbR}\big| \bbP(\widehat{\calT}_{\calJ}\leq z)- \bbP(\calT_{\calJ}\leq z) \big| \rightarrow 0.
    \end{align*}
\end{lemma}
\begin{proof}
    By the definition of $\widehat{\calT}_{\calJ}$ and $\calT_{\calJ}$, 
\begin{align*}
|\widehat{\calT}_{\calJ}-\calT_{\calJ}| \leq \max_{j \in \calJ} \max_{j'\neq j} \bigg|\frac{ 1 }{\sqrt{\widehat{v}_{\bar{\calL}(j,j')}}}\left(\frac{1}{d_1} \sum_{i=1}^{d_1}\widehat{M}^{\mathrm{NR}}_{i,\bar{\calL}(j,j')}  -\frac{1}{d_1} \sum_{i=1}^{d_1}M^\star_{i,\bar{\calL}(j,j')} \right)-\frac{1}{\sqrt{v^\star_{\bar{\calL}(j,j')}}} \frac{1}{d_1}\sum_{i=1}^{d_1}  \xi_{i, \bar{\calL}(j,j')}   \bigg|    
\end{align*}
where $\xi_{i, \bar{\calL}(j,j')}$ is defined in Section \ref{sec:market_ranking}. Note that
\begin{align*}
    &\bigg|\frac{ 1 }{\sqrt{\widehat{v}_{\bar{\calL}(j,j')}}}\left(\frac{1}{d_1} \sum_{i=1}^{d_1}\widehat{M}^{\mathrm{NR}}_{i,\bar{\calL}(j,j')}  -\frac{1}{d_1} \sum_{i=1}^{d_1}M^\star_{i,\bar{\calL}(j,j')} \right)-\frac{1}{\sqrt{v^\star_{\bar{\calL}(j,j')}}} \frac{1}{d_1}\sum_{i=1}^{d_1}  \xi_{i, \bar{\calL}(j,j')}   \bigg|\\
    &\leq \bigg|\frac{ 1 }{\sqrt{\widehat{v}_{\bar{\calL}(j,j')}}}\bigg|\bigg|\left(\frac{1}{d_1} \sum_{i=1}^{d_1}\widehat{M}^{\mathrm{NR}}_{i,\bar{\calL}(j,j')}  -\frac{1}{d_1} \sum_{i=1}^{d_1}M^\star_{i,\bar{\calL}(j,j')} \right)-  \frac{1}{d_1}\sum_{i=1}^{d_1}  \xi_{i, \bar{\calL}(j,j')}  \bigg|\\
    & \quad +\bigg|\frac{1}{\sqrt{v^\star_{\bar{\calL}(j,j')}}} \frac{1}{d_1}\sum_{i=1}^{d_1}  \xi_{i, \bar{\calL}(j,j')} \bigg| \bigg|1-\sqrt{\frac{v^\star_{\bar{\calL}(j,j')}}{\widehat{v}_{\bar{\calL}(j,j')}}} \bigg|.
\end{align*}
We can show that, with probability at least $1-O(\bar{d}^{-10})$,
\begin{align*}
   \max_{j \in \calJ} \max_{j'\neq j} \bigg|1-\sqrt{\frac{v^\star_{\bar{\calL}(j,j')}}{\widehat{v}_{\bar{\calL}(j,j')}}} \bigg| \leq \max_{j \in \calJ} \max_{j'\neq j} \bigg|1-\frac{v^\star_{\bar{\calL}(j,j')}}{\widehat{v}_{\bar{\calL}(j,j')}} \bigg| \lesssim \frac{\kappa \mu R \sigma^\star_{\max} }{\min\{d_1, d_2(d_2-1)/2\}} \sqrt{\frac{  \bar{d}\log(\bar{d})}{\bar{p} (\sigma^\star_{\min})^2 }}
\end{align*}
by following the proofs of Theorem \ref{thm:asymptoticnormality} and Proposition \ref{prop:feasibleCLT}. As a result, with probability at least $1-O(\bar{d}^{-10})$,
 \begin{align*}
    &\max_{j \in \calJ} \max_{j'\neq j} \bigg|\frac{1}{\sqrt{v^\star_{\bar{\calL}(j,j')}}} \frac{1}{d_1}\sum_{i=1}^{d_1}  \xi_{i, \bar{\calL}(j,j')} \bigg| \bigg|1-\sqrt{\frac{v^\star_{\bar{\calL}(j,j')}}{\widehat{v}_{\bar{\calL}(j,j')}}} \bigg| \leq \max_{j \in \calJ} \max_{j'\neq j} \bigg|\frac{1}{\sqrt{v^\star_{\bar{\calL}(j,j')}}} \frac{1}{d_1}\sum_{i=1}^{d_1}  \xi_{i, \bar{\calL}(j,j')} \bigg| \bigg|1- \frac{v^\star_{\bar{\calL}(j,j')}}{\widehat{v}_{\bar{\calL}(j,j')}}  \bigg|\\
    &\lesssim  \frac{\kappa \mu R \sigma^\star_{\max} }{\min\{d_1, d_2(d_2-1)/2\}} \sqrt{\frac{  \bar{d}\log^2(\bar{d})}{\bar{p} (\sigma^\star_{\min})^2 }} \ll 1.
 \end{align*}
In addition, by following the proof of Theorem \ref{thm:asymptoticnormality}, we have, for any $\epsilon>0,$
\begin{align*}
   \max_{j \in \calJ} \max_{j'\neq j} \bigg|\frac{ 1 }{\sqrt{v^\star_{\bar{\calL}(j,j')}}}\bigg|\bigg|\left(\frac{1}{d_1} \sum_{i=1}^{d_1}\widehat{M}^{\mathrm{NR}}_{i,\bar{\calL}(j,j')}  -\frac{1}{d_1} \sum_{i=1}^{d_1}M^\star_{i,\bar{\calL}(j,j')} \right)-  \frac{1}{d_1}\sum_{i=1}^{d_1}  \xi_{i,\bar{\calL}(j,j')}   \bigg| <\epsilon
\end{align*}
with probability at least $1-O(\bar{d}^{-10})$. Then, for any $\epsilon>0$,
\begin{align*}
    \sup_{z \in \bbR} \big| \bbP(\widehat{\calT}_{\calJ}\leq z)-\bbP(\calT_{\calJ} \leq z)  \big| \leq \bbP (|\widehat{\calT}_{\calJ}-\calT_{\calJ}|>\epsilon) + \sup_{z \in \bbR} \bbP (z < \calT_{\calJ} < z+\epsilon). 
\end{align*}
\end{proof}

\begin{lemma}\label{lem:GGmarket}
    Suppose the assumptions in Theorem \ref{thm:marketbootstrap} hold.
    \begin{align*}
        \sup_{z\in \bbR}\big| \bbP(\widehat{\calG}_{\calJ}\leq z)- \bbP(\calG_{\calJ} \leq z) \big| \rightarrow 0.
    \end{align*}
\end{lemma}
\begin{proof}
By definition,
\begin{align*}
    |\widehat{\calG}_{\calJ}-\calG_{\calJ}| \leq \max_{j \in \calJ} \max_{j'\neq j} \bigg|\frac{1}{\sqrt{\widehat{v}_{\bar{\calL}(j,j')}}} \frac{1}{d_1}\sum_{i=1}^{d_1} \widehat{\xi}_{i,\bar{\calL}(j,j')} Z_i-\frac{1}{\sqrt{v^\star_{\bar{\calL}(j,j')}}} \frac{1}{d_1}\sum_{i=1}^{d_1} \xi_{i,\bar{\calL}(j,j')} Z_i \bigg|.
\end{align*}
Note that 
\begin{align*}
   \bigg|\frac{1}{\sqrt{\widehat{v}_{\bar{\calL}(j,j')}}} \frac{1}{d_1}\sum_{i=1}^{d_1} \widehat{\xi}_{i,\bar{\calL}(j,j')} Z_i-\frac{1}{\sqrt{v^\star_{\bar{\calL}(j,j')}}} \frac{1}{d_1}\sum_{i=1}^{d_1} \xi_{i,\bar{\calL}(j,j')} Z_i \bigg| =\bigg| \sum_{i=1}^{d_1}\Delta_{i,\bar{\calL}(j,j')} Z_i \bigg|
\end{align*}
    where $$\Delta_{i,\bar{\calL}(j,j')} \coloneqq \frac{1}{\sqrt{\widehat{v}_{\bar{\calL}(j,j')}}} \frac{1}{d_1}\widehat{\xi}_{i,\bar{\calL}(j,j')}-\frac{1}{\sqrt{v^\star_{\bar{\calL}(j,j')}}} \frac{1}{d_1} \xi_{i,\bar{\calL}(j,j')}.$$

Note that
\begin{align*}
    |\Delta_{i,\bar{\calL}(j,j')}| \leq \bigg|\frac{1}{\sqrt{\widehat{v}_{\bar{\calL}(j,j')}}} \frac{1}{d_1}\widehat{\xi}_{i,\bar{\calL}(j,j')}-\frac{1}{\sqrt{\widehat{v}_{\bar{\calL}(j,j')}}} \frac{1}{d_1} \xi_{i,\bar{\calL}(j,j')}\bigg| +\bigg|\frac{1}{\sqrt{\widehat{v}_{\bar{\calL}(j,j')}}} \frac{1}{d_1}\xi_{i,\bar{\calL}(j,j')}-\frac{1}{\sqrt{v^\star_{\bar{\calL}(j,j')}}} \frac{1}{d_1} \xi_{i,\bar{\calL}(j,j')}\bigg|.
\end{align*}

By following the proof of Proposition \ref{prop:feasibleCLT} analogously, we have, with probability at least $1-O(\bar{d}^{-10})$,
\begin{align*}
    \max_{j \in \calJ} \max_{j'\neq j}  \bigg|\frac{1}{\sqrt{\widehat{v}_{\bar{\calL}(j,j')}}} \frac{1}{d_1}\widehat{\xi}_{i,\bar{\calL}(j,j')}-\frac{1}{\sqrt{\widehat{v}_{\bar{\calL}(j,j')}}} \frac{1}{d_1} \xi_{i,\bar{\calL}(j,j')}\bigg|  \lesssim \frac{1}{\sqrt{\bar{p}d_1}}\frac{\kappa \mu R \sigma^\star_{\max} }{\min\{d_1, d_2(d_2-1)/2\}} \sqrt{\frac{  \bar{d}\log(\bar{d})}{\bar{p} (\sigma^\star_{\min})^2 }} .
\end{align*}

As shown in the proof of Lemma \ref{lem:TTmarket}, with probability at least $1-O(\bar{d}^{-10})$,
\begin{align*}
   \max_{j \in \calJ} \max_{j'\neq j} \bigg|1-\sqrt{\frac{v^\star_{\bar{\calL}(j,j')}}{\widehat{v}_{\bar{\calL}(j,j')}}} \bigg| \leq \max_{j \in \calJ} \max_{j'\neq j} \bigg|1-\frac{v^\star_{\bar{\calL}(j,j')}}{\widehat{v}_{\bar{\calL}(j,j')}} \bigg| \lesssim \frac{\kappa \mu R \sigma^\star_{\max} }{\min\{d_1, d_2(d_2-1)/2\}} \sqrt{\frac{  \bar{d}\log(\bar{d})}{\bar{p} (\sigma^\star_{\min})^2 }},
\end{align*}
which leads to
\begin{align*}
   \max_{j \in \calJ} \max_{j'\neq j} \bigg|\frac{1}{\sqrt{\widehat{v}_{\bar{\calL}(j,j')}}} \frac{1}{d_1}\xi_{i,\bar{\calL}(j,j')}-\frac{1}{\sqrt{v^\star_{\bar{\calL}(j,j')}}} \frac{1}{d_1} \xi_{i,\bar{\calL}(j,j')}\bigg| \lesssim  \frac{1}{\sqrt{\bar{p}d_1}}\frac{\kappa \mu R \sigma^\star_{\max} }{\min\{d_1, d_2(d_2-1)/2\}} \sqrt{\frac{  \bar{d}\log(\bar{d})}{\bar{p} (\sigma^\star_{\min})^2 }}.
\end{align*}
Therefore, with probability at least $1-O(\bar{d}^{-10})$,
\begin{align*}
   \max_{j \in \calJ} \max_{j'\neq j} |\Delta_{i,\bar{\calL}(j,j')}|\lesssim  \frac{1}{\sqrt{\bar{p}d_1}}\frac{\kappa \mu R \sigma^\star_{\max} }{\min\{d_1, d_2(d_2-1)/2\}} \sqrt{\frac{  \bar{d}\log(\bar{d})}{\bar{p} (\sigma^\star_{\min})^2 }}.
\end{align*}
 
We aim to apply the Bernstein inequality \citep{koltchinskii:2011} to $\sum_{i=1}^{d_1}\Delta_{i,\bar{\calL}(j,j')} Z_i.$ Note that $\bbE \Delta_{i,\bar{\calL}(j,j')} Z_i=0$, and
\begin{align*}
    &\norm{\Delta_{i,\bar{\calL}(j,j')} Z_i}_{\mathrm{subE}} \lesssim \frac{1}{\sqrt{\bar{p}d_1}}\frac{\kappa \mu R \sigma^\star_{\max} }{\min\{d_1, d_2(d_2-1)/2\}} \sqrt{\frac{  \bar{d}\log(\bar{d})}{\bar{p} (\sigma^\star_{\min})^2 }};\\
    &V\coloneqq \norm{\bbE \bigg[\sum_{i=1}^{d_1} \left(\Delta_{i,\bar{\calL}(j,j')}\right)^2 Z_i^2 \bigg]} \lesssim \bar{p}d_1\left(\frac{1}{\sqrt{\bar{p}d_1}}\frac{\kappa \mu R \sigma^\star_{\max} }{\min\{d_1, d_2(d_2-1)/2\}} \sqrt{\frac{  \bar{d}\log(\bar{d})}{\bar{p} (\sigma^\star_{\min})^2 }} \right)^2
\end{align*}
where $\norm{\cdot}_{\mathrm{subE}}$ is the sub-exponential norm. Therefore, with probability at least $1-O(\bar{d}^{-10})$,
\begin{align*}
    \bigg|\sum_{i=1}^{d_1}\Delta_{i,\bar{\calL}(j,j')} Z_i \bigg| &\lesssim \sqrt{V \log (\bar{d})} + \frac{1}{\sqrt{\bar{p}d_1}}\frac{\kappa \mu R \sigma^\star_{\max} }{\min\{d_1, d_2(d_2-1)/2\}} \sqrt{\frac{  \bar{d}\log(\bar{d})}{\bar{p} (\sigma^\star_{\min})^2 }} \log^2(\bar{d})\\
    & \lesssim \frac{\kappa \mu R \sigma^\star_{\max} }{\min\{d_1, d_2(d_2-1)/2\}} \sqrt{\frac{  \bar{d}\log^2(\bar{d})}{\bar{p} (\sigma^\star_{\min})^2 }}  \ll 1.
\end{align*}
Therefore, with probability at least $1-O(\bar{d}^{-9})$, for any $\epsilon>0$, $|\widehat{\calG}_{\calJ}-\calG_{\calJ}| < \epsilon.$ Then, we have, for any $\epsilon>0$,
\begin{align*}
    \sup_{z \in \bbR} \big| \bbP(\widehat{\calG}_{\calJ}\leq z)-\bbP(\calG_{\calJ} \leq z)  \big| \leq \bbP (|\widehat{\calG}_{\calJ}-\calG_{\calJ}|>\epsilon) + \sup_{z \in \bbR} \bbP (z < \calG_{\calJ} < z+\epsilon).
\end{align*}
\end{proof}

\section{Equivalence of convex and nonconvex solutions}\label{sec:sectionA}
The argument for establishing the approximate equivalence between the convex solution $\widehat{\bfL}$ and the nonconvex solution of \eqref{eq:nonconvexLS} follows an approach analogous to that in \cite{chen2020noisy}. The main differences are that we allow heterogeneous $p_i$ and assume approximate low-rankness. In addition, since we consider a generalized BTL model, the noise structure differs from that of the matrix completion model in \cite{chen2020noisy}. This section formally demonstrates that, despite these differences, the approximate equivalence between the two solution concepts still holds. Recall the nonconvex surrogate problem \eqref{eq:nonconvexLS}:
\begin{align}
   f(\bfX, \bfY)\coloneqq \frac{1}{2}  \sum_{(i,k)\in S} p_i^{-1} \left(y_{i,k}- [\bfX \bfY^\top]_{i,k} \right)^2 + \frac{\lambda}{2} \norm{\bfX}_{\mathrm{F}}^2+ \frac{\lambda}{2} \norm{\bfY}_{\mathrm{F}}^2 \label{eq:nonconvex_appen}
\end{align}
where $\bfX \in \bbR^{d_1 \times R}$ and $\bfY \in \bbR^{d_2(d_2-1)/2 \times R}$, where $\lambda>0$ is the same tuning parameter as in \eqref{eq:convexob}. Recall also that $\bfU^\star {\bf \Sigma}^\star \bfV^{\star \top}$ be the singular value decomposition of $\bfL^\star$ such that ${\bf \Sigma}^\star$ is a $R \times R$ diagonal matrix of singular values in non-ascending order. Define $\bfX^\star \coloneqq \bfU^\star ({\bf \Sigma}^\star)^{1/2}$ and $\bfY^\star = \bfV^\star ({\bf \Sigma}^\star)^{1/2}$. Then, the sequence of gradient iterates for \eqref{eq:nonconvex_appen} are defined as follows:

\noindent {\bf Initialization:} $\bfX^0=\bfX^{\star}$ and $\bfY^0=\bfY^{\star}$.

\noindent {\bf Gradient updates:} Compute, for integers $t\geq 0$,
\begin{align}
    \bfX^{t+1} &= \bfX^t - \eta \nabla_\bfX f(\bfX^t, \bfY^t) = \bfX^t - \eta \left(-  \sum_{(i,k)\in S} p_i^{-1}( y_{i,k} -  [\bfX^t\bfY^{t \top}]_{i,k})\bfe_i\bfe_k^\top \bfY^t+ \lambda \bfX^t\right)\label{eq:gradientupdates_up}\\
    \bfY^{t+1} &= \bfY^t - \eta \nabla_\bfY f(\bfX^t, \bfY^t) = \bfY^t - \eta \left(-  \sum_{(i,k)\in S} p_i^{-1}( y_{i,k} -  [\bfX^t\bfY^{t \top}]_{i,k})\bfe_k\bfe_i^\top\bfX^t + \lambda \bfY^t\right)\label{eq:gradientupdates}
\end{align}
where $\eta>0$ is the step size. Here, we are abusing notation: $\bfe_i$ and $\bfe_k$ are $d_1$ and $d_2(d_2-1)/2$ dimensional, respectively. We will keep using these notations for simplicity. We define
\begin{align}
    \bfH^t \coloneqq \argmin_{\bfO \in \calO^{R \times R}} \left( \norm{\bfX^t \bfO - \bfX^\star}_{\mathrm{F}}^2+\norm{\bfY^t \bfO - \bfY^\star}_{\mathrm{F}}^2\right) \quad \text{for each $t\geq 0$} \label{eq:Ht}
\end{align}
where $\calO^{R\times R}$ is the set of $R \times R$ orthonormal matrices.

Before proceeding, we record the following lemma, which states that the realized number of comparisons, i.e., our sample size $|S|$, remains close to its expectation. Its proof is based on the standard Chernoff bounds for independent Bernoulli random variables. Throughout this section, we assume that Assumption \ref{asp:randomness} and Assumption \ref{asp:sieve_assumptions} hold, and that $\lambda = C_{\lambda} \sqrt{ \bar{d}/\bar{p}}$ for some sufficiently large constant $C_{\lambda} > 0$.
  \begin{lemma}\label{lem:numberofedges}
   Suppose $\bar{p}\geq C \frac{\log(d_1d_2(d_2-1))}{d_1 d_2(d_2-1)}$ for some sufficiently large constant $C>0$. Then, with probability at least $1-O((\frac{d_1d_2(d_2-1)}{2})^{-100})$, we have $|S| \in [0.9\bar{p}\frac{d_1d_2(d_2-1)}{2}, 1.1\bar{p}\frac{d_1d_2(d_2-1)}{2}].$
\end{lemma}

\begin{lemma}\label{lem:boundgradient}
Suppose that, for some constant $C>0$, we have $p \geq C \frac{\log(\bar{d})}{\bar{d}}.$ Then, with probability at least $1-O( \bar{d}^{-100} )$, we have \begin{align}
 \norm{ \sum_{(i,k)\in S} (y_{i,k} - L^\star_{i,k})\bfe_i \bfe_k^\top} \lesssim \sqrt{      \bar{p} \bar{d}   } .\label{eq:boundgradientL}
	\end{align} 
\end{lemma}
\begin{proof}
By the definition of $\bfL^\star$ and the triangle inequality, we have
\begin{align*}
    \norm{\sum_{(i,k)\in S} (y_{i,k} - L^\star_{i,k})\bfe_i \bfe_k^\top} \leq \norm{\sum_{(i,k)\in S} (y_{i,k} - \sigma(\Mstar_{i,k}))\bfe_i \bfe_k^\top} + \norm{\sum_{(i,k)\in S} \varepsilon_{i,k}\bfe_i \bfe_k^\top}.
\end{align*}
 The first term on the right-hand side can be bounded by
 \begin{align*}
   \norm{\sum_{(i,k)\in S} (y_{i,k} - \sigma(\Mstar_{i,k}))\bfe_i \bfe_k^\top}\lesssim  \sqrt{\bar{p}\bar{d}},
 \end{align*}
 following a straightforward modification of Corollary 3.12 in \cite{bandeira2016sharp}. The second term is bounded by the assumption on the sieve approximation error and Lemma \ref{lem:numberofedges}. That is,
 \begin{align*}
   \norm{\sum_{(i,k)\in S} \varepsilon_{i,k}\bfe_i \bfe_k^\top}  \leq \norm{\sum_{(i,k)\in S}   \varepsilon_{i,k} \bfe_i \bfe_k^\top}_{\mathrm{F}} \lesssim  \sqrt{\bar{p}\bar{d}}.
 \end{align*}
\end{proof}
 We introduce a number of notations and conditions that are needed throughout the proof. Let $\bfL$ be a $d_1 \times d_2(d_2-1)/2$ matrix with rank $R$ and $\bfU {\bf\Sigma} \bfV^\top$ be the singular value decomposition of $\bfL$. Then the tangent space of $\bfL$, denoted by $T(\bfL)$, is defined as
\[T(\bfL) = \{\bfZ \in \mathbb{R}^{d_1 \times d_2(d_2-1)/2}  | \bfZ= \bfU \bfA^\top+\bfB \bfV^\top \,\,\text{for some}\,\, \bfA \in \mathbb{R}^{d_2(d_2-1)/2 \times R}\,\, \text{and} \,\, \bfB \in \mathbb{R}^{d_1 \times R} \}.\]
We will simply denote $T$ instead of $T(\bfL)$ when there is no risk of confusion. Let $\mathcal{P}_{T}(\cdot)$ be the orthogonal projection onto $T$, that is, 
\[\mathcal{P}_{T}(\bfZ)=\bfU \bfU^\top \bfZ +\bfZ \bfV \bfV^\top-\bfU \bfU^\top \bfZ\bfV \bfV^\top\]
for any $\bfZ \in \mathbb{R}^{d_1\times d_2(d_2-1)/2}$. Let $T^{\perp}$ be the orthogonal complement of $T$ and $\mathcal{P}_{T^{\perp}}(\cdot)$ be the projection onto $T^{\perp}$. Note that
\[\mathcal{P}_{T^{\perp}}(\bfZ)=(\bfI-\bfU \bfU^\top)\bfZ(\bfI-\bfV \bfV^\top).\]

\begin{condition}\label{cond:regularization}
The regularization parameter $\lambda$ satisfies
\begin{enumerate}
    \item[i)] $\norm{\sum_{(i,k)\in S} p_i^{-1}(y_{i,k} - L^\star_{i,k})\bfe_i \bfe_k^\top}<\frac{1}{8} \lambda $.
    \item[ii)] $\norm{\sum_{(i,k)\in S} p_i^{-1}([\bfX\bfY^\top]_{i,k} - L^\star_{i,k})\bfe_i \bfe_k^\top-  \left( \bfX \bfY^\top 
 -\bfL^\star\right)}  <\frac{1}{8}  \lambda $.
\end{enumerate}
\end{condition}

\begin{condition}\label{cond:injectivity}
	Let $T$ be the tangent space of $\bfX\bfY^\top$. For some value $c_{\mathrm{inj}}>0$, we have
	\[ \sum_{(i,k)\in S}p_{\min}^{-1} H_{i,k}^2  \geq c_{\mathrm{inj}} \norm{\bfH}^2_{\mathrm{F}}, \quad \text{for any $\bfH \in T$}. \]
\end{condition}
  
\begin{lemma}\label{LemmaA1}
	Suppose that $(\bfX,\bfY)$ satisfies 
	\begin{align}
		\norm{\nabla f(\bfX,\bfY)}_{\mathrm{F}} \leq c \frac{\sqrt{c_{\inj} p_{\min}}}{\kappa} \lambda \sqrt{\sigma^\star_{\min}} \label{LemmaA1.1}
	\end{align}
	for some sufficiently small constant $c>0$. Additionally, assume that any singular value of $\bfX$ and $\bfY$ exists in the interval $[\sqrt{\sigma^\star_{\min}/2},\sqrt{2\sigma^\star_{\max}} ]$. Then, under Conditions \ref{cond:regularization}-\ref{cond:injectivity}, $\widehat{\bfL}$, a minimizer of \eqref{eq:convexob}, satisfies
	\begin{align*}
		\norm{\bfX\bfY^\top-\widehat{\bfL}}_{\mathrm{F}} \lesssim \frac{\kappa}{c_{\mathrm{inj}}} \frac{1}{\sqrt{\sigma^\star_{\min}}} \norm{\nabla f (\bfX,\bfY)}_{\mathrm{F}}.
	\end{align*}
\end{lemma}

\begin{proof}
Denote the singular value decomposition of $\bfX\bfY^\top$ by $\bfU {\bf\Sigma} \bfV^\top$. $T$ and $T^{\perp}$ denote the tangent space of $\bfX\bfY^\top$ and its orthogonal complement, respectively. The following claim plays a crucial role in the analysis.
	\begin{claim}\label{ClaimA1}
		Suppose the assumptions in Lemma \ref{LemmaA1} hold. Then, 
		\begin{align}
       -  \sum_{(i,k)\in S} p_i^{-1}( y_{i,k} -  [\bfX\bfY^{ \top}]_{i,k})\bfe_k\bfe_i^\top  = -\lambda \bfU \bfV^\top+\mathfrak{R} \label{LemmaA1.4}
		\end{align}
		where $\mathfrak{R}$ is a residual matrix such that
		\begin{align}
\norm{\mathcal{P}_T(\mathfrak{R})}_{\mathrm{F}} \leq 72 \kappa \frac{1}{\sqrt{\sigma^\star_{\min}}} \norm{\nabla f(\bfX,\bfY)}_{\mathrm{F}} \quad \text{and} \quad \norm{\mathcal{P}_{T^{\perp}}(\mathfrak{R})}<\frac{1}{2}\lambda. \label{LemmaA1.5} \end{align}
	\end{claim}
	\begin{enumerate}
		\item In this first step, we show that the gap ${\bf\Delta} \coloneqq\widehat{\bfL}-\bfX\bfY^\top$ primarily resides in the tangent space $T$. By definition,
		\begin{align*}   
        &\frac{1}{2} \sum_{(i,k)\in S} p_i^{-1} \left(y_{i,k} - [\bfX \bfY^\top +{\bf\Delta}]_{i,k} \right)^2 + \lambda \norm{\bfX \bfY^\top +{\bf\Delta}}_* \\
        &\quad \leq \frac{1}{2} \sum_{(i,k)\in S} p_i^{-1} \left(y_{i,k} - [\bfX \bfY^\top]_{i,k} \right)^2 + \lambda \norm{\bfX \bfY^\top}_*,
		\end{align*}
	which leads to
		\begin{align*}
			\frac{1}{2} \sum_{(i,k)\in S} p_i^{-1} \Delta_{i,k}^2 \leq \sum_{(i,k)\in S}p_i^{-1}(y_{i,k}-[\bfX \bfY^\top]_{i,k})\Delta_{i,k}  +\lambda \norm{\bfX\bfY^\top}_*-\lambda \norm{\bfX\bfY^\top+{\bf\Delta}}_*.
		\end{align*}
		Using a subgradient of $\norm{\cdot}_*$ at $\bfX \bfY^\top$, we have
		\begin{align*}
			\norm{\bfX\bfY^\top+{\bf\Delta}}_* \geq \norm{\bfX\bfY^\top}_*+\langle \bfU \bfV^\top+\bfW, {\bf\Delta}\rangle
		\end{align*}
		for any $\bfW \in T^{\perp} $ such that $\norm{\bfW} \leq 1$. We can take $\bfW$ such that $\langle \bfW, {\bf\Delta} \rangle = \norm{\mathcal{P}_{T^{\perp}}({\bf\Delta})}_*$. Then we have
		\begin{align*}
			\frac{1}{2} \sum_{(i,k)\in S} p_i^{-1} \Delta_{i,k}^2 & \leq \sum_{(i,k)\in S}p_i^{-1}(y_{i,k}-[\bfX \bfY^\top]_{i,k})\Delta_{i,k} -\lambda \langle \bfU \bfV^\top, {\bf \Delta}\rangle-\lambda \norm{\calP_{T^\perp}({\bf\Delta})}_*. 
		\end{align*}
		 We invoke Claim \ref{ClaimA1} and have
		\begin{align} 
			0 \leq \frac{1}{2} \sum_{(i,k)\in S} p_i^{-1} \Delta_{i,k}^2 &\leq - \langle \mathfrak{R}, {\bf\Delta} \rangle - \lambda \norm{\mathcal{P}_{T^\perp}({\bf\Delta})}_* \nonumber \\
            & = - \langle \calP_T(\mathfrak{R}), {\bf\Delta} \rangle-\langle \calP_{T^\perp}(\mathfrak{R}), {\bf\Delta} \rangle - \lambda \norm{\mathcal{P}_{T^\perp}({\bf\Delta})}_*,
            \label{LemmaA1.6}
		\end{align}
		and therefore
		\begin{align*}
			 \langle \mathcal{P}_T(\mathfrak{R}), {\bf\Delta} \rangle+\langle \mathcal{P}_{T^{\perp}}(\mathfrak{R}), {\bf\Delta} \rangle + \lambda \norm{\mathcal{P}_T({\bf\Delta})}_* \leq 0.
		\end{align*}
		This inequality yields 
		\begin{align*}
			&-\norm{\mathcal{P}_T(\mathfrak{R})}_{\mathrm{F}} \norm{\mathcal{P}_T({\bf\Delta})}_{\mathrm{F}} - \norm{\mathcal{P}_{T^{\perp}}(\mathfrak{R})} \norm{\mathcal{P}_{T^{\perp}}({\bf\Delta})}_* +\lambda \norm{\mathcal{P}_{T^{\perp}}({\bf\Delta})}_*  \leq 0
		\end{align*}
by the properties of the Frobenius inner product and orthogonality. Once again, we invoke Claim \ref{ClaimA1} to control $\norm{\mathcal{P}_{T^{\perp}}(\mathfrak{R})}$. Then
\begin{align*}
        \norm{\mathcal{P}_T(\mathfrak{R})}_{\mathrm{F}} \norm{\mathcal{P}_T({\bf\Delta})}_{\mathrm{F}} \geq \frac{\lambda}{2} \norm{\mathcal{P}_{T^{\perp}}({\bf\Delta})}_*.
		\end{align*}
		By the assumptions $\norm{\mathcal{P}_T(\mathfrak{R})}_{\mathrm{F}} \leq 72 \kappa \frac{1}{\sqrt{\sigma^\star_{\min}}} \norm{\nabla f(\bfX,\bfY)}_{\mathrm{F}}$ and $\norm{\nabla f(\bfX,\bfY)}_{\mathrm{F}} \leq c \frac{\sqrt{c_{\inj} p_{\min}}}{\kappa} \lambda \sqrt{\sigma^\star_{\min}}$, 
		we have
		\begin{align}
			\norm{\mathcal{P}_{T^{\perp}} ({\bf\Delta})}_* \leq 144 \kappa \frac{1}{\lambda \sqrt{\sigma^\star_{\min}}} \norm{\nabla f(\bfX,\bfY)}_{\mathrm{F}} \norm{\mathcal{P}_T({\bf\Delta})}_{\mathrm{F}} \leq 144 c \sqrt{c_{\inj} p_{\min}} \norm{\mathcal{P}_T({\bf\Delta})}_{\mathrm{F}}. \label{LemmaA1.7}
		\end{align}
		By Condition \ref{cond:injectivity} and sufficiently small $c>0$, we can bound $144 c \sqrt{c_{\inj} p_{\min}}\leq 1$. This and \eqref{LemmaA1.7} result in
		\begin{align}
			\norm{\mathcal{P}_{T^{\perp}} ({\bf\Delta})}_{\mathrm{F}} \leq \norm{\mathcal{P}_{T^{\perp}} ({\bf\Delta})}_* \leq 144 c \sqrt{c_{\inj} p_{\min}} \norm{\mathcal{P}_T({\bf\Delta})}_{\mathrm{F}} \leq \norm{\mathcal{P}_T({\bf\Delta})}_{\mathrm{F}}. \label{LemmaA1.8}
		\end{align}
		
		\item In this second step, we now aim to establish an upper bound $\frac{1}{2} \sum_{(i,k)\in S}  \Delta_{i,k}^2$, which is proportional to $\norm{\nabla f(\bfX, \bfY)}_{\mathrm{F}}$. As shown in the first step, Claim \ref{ClaimA1} and \eqref{LemmaA1.6} imply
		\begin{align*}
			\frac{1}{2} \sum_{(i,k)\in S} p_i^{-1}\Delta^2_{i,k}   \leq \norm{\mathcal{P}_T(\mathfrak{R})}_{\mathrm{F}} \norm{\mathcal{P}_T({\bf\Delta})}_{\mathrm{F}}- \frac{\lambda}{2} \norm{\mathcal{P}_{T^\perp}({\bf\Delta})}_* \leq \norm{\mathcal{P}_T(\mathfrak{R})}_{\mathrm{F}} \norm{\mathcal{P}_T({\bf\Delta})}_{\mathrm{F}}.
		\end{align*}
		Applying Claim \ref{ClaimA1} once again, we reach
		\begin{align}
			\frac{1}{2} \sum_{(i,k)\in S} p_i^{-1}\Delta^2_{i,k}
			  \leq 72 \kappa \frac{1}{\sqrt{\sigma^\star_{\min}}} \norm{\nabla f(\bfX,\bfY)}_{\mathrm{F}}\norm{\bf\Delta}_{\mathrm{F}}.\label{eq:step2bound_old}
		\end{align}		
Then, \eqref{eq:step2bound} implies 
\begin{align}
    \frac{1}{2} \sum_{(i,k)\in S} \Delta^2_{i,k} \leq p_{\max } \frac{1}{2} \sum_{(i,k)\in S} p_i^{-1}\Delta^2_{i,k}
			  \leq 72 \kappa \frac{p_{\max } }{\sqrt{\sigma^\star_{\min}}} \norm{\nabla f(\bfX,\bfY)}_{\mathrm{F}}\norm{\bf\Delta}_{\mathrm{F}}. \label{eq:step2bound}
\end{align}

\item In this last step, we connect $\frac{1}{2} \sum_{(i,k)\in S}  \Delta^2_{i,k}$ and $\norm{\bf\Delta}^2_{\mathrm{F}}$ using the injectivity condition. Condition \ref{cond:injectivity} yields
		\begin{align} 
			\sqrt{\sum_{(i,k)\in S} \Delta_{i,k}^2} &= \sqrt{\sum_{(i,k)\in S} [\calP_T({\bf\Delta})+\calP_{T^\perp}({\bf\Delta})]_{i,k}^2}  \geq \sqrt{\sum_{(i,k)\in S} [\calP_T({\bf\Delta})]_{i,k}^2} -\sqrt{\sum_{(i,k)\in S} [\calP_{T^\perp}({\bf\Delta})]_{i,k}^2}\nonumber\\
			&\geq \sqrt{c_{\inj}p_{\min}} \norm{\mathcal{P}_T({\bf\Delta})}_{\mathrm{F}}-\norm{\mathcal{P}_{T^{\perp}}({\bf\Delta})}_{\mathrm{F}}.\label{eq:deltabound}
		\end{align}
		As $c$ is sufficiently small, \eqref{LemmaA1.7} implies
		\begin{align*}
			\norm{\mathcal{P}_{T^{\perp}} ({\bf\Delta})}_{\mathrm{F}} \leq \norm{\mathcal{P}_{T^{\perp}} ({\bf\Delta})}_* \leq 144 c \sqrt{c_{\inj} p_{\min}} \norm{\mathcal{P}_T({\bf\Delta})}_{\mathrm{F}} \leq \frac{\sqrt{c_{\inj}p_{\min}}}{2} \norm{\mathcal{P}_T({\bf\Delta})}_{\mathrm{F}}.
		\end{align*}
		This bound and \eqref{eq:deltabound} yield
		\[\sqrt{\sum_{(i,k)\in S} \Delta_{i,k}^2} \geq \frac{\sqrt{c_{\inj}p_{\min}}}{2} \norm{\mathcal{P}_T({\bf\Delta})}_{\mathrm{F}}.\]
		In addition, \eqref{LemmaA1.8} leads to
		\[\norm{\bf\Delta}_{\mathrm{F}} \leq \norm{\mathcal{P}_T({\bf\Delta})}_{\mathrm{F}}+\norm{\mathcal{P}_{T^{\perp}}({\bf\Delta})}_{\mathrm{F}} \leq 2 \norm{\mathcal{P}_T({\bf\Delta})}_{\mathrm{F}}.\]
		Therefore, we reach
	\begin{align}
    \sqrt{\sum_{(i,k)\in S} \Delta_{i,k}^2}\geq \frac{\sqrt{c_{\inj}p_{\min}}}{2} \norm{\mathcal{P}_T({\bf\Delta})}_{\mathrm{F}} \geq \frac{\sqrt{c_{\inj}p_{\min}}}{4} \norm{\bf\Delta}_{\mathrm{F}} \label{LemmaA1.10} \end{align}
	\end{enumerate}
	Finally, \eqref{eq:step2bound} and \eqref{LemmaA1.10} together give us
	\[\frac{c_{\inj}p_{\min}}{32} \norm{\bf\Delta}_{\mathrm{F}}^2 \leq \frac{1}{2} \sum_{(i,k)\in S} \Delta_{i,k}^2\leq 72 \kappa \frac{p_{\max}}{\sqrt{\sigma^\star_{\min}}} \norm{\nabla f(\bfX,\bfY)}_{\mathrm{F}}\norm{\bf\Delta}_{\mathrm{F}}, \]
	and therefore,
	\[\norm{{\bf\Delta}}_{\mathrm{F}} \lesssim \frac{\kappa}{c_{\inj}} \frac{1}{\sqrt{\sigma^\star_{\min}}} \norm{\nabla f(\bfX,\bfY)}_{\mathrm{F}}.\]
\end{proof}

\begin{proof}[Proof of Claim \ref{ClaimA1}]	
By definition,
	\begin{align*} \nabla f(\bfX,\bfY) =  \begin{bmatrix}
			-  \sum_{(i,k)\in S} p_i^{-1}( y_{i,k} -  [\bfX\bfY^{ \top}]_{i,k})\bfe_i\bfe_k^\top \bfY+ \lambda \bfX\\
			-  \sum_{(i,k)\in S} p_i^{-1}( y_{i,k} -  [\bfX\bfY^{ \top}]_{i,k})\bfe_k\bfe_i^\top\bfX + \lambda \bfY 
		\end{bmatrix}. \end{align*}
	By the assumption on $\norm{\nabla f(\bfX,\bfY)}_{\mathrm{F}}$, we have
	\begin{align}
		\norm{-  \sum_{(i,k)\in S} p_i^{-1}( y_{i,k} -  [\bfX\bfY^{ \top}]_{i,k})\bfe_i\bfe_k^\top \bfY+ \lambda \bfX}_{\mathrm{F}} \leq c \frac{\sqrt{c_{\inj} p_{\min}}}{\kappa} \lambda \sqrt{\sigma^\star_{\min}},  \label{LemmaA1.2} \\
		\norm{-  \sum_{(i,k)\in S} p_i^{-1}( y_{i,k} -  [\bfX\bfY^{ \top}]_{i,k})\bfe_k\bfe_i^\top\bfX + \lambda \bfY }_{\mathrm{F}} \leq c \frac{\sqrt{c_{\inj} p_{\min}}}{\kappa} \lambda \sqrt{\sigma^\star_{\min}}. \label{LemmaA1.3}
	\end{align}
    
    Note that \eqref{LemmaA1.2} and \eqref{LemmaA1.3} together imply
	\begin{align}
		-  \sum_{(i,k)\in S} p_i^{-1}( y_{i,k} -  [\bfX\bfY^{ \top}]_{i,k})\bfe_i\bfe_k^\top \bfY &= -\lambda \bfX + \bfB_1, \quad \text{and} \label{LemmaA1.11_up} \\
        -  \sum_{(i,k)\in S} p_i^{-1}( y_{i,k} -  [\bfX\bfY^{ \top}]_{i,k})\bfe_k\bfe_i^\top\bfX &= -\lambda \bfY + \bfB_2 \label{LemmaA1.11_below}
	\end{align}
	for some $\bfB_1 \in \mathbb{R}^{d_1 \times R}$, $\bfB_2 \in \mathbb{R}^{d_2(d_2-1)/2 \times R}$ such that $\norm{\bfB_l}_{\mathrm{F}} \leq \norm{\nabla f(\bfX,\bfY)}_{\mathrm{F}}$ for $l=1,2$. In \eqref{LemmaA1.4}, we will bound $\norm{\mathcal{P}_T(\mathfrak{R})}_{\mathrm{F}}$ and $\norm{\mathcal{P}_{T^\perp}(\mathfrak{R})}_{\mathrm{F}}$ in turn.
	\begin{enumerate}
		\item First, we bound $\norm{\mathcal{P}_T(\mathfrak{R})}_{\mathrm{F}}$. By the definition of the projection operator,
		\begin{align*} \norm{\mathcal{P}_T(\mathfrak{R})}_{\mathrm{F}} &= \norm{\bfU \bfU^\top\mathfrak{R}(\bfI-\bfV\bfV^\top)+\mathfrak{R}\bfV\bfV^\top}_{\mathrm{F}} \leq \norm{\bfU^\top\mathfrak{R}(\bfI-\bfV\bfV^\top)}_{\mathrm{F}}+\norm{\mathfrak{R}\bfV}_{\mathrm{F}}\\
			&\leq \norm{\bfU^\top \mathfrak{R}}_{\mathrm{F}}+\norm{\mathfrak{R}\bfV}_{\mathrm{F}}.
		\end{align*}
By Lemma \ref{ClaimA2}, we have $\bfX= \bfU {\bf\Sigma}^{1/2} \bfQ$ and $\bfY= \bfV {\bf\Sigma}^{1/2} (\bfQ^{-1})^\top$ for some invertible $R\times R$ dimensional $\bfQ$. Then, \eqref{LemmaA1.4} and \eqref{LemmaA1.11_up} yield
		\[-\lambda \bfU\bfV^\top \bfY + \mathfrak{R} \bfY=-\lambda \bfX+\bfB_1\]
		and therefore
		\begin{align}
		   \mathfrak{R}\bfV=\lambda \bfU {\bf\Sigma}^{1/2} (\bfI-\bfQ\bfQ^\top){\bf\Sigma}^{-1/2}+\bfB_1\bfQ^\top{\bf\Sigma}^{-1/2}. \label{eq:RV} 
		\end{align}
		This implies
\begin{align}
\norm{\mathfrak{R}\bfV}_{\mathrm{F}}&=\norm{\lambda \bfU {\bf\Sigma}^{1/2} (\bfI-\bfQ\bfQ^\top){\bf\Sigma}^{-1/2}}_{\mathrm{F}}+\norm{\bfB_1 \bfQ^\top {\bf\Sigma}^{-1/2}}_{\mathrm{F}} \nonumber \\
		& \leq \lambda \norm{{\bf\Sigma}^{1/2}} \norm{{\bf\Sigma}^{-1/2}}\norm{\bfQ\bfQ^\top-\bfI}_{\mathrm{F}}+ \norm{\bfQ} \norm{{\bf\Sigma}^{-1/2}}\norm{\bfB_1}_{\mathrm{F}} \nonumber \\            
            &\leq \lambda \norm{{\bf\Sigma}^{1/2}} \norm{{\bf\Sigma}^{-1/2}}\norm{ {\bf\Sigma}_\bfQ^2-\bfI}_{\mathrm{F}}+ \norm{\bfQ} \norm{{\bf\Sigma}^{-1/2}}\norm{\bfB_1}_{\mathrm{F}} \label{LemmaA1.12}
		\end{align} 
where ${\bf\Sigma}_{\bfQ}$ is defined in Lemma \ref{ClaimA2}. By the assumption on the singular values of $\bfX$ and $\bfY$,  
		\begin{align}
			\sigma_{\max}({\bf\Sigma}) = \norm{\bfX\bfY^\top} \leq \norm{\bfX}\norm{\bfY}\leq 2 \sigma^\star_{\max}, \label{LemmaA1.13_up} \\
			\sigma_{\min}({\bf\Sigma}) = \sigma_{\min}(\bfX\bfY^\top) \geq \sigma_{\min}(\bfX)\sigma_{\min}(\bfY)\geq \sigma^\star_{\min}/2. \label{LemmaA1.13}
		\end{align}
		
		Lemma \ref{ClaimA2} ensures that, as long as $c>0$ is sufficiently small, we have
		\begin{align*} 
			\norm{{\bf\Sigma}_{\bfQ} - {\bf\Sigma}_{\bfQ}^{-1}}_{\mathrm{F}} \leq 8 \sqrt{\kappa} \frac{1}{\lambda \sqrt{\sigma^\star_{\min}}} \norm{\nabla f(\bfX,\bfY)}_{\mathrm{F}} \leq 8 c \sqrt{\frac{c_{\inj}p_{\min}}{\kappa}} \ll 1.
		\end{align*}
		 It is straight forward to see that this implies $\norm{\bfQ} = \norm{{\bf\Sigma}_\bfQ}\leq 2$. Combining it with Lemma \ref{ClaimA2}, \eqref{LemmaA1.12}, \eqref{LemmaA1.13_up}, and \eqref{LemmaA1.13}, we can reach
		\begin{align*}
			\norm{\mathfrak{R}\bfV}_{\mathrm{F}}  & \leq \lambda \sqrt{2 \sigma^\star_{\max}} \sqrt{\frac{2}{\sigma^\star_{\min}}} \norm{{\bf\Sigma}_\bfQ}\norm{{\bf\Sigma}_\bfQ-{\bf\Sigma}_\bfQ^{-1}}_{\mathrm{F}}+ 2 \sqrt{\frac{2}{\sigma^\star_{\min}}} \norm{\nabla f(\bfX, \bfY)}_{\mathrm{F}} \\
			& \leq 36 \kappa \frac{1}{\sqrt{\sigma^\star_{\min}}} \norm{\nabla f(\bfX, \bfY)}_{\mathrm{F}}.
		\end{align*}
		As the same upper bound for $\norm{\bfU^\top\mathfrak{R}}_{\mathrm{F}}$ can be obtained analogously, we have
	\[\norm{\mathcal{P}_T(\mathfrak{R})}_{\mathrm{F}} \leq 72 \kappa \frac{1}{\sqrt{\sigma^\star_{\min}}} \norm{\nabla f(\bfX, \bfY)}_{\mathrm{F}}.\]
		\item We now turn to $\norm{\mathcal{P}_{T^{\perp}}(\mathfrak{R})}_{\mathrm{F}}$. We first define an alternative residual $\widetilde{\mathfrak{R}}$ for analyzing $\norm{\mathcal{P}_{T^{\perp}}(\mathfrak{R})}_{\mathrm{F}}$. We begin by rewriting \eqref{LemmaA1.11_up} and \eqref{LemmaA1.11_below} as follows:
		\begin{align*}
			& \left[\bfL^\star+ \sum_{(i,k)\in S}p_i^{-1}\left(y_{i,k} -L^\star_{i,k}\right) \bfe_i \bfe_k^\top +  \sum_{(i,k)\in S}p_i^{-1}\left(L^\star_{i,k} -[\bfX \bfY^\top]_{i,k}\right) \bfe_i \bfe_k^\top -(\bfL^\star-\bfX \bfY^\top) \right]\bfY \\
  &\quad =  \bfX\bfY^\top\bfY+\lambda \bfX - \bfB_1, \\
		& \left[\bfL^\star+ \sum_{(i,k)\in S}p_i^{-1}\left(y_{i,k} -L^\star_{i,k}\right) \bfe_i \bfe_k^\top +  \sum_{(i,k)\in S}p_i^{-1}\left(L^\star_{i,k} -[\bfX \bfY^\top]_{i,k}\right) \bfe_i \bfe_k^\top -(\bfL^\star-\bfX \bfY^\top) \right]^\top\bfX \\
  &\quad =  \bfY\bfX^\top\bfX+\lambda \bfY - \bfB_2.
		\end{align*}
		Similar to \eqref{eq:RV}, we further rewrite them using Lemma \ref{ClaimA2}.
		\begin{align}
			&\left[\bfL^\star+ \sum_{(i,k)\in S}p_i^{-1}\left(y_{i,k} -L^\star_{i,k}\right) \bfe_i \bfe_k^\top +  \sum_{(i,k)\in S}p_i^{-1}\left(L^\star_{i,k} -[\bfX \bfY^\top]_{i,k}\right) \bfe_i \bfe_k^\top -(\bfL^\star-\bfX \bfY^\top) \right]\bfV \nonumber\\
  &\quad =  \bfU {\bf\Sigma} +\lambda \bfU {\bf\Sigma}^{1/2}\bfQ \bfQ^\top{\bf\Sigma}^{-1/2}- \bfB_1 \bfQ^\top {\bf\Sigma}^{-1/2}  , \label{eq:USigma}\\
		& \left[\bfL^\star+ \sum_{(i,k)\in S}p_i^{-1}\left(y_{i,k} -L^\star_{i,k}\right) \bfe_i \bfe_k^\top +  \sum_{(i,k)\in S}p_i^{-1}\left(L^\star_{i,k} -[\bfX \bfY^\top]_{i,k}\right) \bfe_i \bfe_k^\top -(\bfL^\star-\bfX \bfY^\top) \right]^\top\bfU \nonumber \\
  &\quad =  \bfV {\bf\Sigma} +\lambda \bfV {\bf\Sigma}^{1/2}(\bfQ^{-1})^\top \bfQ^{-1}{\bf\Sigma}^{-1/2}- \bfB_2 \bfQ^{-1} {\bf\Sigma}^{-1/2}.\label{eq:VSigma}
		\end{align}
    Motivated by the above two equations, we consider the following equation:
		\begin{align}
			&\bfL^\star+ \sum_{(i,k)\in S}p_i^{-1}\left(y_{i,k} -L^\star_{i,k}\right) \bfe_i \bfe_k^\top +  \sum_{(i,k)\in S}p_i^{-1}\left(L^\star_{i,k} -[\bfX \bfY^\top]_{i,k}\right) \bfe_i \bfe_k^\top -(\bfL^\star-\bfX \bfY^\top) \nonumber \\
 & \quad = \bfU {\bf\Sigma}\bfV^\top +\lambda \bfU {\bf\Sigma}^{1/2}\bfQ \bfQ^\top{\bf\Sigma}^{-1/2}\bfV^\top+ \widetilde{\mathfrak{R}} \label{LemmaA1.14}
		\end{align}
		for some residual $\widetilde{\mathfrak{R}} \in \mathbb{R}^{d_1 \times d_2(d_2-1)/2}$. Then, \eqref{LemmaA1.4} and \eqref{LemmaA1.14} result in the following equality for $\mathcal{P}_{T^{\perp}}(\mathfrak{R})$ and $\mathcal{P}_{T^{\perp}}(\widetilde{\mathfrak{R}})$:
		\begin{align*}
			&\mathcal{P}_{T^{\perp}}(\mathfrak{R}) \\
            &= \mathcal{P}_{T^{\perp}}\left(-  \sum_{(i,k)\in S} p_i^{-1}( y_{i,k} -  [\bfX\bfY^{ \top}]_{i,k})\bfe_k\bfe_i^\top \right) \\
			&=-\mathcal{P}_{T^{\perp}}\left(\bfL^\star+ \sum_{(i,k)\in S}p_i^{-1}\left(y_{i,k} -L^\star_{i,k}\right) \bfe_i \bfe_k^\top +  \sum_{(i,k)\in S}p_i^{-1}\left(L^\star_{i,k} -[\bfX \bfY^\top]_{i,k}\right) \bfe_i \bfe_k^\top -(\bfL^\star-\bfX \bfY^\top)\right)  \\
			& = - \mathcal{P}_{T^{\perp}}(\widetilde{\mathfrak{R}}).
		\end{align*}
    Observe that \eqref{eq:USigma}, \eqref{eq:VSigma}, and \eqref{LemmaA1.14} imply
		\[\widetilde{\mathfrak{R}}\bfV= -\bfB_1 \bfQ^\top {\bf\Sigma}^{-1/2}, \quad \widetilde{\mathfrak{R}}^\top \bfU=\lambda \bfV {\bf\Sigma}^{1/2}(\bfQ^{-1})^\top \bfQ^{-1}{\bf\Sigma}^{-1/2}-\lambda \bfV {\bf\Sigma}^{-1/2}\bfQ\bfQ^\top{\bf\Sigma}^{1/2}-\bfB_2 \bfQ^{-1}{\bf\Sigma}^{-1/2}.\]
	We follow the argument for bounding	$\norm{\mathcal{P}_T(\mathfrak{R})}$ analogously and have
    \begin{align}
        \norm{\mathcal{P}_T(\widetilde{\mathfrak{R}})}\leq \norm{\mathcal{P}_T(\widetilde{\mathfrak{R}})}_{\mathrm{F}} \leq \norm{\bfU^\top\widetilde{\mathfrak{R}}}_{\mathrm{F}}+\norm{\widetilde{\mathfrak{R}} \bfV}_{\mathrm{F}} \lesssim c \sqrt{c_{\inj}p}\lambda < \lambda/4, \label{eq:PTR}
    \end{align}
		as $c$ is sufficiently small.
\item Now we provide an upper bound for $\norm{\mathcal{P}_{T^{\perp}}(\widetilde{\mathfrak{R}})}_{\mathrm{F}}$ to complete the proof. Note first that Condition \ref{cond:regularization} implies
\begin{align}
    \norm{ \sum_{(i,k)\in S}p_i^{-1}\left(y_{i,k} -L^\star_{i,k}\right) \bfe_i \bfe_k^\top +  \sum_{(i,k)\in S}p_i^{-1}\left(L^\star_{i,k} -[\bfX \bfY^\top]_{i,k}\right) \bfe_i \bfe_k^\top -(\bfL^\star-\bfX \bfY^\top)} < \lambda /4 \label{LemmaA1.16}.
\end{align}

Weyl's inequality, \eqref{LemmaA1.14}, \eqref{eq:PTR}, and \eqref{LemmaA1.16} together yield
		\begin{align*}
			&\sigma_l \left(\bfU {\bf\Sigma}\bfV^\top +\lambda \bfU {\bf\Sigma}^{1/2}\bfQ \bfQ^\top{\bf\Sigma}^{-1/2}\bfV^\top+ \calP_{T^\perp}(\widetilde{\mathfrak{R}}) \right)\\
			&\quad \leq \norm{ \sum_{(i,k)\in S}p_i^{-1}\left(y_{i,k} -L^\star_{i,k}\right) \bfe_i \bfe_k^\top +  \sum_{(i,k)\in S}p_i^{-1}\left(L^\star_{i,k} -[\bfX \bfY^\top]_{i,k}\right) \bfe_i \bfe_k^\top -(\bfL^\star-\bfX \bfY^\top)-\calP_{T}(\widetilde{\mathfrak{R}})}  \\
            & \quad \quad +  \sigma_l(\bfL^{\star})\\
			& \quad  < \frac{1}{2}\lambda, 
		\end{align*}
		for any $l \geq R+1$, where $\sigma_l(\cdot)$ denote the $l$th largest singular value of a given matrix. Additionally,  Lemma \ref{ClaimA2} reveals that
		\begin{align}
			\norm{{\bf\Sigma}^{1/2}\bfQ \bfQ^\top{\bf\Sigma}^{-1/2}-\bfI}&=\norm{{\bf\Sigma}^{1/2}(\bfQ\bfQ^\top-\bfI){\bf\Sigma}^{-1/2}} \leq \norm{{\bf\Sigma}^{1/2}}\norm{{\bf\Sigma}^{-1/2}}\norm{\bfQ\bfQ^\top-\bfI}_{\mathrm{F}} \nonumber \\
			& \leq \norm{{\bf\Sigma}^{1/2}}\norm{{\bf\Sigma}^{-1/2}}\norm{{\bf\Sigma}_\bfQ}\norm{{\bf\Sigma}_\bfQ-{\bf\Sigma}_\bfQ^{-1}}_{\mathrm{F}} \nonumber \\
			&\leq 2 \sqrt{2 \sigma^\star_{\max}} \sqrt{2/\sigma^\star_{\min}}8c\sqrt{c_{\inj}p_{\min}/\kappa} \leq \frac{1}{10}, \label{LemmaA1.17}
		\end{align}
		as $c$ is sufficiently small. Finally, \eqref{LemmaA1.17} and Weyl's inequality give
		\begin{align*}
			\sigma_l\left(\bfU {\bf\Sigma}\bfV^\top +\lambda \bfU {\bf\Sigma}^{1/2}\bfQ \bfQ^\top{\bf\Sigma}^{-1/2}\bfV^\top \right) & \geq  \sigma_{R}\left( \bfU \left({\bf\Sigma}+\lambda \bfI\right)\bfV^\top  +\lambda \bfU \left( {\bf\Sigma}^{1/2}\bfQ \bfQ^\top{\bf\Sigma}^{-1/2}-\bfI  \right)\bfV^\top \right) \\
			& \geq \sigma_R({\bf\Sigma} + \lambda \bfI) -\lambda \norm{{\bf\Sigma}^{1/2}\bfQ\bfQ^\top{\bf\Sigma}^{-1/2} -\bfI} \\
			& \geq \lambda-\frac{1}{10}\lambda   > \frac{1}{2} \lambda,
		\end{align*}
		for any $l \leq R$. Then, by the orthogonality of $\bfU {\bf\Sigma}\bfV^\top +\lambda \bfU {\bf\Sigma}^{1/2}\bfQ \bfQ^\top{\bf\Sigma}^{-1/2}\bfV^\top$ and $\mathcal{P}_{T^{\perp}}(\widetilde{\mathfrak{R}})$, we must have $\norm{\mathcal{P}_{T^{\perp}}(\widetilde{\mathfrak{R}})} < \lambda/2$. Consequently, we have $\norm{\mathcal{P}_{T^{\perp}}(\mathfrak{R})}<\lambda/2.$

\end{enumerate}
\end{proof}

\begin{lemma}\label{ClaimA2}
		Suppose the assumptions in Lemma \ref{LemmaA1} hold and let $\bfU {\bf\Sigma}\bfV^\top$ be the SVD of $\bfX\bfY^\top$. Then, there exists an invertible matrix $\bfQ \in \mathbb{R}^{R \times R}$ such that $\bfX = \bfU {\bf\Sigma}^{1/2} \bfQ$, $\bfY= \bfV {\bf\Sigma}^{1/2} (\bfQ^{-1})^\top$ and 
		\begin{align*} 
			\norm{{\bf\Sigma}_{\bfQ} - {\bf\Sigma}_{\bfQ}^{-1}}_{\mathrm{F}} \leq 8 \sqrt{\kappa} \frac{1}{\lambda \sqrt{\sigma^\star_{\min}}} \norm{\nabla f(\bfX,\bfY)}_{\mathrm{F}} \leq 8 c \sqrt{\frac{c_{\inj}p_{\min}}{\kappa}}
		\end{align*}
		where $\bfU_\bfQ {\bf\Sigma}_\bfQ \bfV_\bfQ^\top$ is the SVD of $\bfQ.$
	\end{lemma}

\begin{proof} 
Recall \eqref{LemmaA1.11_up} and \eqref{LemmaA1.11_below} and write
	\begin{align}
	&	\norm{\bfX^\top \bfX - \bfY^\top\bfY}_{\mathrm{F}} \nonumber \\
        & = \frac{1}{\lambda} \norm{\bfX^\top \left(\bfB_1+\sum_{(i,k)\in S} p_i^{-1}( y_{i,k} -  [\bfX\bfY^{ \top}]_{i,k})\bfe_i\bfe_k^\top \bfY \right) -\left(\bfB_2+\sum_{(i,k)\in S} p_i^{-1}( y_{i,k} -  [\bfX\bfY^{ \top}]_{i,k})\bfe_k\bfe_i^\top\bfX 
 \right)^\top\bfY}_{\mathrm{F}} \nonumber \\
		& =\frac{1}{\lambda} \norm{\bfX^\top \bfB_1-\bfB_2^\top \bfY}_{\mathrm{F}} \nonumber \\
		& \leq \frac{1}{\lambda} \norm{\bfX}\norm{\bfB_1}_{\mathrm{F}} + \frac{1}{\lambda} \norm{\bfY}\norm{\bfB_2}_{\mathrm{F}} \nonumber \\
		& \leq 2 \frac{1}{\lambda}\sqrt{2\sigma^\star_{\max}} \norm{\nabla f(\bfX, \bfY)}_{\mathrm{F}}. \label{LemmaA1.18}
	\end{align}
	The last line follows from the assumption on the singular values of $\bfX$ and $\bfY$. We invoke Lemma 20 in \cite{chen2020noisy} and claim that there exists an invertible $\bfQ \in \mathbb{R}^{R \times R}$ such that $\bfX=\bfU {\bf\Sigma}^{1/2} \bfQ, \bfY=\bfV {\bf\Sigma}^{1/2}(\bfQ^{-1})^\top$ and 
	\begin{align*}
		\norm{{\bf\Sigma}_\bfQ - {\bf\Sigma}_\bfQ^{-1}}_{\mathrm{F}} &\leq \frac{1}{\sigma_{\min}({\bf\Sigma})} \norm{\bfX^\top \bfX-\bfY^\top\bfY}_{\mathrm{F}},
	\end{align*}
	where ${\bf\Sigma}_\bfQ$ is from SVD of $\bfQ$. From \eqref{LemmaA1.1}, \eqref{LemmaA1.13}, and \eqref{LemmaA1.18}, we reach
	\begin{align*}
		&\frac{1}{\sigma_{\min}({\bf\Sigma})} \norm{\bfX^\top \bfX-\bfY^\top\bfY}_{\mathrm{F}} \\
        &\quad \leq \frac{2}{\sigma^\star_{\min}} \frac{2}{\lambda} \sqrt{2 \sigma^\star_{\max}}\norm{\nabla f(\bfX, \bfY)}_{\mathrm{F}}  \leq 8 \sqrt{\kappa} \frac{1}{\lambda \sqrt{\sigma^\star_{\min}}} \norm{\nabla f(\bfX, \bfY)}_{\mathrm{F}}  \leq 8c \sqrt{c_{\inj}p_{\min}/\kappa}
	\end{align*}
	as claimed.
\end{proof}

\begin{lemma}\label{LemmaA4}
	 Let $T$ denote the tangent space of $\bfX\bfY^\top$. Then, with probability at least $1-O(\bar{d}^{-100})$,
	\begin{align*}
		&\norm{\sum_{(i,k)\in S} p_i^{-1}([\bfX\bfY^\top]_{i,k} - L^\star_{i,k})\bfe_i \bfe_k^\top-  \left( \bfX \bfY^\top 
 -\bfL^\star\right)}  <\frac{1}{8}  \lambda  \quad \text{(Condition \ref{cond:regularization} (ii))} \\
		& \sum_{(i,k)\in S}p_{\min}^{-1} H_{i,k}^2  \geq c_{\mathrm{inj}} \norm{\bfH}^2_{\mathrm{F}}, \quad \text{for any $\bfH \in T$}.  \quad \text{(Condition \ref{cond:injectivity} with $c_{\inj}=1/(32\kappa)$)}
	\end{align*}
    uniformly hold for any $(\bfX, \bfY)$ satisfying
	\begin{align}
		&\max\Big\{\norm{\bfX-\bfX^\star}_{2, \infty}, \norm{\bfY-\bfY^\star}_{2, \infty}\Big\} \nonumber \\
		&\quad \leq C \kappa \sqrt{\frac{\bar{d} \log(\bar{d})}{\bar{p} (\sigma^\star_{\min})^2 }}    \max\Big\{\norm{\bfX^\star}_{2, \infty}, \norm{\bfY^\star}_{2, \infty}\Big\} \label{LemmaA4.1}
	\end{align}
	for some constant $C >0$.
\end{lemma}

\begin{proof}
	The results follow from Lemma \ref{LemmaA5} and Lemma \ref{LemmaA6}.
\end{proof}

\begin{lemma}\label{LemmaA5}
	 Let $T$ denote the tangent space of $\bfX\bfY^\top$. Then, with probability at least $1-O(\bar{d}^{-100})$,
	\begin{align*}
		\sum_{(i,k)\in S}p_{\min}^{-1} H_{i,k}^2  \geq \frac{1}{32\kappa} \norm{\bfH}^2_{\mathrm{F}}, \quad \text{for any $\bfH \in T$}.  \quad \text{(Condition \ref{cond:injectivity} with $c_{\inj}=1/(32\kappa)$)}
	\end{align*}
	holds uniformly for any $(\bfX, \bfY)$ satisfying
	\begin{align}
		\max\Big\{\norm{\bfX-\bfX^\star}_{2, \infty}, \norm{\bfY-\bfY^\star}_{2, \infty}\Big\} \leq \frac{c}{\kappa \sqrt{\bar{d}}}\norm{\bfX^\star} \label{LemmaA5.1}
	\end{align}
	where $c$ is some sufficiently small constant.
\end{lemma}


\begin{proof}[Proof of Lemma \ref{LemmaA5}]
	Any $\bfH \in T$ can be re-written as $\bfX \bfA^\top+\bfB \bfY^\top$ for some $\bfA \in \mathbb{R}^{d_2(d_2-1)/2 \times R}$ and $\bfB \in \mathbb{R}^{d_1 \times R}$. We pin down $(\bfA,\bfB)$ as follows:
	\begin{align}
		(\bfA,\bfB) \coloneqq & \argmin_{\widetilde{\bfA}, \widetilde{\bfB}} \frac{1}{2} \norm{\widetilde{\bfA}}_{\mathrm{F}}^2 + \frac{1}{2} \norm{\widetilde{\bfB}}_{\mathrm{F}}^2 \quad  \text{s.t.}\quad \bfH= \bfX \widetilde{\bfA}^\top+\widetilde{\bfB} \bfY^\top. \label{LemmaA5.2}
	\end{align}
	The first order condition of \eqref{LemmaA5.2} leads to a useful equality:
	\begin{align}
		\bfX^\top \bfB=\bfA^\top \bfY. \label{LemmaA5.3}
	\end{align}
	We plan to establish the following two inequalities in turn, which will complete the proof.
	\begin{align}
		&\norm{\bfH}_{\mathrm{F}}^2 \leq 8 \sigma^\star_{\max}(\norm{\bfA}_{\mathrm{F}}^2+\norm{\bfB}_{\mathrm{F}}^2), \label{LemmaA5.4} \\
		& \frac{1}{2} \sum_{(i,k)\in S}p_{\min}^{-1} H_{i,k}^2    \geq \frac{\sigma^\star_{\min}}{8}(\norm{\bfA}_{\mathrm{F}}^2+\norm{\bfB}_{\mathrm{F}}^2). \label{LemmaA5.5}
	\end{align}
	\begin{enumerate}
		\item We begin with \eqref{LemmaA5.4}. Note first that
         \begin{align*}
			\norm{\bfH}_{\mathrm{F}}^2 &= \norm{ \bfX \bfA^\top+\bfB \bfY^\top}_{\mathrm{F}}^2 \leq 2 \left(\norm{ \bfX \bfA^\top}_{\mathrm{F}}^2+ \norm{\bfB \bfY^\top}_{\mathrm{F}}^2 \right) \leq 2 \left(\norm{\bfX}^2\norm{\bfA}_{\mathrm{F}}^2 + \norm{\bfY}^2 \norm{\bfB}_{\mathrm{F}}^2 \right) \\
			&\leq 2 \max\Big\{\norm{\bfX}^2,  \norm{\bfY}^2\Big\} \left(\norm{\bfA}_{\mathrm{F}}^2+\norm{\bfB}_{\mathrm{F}}^2 \right). 
		\end{align*}
		The condition \eqref{LemmaA5.1} ensures that
		\begin{align*}
			\norm{\bfX} &\leq \norm{\bfX^\star} + \norm{\bfX-\bfX^\star} \leq \norm{\bfX^\star}+\norm{\bfX-\bfX^\star}_{\mathrm{F}}\\
			&\leq \norm{\bfX^\star}+\sqrt{d_1} \norm{\bfX-\bfX^\star}_{2, \infty}
			\leq \norm{\bfX^\star}+\frac{c}{\kappa}\norm{\bfX^\star}\\
			&\leq 2 \norm{\bfX^\star} \leq 2 \sqrt{\sigma^\star_{\max}}
		\end{align*}
		as $c>0$ is sufficiently small. The other case, $\norm{\bfY} \leq2 \sqrt{\sigma^\star_{\max}}$, can be shown similarly. 		
		\item We now turn to \eqref{LemmaA5.5}. To begin with, we decompose as
		\begin{align*}
		\frac{1}{2} \sum_{(i,k)\in S}p_{\min}^{-1} H_{i,k}^2 & \geq \frac{1}{2} \sum_{(i,k)\in S}p_i^{-1} H_{i,k}^2 = 	\frac{1}{2}  \sum_{(i,k)\in S}p_i^{-1}[\bfX \bfA^\top+\bfB \bfY^\top]_{i,k}^2\\
        & = \underbrace{\frac{1}{2}  \sum_{(i,k)\in S}p_i^{-1}[\bfX \bfA^\top+\bfB \bfY^\top]_{i,k}^2 -\frac{1}{2}\norm{\bfX \bfA^\top+\bfB \bfY^\top}_{\mathrm{F}}^2}_{\coloneqq a_1}  \\
			& + \underbrace{\frac{1}{2}\norm{\bfX \bfA^\top+\bfB \bfY^\top}_{\mathrm{F}}^2}_{\coloneqq a_2}.
		\end{align*}
		For $a_2$, we leverage \eqref{LemmaA5.3} and have
		\begin{align*}
			a_2 = \frac{1}{2}\norm{\bfX \bfA^\top+\bfB \bfY^\top}_{\mathrm{F}}^2 &= \frac{1}{2}\norm{\bfX \bfA^\top}_{\mathrm{F}}^2 + \frac{1}{2}\norm{\bfB \bfY^\top}_{\mathrm{F}}^2 + tr(\bfX^\top \bfB \bfY^\top \bfA) \\
            &= \frac{1}{2}\norm{\bfX \bfA^\top}_{\mathrm{F}}^2 + \frac{1}{2}\norm{\bfB \bfY^\top}_{\mathrm{F}}^2 + \norm{\bfX^\top \bfB}_{\mathrm{F}}^2\\
            & \geq \frac{1}{2}\norm{\bfX \bfA^\top}_{\mathrm{F}}^2 + \frac{1}{2}\norm{\bfB \bfY^\top}_{\mathrm{F}}^2.
		\end{align*}
        Then, use the condition \eqref{LemmaA5.1} to rewrite it in terms of the true $\bfX^\star$ and $\bfY^\star$.
		\[a_2 \geq \frac{1}{2}\norm{\bfX^\star \bfA^\top}_{\mathrm{F}}^2 + \frac{1}{2}\norm{\bfB \bfY^{\star\top}}_{\mathrm{F}}^2 -\frac{1}{100}\sigma^\star_{\min} \left(\norm{\bfA}_{\mathrm{F}}^2+\norm{\bfB}_{\mathrm{F}}^2\right).\]
Let us bound $a_1$ now. Denote ${\bf\Delta}_\bfX\coloneqq\bfX-\bfX^\star$ and ${\bf\Delta}_\bfY\coloneqq\bfY-\bfY^\star$, and write
		\begin{align*}
			a_1 =&\frac{1}{2}  \sum_{(i,k)\in S}p_i^{-1}[\bfX^\star \bfA^\top+\bfB \bfY^{\star\top}+{\bf\Delta}_\bfX \bfA^\top+\bfB {\bf\Delta}_\bfY^\top]_{i,k}^2   -\frac{1}{2}\norm{\bfX^\star \bfA^\top+\bfB \bfY^{\star\top}+{\bf\Delta}_\bfX \bfA^\top+\bfB {\bf\Delta}_\bfY^\top}_{\mathrm{F}}^2\\
			=& \underbrace{\frac{1}{2}  \sum_{(i,k)\in S}p_i^{-1} [\bfX^\star \bfA^\top+\bfB \bfY^{\star\top}]_{i,k}^2 -\frac{1}{2}\norm{\bfX^\star \bfA^\top+\bfB \bfY^{\star\top} }_{\mathrm{F}}^2}_{\coloneqq b_1} \\
			& +\underbrace{\frac{1}{2}  \sum_{(i,k)\in S}p_i^{-1} [{\bf\Delta}_\bfX \bfA^\top]_{i,k}^2 -\frac{1}{2}\norm{ {\bf\Delta}_\bfX \bfA^\top }_{\mathrm{F}}^2}_{\coloneqq b_2}
			+\underbrace{ \frac{1}{2}  \sum_{(i,k)\in S}p_i^{-1}[\bfB {\bf\Delta}_\bfY^\top]_{i,k}^2 -\frac{1}{2}\norm{ \bfB {\bf\Delta}_\bfY^\top}_{\mathrm{F}}^2}_{\coloneqq b_3} \\
			& + \underbrace{\sum_{(i,k)\in S} p_i^{-1}  [{\bf\Delta}_\bfX \bfA^\top]_{i,k} [\bfB {\bf\Delta}_\bfY^\top]_{i,k} -  \langle {\bf\Delta}_\bfX \bfA^\top, \bfB {\bf\Delta}_\bfY^\top \rangle}_{\coloneqq b_4} \\
			& +\underbrace{\sum_{(i,k)\in S} p_i^{-1} [\bfX^\star \bfA^\top+\bfB \bfY^{\star\top}]_{i,k} [{\bf\Delta}_\bfX \bfA^\top+\bfB {\bf\Delta}_\bfY^\top]_{i,k} -  \langle \bfX^\star \bfA^\top+\bfB \bfY^{\star\top}, {\bf\Delta}_\bfX \bfA^\top+\bfB {\bf\Delta}_\bfY^\top \rangle}_{\coloneqq b_5}.
		\end{align*}
		We will bound $b_1$ through $b_5$ in turn.
		\begin{enumerate}
			\item By marginally modifying Section 4.2 of \cite{candes2009exact} to allow for heterogeneous $p_i$, we can show that, with probability at least $1-O(\bar{d}^{-100})$,
			\[|b_1| \leq \frac{1}{64} \norm{\bfX^\star \bfA^\top+\bfB \bfY^{\star\top}}_{\mathrm{F}}^2 \leq \frac{1}{32}\left( \norm{\bfX^\star \bfA^\top }_{\mathrm{F}}^2+\norm{ \bfB \bfY^{\star\top}}_{\mathrm{F}}^2 \right) \]
			provided that $p_{\min}d_1d_2(d_2-1) \gg \mu R \bar{d} \log(\bar{d})$.
			\item By following the proof of Lemma 9 in \cite{zheng2016convergence} with straightforward modifications, we have with probability at least $1-O(\bar{d}^{-100})$,
			\begin{align*}
				\frac{1}{2}  \sum_{(i,k)\in S} p_i^{-1}  [{\bf\Delta}_\bfX \bfA^\top]_{i,k}^2 & \leq d_1 \norm{\bf\Delta_\bfX}^2_{2, \infty} \norm{\bfA}_{\mathrm{F}}^2 \\
				\frac{1}{2}  \sum_{(i,k)\in S} p_i^{-1} [{\bfB \bf\Delta}_\bfY^\top]_{i,k}^2 &\leq \frac{d_2(d_2-1)}{2}\norm{\bf\Delta_\bfY}^2_{2, \infty} \norm{\bfB}_{\mathrm{F}}^2
			\end{align*}
                as long as $p_{\min} \geq C \frac{\log (d_1)}{d_1}$ and $p_{\min} \geq C \frac{\log (d_2(d_2-1))}{d_2(d_2-1)}$ for some $C>0.$ Also, it is straightforward that
                \begin{align*}
                    \frac{1}{2} \norm{ {\bf\Delta}_\bfX \bfA^\top }_{\mathrm{F}}^2 &\leq \frac{1}{2} d_1 \norm{\bf\Delta_\bfX}^2_{2, \infty} \norm{\bfA}_{\mathrm{F}}^2, \\
                    \frac{1}{2} \norm{ \bfB{\bf\Delta}_\bfY^\top }_{\mathrm{F}}^2 &\leq \frac{1}{2} d_2(d_2-1) \norm{\bf\Delta_\bfY}^2_{2, \infty} \norm{\bfB}_{\mathrm{F}}^2
                \end{align*}
            Then, the condition \eqref{LemmaA5.1} ensures that
			\begin{align*}
				 |b_2|+|b_3| \leq  \frac{1}{100}\sigma^\star_{\min} \left(\norm{\bfA}_{\mathrm{F}}^2+ \norm{\bfB}_{\mathrm{F}}^2\right).
			\end{align*}
			\item By Lemma 21 in \cite{chen2020noisy} and Lemma 3.2 in \cite{keshavan2010matrix}, we have
			\begin{align*}
				|b_4| &\leq \norm{\sum_{(i,k)\in S} p_i^{-1}   \bfe_i \bfe_k^\top -{\bf 1}{\bf 1}^\top} \norm{{\bf\Delta}_\bfX}_{2,\infty}\norm{\bfA}_{\mathrm{F}}\norm{{\bf\Delta}_\bfY}_{2,\infty}\norm{\bfB}_{\mathrm{F}}  \\
                & \lesssim \sqrt{\frac{\bar{d}}{\bar{p}}}  \norm{{\bf\Delta}_\bfX}_{2,\infty}\norm{\bfA}_{\mathrm{F}}\norm{{\bf\Delta}_\bfY}_{2,\infty}\norm{\bfB}_{\mathrm{F}},
			\end{align*}
 with probability at least $1-O(\bar{d}^{-100})$. Then the condition \eqref{LemmaA5.1}, $\bar{p} \bar{d} \gg 1$, and the elementary inequality $2ab \leq a^2+b^2$ lead to 
 \begin{align*}
     |b_4| \lesssim \frac{1}{100}\sigma^\star_{\min} \left(\norm{\bfA}_{\mathrm{F}}^2+ \norm{\bfB}_{\mathrm{F}}^2\right).
 \end{align*}
			\item To establish a bound for the term $b_5$, we only consider the following term as the argument for the remaining terms are nearly identical:
            \begin{align*}
                &\sum_{(i,k) \in S} p_i^{-1} [\bfX^\star \bfA^\top]_{i,k} [{\bf\Delta}_\bfX \bfA^\top]_{i,k} -  \langle \bfX^\star \bfA^\top, {\bf\Delta}_\bfX \bfA^\top \rangle \\
                & \quad  \leq   \norm{\sum_{(i,k) \in S}  p_i^{-1/2}[\bfX^\star \bfA^\top]_{i,k}}_{\mathrm{F}}   \norm{\sum_{(i,k)\in S} p_i^{-1/2} [{\bf\Delta}_\bfX \bfA^\top]_{i,k}}_{\mathrm{F}} + \norm{\bfX^\star \bfA^\top}_{\mathrm{F}}\norm{{\bf\Delta}_\bfX \bfA^\top}_{\mathrm{F}}.
            \end{align*}
We bound the first term similarly to $|b_2|$ and $|b_3|$, and the second term similarly to $|b_1|$. We omit the details for brevity. With probability at least $1-O(\bar{d}^{-100})$,
			\[|b_5|\leq \frac{1}{100} \sigma^\star_{\min}\left(\norm{\bfA}_{\mathrm{F}}^2+\norm{\bfB}_{\mathrm{F}}^2 \right).\]
			
			\item By collecting the bounds on $b_1$ through $b_5$, we obtain, with probability at least $1-O(\bar{d}^{-100})$,
			\begin{align*}
				|a_1| \leq \sum_{l=1}^5 |b_l| \leq \frac{1}{32}\left(\norm{\bfX^\star\bfA^\top}_{\mathrm{F}}^2 +\norm{\bfB \bfY^{\star\top}}_{\mathrm{F}}^2\right)+\frac{1}{25} \sigma^\star_{\min} \left(\norm{\bfA}_{\mathrm{F}}^2 +\norm{\bfB}_{\mathrm{F}}^2\right)
			\end{align*}
		\end{enumerate}
		Finally, combining the bounds on $a_1$ and $a_2$, we can complete the proof.  
	\end{enumerate}
\end{proof}

\begin{lemma}\label{LemmaA6}
	With probability at least $1-O(\bar{d}^{-100})$, 
	\begin{gather*}
		  \norm{\sum_{(i,k)\in S} p_i^{-1}([\bfX\bfY^\top]_{i,k} - L^\star_{i,k})\bfe_i \bfe_k^\top-  \left( \bfX \bfY^\top 
 -\bfL^\star\right)}  <\frac{1}{8}  \lambda
	\end{gather*}
	holds uniformly for all $(\bfX, \bfY)$ satisfying \eqref{LemmaA4.1}.
\end{lemma}

\begin{proof}
	Using the decomposition $\bfX\bfY^\top-\bfL^\star=(\bfX-\bfX^\star)\bfY^\top+\bfX^\star(\bfY-\bfY^\star)^\top$, we can write
	\begin{align*}
		 &\norm{\sum_{(i,k)\in S} p_i^{-1}([\bfX\bfY^\top]_{i,k} - L^\star_{i,k})\bfe_i \bfe_k^\top-  \left( \bfX \bfY^\top 
 -\bfL^\star\right)} \\
 &\quad \leq  \norm{\sum_{(i,k)\in S} p_i^{-1}  [(\bfX-\bfX^\star)\bfY^\top]_{i,k} -(\bfX-\bfX^\star)\bfY^\top 
 } +  \norm{\sum_{(i,k) \in S} p_i^{-1} [\bfX^\star(\bfY-\bfY^\star)^\top]_{i,k}-\bfX^\star(\bfY-\bfY^\star)^\top}
	\end{align*}
	For the first term, applying Lemma 4.5 in \cite{chen2017memory}, Lemma 3.2 in \cite{keshavan2010matrix}, the condition \eqref{LemmaA4.1}, and $\norm{\bfY}_{2,\infty} \leq  2 \norm{\bfY^\star}_{2,\infty}$, we have
	\begin{align*}
	&	\norm{\sum_{(i,k)\in S} p_i^{-1}  [(\bfX-\bfX^\star)\bfY^\top]_{i,k} - (\bfX-\bfX^\star)\bfY^\top 
  }  \leq \norm{\sum_{(i,k)\in S} p_i^{-1} \bfe_i \bfe_k^\top -{\bf 1}{\bf 1}^\top}\norm{\bfX-\bfX^\star}_{2,\infty} \norm{\bfY}_{2,\infty} \\
		&\lesssim \sqrt{\frac{\bar{d}}{\bar{p}}}\norm{\bfX-\bfX^\star}_{2,\infty}\norm{\bfY}_{2,\infty} \ll \lambda,
	\end{align*}
    with probability at least $1-O(\bar{d}^{-100})$. The second term can be bounded analogously.
\end{proof}
 
\section{Analysis of the nonconvex gradient descent iterates}\label{sec:sectionB}
In this section, we analyze the error bounds of the nonconvex gradient descent iterates of \eqref{eq:nonconvex_appen} using leave-one-out techniques, a result that may be of independent interest. As in Section \ref{sec:sectionA}, the architecture of the analysis is the same as that of \cite{chen2020noisy}. Again, the main differences are that we allow heterogeneous $p_i$, assume approximate low-rankness, and the noise structure is different from that in \cite{chen2020noisy} as we consider the BTL model. Throughout this section, we assume that Assumption \ref{asp:randomness} and Assumption \ref{asp:sieve_assumptions} hold, and that $\lambda = C_{\lambda} \sqrt{ \bar{d}/\bar{p}}$ for some sufficiently large constant $C_{\lambda} > 0$. As the proofs of the lemmas in this section are similar to those in \cite{chen2020noisy}, we omit them for brevity. The complete proofs are available upon request. For notational simplicity, we denote
\begin{align*}
	\bfF^t \coloneqq \begin{bmatrix}
		\bfX^t\\
		\bfY^t
	\end{bmatrix} \in \mathbb{R}^{\bar{d} \times R}
	\quad \text{and} \quad 
	\bfF^\star \coloneqq \begin{bmatrix}
		\bfX^\star\\
		\bfY^\star
	\end{bmatrix} \in \mathbb{R}^{\bar{d} \times R}.
\end{align*}

Note that we have the following properties of $\bfF^\star.$ 
\begin{align}
	&\sigma_1(\bfF^\star)  =\sqrt{2 \sigma^\star_{\max}}; \qquad \sigma_R(\bfF^\star)=\sqrt{2\sigma^\star_{\min}};  \label{Prelim1} \\
	&\norm{\bfF^\star}_{2, \infty}=\max \{\norm{\bfX^\star}_{2,\infty}, \norm{\bfY^\star}_{2, \infty}\} \leq \sqrt{\frac{\mu R \sigma^\star_{\max}}{\min\{d_1,d_2(d_2-1)/2\}}}. \label{Prelim2}
\end{align}

We define the leave-one-out problems. First, define $S_{l,\cdot} \coloneqq \{(i,k)\in S | i= l\}$ for $l=1, \ldots, d_1$ and $S_{\cdot, l}\coloneqq \{(i,k)\in S | k= l-d_1\}$ for $l=d_1+1, \ldots, \bar{d}$. We then define 
     \begin{align*}
   f^{(l)}(\bfX, \bfY) \coloneqq &  \frac{1}{2}\sum_{(i,k) 
   \in S \setminus S_{l,\cdot}}   p_i^{-1} \left(y_{i,k}- [\bfX \bfY^\top]_{i,k} \right)^2 +\frac{1}{2} \sum_{k=1}^{d_2(d_2-1)/2}    \left(\sigma(\Mstar_{l,k})-[\bfX \bfY^\top]_{l,k} \right)^2  \\
   & \quad + \frac{\lambda}{2} \norm{\bfX}_{\mathrm{F}}^2 + \frac{\lambda}{2} \norm{\bfY}_{\mathrm{F}}^2
 \end{align*}  
 for $l =1, \ldots, d_1$, and 
  \begin{align*}
      f^{(l)}(\bfX, \bfY) \coloneqq &  \frac{1}{2}\sum_{(i,k)\in S \setminus S_{\cdot,l}} p_i^{-1} \left(y_{i,k}- [\bfX \bfY^\top]_{i,k}  \right)^2 +  \frac{1}{2}\sum_{i=1}^{d_1}  \left(\sigma(\Mstar_{i,l-d_1})-[\bfX \bfY^\top]_{i,l-d_1} \right)^2    \\
      & \quad + \frac{\lambda}{2} \norm{\bfX}_{\mathrm{F}}^2 + \frac{\lambda}{2} \norm{\bfY}_{\mathrm{F}}^2.
\end{align*} 
 for $l = d_1+1, \ldots, \bar{d}.$ We set the number of iterations as $t_0 = \bar{d}^{20}$. Then, the gradient descent iterates for $f^{(l)}(\cdot, \cdot)$ are constructed as follows:
 
\noindent {\bf Initialization:} $\bfX^{0,(l)}=\bfX^{\star}$ and $\bfY^{0,(l)}=\bfY^{\star}$.

\noindent {\bf Gradient updates:} for $t=0,1, \ldots, t_0-1$, compute
\begin{align}
    \bfX^{t+1,(l)} = \bfX^{t,(l)} - \eta \nabla_\bfX f^{(l)}(\bfX^{t,(l)}, \bfY^{t,(l)}); \quad     \bfY^{t+1,(l)} = \bfY^{t,(l)} - \eta \nabla_\bfY f^{(l)}(\bfX^{t,(l)}, \bfY^{t,(l)}).\label{eq:LOOgradientupdates}
\end{align}
 Denoting $\bfF^{t,(l)} \coloneqq\begin{bmatrix}
     \bfX^{t,(l)}\\
     \bfY^{t,(l)}
 \end{bmatrix},$  we can also define for each $t$ and $l$,
\begin{align*}
	\bfH^{t, (l)} &\coloneqq \argmin_{\bfO \in \mathcal{O}^{r \times r}}\norm{
		\bfF^{t, (l)}
		\bfO - 
		\bfF^{\star}}_{\mathrm{F}}; \quad  \bfR^{t, (l)} \coloneqq \argmin_{\bfO \in \mathcal{O}^{r \times r}}\norm{
		\bfF^{t, (l)}
		\bfO - 
		\bfF^{t} \bfH^{t}}_{\mathrm{F}}.
\end{align*}

\begin{lemma}\label{LemmaB1}
	With probability at least $1-O(\bar{d}^{-10})$, the gradient descent iterates defined in \eqref{eq:gradientupdates_up}, \eqref{eq:gradientupdates}, and \eqref{eq:Ht} satisfy the following:
	\begin{align}
		&\max_t \norm{\bfF^t \bfH^t-\bfF^\star}_{\mathrm{F}}   \leq C_{\mathrm{F}} \sqrt{\frac{\bar{d}}{\bar{p} (\sigma^\star_{\min})^2 }}
 \norm{\bfF^\star}_{\mathrm{F}}, \label{LemmaB1.1}
		\\ 
		&\max_t    \norm{\bfF^t \bfH^t-\bfF^\star}  \leq C_{\op} \sqrt{\frac{\bar{d}}{\bar{p} (\sigma^\star_{\min})^2 }} \norm{\bfF^\star}, \label{LemmaB1.2}
		\\
		&\max_t  \norm{\bfF^t \bfH^t-\bfF^\star}_{2, \infty}   \leq C_{\infty}  \kappa \sqrt{\frac{\bar{d} \log(\bar{d})}{\bar{p} (\sigma^\star_{\min})^2 }}  \norm{\bfF^\star}_{2, \infty}\\
 & \max_t \max_{1\leq l\leq \bar{d}}  \norm{\bfF^{t} \bfH^{t}-\bfF^{t,(l)}\bfR^{t,(l)}}_{\mathrm{F}} \leq C_{\mathrm{HR}}   \sqrt{\frac{\bar{d}\log(\bar{d})}{\bar{p} (\sigma^\star_{\min})^2 }}\norm{\bfF^\star}_{2, \infty} \\
 & \min_{0\leq t < t_0}\norm{\nabla f (\bfX^t,\bfY^t)}_{\mathrm{F}} \leq \frac{1}{\bar{d}^{5}} \lambda \sqrt{\sigma^\star_{\min}}\label{newLemma2}
	\end{align}
	where $C_{\mathrm{F}}, C_{\op}, C_{\infty}, C_{\mathrm{HR}}>0$ are some absolute constants.
\end{lemma}

\begin{lemma}\label{LemmaB11}
With probability at least $1-O(\bar{d}^{-10})$, the followings hold for all $t\leq t_0$:
	\begin{enumerate}
		\item[(i)]  
		\begin{align}
			\norm{\bfF^{t,(l)}\bfR^{t,(l)}-\bfF^\star}_{2, \infty} &\leq  (C_{\infty}\kappa+C_{\mathrm{HR}})\sqrt{\frac{\bar{d} \log(\bar{d})}{\bar{p} (\sigma^\star_{\min})^2 }}\norm{\bfF^\star}_{2, \infty}, \label{LemmaB11.1}  \\ 
			\norm{\bfF^{t,(l)}\bfR^{t,(l)}-\bfF^\star} &\leq 2C_{\op}\sqrt{\frac{ \bar{d} }{\bar{p} (\sigma^\star_{\min})^2 }}\norm{\bfX^\star}.  \label{LemmaB11.2}
		\end{align}
		\item[(ii)]  
		\begin{gather}
			\norm{\bfF^t \bfH^t-\bfF^\star} \leq \norm{\bfX^\star}, \quad \norm{\bfF^t \bfH^t-\bfF^\star}_{\mathrm{F}} \leq \norm{\bfX^\star}_{\mathrm{F}}, \quad \norm{\bfF^t \bfH^t-\bfF^\star}_{2, \infty} \leq \norm{\bfF^\star}_{2, \infty}, \label{LemmaB11.3} \\
			\norm{\bfF^t} \leq 2 \norm{\bfX^\star}, \quad \norm{\bfF^t}_{\mathrm{F}} \leq 2 \norm{\bfX^\star}_{\mathrm{F}}, \quad
			\norm{\bfF^t}_{2, \infty} \leq 2 \norm{\bfF^\star}_{2, \infty}. \label{LemmaB11.4}
		\end{gather}
		\item[(iii)]  
		\[
		\norm{\bfF^t \bfH^t -\bfF^{t,(l)}\bfH^{t,(l)}}_{\mathrm{F}} \leq 5 \kappa \norm{\bfF^t \bfH^t -\bfF^{t,(l)}\bfR^{t,(l)}}_{\mathrm{F}}.
		\]
		\item[(iv)]  \eqref{LemmaB11.3}, \eqref{LemmaB11.4} also hold for the leave-one-out gradient descent iterates $\bfF^{t,(l)}\bfH^{t,(l)}.$ Additionally, we have
		\[\sigma^\star_{\min}/2 \leq \sigma_{\min} \left( (\bfY^{t,(l)} \bfH^{t,(l)})^\top \bfY^{t,(l)} \bfH^{t,(l)}\right) \leq \sigma_{\max} \left( (\bfY^{t,(l)} \bfH^{t,(l)})^\top \bfY^{t,(l)} \bfH^{t,(l)}\right) \leq 2 \sigma^\star_{\max}.\]
	\end{enumerate}
\end{lemma}

\end{document}